\pdfoutput=1

\PassOptionsToPackage{prologue,dvipsnames}{xcolor}

\documentclass[acmsmall]{acmart}

\setcopyright{rightsretained}
\acmDOI{10.1145/1111111}
\acmYear{2025}
\copyrightyear{2025}
\acmJournal{PACMPL}
\acmVolume{1}
\acmNumber{POPL}
\acmArticle{1}
\acmMonth{1}

\usepackage[appendix=append,bibliography=common]{apxproof}
\usepackage{enumitem} 
\usepackage{ifdraft}
\usepackage[nameinlink]{cleveref}
\usepackage{xspace}
\usepackage{url}
\usepackage{varwidth}
\usepackage{galois}
\usepackage{relsize} 
\usepackage{xfp} 
\usepackage{mathtools} 
\usepackage{trimclip} 
\usepackage{mathpartir} 
\usepackage{subcaption}
\usepackage{witharrows}
\usepackage{mathbbol} 
\usepackage[T1]{fontenc} 

\usepackage{utf8-symbols}
\theoremstyle{plain} 
\newtheoremrep{theorem}{Theorem}
\newtheoremrep{lemma}[theorem]{Lemma}
\newtheoremrep{corollary}[theorem]{Corollary}
\newtheoremrep{definition}[theorem]{Definition}

\newtheorem{example}[theorem]{Example}

\newtheorem{abbreviation}[theorem]{Abbreviation}

\counterwithin*{equation}{section}

\newcommand\ie{i.e.\xspace}
\newcommand\eg{e.g.\xspace}
\newcommand\cf{cf.\xspace}

\newcommand\vs{vs.\xspace}
\newcommand\wrt{wrt.\xspace}

\newsavebox{\foobox}
\newcommand{\slantbox}[2][.5]{%
  \mbox{%
    \sbox{\foobox}{#2}%
    \hskip\wd\foobox
    \pdfsave
    \pdfsetmatrix{1 0 #1 1}%
    \llap{\usebox{\foobox}}%
    \pdfrestore
  }%
}

\providetoggle{hidetodos}
\settoggle{hidetodos}{false}

\newcommand{\many}[1]{\overline{#1}}
\newcommand{\wild}{\ensuremath{\mathunderscore}}
\newcommand{\pfun}{\rightharpoonup}
\newcommand{\ruleform}[1]{\fbox{$#1$}}
\newcommand{\highlight}[1]{\setlength{\fboxsep}{2pt}\colorbox[gray]{0.8}{\ensuremath{#1}}}

\newcommand{\fn}[2]{\ensuremath{λ#1.\ #2}}

\makeatletter
\newcommand{\superimpose}[3][\mathord]{#1{\mathpalette\superimpose@{{#2}{#3}}}}
\newcommand{\superimpose@}[2]{\superimpose@@{#1}#2}
\newcommand{\superimpose@@}[3]{%
  \ooalign{%
    \hfil$\m@th#1#2$\hfil\cr
    \hfil$\m@th#1#3$\hfil\cr
  }%
}
\makeatother
\makeatletter

\newcommand{\lessequiv}{\mathrel{\raisebox{-1pt}{$\gl@over[0.25pt]{\vartriangleleft}{\sim}$}}}
\newcommand{\impequiv}{\mathrel{\raisebox{-1.25pt}{$\gl@over[0.25pt]{\superimpose{\raisebox{1.3pt}{\clipbox{0pt 3pt 0pt 3pt}{$\mid$}}}{\rotatebox[origin=c]{90}{$\lozenge$}}}{\sim}$}}}

\newcommand{\gl@over}[3][1pt]{%
  \vcenter{\m@th\offinterlineskip\ialign{%
    \hfil$##$\hfil\cr #2\cr \noalign{\vskip#1} #3\cr
  }}%
}
\makeatother

\newcommand{\later}{{\mathop{\vcenter{\hbox{$\scriptscriptstyle\blacktriangleright$}}}}}
\newcommand{\purelater}{\mathop{\mathsf{next}}}
\newcommand{\aplater}{\mathbin{\circledast}}
\newcommand{\lBrace}{{\{\hspace{-0.22em}\mid}}
\newcommand{\rBrace}{{\mid\hspace{-0.22em}\}}}
\DeclarePairedDelimiter\idiom{\lBrace}{\rBrace}

\makeatletter
\def\arcr{\@arraycr}
\makeatother

\newcommand{\Con}{\mathsf{Con}}
\newcommand{\Var}{\mathsf{Var}}
\newcommand{\Exp}{\mathsf{Exp}}

\newcommand{\Val}{\mathsf{Val}}

\newcommand{\px}{\mathsf{x}}
\newcommand{\py}{\mathsf{y}}
\newcommand{\pz}{\mathsf{z}}
\newcommand{\pv}{\mathsf{v}}
\newcommand{\pe}{\mathsf{e}}

\newcommand{\pE}{\mathsf{E}}

\newcommand{\Lam}[2]{\mathbf{\bar{\lambdaup}} #1. #2}
\newcommand{\Let}[3]{\mathbf{let}~{#1}\nobreak=\nobreak{#2}~\mathbf{in}~{#3}}
\newcommand{\Letmany}[3]{\many{\mathbf{let}~{#1}\nobreak=\nobreak{#2}~\mathbf{in}}~{#3}}

\newcommand{\Case}[2]{\mathbf{case}~{#1}~\mathbf{of}~{#2}}
\newcommand{\Sel}[1][]{\many{K~\many{\px} \rightarrow \pe_{#1}}}
\newcommand{\SelArity}{\many{K~\many{\px}^{α_K} \rightarrow \pe}}
\newcommand{\hole}{\square}
\newcommand{\dom}{\mathop{\mathsf{dom}}}

\newcommand{\ttrue}{\mathtt{tt}}

\newcommand{\Events}{\mathbb{Ev}}
\newcommand{\Traces}{\mathbb{T}}

\newcommand{\ctrl}{\mathit{ctrl}}
\newcommand{\cont}{\mathit{cont}}
\newcommand{\init}{\mathit{init}}

\newcommand{\Values}{\mathbb{V}}
\newcommand{\Environments}{\mathbb{E}}
\newcommand{\Heaps}{\mathbb{H}}
\newcommand{\Continuations}{\mathbb{K}}
\newcommand{\EContexts}{\mathbb{EC}}

\newcommand{\Domain}{{\mathbb{D}}}
\newcommand{\tr}{{\ensuremath{\tilde{ρ}}}}
\newcommand{\tm}{{\ensuremath{\tilde{μ}}}}

\newcommand{\progressto}{\rightsquigarrow}

\newcommand{\progresstorefl} {$\progressto$\textsc{-Refl}\xspace}
\newcommand{\progresstotrans}{$\progressto$\textsc{-Trans}\xspace}
\newcommand{\progresstoext}  {$\progressto$\textsc{-Ext}\xspace}
\newcommand{\progresstomemo} {$\progressto$\textsc{-Memo}\xspace}

\newcommand{\pow}[1]{\wp(#1)}

\DeclareMathSymbol{\bbcolon}{\mathpunct}{bbold}{"3A}
\DeclareMathSymbol{\bbquestionmark}{\mathpunct}{bbold}{"3F}
\DeclareMathSymbol{\bblparen}{\mathpunct}{bbold}{"28}
\DeclareMathSymbol{\bbrparen}{\mathpunct}{bbold}{"29}

\DeclareMathOperator*{\Lub}{\bigsqcup}

\newcommand{\fix}{\mathop\mathsf{fix}}
\newcommand{\lfp}{\mathop\mathsf{lfp}}

\newcommand{\Absence}{\mathsf{Absence}}
\newcommand{\Uses}{\mathsf{Uses}}
\newcommand{\Args}{\mathsf{Args}}
\newcommand{\argcons}{\mathbin{\bbcolon}}
\newcommand{\rep}[1]{\mathsf{Rep}~#1}

\newcommand{\repU}{\mathsf{Rep}~\aU}
\newcommand{\both}{\mathbin{\&}}

\newcommand{\aA}{\mathsf{A}}
\newcommand{\aU}{\mathrlap{\mathsf{U}}\hphantom{\aA}}
\newcommand{\AbsTy}{\mathsf{AbsTy}}
\newcommand{\denot}[1]{\llbracket {#1} \rrbracket}
\newcommand{\semabs}[1]{\mathcal{A}\denot{#1}}

\newcommand{\semabsS}[1]{\mathcal{C}\denot{#1}}

\newcommand{\Addresses}{\mathsf{Addr}}
\newcommand{\pa}{\mathsf{a}}
\newcommand{\smallstep}[1][\hspace{1.5ex}]{\xhookrightarrow{#1}}
\newcommand{\States}{\mathbb{S}}
\newcommand{\STraces}{\mathbb{S}^{\infty}}
\newcommand{\pushF}{\cdot}
\newcommand{\StopF}{\mathbf{stop}}
\newcommand{\ApplyF}{\mathbf{ap}}
\newcommand{\SelF}{\mathbf{sel}}
\newcommand{\UpdateF}{\mathbf{upd}}
\newcommand{\AppIT}{\textsc{App}_1}
\newcommand{\AppET}{\textsc{App}_2}

\newcommand{\CaseIT}{\textsc{Case}_1}
\newcommand{\CaseET}{\textsc{Case}_2}
\newcommand{\LookupT}{\textsc{Look}}
\newcommand{\UpdateT}{\textsc{Upd}}

\newcommand{\LetOT}{\textsc{Let}_0}
\newcommand{\LetIT}{\textsc{Let}_1}

\newcommand{\interior}[1]{#1\,\mathsf{inter}}
\newcommand{\maxtrace}[1]{#1\,\mathsf{max}}

\settoggle{hidetodos}{true}

\makeatletter
\@ifundefined{lhs2tex.lhs2tex.sty.read}%
  {\@namedef{lhs2tex.lhs2tex.sty.read}{}%
   \newcommand\SkipToFmtEnd{}%
   \newcommand\EndFmtInput{}%
   \long\def\SkipToFmtEnd#1\EndFmtInput{}%
  }\SkipToFmtEnd

\newcommand\ReadOnlyOnce[1]{\@ifundefined{#1}{\@namedef{#1}{}}\SkipToFmtEnd}
\usepackage{amstext}
\usepackage{amssymb}
\usepackage{stmaryrd}
\DeclareFontFamily{OT1}{cmtex}{}
\DeclareFontShape{OT1}{cmtex}{m}{n}
  {<5><6><7><8>cmtex8
   <9>cmtex9
   <10><10.95><12><14.4><17.28><20.74><24.88>cmtex10}{}
\DeclareFontShape{OT1}{cmtex}{m}{it}
  {<-> ssub * cmtt/m/it}{}

\DeclareFontShape{OT1}{cmtt}{bx}{n}
  {<5><6><7><8>cmtt8
   <9>cmbtt9
   <10><10.95><12><14.4><17.28><20.74><24.88>cmbtt10}{}
\DeclareFontShape{OT1}{cmtex}{bx}{n}
  {<-> ssub * cmtt/bx/n}{}

\newcommand{\anonymous}{\kern0.06em \vbox{\hrule\@width.5em}}
\newcommand{\plus}{\mathbin{+\!\!\!+}}
\newcommand{\bind}{\mathbin{>\!\!\!>\mkern-6.7mu=}}
\newcommand{\sequ}{\mathbin{>\!\!\!>}}
\renewcommand{\leq}{\leqslant}
\renewcommand{\geq}{\geqslant}
\usepackage{polytable}

\@ifundefined{mathindent}%
  {\newdimen\mathindent\mathindent\leftmargini}%
  {}%

\def\resethooks{%
  \global\let\SaveRestoreHook\empty
  \global\let\ColumnHook\empty}
\newcommand*{\savecolumns}[1][default]%
  {\g@addto@macro\SaveRestoreHook{\savecolumns[#1]}}
\newcommand*{\restorecolumns}[1][default]%
  {\g@addto@macro\SaveRestoreHook{\restorecolumns[#1]}}
\newcommand*{\aligncolumn}[2]%
  {\g@addto@macro\ColumnHook{\column{#1}{#2}}}

\resethooks

\newcommand{\onelinecommentchars}{\quad-{}- }
\newcommand{\commentbeginchars}{\enskip\{-}
\newcommand{\commentendchars}{-\}\enskip}

\newcommand{\visiblecomments}{%
  \let\onelinecomment=\onelinecommentchars
  \let\commentbegin=\commentbeginchars
  \let\commentend=\commentendchars}

\newcommand{\invisiblecomments}{%
  \let\onelinecomment=\empty
  \let\commentbegin=\empty
  \let\commentend=\empty}

\visiblecomments

\newlength{\blanklineskip}
\newcommand{\hsindent}[1]{\quad}
\let\hspre\empty
\let\hspost\empty

\EndFmtInput
\makeatother
\ReadOnlyOnce{polycode.fmt}%
\makeatletter

\newcommand{\hsnewpar}[1]%
  {{\parskip=0pt\parindent=0pt\par\vskip #1\noindent}}

\newcommand{\hscodestyle}{}

\newcommand{\sethscode}[1]%
  {\expandafter\let\expandafter\hscode\csname #1\endcsname
   \expandafter\let\expandafter\endhscode\csname end#1\endcsname}

  {\par\noindent
   \advance\leftskip\mathindent
   \hscodestyle
   \let\\=\@normalcr
   \let\hspre\(\let\hspost\)%
   \pboxed}%
  {\endpboxed\)%
   \par\noindent
   \ignorespacesafterend}

  {\hsnewpar\abovedisplayskip
   \advance\leftskip\mathindent
   \hscodestyle
   \let\hspre\(\let\hspost\)%
   \pboxed}%
  {\endpboxed%
   \hsnewpar\belowdisplayskip
   \ignorespacesafterend}

  {\hsnewpar\abovedisplayskip
   \advance\leftskip\mathindent
   \hscodestyle
   \let\\=\@normalcr
   \(\pboxed}%
  {\endpboxed\)%
   \hsnewpar\belowdisplayskip
   \ignorespacesafterend}

\newcommand{\plainhs}{\sethscode{plainhscode}}

\plainhs

  {\hsnewpar\abovedisplayskip
   \advance\leftskip\mathindent
   \hscodestyle
   \let\\=\@normalcr
   \(\parray}%
  {\endparray\)%
   \hsnewpar\belowdisplayskip
   \ignorespacesafterend}

  {\parray}{\endparray}

  {\(\parray}{\endparray\)}

\def\codeframewidth{\arrayrulewidth}
\RequirePackage{calc}

  {\parskip=\abovedisplayskip\par\noindent
   \hscodestyle
   \arrayrulewidth=\codeframewidth
   \tabular{@{}|p{\linewidth-2\arraycolsep-2\arrayrulewidth-2pt}|@{}}%
   \hline\framedhslinecorrect\\{-1.5ex}%
   \let\endoflinesave=\\
   \let\\=\@normalcr
   \(\pboxed}%
  {\endpboxed\)%
   \framedhslinecorrect\endoflinesave{.5ex}\hline
   \endtabular
   \parskip=\belowdisplayskip\par\noindent
   \ignorespacesafterend}

\newcommand{\framedhslinecorrect}[2]%
  {#1[#2]}

  {\(\def\column##1##2{}%
   \let\>\undefined\let\<\undefined\let\\\undefined
   \newcommand\>[1][]{}\newcommand\<[1][]{}\newcommand\\[1][]{}%
   \def\fromto##1##2##3{##3}%
   }{\) }%

  {\let\orighscode=\hscode
   \let\origendhscode=\endhscode
   \def\endhscode{\def\hscode{\endgroup\def\@currenvir{hscode}\\}\begingroup}
   \orighscode\def\hscode{\endgroup\def\@currenvir{hscode}}}%
  {\origendhscode
   \global\let\hscode=\orighscode
   \global\let\endhscode=\origendhscode}%

\makeatother
\EndFmtInput
\ReadOnlyOnce{forall.fmt}%
\makeatletter

\let\HaskellResetHook\empty
\newcommand*{\AtHaskellReset}[1]{%
  \g@addto@macro\HaskellResetHook{#1}}
\newcommand*{\HaskellReset}{\HaskellResetHook}

\newcommand\hsforall{\global\let\hsdot=\hsperiodonce}
\newcommand\hsexists{\global\let\hsdot=\hsperiodonce}
\newcommand*\hsperiodonce[2]{#2\global\let\hsdot=\hscompose}
\newcommand*\hscompose[2]{#1}

\AtHaskellReset{\global\let\hsdot=\hscompose}

\HaskellReset

\makeatother
\EndFmtInput
\newcommand{\kwcolor}[1]{{\color{BlueViolet} #1}}
\newcommand{\varcolor}[1]{{\color{Sepia} #1}}
\newcommand{\concolor}[1]{{\color{OliveGreen} #1}}
\newcommand{\keyword}[1]{\kwcolor{\mathbf{#1}}}
\newcommand{\varid}[1]{\varcolor{\mathit{#1}}}
\newcommand{\conid}[1]{\concolor{\mathsf{#1}}}

\renewcommand{\commentbegin}{\ensuremath{\;\Lbag\ }}
\renewcommand{\commentend}{\ensuremath{\Rbag\;}}

\floatstyle{plaintop}
\restylefloat{table}

\begin{document}

\setlength{\pdfpageheight}{\paperheight}
\setlength{\pdfpagewidth}{\paperwidth}

\title{Abstracting Denotational Interpreters}
\subtitle{A Pattern for Sound, Compositional and Higher-order Static Program Analysis}

\author{Sebastian Graf}
\affiliation{%
  \institution{Karlsruhe Institute of Technology}
  \city{Karlsruhe}
  \country{Germany}
}
\email{sgraf1337@gmail.com}

\author{Simon Peyton Jones}
\affiliation{%
  \institution{Epic Games}
  \city{Cambridge}
  \country{UK}
}
\email{simon.peytonjones@gmail.com}

\author{Sven Keidel}
\affiliation{%
  \institution{TU Darmstadt}
  \city{Darmstadt}
  \country{Germany}
}
\email{sven.keidel@tu-darmstadt.de}

\begin{abstract}
  We explore \emph{denotational interpreters}:
  denotational semantics that produce coinductive traces of a corresponding
  small-step operational semantics.
  By parameterising our denotational interpreter over the semantic domain
  and then varying it, we recover \emph{dynamic semantics} with different
  evaluation strategies as well as \emph{summary-based static analyses} such as type
  analysis, all from the same generic interpreter.
  Among our contributions is the first denotational semantics for call-by-need
  that is provably adequate in a strong, compositional sense.
  The generated traces lend themselves well to describe \emph{operational properties}
  such as how often a variable is evaluated, and hence enable static analyses
  abstracting these operational properties.
  Since static analysis and dynamic semantics share the same generic interpreter
  definition, soundness proofs via abstract interpretation decompose into
  showing small abstraction laws about the abstract domain, thus obviating
  complicated ad-hoc preservation-style proof frameworks.
\end{abstract}

\ccsdesc[500]{Software and its engineering~Semantics}
\ccsdesc[500]{Software and its engineering~Automated static analysis}
\ccsdesc[300]{Software and its engineering~Compilers}
\ccsdesc[300]{Software and its engineering~Procedures, functions and subroutines}
\ccsdesc[100]{Software and its engineering~Functional languages}
\ccsdesc[100]{Software and its engineering~Software maintenance tools}

\keywords{Programming language semantics, Abstract Interpretation, Static Program Analysis}  

\maketitle

\section{Introduction}
\label{sec:introduction}

A \emph{static program analysis} infers facts about a program, such
as ``this program is well-typed'', ``this higher-order function is always called
with argument $\Lam{x}{x+1}$'' or ``this program never evaluates $x$''.
In a functional-language setting, such static analyses are
often defined \emph{compositionally} on the input term: the result of analysing
a term is obtained by analysing its subterms separately, and combining the results.
For example, consider the claim ``\ensuremath{(\varid{even}\;\mathrm{42})} has type \ensuremath{\conid{Bool}}''.
Type analysis separately computes \ensuremath{\varid{even}\mathbin{::}\conid{Int}\to \conid{Bool}} and \ensuremath{\mathrm{42}\mathbin{::}\conid{Int}}, and then
combines these results to to produce the result type \ensuremath{\varid{even}\;\mathrm{42}\mathbin{::}\conid{Bool}},
without looking at the definition of \ensuremath{\varid{even}} again.

In order to prove the analysis sound, it is helpful to pick a language
semantics that is also compositional, such as a \emph{denotational
semantics}~\citep{ScottStrachey:71}; then the semantics and the analysis ``line
up'' and the soundness proof is relatively straightforward.
Indeed, one can often break up the proof into manageable subgoals by regarding
the analysis as an \emph{abstract interpretation} of the compositional
semantics~\citep{Cousot:21}.

Alas, traditional denotational semantics does not model operational details --
and yet those details might be the whole point of the analysis.
For example, we might want to ask ``How often does $\pe$ evaluate its free
variable $x$?'', but a standard denotational semantics simply does not express
the concept of ``evaluating a variable''.
So we are typically driven to use an \emph{operational
semantics}~\citep{Plotkin:81}, which directly models operational details like
the stack and heap, and sees program execution as a sequence of machine states.
Now we have two unappealing alternatives:
\begin{itemize}
\item Develop a difficult, ad-hoc soundness proof, one that links a
  non-compositional operational semantics with a compositional analysis.
\item Reimagine and reimplement the analysis as an abstraction of the
  reachable states of an operational semantics.
  This is the essence of the \emph{Abstracting Abstract Machines} (AAM)
  \cite{aam} recipe.
  A very fruitful framework, but one that follows the \emph{call strings}
  approach~\citep{SharirPnueli:78}, reanalysing function bodies at call sites.
  Hence the new analysis becomes non-modular, leading to scalability problems
  for a compiler.
\end{itemize}

In this paper, we resolve the tension by exploring \emph{denotational
interpreters}: total, mathematical objects that live at the intersection of
structurally-defined \emph{definitional interpreters}~\citep{Reynolds:72} and
denotational semantics.
Our denotational interpreters generate small-step traces embellished with
arbitrary operational detail and enjoy a straightforward encoding in typical
higher-order programming languages.
Static analyses arise as instantiations of the same generic interpreter,
enabling succinct, shared and modular soundness proofs just like for AAM or
big-step definitional interpreters~\citep{adi,Keidel:18}.
However, the shared, compositional structure enables a wide range of summary
mechanisms in static analyses that we think are beyond the reach of
non-compositional reachable-states abstractions like AAM.

We make the following contributions:
\begin{itemize}
\item
  We use a concrete example (absence analysis) to argue for
  the usefulness of compositional, summary-based analysis in \Cref{sec:problem}
  and we demonstrate the difficulty of conducting an ad-hoc soundness proof
  \wrt a non-compositional small-step operational semantics.
\item \Cref{sec:interp} walks through the definition of our generic
  denotational interpreter and its type class algebra in Haskell.
  We demonstrate the ease with which different instances of our interpreter
  endow our object language with call-by-name, call-by-need and call-by-value
  evaluation strategies, each producing (abstractions of) small-step
  abstract machine traces.
\item A concrete instantiation of a denotational interpreter is \emph{total}
  if it coinductively yields a (possibly infinite) trace for every input
  program, including ones that diverge.
  \Cref{sec:totality} proves that the by-name and by-need instantiations are
  total by embedding the generic interpreter and its instances in Guarded Cubical
  Agda.
\item \Cref{sec:adequacy} proves that the by-need instantiation of our
  denotational interpreter adequately generates an abstraction of a trace
  in the lazy Krivine machine~\citep{Sestoft:97}, preserving its length as well
  as arbitrary operational information about each transition taken.
\item In \Cref{sec:abstraction} we instantiate the generic interpreter with
  finite, abstract semantic domains.
  In doing so, we recover summary-based usage analysis, a generalisation
  of absence analysis in \Cref{sec:problem}, as well as \citeauthor{Milner:78}'s
  type analysis.
  The Appendix contains further examples, such as 0CFA control-flow analysis and
  Demand Analysis of the Glasgow Haskell Compiler.
\item In \Cref{sec:soundness}, we apply abstract interpretation to characterise
  a set of abstraction laws that the type class instances of an abstract
  domain must satisfy in order to soundly approximate by-name and by-need
  interpretation.
  None of the proof obligations mention the generic interpreter, and, more
  remarkably, none of the laws mention the concrete semantics or the Galois
  connection either!
  This enables us to prove usage analysis sound \wrt the by-name
  and by-need semantics in half a page, building on reusable
  semantics-specific theorems.
\item
  We compare to the enormous body of related approaches in \Cref{sec:related-work}.
\end{itemize}

\section{The Problem We Solve}
\label{sec:problem}

What is so difficult about proving a compositional analysis sound
\wrt a non-compositional small-step operational semantics?
We will demonstrate the challenges in this section, by way of a simplified \emph{absence
analysis}~\citep{SPJ:94}, a higher-order form of neededness analysis to inform
removal of dead code in a compiler.

\subsection{Object Language}
\label{sec:lang}

To set the stage, we start by defining the object language of this work, an
untyped lambda calculus with \emph{\textbf{recursive}} let bindings and
algebraic data types:
\[
\arraycolsep=3pt
\begin{array}{rrclcrrclcl}
  \text{Variables}    & \px, \py & ∈ & \Var        &     & \quad \text{Constructors} &        K & ∈ & \Con        &     & \text{with arity $α_K ∈ ℕ$} \\
  \text{Values}       &      \pv & ∈ & \Val        & ::= & \highlight{\Lam{\px}{\pe}} \mid K~\many{\px}^{α_K} \\
  \text{Expressions}  &      \pe & ∈ & \Exp        & ::= & \multicolumn{6}{l}{\highlight{\px \mid \pv \mid \pe~\px \mid \Let{\px}{\pe_1}{\pe_2}} \mid \Case{\pe}{\SelArity}}
\end{array}
\]
This language is very similar to that of \citet{Launchbury:93} and \citet{Sestoft:97}.
It is factored into \emph{A-normal form}, that is, the arguments of applications
are restricted to be variables, so the difference between lazy and eager
semantics is manifest in the semantics of $\mathbf{let}$.
Note that $\Lam{x}{x}$ (with an overline) denotes syntax, whereas $\fn{x}{x+1}$
denotes an anonymous mathematical function.
In this section, only the highlighted parts are relevant and $\mathbf{let}$ is
considered non-recursive, but the interpreter definition in \Cref{sec:interp}
supports data types and recursive $\mathbf{let}$ as well.
Throughout the paper we assume that all bound program variables are distinct.

\subsection{Absence Analysis}
\label{sec:absence}

\begin{figure}
  \[\ruleform{ \semabs{\wild}_{\wild} \colon \Exp → (\Var \pfun \AbsTy) → \AbsTy }\]
  \\[-0.5em]
  \begin{minipage}[t]{0.47\textwidth}
  \arraycolsep=0pt
  \abovedisplayskip=0pt
  \[\begin{array}{rcl}
    \semabs{\px}_ρ & {}={} & ρ(\px) \\
    \semabs{\Lam{\px}{\pe}}_ρ & {}={} & \mathit{fun}_{\px}( \fn{θ}{\semabs{\pe}_{ρ[\px ↦ θ]}}) \\
    \semabs{\pe~\px}_ρ & {}={} & \mathit{app}(\semabs{\pe}_{ρ})(ρ(\px)) \\
    \semabs{\Let{\px}{\pe_1}{\pe_2}}_ρ & {}={} & \semabs{\pe_2}_{ρ[\px ↦ \px \both \semabs{\pe_1}_ρ]} \\
    \\[-0.8em]
    \multicolumn{3}{c}{\mathit{fun}_{\px}( f) {}={} \langle φ[\px↦\aA], φ(\px) \argcons π \rangle} \\
    \multicolumn{3}{c}{\qquad\qquad\text{where } \langle φ, π \rangle = f(\langle [\px↦\aU], \repU \rangle)} \\
    \multicolumn{3}{c}{\mathit{app}(\langle φ_f, a \argcons π \rangle)(\langle φ_a, \wild \rangle) = \langle φ_f ⊔ (a * φ_a), π \rangle} \\
  \end{array}\]
  \end{minipage}%
  \hfill
  \begin{minipage}[t]{0.50\textwidth}
  \arraycolsep=0pt
  \abovedisplayskip=0pt
  \[\begin{array}{c}
  \begin{array}{rclcl}
    a & {}∈{} & \Absence & {}::={} & \aA \mid \aU \\
    φ & {}∈{} & \Uses    & {}={} & \Var \to \Absence \\
    π & {}∈{} & \Args    & {}::={} & a \argcons π \mid \rep{a} \\
    θ & {}∈{} & \AbsTy   & {}::={} & \langle φ, π \rangle \\
    \\[-0.9em]
    \multicolumn{5}{c}{\rep{a} \equiv a \argcons \rep{a}} \\
  \end{array} \\
  \\[-0.9em]
  \begin{array}{l}
    \aA * φ = [] \quad
    \aU * φ = φ  \\
    \px \both \langle φ, π \rangle = \langle φ[\px↦\aU], π \rangle
  \end{array}
  \\[-0.5em]
  \end{array}\]
  \end{minipage}%
  \caption{Absence analysis}
  \label{fig:absence}
\end{figure}

In order to define and explore absence analysis in this subsection, we must
clarify what absence means, semantically.
A variable $\px$ is \emph{absent} in an expression $\pe$ when
$\pe$ never evaluates $\px$, regardless of the context in which $\pe$
appears.
Otherwise, the variable $\px$ is \emph{used} in $\pe$.

\Cref{fig:absence} defines an absence analysis $\semabs{\pe}_ρ$ for lazy
program semantics that conservatively approximates semantic absence.
For illustrative purposes, our analysis definition only works for
the special case of non-recursive $\mathbf{let}$, but later sections will assume
recursive let semantics.%
\footnote{Given an order that we will define in due course, the
generalised definition for recursive as well as non-recursive let is
$\semabs{\Let{\px}{\pe_1}{\pe_2}}_ρ = \semabs{\pe_2}_{ρ[\px ↦
\lfp(\fn{θ}{\px \both \semabs{\pe_1}_{ρ[\px↦θ]}})]}$.}
It takes an environment $ρ \in \Var \pfun \AbsTy$ containing absence
information about the free variables of $\pe$ and returns
an \emph{absence type} $\langle φ, π \rangle \in \AbsTy$; an abstract
representation of $\pe$.
The first component $φ \in \Uses$ of the absence type captures how $\pe$ uses its free
variables by associating an $\Absence$ flag with each variable.
When $φ(\px) = \aA$, then $\px$ is absent in $\pe$; otherwise, $φ(\px) = \aU$
and $\px$ might be used in $\pe$.
The second component $π \in \Args$ of the absence type describes how $\pe$ uses
actual arguments supplied at application sites.
For example, function $f \triangleq \Lam{x}{y}$ has absence type $\langle [y ↦ \aU], \aA \argcons \repU \rangle$.
Mapping $[y ↦ \aU]$ indicates that $f$ may use its free variable $y$.
The literal notation $[y ↦ \aU]$ maps any variable other than $y$ to $\aA$.
Furthermore, $\aA \argcons \repU$ indicates that $f$'s first argument is absent and all further arguments are potentially used.
The element $\repU$ denotes an infinite repetition of $\aU$, as expressed by the
non-syntactic equality $\repU \equiv \aU \argcons \repU$.

We illustrate the analysis at the example expression
$\pe \triangleq \Let{k}{\Lam{y}{\Lam{z}{y}}}{k~x_1~x_2}$, where the initial
environment for $\pe$, $ρ_\pe(\px) \triangleq \langle [\px ↦ \aU], \repU \rangle$,
declares the free variables of $\pe$ with a pessimistic argument description $\repU$.
\begin{DispWithArrows}[fleqn,mathindent=0em]
      & \semabs{\Let{k}{\Lam{y}{\Lam{z}{y}}}{k~x_1~x_2}}_{ρ_{\pe}} \label{eq:abs-ex-let}
        \Arrow{Unfold $\semabs{\Let{\px}{\pe_1}{\pe_2}}$. NB: Lazy Let!} \\
  ={} & \semabs{k~x_1~x_2}_{ρ_\pe[k↦k \both \semabs{\Lam{y}{\Lam{z}{y}}}_{ρ_\pe}]}
        \Arrow{Unf. $\semabs{\wild}$, $ρ_1 \triangleq ρ_\pe[k↦k \! \both \! \semabs{\Lam{y}{\Lam{z}{y}}}_{ρ_\pe}]$} \\
  ={} & \mathit{app}(\mathit{app}(ρ_1(k))(ρ_1(x_1)))(ρ_1(x_2))
        \Arrow{Unfold $ρ_1(k)$} \\
  ={} & \mathit{app}(\mathit{app}(k \both \semabs{\Lam{y}{\Lam{z}{y}}}_{ρ_1})(ρ_1(x_1)))(ρ_1(x_2))
        \Arrow{Unfold $\semabs{\Lam{\px}{\pe}}$ twice, $\semabs{\px}$} \\
  ={} & \mathit{app}(\mathit{app}(k \both \mathit{fun}_{y}(\fn{θ_y}{\mathit{fun}_{z}(\fn{θ_z}{θ_y})}))(...))(...) \label{eq:abs-ex-summarise}
        \Arrow{Unfold $\mathit{fun}$ twice, simplify} \\
  ={} & \mathit{app}(\mathit{app}(\langle [k ↦ \aU], \highlight{\aU} \argcons \aA \argcons \repU \rangle)(\highlight{ρ_1(x_1)}))(...) \label{eq:abs-apply1}
        \Arrow{Unfold $\mathit{app}$, $ρ_1(x_1)=ρ_{\pe}(x_1)$, simplify} \\
  ={} & \mathit{app}(\langle [k ↦ \aU,x_1↦\aU], \highlight{\aA} \argcons \repU \rangle)(\highlight{ρ_1(x_2)}) \label{eq:abs-apply2}
        \Arrow{Unfold $\mathit{app}$, simplify} \\
  ={} & \langle [k ↦ \aU,x_1↦\aU], \repU \rangle
\end{DispWithArrows}
Let us look at the steps in a bit more detail.
Step \labelcref{eq:abs-ex-let} extends the environment with
an absence type for the let right-hand side of $k$.
The steps up until \labelcref{eq:abs-ex-summarise} successively expose
applications of the $\mathit{app}$ and $\mathit{fun}$ helper functions applied
to environment entries for the involved variables.
Step \labelcref{eq:abs-ex-summarise} then computes the absence type
$\mathit{fun}_y(\fn{θ_y}{\mathit{fun}_z(\fn{θ_z}{θ_y})}) = \langle [], \aU \argcons \aA \argcons \repU \rangle$.
The $\Uses$ component is empty because $\Lam{y}{\Lam{z}{y}}$ has no free variables,
and $k \both ...$ will add $[k↦\aU]$ as the single use.
The $\mathit{app}$ steps \labelcref{eq:abs-apply1,eq:abs-apply2} simply zip up
the uses of arguments $ρ_1(x_1)$ and $ρ_1(x_2)$ with the $\Absence$ flags
in $\aU \argcons \aA \argcons \repU$ as highlighted, adding the
$\Uses$ from $ρ_1(x_1) = \langle [x_1 ↦ \aU], \repU \rangle$ but \emph{not}
from $ρ_1(x_2)$, because the first actual argument ($x_1$) is used whereas the
second ($x_2$) is absent.
The join on $\Uses$ follows pointwise from the order $\aA ⊏ \aU$, \ie $(φ_1
⊔ φ_2)(\px) \triangleq φ_1(\px) ⊔ φ_2(\px)$.

The analysis result $[k ↦ \aU,x_1↦\aU]$ infers $k$ and $x_1$ as
potentially used and $x_2$ as absent, despite $x_2$ occurring in argument position,
thanks to the bookkeeping of $\Args$.



\subsection{Compositionality, Summaries and Modularity}
\label{sec:summaries}

The absence type $\langle [], \aU \argcons \aA \argcons \repU \rangle$
above is a finite \emph{summary} of the lambda term $\Lam{y}{\Lam{z}{y}}$.
Let us define what we mean by ``summary'' in order to understand what is so
appealing about a summary-based analysis such as $\semabs{\wild}$.

Just as a denotational semantics, the interpreter $\semabs{\wild}$
\emph{denotates} a term in a \emph{semantic domain} ($\AbsTy$).
This interpretation is \emph{compositional}:
the denotation of a term is a recombination of the denotations of its subterms.

In order for a denotational semantics to faithfully and compositionally
denotate lambda terms, its semantic domain must contain infinite elements.
However, every element of the semantic domain $\AbsTy$ of absence analysis is
\emph{finitely representable} (data, even!), so the denotation of lambda terms
must be approximate in some sense.
We call such a finitely representable and thus approximate denotation
a \emph{summary}.

The approximate nature of summaries is best appreciated when analysing
beta redexes such as $\semabs{(\Lam{\px}{\pe})~\py}$, which invokes the
\emph{summary mechanism}.
In the definition of $\semabs{\wild}$, we took care to explicate this mechanism
via the adjoint functions $\mathit{fun}$ and $\mathit{app}$.
For the summary mechanism to be sound, we must have
$\semabs{\pe[\py / \px]} ⊑ \semabs{(\Lam{\px}{\pe})~\py}$
(where $\pe[\py / \px]$ substitutes $\py$ for $\px$ in $\pe$).

To support efficient separate compilation, a compiler analysis must be
\emph{modular}, and summaries are \emph{indispensable} in achieving that.
Let us say that our example function $k = (\Lam{y}{\Lam{z}{y}})$ is defined in
module A and there is a use site $(k~x_1~x_2)$ in module B.
Then a \emph{modular analysis} must not reanalyse A.$k$ at its use site in B.
Our analysis $\semabs{\wild}$ facilitates that easily, because it can
serialise the $\AbsTy$ summary for $k$ into module A's signature file.

The same way summaries enable efficient \emph{inter}-module compilation,
they enable efficient \emph{intra}-module compilation:
Compositionality implies that $\semabs{\Let{f}{\Lam{x}{\pe_{\mathit{big}}}}{f~f~f}}$
is a function of $\semabs{\Lam{x}{\pe_{\mathit{big}}}}$, and finite summaries
prevent having to reanalyse $\pe_{\mathit{big}}$ repeatedly for each call of $f$.

This is why summary-based analyses are great: they scale.

\subsection{Summaries \vs Abstracting Abstract Machines}

Now, instead of coming up with a summary mechanism, we could simply have
``inlined'' $k$ during analysis of the example above to see that $x_2$ is absent
in a simple first-order sense.
The \emph{call strings} approach to interprocedural program
analysis~\citep{SharirPnueli:78} turns this idea into a static analysis,
and the AAM recipe could be used to derive an absence analysis based on call
strings that is sound by construction.
Alas, following this paths gives up on modularity, because a call-strings-based
analysis would need to invoke the function
$(\fn{θ_y}{\fn{θ_z}{θ_y}}) : \AbsTy \to \AbsTy \to \AbsTy$ that describes
$k$'s inline expansion \emph{at every use site}, leading to scalability problems in a compiler.

\subsection{Problem: Proving Soundness of Summary-Based Analyses}

In this subsection, we demonstrate the difficulty of proving soundness for summary-based
analyses. For absence analysis, we have proved (in the Appendix) the following
correctness statement:

\begin{theoremrep}[$\semabs{\wild}$ infers absence]
  \label{thm:absence-correct}
  If $\semabs{\pe}_{ρ_\pe} = \langle φ, π \rangle$ and $φ(\px) = \aA$,
  then $\px$ is absent in $\pe$.
\end{theoremrep}
\begin{proof}
  See \hyperlink{proof:absence-correct}{the proof at the end of this section}.
\end{proof}

As the first step, we must define precisely what absence (used in the theorem statement)
\emph{means}.
One plausible definition is in terms of the standard operational semantics in
\Cref{sec:op-sem}:

\begin{definitionrep}[Absence]
  \label{defn:absence}
  A variable $\px$ is \emph{used} in an expression $\pe$
  if and only if there exists a trace
  $(\Let{\px}{\pe'}{\pe},ρ,μ,κ) \smallstep^* ... \smallstep[\LookupT(\px)] ...$
  that looks up the heap entry of $\px$, \ie it evaluates $\px$.
  Otherwise, $\px$ is \emph{absent} in $\pe$.
\end{definitionrep}
Absence of $\px$ means that $\px$ is not looked up \emph{regardless of
the context in which $\pe$ is used}, to justify rewrites via contextual
improvement~\citep{MoranSands:99}.
Furthermore, we must prove that the summary mechanism approximates beta
reduction, captured syntactically in the following \emph{substitution
lemma}~\citep{tapl}:%

\begin{toappendix}
Note that for the proofs we assume the recursive let definition
\[
  \semabs{\Let{\px}{\pe_1}{\pe_2}}_ρ = \semabs{\pe_2}_{ρ[\px ↦ \lfp(\fn{θ}{\px \both \semabs{\pe_1}_{ρ[\px↦θ]}})]}.
\]
The partial order on $\AbsTy$ necessary for computing the least fixpoint $\lfp$
follows structurally from $\aA ⊏ \aU$ (\ie product order, pointwise order).

\begin{abbreviation}
  The syntax $θ.φ$ for an $\AbsTy$ $θ = \langle φ, π \rangle$
  returns the $φ$ component of $θ$.
  The syntax $θ.π$
  returns the $π$ component of $θ$.
\end{abbreviation}

\begin{definition}[Abstract substitution]
  \label{defn:abs-subst}
  We call $φ[\px \Mapsto φ'] \triangleq φ[\px↦\aA] ⊔ (φ(\px) * φ')$ the
  \emph{abstract substitution} operation on $\Uses$
  and overload this notation for $\AbsTy$, so that
  $(\langle φ, π \rangle)[\px \Mapsto φ_\py] \triangleq \langle φ[\px \Mapsto φ_\py], π \rangle$.
\end{definition}

Abstract substitution is useful to give a concise description of the effect of
syntactic substitution:
\begin{lemma}
  \label{thm:abs-syn-subst}
  $\semabs{(\Lam{\px}{\pe})~\py}_ρ = (\semabs{\pe}_{ρ[\px↦\langle [\px↦\aU], \repU \rangle]})[\px \Mapsto ρ(\py).φ]$.
\end{lemma}
\begin{proof}
Follows by unfolding the application and lambda case and then refolding abstract substitution.
\end{proof}

\begin{lemma}
\label{thm:lambda-bound-absent}
Lambda-bound uses do not escape their scope. That is, when $\px$ is lambda-bound in $\pe$, it is
\[
  (\semabs{\pe}_ρ).φ(\px) = \aA.
\]
\end{lemma}
\begin{proof}
By induction on $\pe$. In the lambda case, any use of $\px$ is cleared to $\aA$
when returning.
\end{proof}

\begin{lemma}
\label{thm:lambda-commutes-absence}
$\semabs{(\Lam{\px}{\Lam{\py}{\pe}})~\pz}_ρ = \semabs{\Lam{\py}{((\Lam{\px}{\pe})~\pz)}}_ρ$.
\end{lemma}
\begin{proof}
\begin{DispWithArrows*}[fleqn,mathindent=0em]
      & \semabs{(\Lam{\px}{\Lam{\py}{\pe}})~\pz}_ρ
      \Arrow{Unfold $\semabs{\wild}$, \Cref{thm:abs-syn-subst}} \\
  ={} & (\mathit{fun}_\py(\fn{θ_\py}{\semabs{\pe}_{ρ[\px↦\langle [\px↦\aU], \repU \rangle,\py↦θ_\py]}}))[\px \Mapsto ρ(\pz).φ]
      \Arrow{$ρ(\pz)(\py) = \aA$ by \Cref{thm:lambda-bound-absent}, $\px \not= \py \not= \pz$} \\
  ={} & \mathit{fun}_\py(\fn{θ_\py}{(\semabs{\pe}_{ρ[\px↦\langle [\px↦\aU], \repU \rangle,\py↦θ_\py]})[\px \Mapsto ρ(\pz).φ]})
      \Arrow{Refold $\semabs{\wild}$} \\
  ={} & \semabs{\Lam{\py}{((\Lam{\px}{\pe})~\pz)}}_ρ
\end{DispWithArrows*}
\end{proof}

\begin{lemma}
\label{thm:push-app-absence}
$\semabs{(\Lam{\px}{\pe})~\py~\pz}_ρ = \semabs{(\Lam{\px}{\pe~\pz})~\py}_ρ$.
\end{lemma}
\begin{proof}
\begin{DispWithArrows*}[fleqn,mathindent=0em]
      & \semabs{(\Lam{\px}{\pe})~\py~\pz}_ρ
      \Arrow{Unfold $\semabs{\wild}$, \Cref{thm:abs-syn-subst}} \\
  ={} & \mathit{app}((\semabs{\pe}_{ρ[\langle [\px↦\aU], \repU \rangle]})[\px \Mapsto ρ(\py).φ])(ρ(\pz))
      \Arrow{$ρ(\pz)(\px) = \aA$ by \Cref{thm:lambda-bound-absent}, $\py \not= \px \not= \pz$} \\
  ={} & \mathit{app}(\semabs{\pe}_{ρ[\langle [\px↦\aU], \repU \rangle]})(ρ(\pz))[\px \Mapsto ρ(\py).φ]
      \Arrow{Refold $\semabs{\wild}$} \\
  ={} & \semabs{(\Lam{\px}{\pe~\pz})~\py}_ρ
\end{DispWithArrows*}
\end{proof}

\begin{lemma}
\label{thm:push-let-absence}
$\semabs{\Let{\pz}{(\Lam{\px}{\pe_1})~\py}{(\Lam{\px}{\pe_2})~\py}}_ρ = \semabs{(\Lam{\px}{\Let{\pz}{\pe_1}{\pe_2}})~\py}_ρ$.
\end{lemma}
\begin{proof}
The key of this lemma is that it is equivalent to postpone the abstract
substitution from the let RHS $\pe_1$ to the let body $\pe_2$.
This can easily be proved by induction on $\pe_2$, which we omit here, but
indicate the respective step below as ``hand-waving''.
Note that we assume the (more general) recursive let semantics as defined at the
begin of this section.

\begin{DispWithArrows*}[fleqn,mathindent=1em]
      & \semabs{\Let{\pz}{(\Lam{\px}{\pe_1})~\py}{(\Lam{\px}{\pe_2})~\py}}_ρ
      \Arrow{Unfold $\semabs{\wild}$} \\
  ={} & \semabs{(\Lam{\px}{\pe_2})~\py}_{ρ[\pz↦\lfp(\fn{θ}{\pz \both \semabs{(\Lam{\px}{\pe_1})~\py}_{ρ[\pz ↦ θ]}})]}
      \Arrow{\Cref{thm:abs-syn-subst}} \\
  ={} & (\semabs{\pe_2}_{ρ[\px↦\langle [\px ↦ \aU], \repU \rangle,\pz↦\lfp(\fn{θ}{\pz \both (\semabs{\pe_1}_{ρ[\px↦\langle [\px ↦ \aU], \repU \rangle, \pz ↦ θ]})[\px \Mapsto ρ(\py).φ]})]})[\px \Mapsto ρ(\py).φ]
      \Arrow{Hand-waving above} \\
  ={} & (\semabs{\pe_2}_{ρ[\px↦\langle [\px ↦ \aU], \repU \rangle,\pz↦\lfp(\fn{θ}{\pz \both \semabs{\pe_1}_{ρ[\px↦\langle [\px ↦ \aU], \repU \rangle, \pz ↦ θ]}})]})[\px \Mapsto ρ(\py).φ]
      \Arrow{Refold $\semabs{\wild}$} \\
  ={} & (\semabs{\Let{\pz}{\pe_1}{\pe_2}}_{ρ[\px↦\langle [\px ↦ \aU], \repU \rangle]})[\px \Mapsto ρ(\py).φ]
      \Arrow{\Cref{thm:abs-syn-subst}} \\
  ={} & \semabs{(\Lam{\px}{\Let{\pz}{\pe_1}{\pe_2}})~\py}_ρ
\end{DispWithArrows*}
\end{proof}
\end{toappendix}

\begin{lemmarep}[Substitution]
\label{thm:absence-subst}
$\semabs{\pe}_{ρ[\px↦ρ(\py)]} ⊑ \semabs{(\Lam{\px}{\pe})~\py}_ρ$.
\end{lemmarep}
\begin{proof}
By induction on $\pe$.
\begin{itemize}
  \item \textbf{Case }$\pz$:
    When $\px \not= \pz$, then $\pz$ is bound outside the lambda and can't
    possibly use $\px$, so $ρ(\pz).φ(\px) = \aA$.
    We have
    \begin{DispWithArrows*}[fleqn,mathindent=4em]
        & \semabs{\pz}_{ρ[\px↦ρ(\py)]}
        \Arrow{$\px \not= \pz$} \\
    ={} & ρ(\pz)
        \Arrow{Refold $\semabs{\wild}$} \\
    ={} & \semabs{\pz}_{ρ[\px↦\langle [\px ↦ \aU], \repU \rangle]}
        \Arrow{$ρ(\pz).φ(\px) = \aA$} \\
    ={} & (\semabs{\pz}_{ρ[\px↦\langle [\px ↦ \aU], \repU \rangle]})[\px \Mapsto ρ(\py).φ]
        \Arrow{\Cref{thm:abs-syn-subst}} \\
    ={} & \semabs{(\Lam{\px}{\pz})~\py}_ρ
    \end{DispWithArrows*}
    Otherwise, we have $\px = \pz$,
    thus $ρ(\px) = \langle [\px ↦ \aU], π = \repU \rangle$,
    and thus
    \begin{DispWithArrows*}[fleqn,mathindent=4em]
        & \semabs{\pz}_{ρ[\px↦ρ(\py)]}
        \Arrow{$\px = \pz$} \\
    ={} & ρ(\py)
        \Arrow{$π ⊑ \repU$} \\
    ⊑{} & \langle ρ(\py).φ, \repU \rangle
        \Arrow{Definition of $\wild[\wild\Mapsto\wild]$} \\
    ={} & (\langle [\px ↦ \aU], \repU \rangle)[\px ↦ ρ(\py).φ]
        \Arrow{Refold $\semabs{\wild}$} \\
    ={} & (\semabs{\pz}_{ρ[\px↦\langle [\px ↦ \aU], \repU \rangle]})[\px \Mapsto ρ(\py).φ]
        \Arrow{\Cref{thm:abs-syn-subst}} \\
    ={} & \semabs{(\Lam{\px}{\pz})~\py}_ρ
    \end{DispWithArrows*}

  \item \textbf{Case }$\Lam{\pz}{\pe'}$:
    \begin{DispWithArrows*}[fleqn,mathindent=4em]
        & \semabs{\Lam{\pz}{\pe'}}_{ρ[\px↦ρ(\py)]}
        \Arrow{Unfold $\semabs{\wild}$} \\
    ={} & \mathit{fun}_\pz(\fn{θ_\pz}{\semabs{\pe'}_{ρ[\pz↦θ_\pz, \px↦ρ(\py)]}})
        \Arrow{Induction hypothesis, monotonicity} \\
    ⊑{} & \mathit{fun}_\pz(\fn{θ_\pz}{\semabs{(\Lam{\px}{\pe'})~\py}_{ρ[\pz↦θ_\pz]}})
        \Arrow{Refold $\semabs{\wild}$} \\
    ={} & \semabs{\Lam{\pz}{((\Lam{\px}{\pe'})~\py)}}_ρ
        \Arrow{\Cref{thm:lambda-commutes-absence}} \\
    ={} & \semabs{(\Lam{\px}{\Lam{\pz}{\pe'}})~\py}_ρ
    \end{DispWithArrows*}

  \item \textbf{Case }$\pe'~\pz$:
    When $\px = \pz$:
    \begin{DispWithArrows*}[fleqn,mathindent=4em]
        & \semabs{\pe'~\pz}_{ρ[\px↦ρ(\py)]}
        \Arrow{Unfold $\semabs{\wild}$} \\
    ={} & \mathit{app}(\semabs{\pe'}_{ρ[\px↦ρ(\py)]})(ρ(\py))
        \Arrow{Induction hypothesis, monotonicity} \\
    ⊑{} & \mathit{app}(\semabs{(\Lam{\px}{\pe'})~\py}_ρ)(ρ(\py))
        \Arrow{Refold $\semabs{\wild}$} \\
    ={} & \semabs{(\Lam{\px}{\pe'})~\py~\py}_ρ
        \Arrow{\Cref{thm:push-app-absence}} \\
    ={} & \semabs{(\Lam{\px}{\pe'~\py})~\py}_ρ
        \Arrow{Compositionality in $(\Lam{\px}{\pe'~\hole})~\py$} \\
    ={} & \semabs{(\Lam{\px}{\pe'~\px})~\py}_ρ
        \Arrow{$\px = \pz$} \\
    ={} & \semabs{(\Lam{\px}{\pe'~\pz})~\py}_ρ
    \end{DispWithArrows*}
    When $\px \not= \pz$:
    \begin{DispWithArrows*}[fleqn,mathindent=4em]
        & \semabs{\pe'~\pz}_{ρ[\px↦ρ(\py)]}
        \Arrow{Unfold $\semabs{\wild}$} \\
    ={} & \mathit{app}(\semabs{\pe'}_{ρ[\px↦ρ(\py)]})(ρ(\pz))
        \Arrow{Induction hypothesis, monotonicity} \\
    ⊑{} & \mathit{app}(\semabs{(\Lam{\px}{\pe'})~\py}_ρ)(ρ(\pz))
        \Arrow{Refold $\semabs{\wild}$} \\
    ={} & \semabs{(\Lam{\px}{\pe'})~\py~\pz}_ρ
        \Arrow{\Cref{thm:push-app-absence}} \\
    ={} & \semabs{(\Lam{\px}{\pe'~\pz})~\py}_ρ
    \end{DispWithArrows*}

  \item \textbf{Case }$\Let{\pz}{\pe_1}{\pe_2}$:
    \begin{DispWithArrows*}[fleqn,mathindent=4em]
        & \semabs{\Let{\pz}{\pe_1}{\pe_2}}_{ρ[\px↦ρ(\py)]}
        \Arrow{Unfold $\semabs{\wild}$} \\
    ={} & \semabs{\pe_2}_{ρ[\px↦ρ(\py),\pz↦\lfp(\fn{θ}{\pz \both \semabs{\pe_1}_{ρ[\px↦ρ(\py),\pz ↦ θ]}})]}
        \Arrow{Induction hypothesis, monotonicity} \\
    ⊑{} & \semabs{(\Lam{\px}{\pe_2})~\py}_{ρ[\pz↦\lfp(\fn{θ}{\pz \both \semabs{(\Lam{\px}{\pe_1})~\py}_{ρ[\pz ↦ θ]}})]}
        \Arrow{Refold $\semabs{\wild}$} \\
    ={} & \semabs{\Let{\pz}{(\Lam{\px}{\pe_1})~\py}{(\Lam{\px}{\pe_2})~\py}}_ρ
        \Arrow{\Cref{thm:push-let-absence}} \\
    ={} & \semabs{(\Lam{\px}{\Let{\pz}{\pe_1}{\pe_2}})~\py}_ρ
    \end{DispWithArrows*}
\end{itemize}
\end{proof}

\begin{toappendix}
Whenever there exists $ρ$ such that $ρ(\px).φ \not⊑ (\semabs{\pe}_ρ).φ$
(recall that $θ.φ$ selects the $\Uses$ in the first field of the pair $θ$),
then also $ρ_\pe(\px).φ \not⊑ \semabs{\pe}_{ρ_\pe}$.
The following lemma captures this intuition:

\begin{lemma}[Diagonal factoring]
\label{thm:diag-fact}
Let $ρ$ and $ρ_Δ$ be two environments such that
$\forall \px.\ ρ(\px).π = ρ_Δ(\px).π$.

If $ρ_Δ.φ(\px) ⊑ ρ_Δ.φ(\py)$ if and only if $\px = \py$,
then every instantiation of $\semabs{\pe}$ factors through $\semabs{\pe}_{ρ_Δ}$,
that is,
\[
  \semabs{\pe}_ρ = (\semabs{\pe}_{ρ_Δ})[\many{\px \Mapsto ρ(\px).φ}]
\]
\end{lemma}
\begin{proof}
By induction on $\pe$.
\begin{itemize}
  \item \textbf{Case $\pe = \py$}:
    We assert $\semabs{\py}_ρ = ρ(\py) = ρ_Δ(\py)[\py \Mapsto ρ(\py).φ]$ by simple unfolding.

  \item \textbf{Case $\pe = \pe'~\py$}:
    \begin{DispWithArrows*}[fleqn,mathindent=3em]
          & \semabs{\pe'~\py}_ρ
          \Arrow{Unfold $\semabs{\wild}$} \\
      ={} & \mathit{app}(\semabs{\pe'}_ρ,ρ(\py))
          \Arrow{Induction hypothesis, variable case} \\
      ={} & \mathit{app}((\semabs{\pe'}_{ρ_Δ})[\many{\px \Mapsto ρ(\px).φ}], ρ_Δ(\py)[\many{\px \Mapsto ρ(\px).φ}]).
          \Arrow{$⊔$ and $*$ commute with $\wild[\wild\Mapsto\wild]$} \\
      ={} & \mathit{app}(\semabs{\pe'}_{ρ_Δ},ρ_Δ(\py))[\many{\px \Mapsto ρ(\px).φ}]
          \Arrow{Refold $\semabs{\wild}$} \\
      ={} & (\semabs{\pe'~\py}_{ρ_Δ})[\many{\px \Mapsto ρ(\px).φ}]
    \end{DispWithArrows*}

  \item \textbf{Case $\pe = \Lam{\py}{\pe'}$}:
    Note that $\px \not= \py$ because $\py$ is not free in $\pe$.
    \begin{DispWithArrows*}[fleqn,mathindent=3em]
          & \semabs{\Lam{\py}{\pe'}}_ρ
          \Arrow{Unfold $\semabs{\wild}$} \\
      ={} & \mathit{lam}_\py(\fn{θ}{\semabs{\pe'}_{ρ[\py↦θ]}})
          \Arrow{Property of $\mathit{lam}_\py$} \\
      ={} & \mathit{lam}_\py(\fn{θ}{(\semabs{\pe'}_{{ρ}[\py↦\langle [\py ↦ \aU], \repU \rangle]})})
        \Arrow{Induction hypothesis} \\
      ={} & \mathit{lam}_\py(\fn{θ}{(\semabs{\pe'}_{{ρ_Δ}[\py↦\langle [\py ↦ \aU], \repU \rangle]})[\many{\px \Mapsto ρ(\px).φ}, \py \Mapsto [\py ↦ \aU]]})
          \Arrow{$θ[\py \Mapsto [\py ↦ \aU]] = θ$} \\
      ={} & \mathit{lam}_\py(\fn{θ}{(\semabs{\pe'}_{{ρ_Δ}[\py↦\langle [\py ↦ \aU], \repU \rangle]})[\many{\px \Mapsto ρ(\px).φ}]})
          \Arrow{$θ[\py \Mapsto [\py ↦ \aU]] = θ$} \\
      ={} & \mathit{lam}_\py(\fn{θ}{(\semabs{\pe'}_{{ρ_Δ}[\py↦θ]})[\many{\px \Mapsto ρ(\px).φ}]})
          \Arrow{Property of $\mathit{lam}_\py$} \\
      ={} & \mathit{lam}_\py(\fn{θ}{\semabs{\pe'}_{{ρ_Δ}[\py↦θ]}})[\many{\px \Mapsto ρ(\px).φ}]
          \Arrow{Refold $\semabs{\wild}$} \\
      ={} & (\semabs{\Lam{\py}{\pe'}}_{ρ_Δ})[\many{\px \Mapsto ρ(\px).φ}]
    \end{DispWithArrows*}

  \item \textbf{Case }$\Let{\py}{\pe_1}{\pe_2}$:
    Note that $\px \not= \py$ because $\py$ is not free in $\pe$.
    \begin{DispWithArrows*}[fleqn,mathindent=4em]
        & \semabs{\Let{\py}{\pe_1}{\pe_2}}_ρ
        \Arrow{Unfold $\semabs{\wild}$} \\
    ={} & \semabs{\pe_2}_{ρ[\py↦\lfp(\fn{θ}{\py \both \semabs{\pe_1}_{ρ[\py ↦ θ]}})]}
        \Arrow{Induction hypothesis} \\
    ={} & \semabs{\pe_2}_{ρ[\py↦\lfp(\fn{θ}{\py \both (\semabs{\pe_1}_{{ρ_Δ}[\py ↦ \langle [\py ↦ \aU], θ.π \rangle]})[\many{\px \Mapsto ρ(\px).φ}, \py \Mapsto θ.φ]})]}
        \Arrow{Again, backwards} \\
    ={} & \semabs{\pe_2}_{ρ[\py↦\lfp(\fn{θ}{\py \both (\semabs{\pe_1}_{{ρ_Δ}[\py ↦ θ]})[\many{\px \Mapsto ρ(\px).φ}]})]} \\
        & \text{\emph{Similarly for $\pe_2$, hand-waving to push out the subst as in \Cref{thm:push-let-absence}}} \\
    ={} & (\semabs{\pe_2}_{ρ_Δ[\py↦\lfp(\fn{θ}{\py \both \semabs{\pe_1}_{{ρ_Δ}[\py ↦ θ]}})]})[\many{\px \Mapsto ρ(\px).φ}]
        \Arrow{Refold $\semabs{\wild}$} \\
    ={} & (\semabs{\Let{\py}{\pe_1}{\pe_2}}_{ρ_Δ})[\many{\px \Mapsto ρ(\px).φ}]
    \end{DispWithArrows*}
\end{itemize}
\end{proof}

For the purposes of the preservation proof, we will write $\tr$ with a tilde to
denote that abstract environment of type $\Var \to \AbsTy$, to disambiguate it
from a concrete environment $ρ$ from the LK machine.

\begin{figure}
\arraycolsep=0pt
\[\begin{array}{rcl}
  \multicolumn{3}{c}{ \ruleform{ \semabsS{\wild} \colon \States → \AbsTy } } \\
  \\[-0.5em]
  \semabsS{(\pe,ρ,μ,κ)} & {}={} & \mathit{apps}_μ(κ,\semabs{\pe}_{α(μ) \circ ρ}) \\
                   α(μ) & {}={} & \lfp(\fn{\tm}{[ \pa ↦ \px \both \semabs{\pe'}_{\tm \circ ρ'} \mid μ(\pa) = (\px,ρ',\pe') ]}) \\
             \mathit{apps}_μ(\StopF,θ) & {}={} & θ \\
             \mathit{apps}_μ(\ApplyF(\pa) \pushF κ,θ) & {}={} & \mathit{apps}_μ(κ,\mathit{app}(θ,α(μ)(\pa))) \\
             \mathit{apps}_μ(\UpdateF(\pa) \pushF κ,θ) & {}={} & \mathit{apps}_μ(κ,θ)
  \\[-0.5em]
\end{array}\]
\caption{Absence analysis extended to small-step configurations}
  \label{fig:absence-ext}
\end{figure}

In \Cref{fig:absence-ext}, we give the extension of $\semabsS{\wild}$ to whole
machine configurations $σ$.
Although $\semabsS{\wild}$ looks like an entirely new definition, it is
actually derivative of $\semabs{\wild}$ via a context lemma à la
\citet[Lemma 3.2]{MoranSands:99}:
The environments $ρ$ simply govern the transition from syntax to
operational representation in the heap.
The bindings in the heap are to be treated as mutually recursive let bindings,
hence a fixpoint is needed.
For safety properties such as absence, a least fixpoint is appropriate.
Apply frames on the stack correspond to the application case of $\semabs{\wild}$
and invoke the summary mechanism.
Update frames are ignored because our analysis is not heap-sensitive.

Now we can prove that $\semabsS{\wild}$ is preserved/improves during reduction:

\begin{lemma}[Preservation of $\semabsS{\wild}$]
\label{thm:preserve-absent}
If $σ_1 \smallstep σ_2$, then $\semabsS{σ_1} ⊒ \semabsS{σ_2}$.
\end{lemma}
\begin{proof}
By cases on the transition.
\begin{itemize}
  \item \textbf{Case }$\LetIT$: Then $\pe = \Let{\py}{\pe_1}{\pe_2}$ and
    \[
      (\Let{\py}{\pe_1}{\pe_2},ρ,μ,κ) \smallstep (\pe_2,ρ[\py↦\pa],μ[\pa↦(\py,ρ[\py↦\pa],\pe_1)],κ).
    \]
    Abbreviating $ρ_1 \triangleq ρ[\py↦\pa], μ_1 \triangleq μ[\pa ↦ (\py,ρ_1,\pe_1)$, we have
    \begin{DispWithArrows*}[fleqn,mathindent=3em]
           & \semabsS{σ_1} \Arrow{Unfold $\semabsS{σ_1}$} \\
      {}={}& \mathit{apps}_μ(κ)(\semabs{\Let{\py}{\pe_1}{\pe_2}}_{α(μ) \circ ρ}) \Arrow{Unfold $\semabs{\Let{\py}{\pe_1}{\pe_2}}$} \\
      {}={}& \mathit{apps}_μ(κ)(\semabs{\pe_2}_{(α(μ) \circ ρ)[\py↦\py \both \lfp(\fn{θ}{\semabs{\pe_1}_{(α(μ) \circ ρ)[\py↦θ]}})]}) \Arrow{Move fixpoint outwards, refold $α$} \\
      {}={}& \mathit{apps}_{μ_1}(κ)(\semabs{\pe_2}_{α(μ_1) \circ ρ_1}) \Arrow{Refold $\semabsS{σ_2}$} \\
      {}={}& \semabsS{σ_2}
    \end{DispWithArrows*}

  \item \textbf{Case }$\AppIT$:
    Then $(\pe'~\py,ρ,μ,κ) \smallstep (\pe',ρ,μ,\ApplyF(ρ(\py)) \pushF κ)$.
    \begin{DispWithArrows*}[fleqn,mathindent=3em]
           & \semabsS{σ_1} \Arrow{Unfold $\semabsS{σ_1}$} \\
      {}={}& \mathit{apps}_μ(κ)(\semabs{\pe'~\py}_{α(μ) \circ ρ}) \Arrow{Unfold $\semabs{\pe'~\py}_{(α(μ) \circ ρ)}$} \\
      {}={}& \mathit{apps}_μ(κ)(\mathit{app}(\semabs{\pe'}_{α(μ) \circ ρ}, α(μ)(ρ(\py)))) \Arrow{Rearrange} \\
      {}={}& \mathit{apps}_μ(\ApplyF(ρ(\py)) \pushF κ)(\semabs{\pe'}_{α(μ) \circ ρ}) \Arrow{Refold $\semabsS{σ_2}$} \\
      {}={}& \semabsS{σ_2}
    \end{DispWithArrows*}

  \item \textbf{Case }$\AppET$:
    Then $(\Lam{\py}{\pe'},ρ,μ,\ApplyF(\pa) \pushF κ) \smallstep (\pe',ρ[\py↦\pa],μ,κ)$.
    \begin{DispWithArrows*}[fleqn,mathindent=3em]
           & \semabsS{σ_1} \Arrow{Unfold $\semabsS{σ_1}$} \\
      {}={}& \mathit{apps}_μ(\ApplyF(\pa) \pushF κ)(\semabs{\Lam{\py}{\pe'}}_{α(μ) \circ ρ}) \Arrow{Unfold $\mathit{apps}$} \\
      {}={}& \mathit{apps}_μ(κ)(\mathit{app}(\semabs{\Lam{\py}{\pe'}}_{α(μ) \circ ρ}, α(μ)(\pa))) \Arrow{Unfold RHS of \Cref{thm:absence-subst}} \\
      {}⊒{}& \mathit{apps}_μ(κ)(\semabs{\pe'}_{(α(μ) \circ ρ)[\py↦α(μ)(\pa)]}) \Arrow{Rearrange} \\
      {}={}& \mathit{apps}_μ(κ)(\semabs{\pe'}_{(α(μ) \circ ρ[\py↦\pa])}) \Arrow{Refold $\semabsS{σ_2}$} \\
      {}={}& \semabsS{σ_2}
    \end{DispWithArrows*}

  \item \textbf{Case }$\LookupT$:
    Then $\pe = \py$, $\pa \triangleq ρ(\py)$, $(\pz,ρ',\pe') \triangleq μ(\pa)$ and
    $(\py,ρ,μ,κ) \smallstep (\pe',ρ',μ,\UpdateF(\pa) \pushF κ)$.
    \begin{DispWithArrows*}[fleqn,mathindent=3em]
           & \semabsS{σ_1} \Arrow{Unfold $\semabsS{σ_1}$} \\
      {}={}& \mathit{apps}_μ(κ)(\semabs{\py}_{α(μ) \circ ρ}) \Arrow{Unfold $\semabs{\py}$} \\
      {}={}& \mathit{apps}_μ(κ)((α(μ) \circ ρ)(\py)) \Arrow{Unfold $α$} \\
      {}={}& \mathit{apps}_μ(κ)(\pz \both \semabs{\pe'}_{α(μ) \circ ρ'}) \Arrow{Drop $[\pz↦\aU]$} \\
      {}⊒{}& \mathit{apps}_μ(κ)(\semabs{\pe'}_{α(μ) \circ ρ'}) \Arrow{Definition of $\mathit{apps}_μ$} \\
      {}={}& \mathit{apps}_μ(\UpdateF(\pa) \pushF κ)(\semabs{\pe'}_{α(μ) \circ ρ'}) \Arrow{Refold $\semabsS{σ_2}$} \\
      {}={}& \semabsS{σ_2}
    \end{DispWithArrows*}

  \item \textbf{Case }$\UpdateT$:
    Then $(\pv, ρ, μ[\pa↦(\py,ρ',\pe')], \UpdateF(\pa) \pushF κ) \smallstep (\pv,ρ,μ[\pa↦(\py,ρ,\pv)],κ)$.

    This case is a bit hand-wavy and shows how heap update during by-need
    evaluation is dreadfully complicated to handle, even though
    $\semabs{\wild}$ is heap-less and otherwise correct \wrt by-name
    evaluation.
    The culprit is that in order to show $\semabsS{σ_2} ⊑ \semabsS{σ_1}$, we
    have to show
    \begin{equation}
      \semabs{\pv}_{α(μ) \circ ρ} ⊑ \semabs{\pe'}_{α(μ') \circ ρ'}. \label{eqn:absent-upd}
    \end{equation}

    Intuitively, this is somewhat clear, because $μ$ ``evaluates to'' $μ'$ and
    $\pv$ is the value of $\pe'$, in the sense that there exists
    $σ'=(\pe',ρ',μ',κ)$ such that $σ' \smallstep^* σ_1 \smallstep σ_2$.

    Alas, who guarantees that such a $σ'$ actually exists?
    We would need to rearrange the lemma for that and argue by step indexing
    (a.k.a. coinduction) over prefixes of \emph{maximal traces} (to be
    rigorously defined later).
    That is, we presume that the statement
    \[
      \forall n.\ σ_0 \smallstep^n σ_2 \Longrightarrow \semabsS{σ_2} ⊑ \semabsS{σ_0}
    \]
    has been proved for all $n < k$ and proceed to prove it for $n = k$.
    So we presume $σ_0 \smallstep^{k-1} σ_1 \smallstep σ_2$ and $\semabsS{σ_1} ⊑ \semabsS{σ_0}$
    to arrive at a similar setup as before, only with a stronger assumption
    about $σ_1$.
    Specifically, due to the balanced stack discipline we know that
    $σ_0 \smallstep^{k-1} σ_1$ factors over $σ'$ above.
    We may proceed by induction over the balanced stack discipline (we will see
    in \Cref{sec:adequacy} that this amounts to induction over the big-step
    derivation) of the trace $σ' \smallstep^* σ_1$ to show \Cref{eqn:absent-upd}.

    This reasoning was not specific to $\semabs{\wild}$ at all.
    We will show a more general result in \Cref{thm:abstract-by-need}
    that can be reused across many more analyses.

    Assuming \Cref{eqn:absent-upd} has been proved, we proceed
    \begin{DispWithArrows*}[fleqn,mathindent=3em]
           & \semabsS{σ_1} \Arrow{Unfold $\semabsS{σ_1}$} \\
      {}={}& \mathit{apps}_μ(\UpdateF(\pa) \pushF κ)(\semabs{\pv}_{α(μ) \circ ρ}) \Arrow{Definition of $\mathit{apps}_μ$} \\
      {}={}& \mathit{apps}_μ(κ)(\semabs{\pv}_{α(μ) \circ ρ}) \Arrow{Above argument that $\semabs{\pv}_{α(μ) \circ ρ} ⊑ \semabs{\pe'}_{α(μ') \circ ρ'}$} \\
      {}⊒{}& \mathit{apps}_{μ[\pa↦(\py,ρ,\pv)]}(κ)(\semabs{\pv}_{α(μ[\pa↦(\py,ρ,\pv)]) \circ ρ}) \Arrow{Refold $\semabsS{σ_2}$} \\
      {}={}& \semabsS{σ_2}
    \end{DispWithArrows*}
\end{itemize}
\end{proof}

\noindent
We conclude with the \hypertarget{proof:absence-correct}{proof} for \Cref{thm:absence-correct}:
\begin{proof}
We show the contraposition, that is,
if $\px$ is used in $\pe$, then $φ(\px) = \aU$.

Since $\px$ is used in $\pe$, there exists a trace
\[
  (\Let{\px}{\pe'}{\pe},ρ,μ,κ) \smallstep (\pe,ρ_1,μ_1,κ) \smallstep^* (\py,ρ'[\py↦\pa],μ',κ') \smallstep[\LookupT(\px)] ...,
\]
where $ρ_1 \triangleq ρ[\px↦\pa]$, $μ_1 \triangleq μ[\pa↦(\px,ρ[\px↦\pa],\pe')]$.
Without loss of generality, we assume the trace prefix ends at the first lookup
at $\pa$, so $μ'(\pa) = μ_1(\pa) = (\px, ρ_1,\pe')$.
If that was not the case, we could just find a smaller prefix with this property.

Let us abbreviate $\tr \triangleq (α(μ_1) \circ ρ_1)$.
Under the above assumptions, $\tr(\py).φ(\px) = \aU$ implies $\px = \py$ for all
$\py$, because $μ_1(\pa)$ is the only heap entry in which $\px$ occurs by our
shadowing assumptions on syntax.
By unfolding $\semabsS{\wild}$ and $\semabs{\py}$ we can see that
\[
  [\px ↦ \aU] ⊑ α(μ_1)(\pa).φ = α(μ')(\pa).φ = \semabs{\py}_{α(μ') \circ ρ'[\py↦\pa]}.φ ⊑ (\semabsS{(\py,ρ'[\py↦\pa],μ',κ')}).φ.
\]
By \Cref{thm:preserve-absent}, we also have
\[
  (\semabsS{(\py,ρ'[\py↦\pa],μ',κ')}).φ ⊑ (\semabsS{(\pe,ρ_1,μ_1,κ)}).φ.
\]
And with transitivity, we get $[\px ↦ \aU] ⊑ (\semabsS{(\pe,ρ_1,μ_1,κ)}).φ$.
Since there was no other heap entry for $\px$ and $\pa$ cannot occur in $κ$ or
$ρ$ due to well-addressedness, we have $[\px ↦ \aU] ⊑ (\semabsS{(\pe,ρ_1,μ_1,κ)}).φ$ if
and only if $[\px ↦ \aU] ⊑ (\semabs{\pe}_{\tr}).φ$.
With \Cref{thm:diag-fact}, we can decompose
\begin{DispWithArrows*}[fleqn,mathindent=1em]
       & [\px ↦ \aU] \Arrow{Above result} \\
  {}⊑{}& (\semabs{\pe}_{\tr}).φ \Arrow{$\tr_Δ(\px) \triangleq \langle [\px ↦ \aU], \tr(\px).π \rangle$, \Cref{thm:diag-fact}} \\
  {}={}& ((\semabs{\pe}_{\tr_Δ})[\many{\py \Mapsto \tr(\py).φ}]).φ \Arrow{$π ⊑ \repU$, hence $\tr_Δ ⊑ \tr_\pe$} \\
  {}⊑{}& ((\semabs{\pe}_{\tr_\pe})[\many{\py \Mapsto \tr(\py).φ}]).φ \Arrow{Definition of $\wild[\wild \Mapsto \wild]$} \\
  {}={}& \Lub \{ \tr(\py).φ \mid \semabs{\pe}_{\tr_\pe}.φ(\py) = \aU \}
\end{DispWithArrows*}
But since $\tr(\py).φ(\px) = \aU$ implies $\px = \py$ (refer to definition of
$\tr$), we must have $(\semabs{\pe}_{\tr_\pe}).φ(\px) = \aU$, as required.
\end{proof}
\end{toappendix}

\Cref{defn:absence} and the substitution \Cref{thm:absence-subst} will make
a reappearance in \Cref{sec:soundness}.
They are necessary components of a soundness proof.
Building on these definitions, we may finally attempt the proof for
\Cref{thm:absence-correct}.
We suggest for the reader to have a cursory look at the proof in the Appendix.
The proof is exemplary of far more ambitious proofs such as in
\citet{Sergey:14} and \citet[Section 4]{Breitner:16}.
Though seemingly disparate, these proofs all follow an established
preservation-style proof technique at heart.
The proof of \citet{Sergey:14} for a generalisation of $\semabs{\wild}$
is roughly structured as follows (starred* references of Figures and Lemmas
refer to \citet{Sergey:14}):

\begin{enumerate}
  \item Instrument a standard call-by-need semantics (a variant of our reference
    in \Cref{sec:op-sem}) such that heap lookups decrement a per-address
    counter; when heap lookup is attempted and the counter is 0, the machine is stuck.
    For absence, the instrumentation is simpler: the $\LookupT$
    transition in \Cref{fig:lk-semantics} carries the let-bound variable that is
    looked up.
  \item Give a declarative type system that characterises the results of the
    analysis (\ie $\semabs{\wild}$) in a lenient (upwards closed) way.
    In case of \Cref{thm:absence-correct}, we define an analysis function on
    machine configurations for the proof (\Cref{fig:absence-ext}).
  \item Prove that evaluation of well-typed terms in the instrumented
    semantics is bisimilar to evaluation of the term in the standard semantics,
    \ie does not get stuck when the standard semantics would not.
    A classic \emph{logical relation}~\citep{Nielson:99}.
\end{enumerate}
Alas, the effort in comprehending such a proof in detail, let alone formulating
it, is enormous.
\begin{itemize}
  \item
    The instrumentation (1) can be semantically non-trivial; for example the
    semantics in \citet{Sergey:14} becomes non-deterministic.
    Does this instrumentation still express the desired semantic property?
  \item Step (2) all but duplicates a complicated analysis
    definition (\ie $\semabs{\wild}$) into a type system (in Figure 7*) with
    subtle adjustments expressing invariants for the preservation proof.
  \item
    Furthermore, step (2) extends this type system to small-step machine
    configurations (in Figure 13*), \ie stacks and heaps, the scoping of which
    is mutually recursive.%
    \footnote{We believe that this extension can always be derived systematically from a
    context lemma~\citep[Lemma 3.2]{MoranSands:99} and imitating what the type
    system does on the closed expression derivable from a configuration via the
    context lemma.}
    Another page worth of Figures; the amount of duplicated proof artifacts is
    staggering.
    In our case, the analysis function on machine configurations is about as
    long as on expressions.
  \item
    This is all setup before step (3) proves interesting properties about the
    semantic domain of the analysis.
    Among the more interesting properties is the \emph{substitution lemma} A.8*
    to be applied during beta reduction; exactly as in our proof.
  \item
    While proving that a single step $σ_1 \smallstep σ_2$ preserves analysis
    information in step (3), we noticed that we actually got stuck in the $\UpdateT$
    case, and would need to redo the proof using step-indexing~\citep{AppelMcAllester:01}.
    This case mutates the heap and thus is notoriously difficult; we give a
    proper account in \Cref{thm:abstract-by-need}.

    Although the proof in \citet{Sergey:14} is perceived as detailed and
    rigorous, it is quite terse in the corresponding \textsc{EUpd} case of the
    single-step safety proof in Lemma A.6*.
\end{itemize}
\noindent
There are two main problems to address, and we believe the first causes the second.
\begin{enumerate}
  \item
    Although analysis and semantics are individually simple, it is conceptually
    difficult to connect them, causing an explosion of formal artefacts.
    This is because one is compositional while the other is not.
  \item
    Compared to analysis and semantics, the soundness proof is rather
    complicated because it is \emph{entangled}:
    The parts of the proof that concern the domain of the analysis are drowned in
    coping with semantic subtleties that ultimately could be shared with similar
    analyses.
\end{enumerate}

Abstract interpretation~\citep{Cousot:77} provides a framework to tackle problem
(2), but its systematic applications seem to require a structurally matching
semantics.
For example, the book of \citet{Cousot:21} starts from a \emph{compositional},
trace-generating semantics for an imperative first-order language to derive
compositional analyses.

In this work we present the \textbf{\emph{denotational interpreter}} design
pattern to solve both problems above.
Inspired by \citeauthor{Cousot:21}, we define a \textbf{\emph{compositional
semantics that exhibits operational detail}} for higher-order languages;
one in which it is possible to express \emph{operational properties} such as
\emph{usage cardinality}, \ie ``$\pe$ evaluates $\px$ at most $u$ times'', as
required in \citet{Sergey:14}.%
\footnote{Useful applications of the ``at most once'' cardinality are given in
\citet{Turner:95,Sergey:14}, motivating inlining into function bodies that are
called at most once, for example.}

The example of usage analysis in \Cref{sec:abstraction} (generalising
$\semabs{\wild}$) demonstrates that we can \textbf{\emph{derive summary-based
analyses as an abstract interpretation}} from our semantics.

Since both semantics and analysis are \textbf{\emph{derived from the same
generic interpreter}}, solving problem (1), we can prove usage analysis to be an
\emph{abstract interpretation} of call-by-need semantics.
Doing so disentangles the preservation proof such that the proof
for usage analysis in \Cref{thm:usage-abstracts-need} takes no more than a
semantic substitution lemma and a bit of plumbing, solving problem (2).

\section{Reference Semantics: Lazy Krivine Machine}
\label{sec:op-sem}

\begin{figure}
\[\begin{array}{c}
 \arraycolsep3pt
 \begin{array}{rrclcl@{\hspace{1.5em}}rrcrclcl}
  \text{Addresses}     & \pa & ∈ & \Addresses     & \simeq & ℕ
  &
  \text{States}        & σ   & ∈ & \States        & =      & \Exp \times \Environments \times \Heaps \times \Continuations
  \\
  \text{Environments}  & ρ   & ∈ & \Environments  & =      & \Var \pfun \Addresses
  &
  \text{Heaps}         & μ   & ∈ & \Heaps         & =      & \Addresses \pfun \Var \times \Environments \times \Exp
  \\
  \text{Continuations} & κ   & ∈ & \Continuations & ::=    & \multicolumn{7}{l}{\StopF \mid \ApplyF(\pa) \pushF κ \mid \SelF(ρ,\SelArity) \pushF κ \mid \UpdateF(\pa) \pushF κ} \\
 \end{array} \\
 \\[-0.5em]
\end{array}\]

\newcolumntype{L}{>{$}l<{$}} 
\newcolumntype{R}{>{$}r<{$}} 
\newcolumntype{C}{>{$}c<{$}} 
\resizebox{\textwidth}{!}{%
\begin{tabular}{LR@{\hspace{0.4em}}C@{\hspace{0.4em}}LL}
\toprule
\text{Rule} & σ_1 & \smallstep & σ_2 & \text{where} \\
\midrule
\LetIT & (\Let{\px}{\pe_1}{\pe_2},ρ,μ,κ) & \smallstep & (\pe_2,ρ',μ[\pa↦(\px,ρ',\pe_1)], κ) & \pa \not∈ \dom(μ),\ ρ'\! = ρ[\px↦\pa] \\
\AppIT & (\pe~\px,ρ,μ,κ) & \smallstep & (\pe,ρ,μ,\ApplyF(\pa) \pushF κ) & \pa = ρ(\px) \\
\CaseIT & (\Case{\pe_s}{\Sel[r]},ρ,μ,κ) & \smallstep & (\pe_s,ρ,μ,\SelF(ρ,\Sel[r]) \pushF κ) & \\
\LookupT(\py) & (\px, ρ, μ, κ) & \smallstep & (\pe, ρ', μ, \UpdateF(\pa) \pushF κ) & \pa = ρ(\px),\ (\py,ρ',\pe) = μ(\pa) \\
\AppET & (\Lam{\px}{\pe},ρ,μ, \ApplyF(\pa) \pushF κ) & \smallstep & (\pe,ρ[\px ↦ \pa],μ,κ) &  \\
\CaseET & (K'~\many{y},ρ,μ, \SelF(ρ',\Sel) \pushF κ) & \smallstep & (\pe_i,ρ'[\many{\px_i ↦ \pa}],μ,κ) & K_i = K',\ \many{\pa = ρ(\py)} \\
\UpdateT & (\pv, ρ, μ, \UpdateF(\pa) \pushF κ) & \smallstep & (\pv, ρ, μ[\pa ↦ (\px,ρ,\pv)], κ) & μ(\pa) = (\px,\wild,\wild) \\
\bottomrule
\end{tabular}
} 
\caption{Lazy Krivine transition semantics $\smallstep$}
  \label{fig:lk-semantics}
\end{figure}

Before we get to introduce our novel denotational interpreters, let us
recall the semantic ground truth of this work and others \citep{Sergey:14,
Breitner:16}: The Mark II machine of \citet{Sestoft:97} given in
\Cref{fig:lk-semantics}, a small-step operational semantics.
It is a Lazy Krivine (LK) machine implementing call-by-need.
(A close sibling for call-by-value would be a CESK machine \citep{Felleisen:87}.)
A reasonable call-by-name semantics can be recovered by removing the $\UpdateT$
rule and the pushing of update frames in $\LookupT$.

The configurations $σ$ in this transition system resemble abstract machine
states, consisting of a control expression $\pe$, an environment $ρ$ mapping
lexically-scoped variables to their current heap address, a heap $μ$ listing a
closure for each address, and a stack of continuation frames $κ$.
There is one harmless non-standard extension: For $\LookupT$
transitions, we take note of the let-bound variable $\py$ which allocated the
heap binding that the machine is about to look up. The association from address
to let-bound variable is maintained in the first component of a heap entry
triple and requires slight adjustments of the $\LetIT$, $\LookupT$ and
$\UpdateT$ rules.

The notation $f ∈ A \pfun B$ used in the definition of $ρ$ and $μ$ denotes a
finite map from $A$ to $B$, a partial function where the domain $\dom(f)$ is
finite.
The literal notation $[a_1↦b_1,...,a_n↦b_n]$ denotes a finite map with domain
$\{a_1,...,a_n\}$ that maps $a_i$ to $b_i$. Function update $f[a ↦ b]$
maps $a$ to $b$ and is otherwise equal to $f$.

The initial machine state for a closed expression $\pe$ is given by the
injection function $\init(\pe) = (\pe,[],[],\StopF)$ and
the final machine states are of the form $(\pv,\wild,\wild,\StopF)$.
We bake into $σ∈\States$ the simplifying invariant of \emph{well-addressedness}:
Any address $\pa$ occurring in $ρ$, $κ$ or the range of $μ$ must be an element of
$\dom(μ)$.
It is easy to see that the transition system maintains this invariant and that
it is still possible to observe scoping errors which are thus confined to lookup
in $ρ$.

We conclude with two example traces. The first one evaluates $\Let{i}{\Lam{x}{x}}{i~i}$:
\begin{align} \label{ex:trace}
  &\arraycolsep3pt
  \begin{array}{lclclclclc}
  (\Let{i}{\Lam{x}{x}}{i~i}, [], [], \StopF) & \smallstep[\LetIT] & (i~i, ρ_1, μ, \StopF) & \smallstep[\AppIT] & (i, ρ_1, μ, κ) & \smallstep[\LookupT(i)] \\
  \highlight{(\Lam{x}{x}, ρ_1, μ, \UpdateF(\pa_1) \pushF κ)} & \highlight{\smallstep[\UpdateT]} & (\Lam{x}{x}, ρ_1, μ, κ) & \smallstep[\AppET] & (x, ρ_2, μ, \StopF) & \smallstep[\LookupT(i)] \\
  \highlight{(\Lam{x}{x}, ρ_1, μ, \UpdateF(\pa_1) \pushF \StopF)} & \highlight{\smallstep[\UpdateT]} & (\Lam{x}{x}, ρ_1, μ, \StopF)
  \end{array} \\ \notag
  &\qquad \text{where} \begin{array}{llll}
    κ = \ApplyF(\pa_1) \pushF \StopF, & ρ_1 = [i ↦ \pa_1], & ρ_2 = [i ↦ \pa_1, x ↦ \pa_1], & μ = [\pa_1 ↦ (i, ρ_1,\Lam{x}{x})] \\
  \end{array} \notag
\end{align}
The corresponding by-name trace simply omits the highlighted update steps.
The last $\LookupT(i)$ transition exemplifies that the lambda-bound variable in
control differs from the let-bound variable used to identify the heap entry.

The second example evaluates $\pe \triangleq \Let{i}{(\Lam{y}{\Lam{x}{x}})~i}{i~i}$,
demonstrating memoisation of $i$:
\begin{align} \label{ex:trace2}
  &\begin{array}{l}
  (\pe, [], [], \StopF)
  \smallstep[\LetIT]
  (i~i, ρ_1, μ_1, \StopF)
  \smallstep[\AppIT]
  (i, ρ_1, μ_1, κ_1)
  \smallstep[\LookupT(i)]
  ((\Lam{y}{\Lam{x}{x}})~i, ρ_1, μ_1, κ_2)
  \\
  \smallstep[\AppIT]
  (\Lam{y}{\Lam{x}{x}}, ρ_1, μ_1, \ApplyF(\pa_1) \pushF κ_2)
  \smallstep[\AppET]
  (\Lam{x}{x}, ρ_2, μ_1, κ_2)
  \smallstep[\UpdateT]
  (\Lam{x}{x}, ρ_2, μ_2, κ_1)
  \\
  \smallstep[\AppET]
  (x, ρ_3, μ_2, \StopF)
  \smallstep[\LookupT(i)]
  (\Lam{x}{x}, ρ_2, μ_2, \UpdateF(\pa_1) \pushF \StopF)
  \smallstep[\UpdateT]
  (\Lam{x}{x}, ρ_2, μ_2, \StopF)
  \end{array} \\ \notag
  &\qquad\text{where } \arraycolsep1pt \begin{array}{ll}
    ρ_1 = [i ↦ \pa_1], & μ_1 = [\pa_1 ↦ (i, ρ_1, (\Lam{y}{\Lam{x}{x}})~i)], ρ_3 = [i ↦ \pa_1, y ↦ \pa_1, x ↦ \pa_1] \\
    ρ_2 = [i ↦ \pa_1, y ↦ \pa_1], & μ_2 = [\pa_1 ↦ (i,ρ_2,\Lam{x}{x})], κ_1 = \ApplyF(\pa_1) \pushF \StopF, κ_2 = \UpdateF (\pa_1) \pushF κ_1
  \end{array} \notag
\end{align}

\section{A Denotational Interpreter}
\label{sec:interp}

In this section, we present a generic \emph{denotational interpreter}%
\footnote{This term was coined by \citet{Might:10}. We find it fitting,
because a denotational interpreter is both a \emph{denotational
semantics}~\citep{ScottStrachey:71} as well as a total \emph{definitional
interpreter}~\citep{Reynolds:72}.}
for a functional language which we instantiate with different semantic domains.
The choice of semantic domain determines the \emph{evaluation strategy}
(call-by-name, call-by-value, call-by-need) and the degree to which
\emph{operational detail} can be observed.
Other semantic domains give rise to useful \emph{summary-based} static
analyses such as usage analysis in \Cref{sec:abstraction}.
The major contribution of our framework is that the derived summary-based
analyses may observe operational detail in an intuitive and semantically
meaningful way.
Adhering to our design pattern pays off in that it enables sharing of soundness
proofs, thus drastically simplifying the soundness proof obligation per derived
analysis (\Cref{sec:soundness}).

Denotational interpreters can be implemented in any higher-order language such as OCaml, Scheme or Java with explicit thunks, but we picked Haskell for convenience.%
\footnote{We extract from this document runnable Haskell files which we add as a Supplement, containing the complete definitions. Furthermore, the (terminating) interpreter outputs are directly generated from this extract.}

\begin{figure}
\begin{minipage}{0.49\textwidth}
\begin{hscode}\SaveRestoreHook
\column{B}{@{}>{\hspre}l<{\hspost}@{}}%
\column{3}{@{}>{\hspre}c<{\hspost}@{}}%
\column{3E}{@{}l@{}}%
\column{6}{@{}>{\hspre}l<{\hspost}@{}}%
\column{E}{@{}>{\hspre}l<{\hspost}@{}}%
\>[B]{}\keyword{data}\;\conid{Exp}{}\<[E]%
\\
\>[B]{}\hsindent{3}{}\<[3]%
\>[3]{}\mathrel{=}{}\<[3E]%
\>[6]{}\conid{Var}\;\conid{Name}\mid \conid{Let}\;\conid{Name}\;\conid{Exp}\;\conid{Exp}{}\<[E]%
\\
\>[B]{}\hsindent{3}{}\<[3]%
\>[3]{}\mid {}\<[3E]%
\>[6]{}\conid{Lam}\;\conid{Name}\;\conid{Exp}\mid \conid{App}\;\conid{Exp}\;\conid{Name}{}\<[E]%
\\
\>[B]{}\hsindent{3}{}\<[3]%
\>[3]{}\mid {}\<[3E]%
\>[6]{}\conid{ConApp}\;\conid{Tag}\;[\mskip1.5mu \conid{Name}\mskip1.5mu]\mid \conid{Case}\;\conid{Exp}\;\conid{Alts}{}\<[E]%
\\
\>[B]{}\keyword{type}\;\conid{Name}\mathrel{=}\conid{String}{}\<[E]%
\\
\>[B]{}\keyword{type}\;\conid{Alts}\mathrel{=}\conid{Tag}\mathbin{:\rightharpoonup}([\mskip1.5mu \conid{Name}\mskip1.5mu],\conid{Exp}){}\<[E]%
\\
\>[B]{}\keyword{data}\;\conid{Tag}\mathrel{=}\mathbin{...};\varid{conArity}\mathbin{::}\conid{Tag}\to \conid{Int}{}\<[E]%
\ColumnHook
\end{hscode}\resethooks
\caption{Syntax}
\label{fig:syntax}
\end{minipage}%
\begin{minipage}{0.51\textwidth}
\begin{hscode}\SaveRestoreHook
\column{B}{@{}>{\hspre}l<{\hspost}@{}}%
\column{16}{@{}>{\hspre}l<{\hspost}@{}}%
\column{E}{@{}>{\hspre}l<{\hspost}@{}}%
\>[B]{}\keyword{type}\;(\mathbin{:\rightharpoonup})\mathrel{=}\conid{Map};\varcolor{\varepsilon}\mathbin{::}\conid{Ord}\;\varid{k}\Rightarrow \varid{k}\mathbin{:\rightharpoonup}\varid{v}{}\<[E]%
\\
\>[B]{}\wild[\wild\mapsto\wild]\mathbin{::}\conid{Ord}\;\varid{k}\Rightarrow (\varid{k}\mathbin{:\rightharpoonup}\varid{v})\to \varid{k}\to \varid{v}\to (\varid{k}\mathbin{:\rightharpoonup}\varid{v}){}\<[E]%
\\
\>[B]{}\wild[\many{\wild\mapsto\wild}]\mathbin{::}\conid{Ord}\;\varid{k}{}\<[16]%
\>[16]{}\Rightarrow (\varid{k}\mathbin{:\rightharpoonup}\varid{v})\to [\mskip1.5mu \varid{k}\mskip1.5mu]\to [\mskip1.5mu \varid{v}\mskip1.5mu]{}\<[E]%
\\
\>[16]{}\to (\varid{k}\mathbin{:\rightharpoonup}\varid{v}){}\<[E]%
\\
\>[B]{}(\mathop{!})\mathbin{::}\conid{Ord}\;\varid{k}\Rightarrow (\varid{k}\mathbin{:\rightharpoonup}\varid{v})\to \varid{k}\to \varid{v}{}\<[E]%
\\
\>[B]{}\varid{dom}\mathbin{::}\conid{Ord}\;\varid{k}\Rightarrow (\varid{k}\mathbin{:\rightharpoonup}\varid{v})\to \conid{Set}\;\varid{k}{}\<[E]%
\\
\>[B]{}(\in )\mathbin{::}\conid{Ord}\;\varid{k}\Rightarrow \varid{k}\to \conid{Set}\;\varid{k}\to \conid{Bool}{}\<[E]%
\\
\>[B]{}(\mathbin{\lhd})\mathbin{::}(\varid{b}\to \varid{c})\to (\varid{a}\mathbin{:\rightharpoonup}\varid{b})\to (\varid{a}\mathbin{:\rightharpoonup}\varid{c}){}\<[E]%
\\
\>[B]{}\varid{assocs}\mathbin{::}(\varid{k}\mathbin{:\rightharpoonup}\varid{v})\to [\mskip1.5mu (\varid{k},\varid{v})\mskip1.5mu]{}\<[E]%
\ColumnHook
\end{hscode}\resethooks
\caption{Environments}
\label{fig:map}
\end{minipage}
\end{figure}

\subsection{Semantic Domain} \label{sec:dna}

Just as traditional denotational semantics, denotational interpreters
assign meaning to programs in some \emph{semantic domain}.
Traditionally, the semantic domain \ensuremath{\conid{D}} comprises \emph{semantic values} such as
base values (integers, strings, etc.) and functions \ensuremath{\conid{D}\to \conid{D}}.
One of the main features of these semantic domains is that they lack
\emph{operational}, or, \emph{intensional detail} that is unnecessary to
assigning each observationally distinct expression a distinct meaning.
For example, it is not possible to observe evaluation cardinality, which
is the whole point of analyses such as usage analysis (\Cref{sec:abstraction}).

A distinctive feature of our work is that our semantic domains are instead
\emph{traces} that describe the \emph{steps} taken by an abstract machine, and
that \emph{end} in semantic values.
It is possible to describe usage cardinality as a property of the traces
thus generated, as required for a soundness proof of usage analysis.
We choose \ensuremath{\conid{D}_{\mathbf{na}}}, defined below, as the first example of such a semantic domain,
because it is simple and illustrative of the approach.
Instantiated at \ensuremath{\conid{D}_{\mathbf{na}}}, our generic interpreter will produce precisely the
traces of the by-\textbf{\textrm{na}}me variant of the Krivine machine in
\Cref{fig:lk-semantics}.%
\footnote{For a realistic implementation, we would define \ensuremath{\conid{D}} as a \ensuremath{\keyword{newtype}} to
keep type class resolution decidable and non-overlapping. We will however stick
to a \ensuremath{\keyword{type}} synonym in this presentation in order to elide noisy wrapping and
unwrapping of constructors.}

\begin{minipage}{0.62\textwidth}
\begin{hscode}\SaveRestoreHook
\column{B}{@{}>{\hspre}l<{\hspost}@{}}%
\column{20}{@{}>{\hspre}c<{\hspost}@{}}%
\column{20E}{@{}l@{}}%
\column{23}{@{}>{\hspre}l<{\hspost}@{}}%
\column{27}{@{}>{\hspre}l<{\hspost}@{}}%
\column{E}{@{}>{\hspre}l<{\hspost}@{}}%
\>[B]{}\keyword{type}\;\conid{D}\;\varcolor{\tau}\mathrel{=}\varcolor{\tau}\;(\conid{Value}\;\varcolor{\tau});{}\<[27]%
\>[27]{}\keyword{type}\;\conid{D}_{\mathbf{na}}\mathrel{=}\conid{D}\;\conid{T}{}\<[E]%
\\
\>[B]{}\keyword{data}\;\conid{T}\;\varid{v}\mathrel{=}\conid{Step}\;\conid{Event}\;(\conid{T}\;\varid{v})\mid \conid{Ret}\;\varid{v}{}\<[E]%
\\
\>[B]{}\keyword{data}\;\conid{Event}{}\<[20]%
\>[20]{}\mathrel{=}{}\<[20E]%
\>[23]{}\conid{Look}\;\conid{Name}\mid \conid{Upd}\mid \conid{App}_{1}\mid \conid{App}_{2}{}\<[E]%
\\
\>[20]{}\mid {}\<[20E]%
\>[23]{}\conid{Let}_{1}\mid \conid{Case}_{1}\mid \conid{Case}_{2}{}\<[E]%
\\
\>[B]{}\keyword{data}\;\conid{Value}\;\varcolor{\tau}\mathrel{=}\conid{Stuck}\mid \conid{Fun}\;(\conid{D}\;\varcolor{\tau}\to \conid{D}\;\varcolor{\tau})\mid \conid{Con}\;\conid{Tag}\;[\mskip1.5mu \conid{D}\;\varcolor{\tau}\mskip1.5mu]{}\<[E]%
\ColumnHook
\end{hscode}\resethooks
\end{minipage}
\begin{minipage}{0.38\textwidth}
\begin{hscode}\SaveRestoreHook
\column{B}{@{}>{\hspre}l<{\hspost}@{}}%
\column{3}{@{}>{\hspre}l<{\hspost}@{}}%
\column{E}{@{}>{\hspre}l<{\hspost}@{}}%
\>[B]{}\keyword{instance}\;\conid{Monad}\;\conid{T}\;\keyword{where}{}\<[E]%
\\
\>[B]{}\hsindent{3}{}\<[3]%
\>[3]{}\varid{return}\;\varid{v}\mathrel{=}\conid{Ret}\;\varid{v}{}\<[E]%
\\
\>[B]{}\hsindent{3}{}\<[3]%
\>[3]{}\conid{Ret}\;\varid{v}\bind \varid{k}\mathrel{=}\varid{k}\;\varid{v}{}\<[E]%
\\
\>[B]{}\hsindent{3}{}\<[3]%
\>[3]{}\conid{Step}\;\varid{e}\;\varcolor{\tau}\bind \varid{k}\mathrel{=}\conid{Step}\;\varid{e}\;(\varcolor{\tau}\bind \varid{k}){}\<[E]%
\ColumnHook
\end{hscode}\resethooks
\end{minipage}
\noindent
A trace \ensuremath{\conid{T}} either returns a value (\ensuremath{\conid{Ret}}) or makes a small-step transition (\ensuremath{\conid{Step}}).
Each step \ensuremath{\conid{Step}\;\varid{ev}\;\varid{rest}} is decorated with an event \ensuremath{\varid{ev}}, which describes what happens in that step.
For example, event \ensuremath{\conid{Look}\;\varid{x}} describes the lookup of variable \ensuremath{\varid{x}\mathbin{::}\conid{Name}} in the environment.
Note that the choice of \ensuremath{\conid{Event}} is use-case (\ie analysis) specific and suggests
a spectrum of intensionality, with \ensuremath{\keyword{data}\;\conid{Event}\mathrel{=}\conid{Unit}} on the more abstract end
of the spectrum and arbitrary syntactic detail attached to each of \ensuremath{\conid{Event}}'s
constructors at the intensional end of the spectrum.%
\footnote{If our language had facilities for input/output and more general
side-effects, we could have started from a more elaborate trace construction
such as (guarded) interaction trees~\citep{interaction-trees,gitrees}.}

A trace in \ensuremath{\conid{D}_{\mathbf{na}}\mathrel{=}\conid{T}\;(\conid{Value}\;\conid{T})} eventually terminates with a \ensuremath{\conid{Value}} that is
either stuck (\ensuremath{\conid{Stuck}}), a function waiting to be applied to a domain value
(\ensuremath{\conid{Fun}}), or a constructor application giving the denotations of its
fields (\ensuremath{\conid{Con}}).
We postpone worries about well-definedness and totality of this encoding to
\Cref{sec:totality}.

\begin{figure}
\begin{minipage}{0.55\textwidth}
\begin{hscode}\SaveRestoreHook
\column{B}{@{}>{\hspre}l<{\hspost}@{}}%
\column{3}{@{}>{\hspre}l<{\hspost}@{}}%
\column{5}{@{}>{\hspre}l<{\hspost}@{}}%
\column{7}{@{}>{\hspre}c<{\hspost}@{}}%
\column{7E}{@{}l@{}}%
\column{8}{@{}>{\hspre}l<{\hspost}@{}}%
\column{10}{@{}>{\hspre}l<{\hspost}@{}}%
\column{11}{@{}>{\hspre}l<{\hspost}@{}}%
\column{12}{@{}>{\hspre}l<{\hspost}@{}}%
\column{16}{@{}>{\hspre}l<{\hspost}@{}}%
\column{23}{@{}>{\hspre}l<{\hspost}@{}}%
\column{25}{@{}>{\hspre}l<{\hspost}@{}}%
\column{26}{@{}>{\hspre}c<{\hspost}@{}}%
\column{26E}{@{}l@{}}%
\column{29}{@{}>{\hspre}l<{\hspost}@{}}%
\column{E}{@{}>{\hspre}l<{\hspost}@{}}%
\>[B]{}\mathcal{S}\denot{\wild}_{\wild}{}\<[7]%
\>[7]{}\mathbin{::}{}\<[7E]%
\>[11]{}(\conid{Trace}\;\varid{d},\conid{Domain}\;\varid{d},\conid{HasBind}\;\varid{d}){}\<[E]%
\\
\>[7]{}\Rightarrow {}\<[7E]%
\>[11]{}\conid{Exp}\to (\conid{Name}\mathbin{:\rightharpoonup}\varid{d})\to \varid{d}{}\<[E]%
\\
\>[B]{}\mathcal{S}\denot{\varid{e}}_{\varcolor{\rho}}\mathrel{=}\keyword{case}\;\varid{e}\;\keyword{of}{}\<[E]%
\\
\>[B]{}\hsindent{3}{}\<[3]%
\>[3]{}\conid{Var}\;\varid{x}{}\<[10]%
\>[10]{}\mid \varid{x}\in \varid{dom}\;\varcolor{\rho}{}\<[23]%
\>[23]{}\to \varcolor{\rho}\mathop{!}\varid{x}{}\<[E]%
\\
\>[10]{}\mid \varid{otherwise}{}\<[23]%
\>[23]{}\to \varid{stuck}{}\<[E]%
\\
\>[B]{}\hsindent{3}{}\<[3]%
\>[3]{}\conid{Lam}\;\varid{x}\;\varid{body}\to \varid{fun}\;\varid{x}\; \mathbin{\$}\lambda \varid{d}\to {}\<[E]%
\\
\>[3]{}\hsindent{2}{}\<[5]%
\>[5]{}\varid{step}\;\conid{App}_{2}\;(\mathcal{S}\denot{\varid{body}}_{(\varcolor{\rho}[\varid{x}\mapsto\varid{d}])}){}\<[E]%
\\
\>[B]{}\hsindent{3}{}\<[3]%
\>[3]{}\conid{App}\;\varid{e}\;\varid{x}{}\<[12]%
\>[12]{}\mid \varid{x}\in \varid{dom}\;\varcolor{\rho}{}\<[25]%
\>[25]{}\to \varid{step}\;\conid{App}_{1}\mathbin{\$}{}\<[E]%
\\
\>[12]{}\hsindent{4}{}\<[16]%
\>[16]{}\varid{apply}\;(\mathcal{S}\denot{\varid{e}}_{\varcolor{\rho}})\;(\varcolor{\rho}\mathop{!}\varid{x}){}\<[E]%
\\
\>[12]{}\mid \varid{otherwise}{}\<[25]%
\>[25]{}\to \varid{stuck}{}\<[E]%
\\
\>[B]{}\hsindent{3}{}\<[3]%
\>[3]{}\conid{Let}\;\varid{x}\;\varid{e}_{1}\;\varid{e}_{2}\to \varid{bind}\; {}\<[E]%
\\
\>[3]{}\hsindent{2}{}\<[5]%
\>[5]{}(\lambda \varid{d}_{1}\to \mathcal{S}\denot{\varid{e}_{1}}_{\varcolor{\rho}[\varid{x}\mapsto\varid{step}\;(\conid{Look}\;\varid{x})\;\varid{d}_{1}]}){}\<[E]%
\\
\>[3]{}\hsindent{2}{}\<[5]%
\>[5]{}(\lambda \varid{d}_{1}\to \varid{step}\;\conid{Let}_{1}\;(\mathcal{S}\denot{\varid{e}_{2}}_{\varcolor{\rho}[\varid{x}\mapsto\varid{step}\;(\conid{Look}\;\varid{x})\;\varid{d}_{1}]})){}\<[E]%
\\
\>[B]{}\hsindent{3}{}\<[3]%
\>[3]{}\conid{ConApp}\;\varid{k}\;\varid{xs}{}\<[E]%
\\
\>[3]{}\hsindent{2}{}\<[5]%
\>[5]{}\mid \varid{all}\;(\in \varid{dom}\;\varcolor{\rho})\;\varid{xs},\varid{length}\;\varid{xs}\mathrel{\clipbox{2.5pt 0pt}{$==$}}\varid{conArity}\;\varid{k}{}\<[E]%
\\
\>[3]{}\hsindent{2}{}\<[5]%
\>[5]{}\to \varid{con}\; \varid{k}\;(\varid{map}\;(\varcolor{\rho}\mathop{!})\;\varid{xs}){}\<[E]%
\\
\>[3]{}\hsindent{2}{}\<[5]%
\>[5]{}\mid \varid{otherwise}{}\<[E]%
\\
\>[3]{}\hsindent{2}{}\<[5]%
\>[5]{}\to \varid{stuck}{}\<[E]%
\\
\>[B]{}\hsindent{3}{}\<[3]%
\>[3]{}\conid{Case}\;\varid{e}\;\varid{alts}\to \varid{step}\;\conid{Case}_{1}\mathbin{\$}{}\<[E]%
\\
\>[3]{}\hsindent{2}{}\<[5]%
\>[5]{}\varid{select}\;(\mathcal{S}\denot{\varid{e}}_{\varcolor{\rho}})\;(\varid{cont}\mathbin{\lhd}\varid{alts}){}\<[E]%
\\
\>[3]{}\hsindent{2}{}\<[5]%
\>[5]{}\keyword{where}{}\<[E]%
\\
\>[5]{}\hsindent{3}{}\<[8]%
\>[8]{}\varid{cont}\;(\varid{xs},\varcolor{e}_r)\;\varid{ds}{}\<[26]%
\>[26]{}\mid {}\<[26E]%
\>[29]{}\varid{length}\;\varid{xs}\mathrel{\clipbox{2.5pt 0pt}{$==$}}\varid{length}\;\varid{ds}{}\<[E]%
\\
\>[26]{}\mathrel{=}{}\<[26E]%
\>[29]{}\varid{step}\;\conid{Case}_{2}\;(\mathcal{S}\denot{\varcolor{e}_r}_{\varcolor{\rho}[\many{\varid{xs}\mapsto\varid{ds}}]}){}\<[E]%
\\
\>[26]{}\mid {}\<[26E]%
\>[29]{}\varid{otherwise}{}\<[E]%
\\
\>[26]{}\mathrel{=}{}\<[26E]%
\>[29]{}\varid{stuck}{}\<[E]%
\ColumnHook
\end{hscode}\resethooks
\end{minipage}%
\begin{minipage}{0.44\textwidth}
\begin{hscode}\SaveRestoreHook
\column{B}{@{}>{\hspre}l<{\hspost}@{}}%
\column{3}{@{}>{\hspre}l<{\hspost}@{}}%
\column{43}{@{}>{\hspre}l<{\hspost}@{}}%
\column{E}{@{}>{\hspre}l<{\hspost}@{}}%
\>[B]{}\keyword{class}\;\conid{Trace}\;\varid{d}\;\keyword{where}{}\<[E]%
\\
\>[B]{}\hsindent{3}{}\<[3]%
\>[3]{}\varid{step}\mathbin{::}\conid{Event}\to \varid{d}\to \varid{d}{}\<[E]%
\\[\blanklineskip]%
\>[B]{}\keyword{class}\;\conid{Domain}\;\varid{d}\;\keyword{where}{}\<[E]%
\\
\>[B]{}\hsindent{3}{}\<[3]%
\>[3]{}\varid{stuck}\mathbin{::}\varid{d}{}\<[E]%
\\
\>[B]{}\hsindent{3}{}\<[3]%
\>[3]{}\varid{fun}\mathbin{::}\conid{Name}\to  (\varid{d}\to \varid{d})\to \varid{d}{}\<[E]%
\\
\>[B]{}\hsindent{3}{}\<[3]%
\>[3]{}\varid{apply}\mathbin{::}\varid{d}\to \varid{d}\to \varid{d}{}\<[E]%
\\
\>[B]{}\hsindent{3}{}\<[3]%
\>[3]{}\varid{con}\mathbin{::} \conid{Tag}\to [\mskip1.5mu \varid{d}\mskip1.5mu]\to \varid{d}{}\<[E]%
\\
\>[B]{}\hsindent{3}{}\<[3]%
\>[3]{}\varid{select}\mathbin{::}\varid{d}\to (\conid{Tag}\mathbin{:\rightharpoonup}([\mskip1.5mu \varid{d}\mskip1.5mu]\to \varid{d}))\to {}\<[43]%
\>[43]{}\varid{d}{}\<[E]%
\\[\blanklineskip]%
\>[B]{}\keyword{class}\;\conid{HasBind}\;\varid{d}\;\keyword{where}{}\<[E]%
\\
\>[B]{}\hsindent{3}{}\<[3]%
\>[3]{}\varid{bind}\mathbin{::} (\varid{d}\to \varid{d})\to (\varid{d}\to \varid{d})\to \varid{d}{}\<[E]%
\ColumnHook
\end{hscode}\resethooks
\\[-2.5em]
\subcaption{Interface of traces and values}
  \label{fig:trace-classes}
\begin{hscode}\SaveRestoreHook
\column{B}{@{}>{\hspre}l<{\hspost}@{}}%
\column{3}{@{}>{\hspre}l<{\hspost}@{}}%
\column{5}{@{}>{\hspre}l<{\hspost}@{}}%
\column{10}{@{}>{\hspre}l<{\hspost}@{}}%
\column{30}{@{}>{\hspre}l<{\hspost}@{}}%
\column{E}{@{}>{\hspre}l<{\hspost}@{}}%
\>[B]{}\keyword{instance}\;\conid{Trace}\;(\conid{T}\;\varid{v})\;\keyword{where}{}\<[E]%
\\
\>[B]{}\hsindent{3}{}\<[3]%
\>[3]{}\varid{step}\mathrel{=}\conid{Step}{}\<[E]%
\\[\blanklineskip]%
\>[B]{}\keyword{instance}\;\conid{Monad}\;\varcolor{\tau}\Rightarrow \conid{Domain}\;(\conid{D}\;\varcolor{\tau})\;\keyword{where}{}\<[E]%
\\
\>[B]{}\hsindent{3}{}\<[3]%
\>[3]{}\varid{stuck}\mathrel{=}\varid{return}\;\conid{Stuck}{}\<[E]%
\\
\>[B]{}\hsindent{3}{}\<[3]%
\>[3]{}\varid{fun}\;\anonymous \; \varid{f}\mathrel{=}\varid{return}\;(\conid{Fun}\;\varid{f}){}\<[E]%
\\
\>[B]{}\hsindent{3}{}\<[3]%
\>[3]{}\varid{apply}\;{}\<[10]%
\>[10]{}\varid{d}\;\varid{a}\mathrel{=}\varid{d}\bind \lambda \varid{v}\to \keyword{case}\;\varid{v}\;\keyword{of}{}\<[E]%
\\
\>[3]{}\hsindent{2}{}\<[5]%
\>[5]{}\conid{Fun}\;\varid{f}\to \varid{f}\;\varid{a};\anonymous \to \varid{stuck}{}\<[E]%
\\
\>[B]{}\hsindent{3}{}\<[3]%
\>[3]{}\varid{con}\; \varid{k}\;\varid{ds}\mathrel{=}\varid{return}\;(\conid{Con}\;\varid{k}\;\varid{ds}){}\<[E]%
\\
\>[B]{}\hsindent{3}{}\<[3]%
\>[3]{}\varid{select}\;\varid{dv}\;\varid{alts}\mathrel{=}\varid{dv}\bind \lambda \varid{v}\to \keyword{case}\;\varid{v}\;\keyword{of}{}\<[E]%
\\
\>[3]{}\hsindent{2}{}\<[5]%
\>[5]{}\conid{Con}\;\varid{k}\;\varid{ds}\mid \varid{k}\in \varid{dom}\;\varid{alts}{}\<[30]%
\>[30]{}\to (\varid{alts}\mathop{!}\varid{k})\;\varid{ds}{}\<[E]%
\\
\>[3]{}\hsindent{2}{}\<[5]%
\>[5]{}\anonymous {}\<[30]%
\>[30]{}\to \varid{stuck}{}\<[E]%
\\[\blanklineskip]%
\>[B]{}\keyword{instance}\;\conid{HasBind}\;\conid{D}_{\mathbf{na}}\;\keyword{where}{}\<[E]%
\\
\>[B]{}\hsindent{3}{}\<[3]%
\>[3]{}\varid{bind}\;\varid{rhs}\;\varid{body}\mathrel{=}\keyword{let}\;\varid{d}\mathrel{=}\varid{rhs}\;\varid{d}\;\keyword{in}\;\varid{body}\;\varid{d}{}\<[E]%
\ColumnHook
\end{hscode}\resethooks
\\[-2.5em]
\subcaption{Concrete by-name semantics for \ensuremath{\conid{D}_{\mathbf{na}}}}
  \label{fig:trace-instances}
\end{minipage}%
\\[-0.5em]
\caption{Denotational Interpreter}
  \label{fig:eval}
\end{figure}

\subsection{The Interpreter}

Traditionally, a denotational semantics is expressed as a mathematical function,
often written \ensuremath{\denot{\varid{e}}_{\varcolor{\rho}}}, to give an expression \ensuremath{\varid{e}\mathbin{::}\conid{Exp}} a meaning, or
\emph{denotation}, in terms of some semantic domain \ensuremath{\conid{D}}.
The environment \ensuremath{\varcolor{\rho}\mathbin{::}\conid{Name}\mathbin{:\rightharpoonup}\conid{D}} gives meaning to the free variables of \ensuremath{\varid{e}},
by mapping each free variable to its denotation in \ensuremath{\conid{D}}.
We sketch the Haskell encoding of \ensuremath{\conid{Exp}} in \Cref{fig:syntax} and the API of
environments and sets in \Cref{fig:map}.
For concise notation, we will use a small number of infix operators: \ensuremath{(\mathbin{:\rightharpoonup})} as
a synonym for finite \ensuremath{\conid{Map}}s, with \ensuremath{\varid{m}\mathop{!}\varid{x}} for looking up \ensuremath{\varid{x}} in \ensuremath{\varid{m}}, \ensuremath{\varcolor{\varepsilon}} for
the empty map, \ensuremath{\varid{m}[\varid{x}\mapsto\varid{d}]} for updates, \ensuremath{\varid{assocs}\;\varid{m}} for a list of key-value pairs
in \ensuremath{\varid{m}}, \ensuremath{\varid{f}\mathbin{\lhd}\varid{m}} for mapping \ensuremath{\varid{f}} over every value in \ensuremath{\varid{m}}, \ensuremath{\varid{dom}\;\varid{m}} for the set of
keys present in the map, and \ensuremath{(\in )} for membership tests in that set.

Our denotational interpreter \ensuremath{\mathcal{S}\denot{\wild}_{\wild}\mathbin{::}\conid{Exp}\to (\conid{Name}\mathbin{:\rightharpoonup}\conid{D}_{\mathbf{na}})\to \conid{D}_{\mathbf{na}}} can
have a similar type as \ensuremath{\denot{\wild}_{\wild}}.
However, to derive both dynamic semantics and static analysis as instances of the same
generic interpreter \ensuremath{\mathcal{S}\denot{\wild}_{\wild}}, we need to vary the type of its semantic domain,
which is naturally expressed using type class overloading, thus:
\[
\ensuremath{\mathcal{S}\denot{\wild}_{\wild}\mathbin{::}(\conid{Trace}\;\varid{d},\conid{Domain}\;\varid{d},\conid{HasBind}\;\varid{d})\Rightarrow \conid{Exp}\to (\conid{Name}\mathbin{:\rightharpoonup}\varid{d})\to \varid{d}}.
\]
We have parameterised the semantic domain \ensuremath{\varid{d}} over three type classes \ensuremath{\conid{Trace}}, \ensuremath{\conid{Domain}} and \ensuremath{\conid{HasBind}}, whose signatures are given in \Cref{fig:trace-classes}.
Each of the three type classes offer knobs that we will tweak to derive
different evaluation strategies as well as static analyses.

\Cref{fig:eval} gives the complete definition of \ensuremath{\mathcal{S}\denot{\wild}_{\wild}} together with instances for domain \ensuremath{\conid{D}_{\mathbf{na}}} that we introduced in \Cref{sec:dna}.
Together this is enough to actually run the denotational interpreter to produce traces.
We use \ensuremath{\varid{read}\mathbin{::}\conid{String}\to \conid{Exp}} as a parsing function and a \ensuremath{\conid{Show}} instance for
\ensuremath{\conid{D}\;\varcolor{\tau}} that displays traces.
For example, we can evaluate the expression $\Let{i}{\Lam{x}{x}}{i~i}$ like
this:
\begin{hscode}\SaveRestoreHook
\column{B}{@{}>{\hspre}l<{\hspost}@{}}%
\column{3}{@{}>{\hspre}l<{\hspost}@{}}%
\column{E}{@{}>{\hspre}l<{\hspost}@{}}%
\>[3]{}\lambda\!\!\mathbin{>}\mathcal{S}\denot{\varid{read}\;\text{\ttfamily \char34 let~i~=~λx.x~in~i~i\char34}}_{\varcolor{\varepsilon}}\mathbin{::}\conid{D}_{\mathbf{na}}{}\<[E]%
\ColumnHook
\end{hscode}\resethooks
$
\LetIT\xhookrightarrow{\hspace{1.1ex}}\AppIT\xhookrightarrow{\hspace{1.1ex}}\LookupT(i)\xhookrightarrow{\hspace{1.1ex}}\AppET\xhookrightarrow{\hspace{1.1ex}}\LookupT(i)\xhookrightarrow{\hspace{1.1ex}}\langle \lambda\rangle 
$,
\\[\belowdisplayskip]
\noindent
where $\langle\lambda\rangle$
means that the trace ends in a \ensuremath{\conid{Fun}} value.
We cannot generally print \ensuremath{\conid{D}_{\mathbf{na}}} or \ensuremath{\conid{Fun}}ctions thereof, but in this case the result would be the value $\Lam{x}{x}$.
This is in direct correspondence to the earlier call-by-name small-step trace
\labelcref{ex:trace} in \Cref{sec:op-sem}.

The definition of \ensuremath{\mathcal{S}\denot{\wild}_{\wild}}, given in \Cref{fig:eval}, is by structural recursion over the input expression.
For example, to get the denotation of \ensuremath{\conid{Lam}\;\varid{x}\;\varid{body}}, we must recursively invoke \ensuremath{\mathcal{S}\denot{\wild}_{\wild}} on \ensuremath{\varid{body}}, extending the environment to bind \ensuremath{\varid{x}} to its denotation.
We wrap that body denotation in \ensuremath{\varid{step}\;\conid{App}_{2}}, to prefix the trace of \ensuremath{\varid{body}} with an \ensuremath{\conid{App}_{2}} event whenever the function is invoked, where \ensuremath{\varid{step}} is a method of class \ensuremath{\conid{Trace}}.
Finally, we use \ensuremath{\varid{fun}} to build the returned denotation; the details necessarily depend on the \ensuremath{\conid{Domain}}, so \ensuremath{\varid{fun}} is a method of class \ensuremath{\conid{Domain}}.
While the lambda-bound \ensuremath{\varid{x}\mathbin{::}\conid{Name}} passed to \ensuremath{\varid{fun}} is ignored in the
\ensuremath{\conid{Domain}\;\conid{D}_{\mathbf{na}}} instance of the concrete by-name semantics, it is useful for
abstract domains such as that of usage analysis (\Cref{sec:abstraction}).
(We refrain from passing field binders or other syntactic tokens in \ensuremath{\varid{select}}
and let binders in \ensuremath{\varid{bind}} as well, because the analyses considered do not need
them.)
The other cases follow a similar pattern; they each do some work, before handing
off to type class methods to do the domain-specific work.

The \ensuremath{\conid{HasBind}} type class defines a particular \emph{evaluation strategy}, as we shall see in
\Cref{sec:evaluation-strategies}.
The \ensuremath{\varid{bind}} method of \ensuremath{\conid{HasBind}} is used to give meaning to recursive let
bindings:
it takes two functionals
for building the denotation of the right-hand side and that of the let body,
given a denotation for the right-hand side.
The concrete implementation for \ensuremath{\varid{bind}} given in \Cref{fig:trace-instances}
computes a \ensuremath{\varid{d}} such that \ensuremath{\varid{d}\mathrel{=}\varid{rhs}\;\varid{d}} and passes the recursively-defined \ensuremath{\varid{d}} to
\ensuremath{\varid{body}}.%
\footnote{Such a \ensuremath{\varid{d}} corresponds to the \emph{guarded fixpoint} of \ensuremath{\varid{rhs}}.
Strict languages can define this fixpoint as \ensuremath{\varid{d}\;()\mathrel{=}\varid{rhs}\;(\varid{d}\;())}.}
Doing so yields a call-by-name evaluation strategy, because the trace \ensuremath{\varid{d}}
will be unfolded at every occurrence of \ensuremath{\varid{x}} in the right-hand side \ensuremath{\varid{e}_{1}}.
We will shortly see examples of eager evaluation strategies that will yield from
\ensuremath{\varid{d}} inside \ensuremath{\varid{bind}} instead of calling \ensuremath{\varid{body}} immediately.

We conclude this subsection with a few examples.
First we demonstrate that our interpreter is \emph{productive}:
we can observe prefixes of diverging traces without risking a looping
interpreter.
To observe prefixes, we use a function \ensuremath{\varid{takeT}\mathbin{::}\conid{Int}\to \conid{T}\;\varid{v}\to \conid{T}\;(\conid{Maybe}\;\varid{v})}:
\ensuremath{\varid{takeT}\;\varid{n}\;\varcolor{\tau}} returns the first \ensuremath{\varid{n}} steps of \ensuremath{\varcolor{\tau}} and replaces the final value
with \ensuremath{\conid{Nothing}} (printed as $...$) if it goes on for longer.
\begin{hscode}\SaveRestoreHook
\column{B}{@{}>{\hspre}l<{\hspost}@{}}%
\column{3}{@{}>{\hspre}l<{\hspost}@{}}%
\column{E}{@{}>{\hspre}l<{\hspost}@{}}%
\>[3]{}\lambda\!\!\mathbin{>}\varid{takeT}\;\mathrm{5}\mathbin{\$}\mathcal{S}\denot{\varid{read}\;\text{\ttfamily \char34 let~x~=~x~in~x\char34}}_{\varcolor{\varepsilon}}\mathbin{::}\conid{T}\;(\conid{Maybe}\;(\conid{Value}\;\conid{T})){}\<[E]%
\ColumnHook
\end{hscode}\resethooks
$
\LetIT\xhookrightarrow{\hspace{1.1ex}}\LookupT(x)\xhookrightarrow{\hspace{1.1ex}}\LookupT(x)\xhookrightarrow{\hspace{1.1ex}}\LookupT(x)\xhookrightarrow{\hspace{1.1ex}}\LookupT(x)\xhookrightarrow{\hspace{1.1ex}}...
$
\begin{hscode}\SaveRestoreHook
\column{B}{@{}>{\hspre}l<{\hspost}@{}}%
\column{3}{@{}>{\hspre}l<{\hspost}@{}}%
\column{E}{@{}>{\hspre}l<{\hspost}@{}}%
\>[3]{}\lambda\!\!\mathbin{>}\varid{takeT}\;\mathrm{9}\mathbin{\$}\mathcal{S}\denot{\varid{read}\;\text{\ttfamily \char34 let~w~=~λy.~y~y~in~w~w\char34}}_{\varcolor{\varepsilon}}\mathbin{::}\conid{T}\;(\conid{Maybe}\;(\conid{Value}\;\conid{T})){}\<[E]%
\ColumnHook
\end{hscode}\resethooks
$
\LetIT\xhookrightarrow{\hspace{1.1ex}}\AppIT\xhookrightarrow{\hspace{1.1ex}}\LookupT(w)\xhookrightarrow{\hspace{1.1ex}}\AppET\xhookrightarrow{\hspace{1.1ex}}\AppIT\xhookrightarrow{\hspace{1.1ex}}\LookupT(w)\xhookrightarrow{\hspace{1.1ex}}\AppET\xhookrightarrow{\hspace{1.1ex}}\AppIT\xhookrightarrow{\hspace{1.1ex}}\LookupT(w)\xhookrightarrow{\hspace{1.1ex}}...
$
\\[\belowdisplayskip]
\noindent
The reason \ensuremath{\mathcal{S}\denot{\wild}_{\wild}} is productive is due to the coinductive nature of \ensuremath{\conid{T}}'s
definition in Haskell.%
\footnote{In a strict language, we need to introduce a thunk in
the definition of \ensuremath{\conid{Step}}, \eg \text{\ttfamily Step~of~event~\char42{}~\char40{}unit~\char45{}\char62{}~\char39{}a~t\char41{}}.}
Productivity requires that the monadic bind operator \ensuremath{(\bind )} for \ensuremath{\conid{T}}
guards the recursion, as in the delay monad of \citet{Capretta:05}.

Data constructor values are printed as $Con(K)$, where $K$ indicates the
\ensuremath{\conid{Tag}}.
Data types allow for interesting ways (type errors) to get \ensuremath{\conid{Stuck}} (\ie the
\textbf{wrong} value of \citet{Milner:78}), printed as $\lightning$:
\begin{hscode}\SaveRestoreHook
\column{B}{@{}>{\hspre}l<{\hspost}@{}}%
\column{3}{@{}>{\hspre}l<{\hspost}@{}}%
\column{E}{@{}>{\hspre}l<{\hspost}@{}}%
\>[3]{}\lambda\!\!\mathbin{>}\mathcal{S}\denot{\varid{read}\;\text{\ttfamily \char34 let~zro~=~Z()~in~let~one~=~S(zro)~in~case~one~of~\char123 ~S(z)~->~z~\char125 \char34}}_{\varcolor{\varepsilon}}\mathbin{::}\conid{D}_{\mathbf{na}}{}\<[E]%
\ColumnHook
\end{hscode}\resethooks
$
\LetIT\xhookrightarrow{\hspace{1.1ex}}\LetIT\xhookrightarrow{\hspace{1.1ex}}\CaseIT\xhookrightarrow{\hspace{1.1ex}}\LookupT(one)\xhookrightarrow{\hspace{1.1ex}}\CaseET\xhookrightarrow{\hspace{1.1ex}}\LookupT(zro)\xhookrightarrow{\hspace{1.1ex}}\langle \mathit{Con}(Z)\rangle 
$
\begin{hscode}\SaveRestoreHook
\column{B}{@{}>{\hspre}l<{\hspost}@{}}%
\column{3}{@{}>{\hspre}l<{\hspost}@{}}%
\column{E}{@{}>{\hspre}l<{\hspost}@{}}%
\>[3]{}\lambda\!\!\mathbin{>}\mathcal{S}\denot{\varid{read}\;\text{\ttfamily \char34 let~zro~=~Z()~in~zro~zro\char34}}_{\varcolor{\varepsilon}}\mathbin{::}\conid{D}_{\mathbf{na}}{}\<[E]%
\ColumnHook
\end{hscode}\resethooks
$
\LetIT\xhookrightarrow{\hspace{1.1ex}}\AppIT\xhookrightarrow{\hspace{1.1ex}}\LookupT(zro)\xhookrightarrow{\hspace{1.1ex}}\langle \lightning\rangle 
$

\subsection{More Evaluation Strategies}
\label{sec:evaluation-strategies}


By varying the \ensuremath{\conid{HasBind}} instance of our type \ensuremath{\conid{D}}, we can endow our language
\ensuremath{\conid{Exp}} with different evaluation strategies.
The appeal of that is, firstly, that it is possible to do so without changing
the interpreter definition, supporting the claim that the denotational
interpreter design pattern is equally suited to model lazy as well as strict
semantics.
More importantly, in order to prove usage analysis sound \wrt by-need evaluation
in \Cref{sec:soundness}, we need to define a semantic domain for call-by-need!
It turns out that the interpreter thus derived is the --- to our knowledge ---
first provably adequate denotational semantics for call-by-need (\Cref{sec:adequacy}).

Although the main body discusses only by-name and by-need semantics,
we provide instances for call-by-value as well as clairvoyant
semantics~\citep{HackettHutton:19} in \Cref{sec:more-eval-strat} as well.

\subsubsection{Call-by-name}

Following a similar approach as~\citet{adi}, we maximise reuse by instantiating
the same \ensuremath{\conid{D}} at different wrappers of \ensuremath{\conid{T}}, rather than reinventing \ensuremath{\conid{Value}} and \ensuremath{\conid{T}}.
Hence we redefine
by-name semantics via the following \ensuremath{\conid{ByName}} newtype wrapper:
\begin{hscode}\SaveRestoreHook
\column{B}{@{}>{\hspre}l<{\hspost}@{}}%
\column{E}{@{}>{\hspre}l<{\hspost}@{}}%
\>[B]{}\mathcal{S}_{\mathbf{name}}\denot{\varid{e}}_{\varcolor{\rho}}\mathrel{=}\mathcal{S}\denot{\varid{e}}_{\varcolor{\rho}}\mathbin{::}\conid{D}\;(\conid{ByName}\;\conid{T}){}\<[E]%
\\
\>[B]{}\keyword{newtype}\;\conid{ByName}\;\varcolor{\tau}\;\varid{v}\mathrel{=}\conid{ByName}\;\{\mskip1.5mu \varid{unByName}\mathbin{::}\varcolor{\tau}\;\varid{v}\mskip1.5mu\}\;\keyword{deriving}\;(\conid{Monad},\conid{Trace}){}\<[E]%
\\
\>[B]{}\keyword{instance}\;\conid{HasBind}\;(\conid{D}\;(\conid{ByName}\;\varcolor{\tau}))\;\keyword{where}\;\varid{bind}\;\varid{rhs}\;\varid{body}\mathrel{=}\keyword{let}\;\varid{d}\mathrel{=}\varid{rhs}\;\varid{d}\;\keyword{in}\;\varid{body}\;\varid{d}{}\<[E]%
\ColumnHook
\end{hscode}\resethooks
We call \ensuremath{\conid{ByName}\;\varcolor{\tau}} a \emph{trace transformer} because it inherits its \ensuremath{\conid{Monad}}
and \ensuremath{\conid{Trace}} instance from \ensuremath{\varcolor{\tau}}, in reminiscence to Galois transformers~\citep{Darais:15}.
The old \ensuremath{\conid{D}_{\mathbf{na}}} can be recovered as \ensuremath{\conid{D}\;(\conid{ByName}\;\conid{T})} and we refer to its
interpreter instance as \ensuremath{\mathcal{S}_{\mathbf{name}}\denot{\varid{e}}_{\varcolor{\rho}}}.

\subsubsection{Call-by-need}
\label{sec:by-need-instance}

\begin{figure}
\begin{hscode}\SaveRestoreHook
\column{B}{@{}>{\hspre}l<{\hspost}@{}}%
\column{3}{@{}>{\hspre}l<{\hspost}@{}}%
\column{7}{@{}>{\hspre}l<{\hspost}@{}}%
\column{10}{@{}>{\hspre}l<{\hspost}@{}}%
\column{21}{@{}>{\hspre}l<{\hspost}@{}}%
\column{25}{@{}>{\hspre}l<{\hspost}@{}}%
\column{57}{@{}>{\hspre}l<{\hspost}@{}}%
\column{65}{@{}>{\hspre}l<{\hspost}@{}}%
\column{E}{@{}>{\hspre}l<{\hspost}@{}}%
\>[B]{}\keyword{type}\;\conid{Addr}\mathrel{=}\conid{Int};\keyword{type}\;\conid{Heap}\;\varcolor{\tau}\mathrel{=}\conid{Addr}\mathbin{:\rightharpoonup}\conid{D}\;\varcolor{\tau};\varid{nextFree}\mathbin{::}\conid{Heap}\;\varcolor{\tau}\to \conid{Addr}{}\<[E]%
\\
\>[B]{}\keyword{newtype}\;\conid{ByNeed}\;\varcolor{\tau}\;\varid{v}\mathrel{=}\conid{ByNeed}\;\{\mskip1.5mu \varid{unByNeed}\mathbin{::}\conid{Heap}\;(\conid{ByNeed}\;\varcolor{\tau})\to \varcolor{\tau}\;(\varid{v},\conid{Heap}\;(\conid{ByNeed}\;\varcolor{\tau}))\mskip1.5mu\}{}\<[E]%
\\[\blanklineskip]%
\>[B]{}\keyword{type}\;\conid{D}_{\mathbf{ne}}\mathrel{=}\conid{D}\;(\conid{ByNeed}\;\conid{T});\keyword{type}\;\conid{Value}_{\mathbf{ne}}\mathrel{=}\conid{Value}\;(\conid{ByNeed}\;\conid{T});\keyword{type}\;\conid{Heap}_{\mathbf{ne}}\mathrel{=}\conid{Heap}\;(\conid{ByNeed}\;\conid{T}){}\<[E]%
\\
\>[B]{}\mathcal{S}_{\mathbf{need}}\denot{\varid{e}}_{\varcolor{\rho}}(\varcolor{\mu})\mathrel{=}\varid{unByNeed}\;(\mathcal{S}\denot{\varid{e}}_{\varcolor{\rho}}\mathbin{::}\conid{D}_{\mathbf{ne}})\;\varcolor{\mu}{}\<[E]%
\\[\blanklineskip]%
\>[B]{}\varid{get}{}\<[7]%
\>[7]{}\mathbin{::}\conid{Monad}\;\varcolor{\tau}\Rightarrow \conid{ByNeed}\;\varcolor{\tau}\;(\conid{Heap}\;(\conid{ByNeed}\;\varcolor{\tau}));{}\<[57]%
\>[57]{}\varid{get}{}\<[65]%
\>[65]{}\mathrel{=}\conid{ByNeed}\;(\lambda \varcolor{\mu}\to \varid{return}\;(\varcolor{\mu},\varcolor{\mu})){}\<[E]%
\\
\>[B]{}\varid{put}{}\<[7]%
\>[7]{}\mathbin{::}\conid{Monad}\;\varcolor{\tau}\Rightarrow \conid{Heap}\;(\conid{ByNeed}\;\varcolor{\tau})\to \conid{ByNeed}\;\varcolor{\tau}\;();\;{}\<[57]%
\>[57]{}\varid{put}\;\varcolor{\mu}{}\<[65]%
\>[65]{}\mathrel{=}\conid{ByNeed}\;(\lambda \anonymous \to \varid{return}\;((),\varcolor{\mu})){}\<[E]%
\\
\>[B]{}\keyword{instance}\;\conid{Monad}\;\varcolor{\tau}\Rightarrow \conid{Monad}\;(\conid{ByNeed}\;\varcolor{\tau})\;\keyword{where}\mathbin{...}{}\<[E]%
\\[\blanklineskip]%
\>[B]{}\keyword{instance}\;(\keyword{\forall}\!\! \hsforall \;\varid{v}\hsdot{\circ }{.\ }\conid{Trace}\;(\varcolor{\tau}\;\varid{v}))\Rightarrow \conid{Trace}\;(\conid{ByNeed}\;\varcolor{\tau}\;\varid{v})\;\keyword{where}\;\varid{step}\;\varid{e}\;\varid{m}\mathrel{=}\conid{ByNeed}\;(\varid{step}\;\varid{e}\hsdot{\circ }{.\ }\varid{unByNeed}\;\varid{m}){}\<[E]%
\\[\blanklineskip]%
\>[B]{}\varid{fetch}\mathbin{::}\conid{Monad}\;\varcolor{\tau}\Rightarrow \conid{Addr}\to \conid{D}\;(\conid{ByNeed}\;\varcolor{\tau});\varid{fetch}\;\varid{a}\mathrel{=}\varid{get}\bind \lambda \varcolor{\mu}\to \varcolor{\mu}\mathop{!}\varid{a}{}\<[E]%
\\[\blanklineskip]%
\>[B]{}\varid{memo}\mathbin{::}\keyword{\forall}\!\! \hsforall \;\varcolor{\tau}\hsdot{\circ }{.\ }(\conid{Monad}\;\varcolor{\tau},\keyword{\forall}\!\! \hsforall \;\varid{v}\hsdot{\circ }{.\ }\conid{Trace}\;(\varcolor{\tau}\;\varid{v}))\Rightarrow \conid{Addr}\to \conid{D}\;(\conid{ByNeed}\;\varcolor{\tau})\to \conid{D}\;(\conid{ByNeed}\;\varcolor{\tau}){}\<[E]%
\\
\>[B]{}\varid{memo}\;\varid{a}\;\varid{d}\mathrel{=}\varid{d}\bind \lambda \varid{v}\to \conid{ByNeed}\;(\varid{upd}\;\varid{v}){}\<[E]%
\\
\>[B]{}\hsindent{3}{}\<[3]%
\>[3]{}\keyword{where}\;{}\<[10]%
\>[10]{}\varid{upd}\;\conid{Stuck}\;{}\<[21]%
\>[21]{}\varcolor{\mu}\mathrel{=}\varid{return}\;(\conid{Stuck}\mathbin{::}\conid{Value}\;(\conid{ByNeed}\;\varcolor{\tau}),\varcolor{\mu}){}\<[E]%
\\
\>[10]{}\varid{upd}\;\varid{v}\;{}\<[21]%
\>[21]{}\varcolor{\mu}\mathrel{=}\varid{step}\;\conid{Upd}\;(\varid{return}\;(\varid{v},\varcolor{\mu}[\varid{a}\mapsto\varid{memo}\;\varid{a}\;(\varid{return}\;\varid{v})])){}\<[E]%
\\[\blanklineskip]%
\>[B]{}\keyword{instance}\;(\conid{Monad}\;\varcolor{\tau},\keyword{\forall}\!\! \hsforall \;\varid{v}\hsdot{\circ }{.\ }\conid{Trace}\;(\varcolor{\tau}\;\varid{v}))\Rightarrow \conid{HasBind}\;(\conid{D}\;(\conid{ByNeed}\;\varcolor{\tau}))\;\keyword{where}{}\<[E]%
\\
\>[B]{}\hsindent{3}{}\<[3]%
\>[3]{}\varid{bind}\;\varid{rhs}\;\varid{body}\mathrel{=}\keyword{do}\;{}\<[25]%
\>[25]{}\varcolor{\mu}\leftarrow \varid{get}{}\<[E]%
\\
\>[25]{}\keyword{let}\;\varid{a}\mathrel{=}\varid{nextFree}\;\varcolor{\mu}{}\<[E]%
\\
\>[25]{}\varid{put}\;\varcolor{\mu}[\varid{a}\mapsto\varid{memo}\;\varid{a}\;(\varid{rhs}\;(\varid{fetch}\;\varid{a}))]{}\<[E]%
\\
\>[25]{}\varid{body}\;(\varid{fetch}\;\varid{a}){}\<[E]%
\ColumnHook
\end{hscode}\resethooks
\\[-1em]
\caption{Call-by-need}
\label{fig:by-need}
\end{figure}

The use of a stateful heap is essential to the call-by-need evaluation strategy
in order to enable memoisation.
So how do we vary \ensuremath{\varcolor{\theta}} such that \ensuremath{\conid{D}\;\varcolor{\theta}} accommodates state?
We certainly cannot perform the heap update by updating entries in \ensuremath{\varcolor{\rho}},
because those entries are immutable once inserted, and we do not want to change
the generic interpreter.
That rules out $\ensuremath{\varcolor{\theta}} \cong \ensuremath{\conid{T}}$ (as for \ensuremath{\conid{ByName}\;\conid{T}}), because then repeated
occurrences of the variable \ensuremath{\varid{x}} must yield the same trace \ensuremath{\varcolor{\rho}\mathop{!}\varid{x}}.
However, the whole point of memoisation is that every evaluation of \ensuremath{\varid{x}} after
the first one leads to a potentially different, shorter trace.
This implies we have to \emph{paramaterise} every occurrence of \ensuremath{\varid{x}} over the
current heap \ensuremath{\varcolor{\mu}} at the time of evaluation, and every evaluation of \ensuremath{\varid{x}} must
subsequently update this heap with its value, so that the next evaluation
of \ensuremath{\varid{x}} returns the value directly.
In other words, we need a representation $\ensuremath{\conid{D}\;\varcolor{\theta}} \cong \ensuremath{\conid{Heap}\to \conid{T}\;(\conid{Value}\;\varcolor{\theta},\conid{Heap})}$.

Our trace transformer \ensuremath{\conid{ByNeed}} in \Cref{fig:by-need} solves this type equation
via \ensuremath{\varcolor{\theta}\triangleq\conid{ByNeed}\;\conid{T}}.
It embeds a standard state transformer monad,
whose key operations \ensuremath{\varid{get}} and \ensuremath{\varid{put}} are given in \Cref{fig:by-need}.

So the denotation of an expression is no longer a trace, but rather a
\emph{stateful function returning a trace} with state \ensuremath{\conid{Heap}\;(\conid{ByNeed}\;\varcolor{\tau})} in
which to allocate call-by-need thunks.
The \ensuremath{\conid{Trace}} instance of \ensuremath{\conid{ByNeed}\;\varcolor{\tau}} simply forwards to that of \ensuremath{\varcolor{\tau}} (\ie often
\ensuremath{\conid{T}}), pointwise over heaps.
Doing so needs a \ensuremath{\conid{Trace}} instance for \ensuremath{\varcolor{\tau}\;(\conid{Value}\;(\conid{ByNeed}\;\varcolor{\tau}),\conid{Heap}\;(\conid{ByNeed}\;\varcolor{\tau}))}, but we
found it more succinct to use a quantified constraint \ensuremath{(\keyword{\forall}\!\! \hsforall \;\varid{v}\hsdot{\circ }{.\ }\conid{Trace}\;(\varcolor{\tau}\;\varid{v}))}, that is, we require a \ensuremath{\conid{Trace}\;(\varcolor{\tau}\;\varid{v})} instance for every choice of \ensuremath{\varid{v}}.
Given that \ensuremath{\varcolor{\tau}} must also be a \ensuremath{\conid{Monad}}, that is not an onerous requirement.

The key part is again the implementation of \ensuremath{\conid{HasBind}} for \ensuremath{\conid{D}\;(\conid{ByNeed}\;\varcolor{\tau})},
because that is the only place where thunks are allocated.
The implementation of \ensuremath{\varid{bind}} designates a fresh heap address \ensuremath{\varid{a}}
to hold the denotation of the right-hand side.
Both \ensuremath{\varid{rhs}} and \ensuremath{\varid{body}} are called with \ensuremath{\varid{fetch}\;\varid{a}}, a denotation that looks up \ensuremath{\varid{a}}
in the heap and runs it.
If we were to omit the \ensuremath{\varid{memo}\;\varid{a}} action explained next, we would thus have
recovered another form of call-by-name semantics based on mutable state instead
of guarded fixpoints such as in \ensuremath{\conid{ByName}} and \ensuremath{\conid{ByValue}}.
The whole purpose of the \ensuremath{\varid{memo}\;\varid{a}\;\varid{d}} combinator then is to \emph{memoise} the
computation of \ensuremath{\varid{d}} the first time we run the computation, via \ensuremath{\varid{fetch}\;\varid{a}} in the
\ensuremath{\conid{Var}} case of \ensuremath{\mathcal{S}_{\mathbf{need}}\denot{\wild}_{\wild}}.
So \ensuremath{\varid{memo}\;\varid{a}\;\varid{d}} yields from \ensuremath{\varid{d}} until it has reached a value, and then \ensuremath{\varid{upd}}ates
the heap after an additional \ensuremath{\conid{Upd}} step.
Repeated access to the same variable will run the replacement \ensuremath{\varid{memo}\;\varid{a}\;(\varid{return}\;\varid{v})}, which immediately yields \ensuremath{\varid{v}} after performing a \ensuremath{\varid{step}\;\conid{Upd}} that does
nothing.%
\footnote{More serious semantics would omit updates after the first
evaluation as an \emph{optimisation}, \ie update with \ensuremath{\varcolor{\mu}[\varid{a}\mapsto\varid{return}\;\varid{v}]},
but doing so complicates relating the semantics to \Cref{fig:lk-semantics},
where omission of update frames for values behaves differently.
For now, our goal is not to formalise this optimisation, but rather to show
adequacy \wrt an established semantics.}

Although the code is carefully written, it is worth stressing how
compact and expressive it is.
We were able to move from traces to stateful traces just by wrapping traces \ensuremath{\conid{T}}
in a state transformer \ensuremath{\conid{ByNeed}}, without modifying the main \ensuremath{\mathcal{S}\denot{\wild}_{\wild}} function at
all.
In doing so, we provide the simplest encoding of a denotational by-need semantics
that we know of.%

Here is an example evaluating $\Let{i}{(\Lam{y}{\Lam{x}{x}})~i}{i~i}$, starting
in an empty \hypertarget{ex:eval-need-trace2}{heap}:
\begin{hscode}\SaveRestoreHook
\column{B}{@{}>{\hspre}l<{\hspost}@{}}%
\column{3}{@{}>{\hspre}l<{\hspost}@{}}%
\column{E}{@{}>{\hspre}l<{\hspost}@{}}%
\>[3]{}\lambda\!\!\mathbin{>}\mathcal{S}_{\mathbf{need}}\denot{\varid{read}\;\text{\ttfamily \char34 let~i~=~(λy.λx.x)~i~in~i~i\char34}}_{\varcolor{\varepsilon}}(\varcolor{\varepsilon})\mathbin{::}\conid{T}\;(\conid{Value}_{\mathbf{ne}},\conid{Heap}_{\mathbf{ne}}){}\<[E]%
\ColumnHook
\end{hscode}\resethooks
$
\LetIT\xhookrightarrow{\hspace{1.1ex}}\AppIT\xhookrightarrow{\hspace{1.1ex}}\LookupT(i)\xhookrightarrow{\hspace{1.1ex}}\AppIT\xhookrightarrow{\hspace{1.1ex}}\AppET\xhookrightarrow{\hspace{1.1ex}}\UpdateT\xhookrightarrow{\hspace{1.1ex}}\AppET\xhookrightarrow{\hspace{1.1ex}}\LookupT(i)\xhookrightarrow{\hspace{1.1ex}}\UpdateT\xhookrightarrow{\hspace{1.1ex}}\langle (\lambda,[0\!\! \mapsto \!\! \wild])\rangle 
$
\\[\belowdisplayskip]
\noindent
This trace is in clear correspondence to the earlier by-need LK trace
\labelcref{ex:trace2}.
We can observe memoisation at play:
Between the first bracket of $\LookupT$ and $\UpdateT$ events, the heap entry
for $i$ goes through a beta reduction before producing a value.
This work is cached, so that the second $\LookupT$ bracket does not do any beta
reduction.

\begin{toappendix}
\label{sec:more-eval-strat}

To show that our denotational interpreter pattern equally well applies to
by-value evaluation strategies, we introduce three more concrete semantic domain
instances for call-by-value in this section.
The first one is a plain old by-value encoding the representation of which is
isomorphic to \ensuremath{\conid{D}\;\conid{T}}, just like for \ensuremath{\conid{D}_{\mathbf{na}}}.
However, this instance is partial for the original recursive \ensuremath{\conid{Let}} construct.
In order to recover a total by-value semantics, our second instance augments
call-by-value with a lazy initialisation technique~\citep{Nakata:06} involving a
mutable heap, thus sharing its representation with \ensuremath{\conid{D}_{\mathbf{ne}}}.
The third and final by-value domain models clairvoyant
call-by-value~\citep{HackettHutton:19}, which unfortunately proves to be
partial, and more fundamentally so than the partial by-value instance.

\subsection{Call-by-value}

\begin{figure}
\begin{hscode}\SaveRestoreHook
\column{B}{@{}>{\hspre}l<{\hspost}@{}}%
\column{3}{@{}>{\hspre}l<{\hspost}@{}}%
\column{36}{@{}>{\hspre}l<{\hspost}@{}}%
\column{41}{@{}>{\hspre}l<{\hspost}@{}}%
\column{69}{@{}>{\hspre}l<{\hspost}@{}}%
\column{E}{@{}>{\hspre}l<{\hspost}@{}}%
\>[B]{}\mathcal{S}_{\mathbf{value}}\denot{\varid{e}}_{\varcolor{\rho}}\mathrel{=}\mathcal{S}\denot{\varid{e}}_{\varcolor{\rho}}\mathbin{::}\conid{D}\;(\conid{ByValue}\;\conid{T}){}\<[E]%
\\[\blanklineskip]%
\>[B]{}\keyword{newtype}\;\conid{ByValue}\;\varcolor{\tau}\;\varid{v}\mathrel{=}\conid{ByValue}\;\{\mskip1.5mu \varid{unByValue}\mathbin{::}\varcolor{\tau}\;\varid{v}\mskip1.5mu\}{}\<[E]%
\\
\>[B]{}\keyword{instance}\;\conid{Monad}\;\varcolor{\tau}\Rightarrow \conid{Monad}\;(\conid{ByValue}\;\varcolor{\tau})\;\keyword{where}\mathbin{...}{}\<[E]%
\\
\>[B]{}\keyword{instance}\;\conid{Trace}\;(\varcolor{\tau}\;\varid{v})\Rightarrow \conid{Trace}\;(\conid{ByValue}\;\varcolor{\tau}\;\varid{v})\;\keyword{where}\; ... {}\<[E]%
\\[\blanklineskip]%
\>[B]{}\keyword{class}\;\conid{Extract}\;\varcolor{\tau}\;\keyword{where}\;\varid{getValue}\mathbin{::}\varcolor{\tau}\;\varid{v}\to \varid{v}{}\<[E]%
\\
\>[B]{}\keyword{instance}\;\conid{Extract}\;\conid{T}\;\keyword{where}\;\varid{getValue}\;(\conid{Ret}\;\varid{v})\mathrel{=}\varid{v};\varid{getValue}\;(\conid{Step}\;\anonymous \;\varcolor{\tau})\mathrel{=}\varid{getValue}\;\varcolor{\tau}{}\<[E]%
\\[\blanklineskip]%
\>[B]{}\keyword{instance}\;(\conid{Trace}\;(\conid{D}\;(\conid{ByValue}\;\varcolor{\tau})),\conid{Monad}\;\varcolor{\tau},\conid{Extract}\;\varcolor{\tau})\Rightarrow \conid{HasBind}\;(\conid{D}\;(\conid{ByValue}\;\varcolor{\tau}))\;\keyword{where}{}\<[E]%
\\
\>[B]{}\hsindent{3}{}\<[3]%
\>[3]{}\varid{bind}\;\varid{rhs}\;\varid{body}\mathrel{=}\varid{step}\;\conid{Let}_{0}\;(\keyword{do}\;{}\<[36]%
\>[36]{}\keyword{let}\;{}\<[41]%
\>[41]{}\varid{d}\mathrel{=}\varid{rhs}\;(\varid{return}\;\varid{v}){}\<[69]%
\>[69]{}\mathbin{::}\conid{D}\;(\conid{ByValue}\;\varcolor{\tau}){}\<[E]%
\\
\>[41]{}\varid{v}\mathrel{=}\varid{getValue}\;(\varid{unByValue}\;\varid{d}){}\<[69]%
\>[69]{}\mathbin{::}\conid{Value}\;(\conid{ByValue}\;\varcolor{\tau}){}\<[E]%
\\
\>[36]{}\varid{v}_{1}\leftarrow \varid{d};\varid{body}\;(\varid{return}\;\varid{v}_{1})){}\<[E]%
\ColumnHook
\end{hscode}\resethooks
\\[-1em]
\caption{Call-by-value }
\label{fig:by-value}
\end{figure}

Call-by-value eagerly evaluates a let-bound RHS and then substitutes its
\emph{value}, rather than the reduction trace that led to the value, into every
use site.

The call-by-value evaluation strategy is implemented with the \ensuremath{\conid{ByValue}} trace transformer shown in \Cref{fig:by-value}.
Function \ensuremath{\varid{bind}} defines a denotation \ensuremath{\varid{d}\mathbin{::}\conid{D}\;(\conid{ByValue}\;\varcolor{\tau})} of the right-hand
side by mutual recursion with \ensuremath{\varid{v}\mathbin{::}\conid{Value}\;(\conid{ByValue}\;\varcolor{\tau})} that we will discuss
shortly.

As its first action, \ensuremath{\varid{bind}} yields a brand-new \ensuremath{\conid{Let}_{0}} event that we assume was
added to the definition of \ensuremath{\conid{Event}}, announcing in the trace that the right-hand
side of a \ensuremath{\conid{Let}} is to be evaluated.
Then monadic bind \ensuremath{\varid{v}_{1}\leftarrow \varid{d};\varid{body}\;(\varid{return}\;\varid{v}_{1})} yields steps from the right-hand
side \ensuremath{\varid{d}} until its value \ensuremath{\varid{v}_{1}\mathbin{::}\conid{Value}\;(\conid{ByValue}\;\varcolor{\tau})} is reached, which is then
passed \ensuremath{\varid{return}}ed (\ie wrapped in \ensuremath{\conid{Ret}}) to the let \ensuremath{\varid{body}}.
Note that the steps in \ensuremath{\varid{d}} are yielded \emph{eagerly}, and only once, rather
than duplicating the trace at every use site in \ensuremath{\varid{body}}, as the by-name form
\ensuremath{\varid{body}\;\varid{d}} would.

To understand the recursive definition of the denotation of the right-hand side \ensuremath{\varid{d}}
and its value \ensuremath{\varid{v}},
consider the case \ensuremath{\varcolor{\tau}\mathrel{=}\conid{T}}.
Then \ensuremath{\varid{return}\mathrel{=}\conid{Ret}} and we get \ensuremath{\varid{d}\mathrel{=}\varid{rhs}\;(\conid{Ret}\;\varid{v})} for the value \ensuremath{\varid{v}} at the end of
the trace \ensuremath{\varid{d}}, as computed by the type class instance method \ensuremath{\varid{getValue}\mathbin{::}\conid{T}\;\varid{v}\to \varid{v}}.
The effect of \ensuremath{\conid{Ret}\;(\varid{getValue}\;(\varid{unByValue}\;\varid{d}))} is that of stripping all \ensuremath{\conid{Step}}s from \ensuremath{\varid{d}}.%

Since nothing about \ensuremath{\varid{getValue}} is particularly special to \ensuremath{\conid{T}}, it lives in its
own type class \ensuremath{\conid{Extract}} so that we get a \ensuremath{\conid{HasBind}} instance for different
types of \ensuremath{\conid{Trace}}s, such as more abstract ones in \Cref{sec:abstraction}.

Let us trace $\Let{i}{(\Lam{y}{\Lam{x}{x}})~i}{i~i}$ for call-by-value:
\begin{hscode}\SaveRestoreHook
\column{B}{@{}>{\hspre}l<{\hspost}@{}}%
\column{3}{@{}>{\hspre}l<{\hspost}@{}}%
\column{E}{@{}>{\hspre}l<{\hspost}@{}}%
\>[3]{}\lambda\!\!\mathbin{>}\mathcal{S}_{\mathbf{value}}\denot{\varid{read}\;\text{\ttfamily \char34 let~i~=~(λy.λx.x)~i~in~i~i\char34}}_{\varcolor{\varepsilon}}{}\<[E]%
\ColumnHook
\end{hscode}\resethooks
$
\LetOT\xhookrightarrow{\hspace{1.1ex}}\AppIT\xhookrightarrow{\hspace{1.1ex}}\AppET\xhookrightarrow{\hspace{1.1ex}}\LetIT\xhookrightarrow{\hspace{1.1ex}}\AppIT\xhookrightarrow{\hspace{1.1ex}}\LookupT(i)\xhookrightarrow{\hspace{1.1ex}}\AppET\xhookrightarrow{\hspace{1.1ex}}\LookupT(i)\xhookrightarrow{\hspace{1.1ex}}\langle \lambda\rangle 
$
\\[\belowdisplayskip]
\noindent
The beta reduction of $(\Lam{y}{\Lam{x}{x}})~i$ now happens once within the
$\LetOT$/$\LetIT$ bracket; the two subsequent $\LookupT$ events immediately halt
with a value.

Alas, this model of call-by-value does not yield a total interpreter!
Consider the case when the right-hand side accesses its value before yielding
one, \eg
\begin{hscode}\SaveRestoreHook
\column{B}{@{}>{\hspre}l<{\hspost}@{}}%
\column{3}{@{}>{\hspre}l<{\hspost}@{}}%
\column{E}{@{}>{\hspre}l<{\hspost}@{}}%
\>[3]{}\lambda\!\!\mathbin{>}\varid{takeT}\;\mathrm{5}\mathbin{\$}\mathcal{S}_{\mathbf{value}}\denot{\varid{read}\;\text{\ttfamily \char34 let~x~=~x~in~x~x\char34}}_{\varcolor{\varepsilon}}{}\<[E]%
\ColumnHook
\end{hscode}\resethooks
$\LetOT\xhookrightarrow{\hspace{1.1ex}}\LookupT(x)\xhookrightarrow{\hspace{1.1ex}}\LetIT\xhookrightarrow{\hspace{1.1ex}}\AppIT\xhookrightarrow{\hspace{1.1ex}}\LookupT(x)\xhookrightarrow{\hspace{1.1ex}}\texttt{\textasciicircum{}CInterrupted}$
\\[\belowdisplayskip]
\noindent
This loops forever unproductively, rendering the interpreter unfit as a
denotational semantics.

\begin{figure}
\begin{hscode}\SaveRestoreHook
\column{B}{@{}>{\hspre}l<{\hspost}@{}}%
\column{3}{@{}>{\hspre}l<{\hspost}@{}}%
\column{25}{@{}>{\hspre}l<{\hspost}@{}}%
\column{E}{@{}>{\hspre}l<{\hspost}@{}}%
\>[B]{}\mathcal{S}_{\mathbf{vinit}}\denot{\varid{e}}_{\varcolor{\rho}}(\varcolor{\mu})\mathrel{=}\varid{unByVInit}\;(\mathcal{S}\denot{\varid{e}}_{\varcolor{\rho}}\mathbin{::}\conid{D}\;(\conid{ByVInit}\;\conid{T}))\;\varcolor{\mu}{}\<[E]%
\\[\blanklineskip]%
\>[B]{}\keyword{newtype}\;\conid{ByVInit}\;\varcolor{\tau}\;\varid{v}\mathrel{=}\conid{ByVInit}\;\{\mskip1.5mu \varid{unByVInit}\mathbin{::}\conid{Heap}\;(\conid{ByVInit}\;\varcolor{\tau})\to \varcolor{\tau}\;(\varid{v},\conid{Heap}\;(\conid{ByVInit}\;\varcolor{\tau}))\mskip1.5mu\}{}\<[E]%
\\
\>[B]{}\keyword{instance}\;(\conid{Monad}\;\varcolor{\tau},\keyword{\forall}\!\! \hsforall \;\varid{v}\hsdot{\circ }{.\ }\conid{Trace}\;(\varcolor{\tau}\;\varid{v}))\Rightarrow \conid{HasBind}\;(\conid{D}\;(\conid{ByVInit}\;\varcolor{\tau}))\;\keyword{where}{}\<[E]%
\\
\>[B]{}\hsindent{3}{}\<[3]%
\>[3]{}\varid{bind}\;\varid{rhs}\;\varid{body}\mathrel{=}\keyword{do}\;{}\<[25]%
\>[25]{}\varcolor{\mu}\leftarrow \varid{get}{}\<[E]%
\\
\>[25]{}\keyword{let}\;\varid{a}\mathrel{=}\varid{nextFree}\;\varcolor{\mu}{}\<[E]%
\\
\>[25]{}\varid{put}\;\varcolor{\mu}[\varid{a}\mapsto\varid{stuck}]{}\<[E]%
\\
\>[25]{}\varid{step}\;\conid{Let}_{0}\;(\varid{memo}\;\varid{a}\;(\varid{rhs}\;(\varid{fetch}\;\varid{a})))\bind \varid{body}\hsdot{\circ }{.\ }\varid{return}{}\<[E]%
\ColumnHook
\end{hscode}\resethooks
\\[-1em]
\caption{Call-by-value with lazy initialisation}
\label{fig:by-value-init}
\end{figure}

\subsection{Lazy Initialisation and Black-holing}
\label{sec:lazy-init}

Recall that our simple \ensuremath{\conid{ByValue}} transformer above yields a potentially looping
interpreter.
Typical strict languages work around this issue in either of two ways:
They enforce termination of the RHS statically (OCaml, ML), or they use
\emph{lazy initialisation} techniques~\citep{Nakata:10,Nakata:06} (Scheme,
recursive modules in OCaml).
We recover a total interpreter using the semantics in \citet{Nakata:10},
building on the same encoding as \ensuremath{\conid{ByNeed}} and initialising the heap with a
\emph{black hole}~\citep{stg} \ensuremath{\varid{stuck}} in \ensuremath{\varid{bind}} as in \Cref{fig:by-value-init}.
\begin{hscode}\SaveRestoreHook
\column{B}{@{}>{\hspre}l<{\hspost}@{}}%
\column{3}{@{}>{\hspre}l<{\hspost}@{}}%
\column{E}{@{}>{\hspre}l<{\hspost}@{}}%
\>[3]{}\lambda\!\!\mathbin{>}\mathcal{S}_{\mathbf{vinit}}\denot{\varid{read}\;\text{\ttfamily \char34 let~x~=~x~in~x~x\char34}}_{\varcolor{\varepsilon}}(\varcolor{\varepsilon})\mathbin{::}\conid{T}\;(\conid{Value}\;\anonymous ,\conid{Heap}\;\anonymous ){}\<[E]%
\ColumnHook
\end{hscode}\resethooks
$
\LetOT\xhookrightarrow{\hspace{1.1ex}}\LookupT(x)\xhookrightarrow{\hspace{1.1ex}}\LetIT\xhookrightarrow{\hspace{1.1ex}}\AppIT\xhookrightarrow{\hspace{1.1ex}}\LookupT(x)\xhookrightarrow{\hspace{1.1ex}}\langle (\lightning,[0\!\! \mapsto \!\! \wild])\rangle 
$

\begin{figure}
\begin{hscode}\SaveRestoreHook
\column{B}{@{}>{\hspre}l<{\hspost}@{}}%
\column{3}{@{}>{\hspre}l<{\hspost}@{}}%
\column{5}{@{}>{\hspre}l<{\hspost}@{}}%
\column{12}{@{}>{\hspre}l<{\hspost}@{}}%
\column{E}{@{}>{\hspre}l<{\hspost}@{}}%
\>[B]{}\mathcal{S}_{\mathbf{clair}}\denot{\varid{e}}_{\varcolor{\rho}}\mathrel{=}\varid{runClair}\mathbin{\$}\mathcal{S}\denot{\varid{e}}_{\varcolor{\rho}}\mathbin{::}\conid{T}\;(\conid{Value}\;(\conid{Clairvoyant}\;\conid{T})){}\<[E]%
\\[\blanklineskip]%
\>[B]{}\keyword{data}\;\conid{Fork}\;\varid{f}\;\varid{a}\mathrel{=}\conid{Empty}\mid \conid{Single}\;\varid{a}\mid \conid{Fork}\;(\varid{f}\;\varid{a})\;(\varid{f}\;\varid{a});\keyword{data}\;\conid{ParT}\;\varid{m}\;\varid{a}\mathrel{=}\conid{ParT}\;(\varid{m}\;(\conid{Fork}\;(\conid{ParT}\;\varid{m})\;\varid{a})){}\<[E]%
\\
\>[B]{}\keyword{instance}\;\conid{Monad}\;\varcolor{\tau}\Rightarrow \conid{Alternative}\;(\conid{ParT}\;\varcolor{\tau})\;\keyword{where}{}\<[E]%
\\
\>[B]{}\hsindent{3}{}\<[3]%
\>[3]{}\varid{empty}\mathrel{=}\conid{ParT}\;(\varid{pure}\;\conid{Empty});\varid{l}\mathbin{<\mspace{-6mu}|\mspace{-6mu}>}\varid{r}\mathrel{=}\conid{ParT}\;(\varid{pure}\;(\conid{Fork}\;\varid{l}\;\varid{r})){}\<[E]%
\\[\blanklineskip]%
\>[B]{}\keyword{newtype}\;\conid{Clairvoyant}\;\varcolor{\tau}\;\varid{a}\mathrel{=}\conid{Clairvoyant}\;(\conid{ParT}\;\varcolor{\tau}\;\varid{a}){}\<[E]%
\\
\>[B]{}\varid{runClair}\mathbin{::}\conid{D}\;(\conid{Clairvoyant}\;\conid{T})\to \conid{T}\;(\conid{Value}\;(\conid{Clairvoyant}\;\conid{T})){}\<[E]%
\\[\blanklineskip]%
\>[B]{}\keyword{instance}\;(\conid{Extract}\;\varcolor{\tau},\conid{Monad}\;\varcolor{\tau},\keyword{\forall}\!\! \hsforall \;\varid{v}\hsdot{\circ }{.\ }\conid{Trace}\;(\varcolor{\tau}\;\varid{v}))\Rightarrow \conid{HasBind}\;(\conid{D}\;(\conid{Clairvoyant}\;\varcolor{\tau}))\;\keyword{where}{}\<[E]%
\\
\>[B]{}\hsindent{3}{}\<[3]%
\>[3]{}\varid{bind}\;\varid{rhs}\;\varid{body}\mathrel{=}\conid{Clairvoyant}\;(\varid{skip}\mathbin{<\mspace{-6mu}|\mspace{-6mu}>}\varid{let'})\bind \varid{body}{}\<[E]%
\\
\>[3]{}\hsindent{2}{}\<[5]%
\>[5]{}\keyword{where}\;{}\<[12]%
\>[12]{}\varid{skip}\mathrel{=}\varid{return}\;(\conid{Clairvoyant}\;\varid{empty}){}\<[E]%
\\
\>[12]{}\varid{let'}\mathrel{=}\varid{fmap}\;\varid{return}\mathbin{\$}\varid{step}\;\conid{Let}_{0}\mathbin{\$}\mathbin{...}\;\varid{fix}\mathbin{...}\varid{rhs}\mathbin{...}\varid{getValue}\mathbin{...}{}\<[E]%
\ColumnHook
\end{hscode}\resethooks
\\[-1em]
\caption{Clairvoyant Call-by-value}
\label{fig:clairvoyant-by-value}
\end{figure}

\subsection{Clairvoyant Call-by-value}
\label{sec:clair}

Clairvoyant call-by-value~\citep{HackettHutton:19} is an alternative to
call-by-need semantics that exploits non-determinism and a cost model to absolve
of the heap.
We can instantiate our interpreter to generate the shortest clairvoyant
call-by-value trace as well, as sketched out in \Cref{fig:clairvoyant-by-value}.
Doing so yields an evaluation strategy that either skips or speculates let
bindings, depending on whether or not the binding is needed:
\begin{hscode}\SaveRestoreHook
\column{B}{@{}>{\hspre}l<{\hspost}@{}}%
\column{3}{@{}>{\hspre}l<{\hspost}@{}}%
\column{E}{@{}>{\hspre}l<{\hspost}@{}}%
\>[3]{}\lambda\!\!\mathbin{>}\mathcal{S}_{\mathbf{clair}}\denot{\varid{read}\;\text{\ttfamily \char34 let~f~=~λx.x~in~let~g~=~λy.f~in~g\char34}}_{\varcolor{\varepsilon}}\mathbin{::}\conid{T}\;(\conid{Value}\;(\conid{Clairvoyant}\;\conid{T})){}\<[E]%
\ColumnHook
\end{hscode}\resethooks
$
\LetIT\xhookrightarrow{\hspace{1.1ex}}\LetOT\xhookrightarrow{\hspace{1.1ex}}\LetIT\xhookrightarrow{\hspace{1.1ex}}\LookupT(g)\xhookrightarrow{\hspace{1.1ex}}\langle \lambda\rangle 
$
\begin{hscode}\SaveRestoreHook
\column{B}{@{}>{\hspre}l<{\hspost}@{}}%
\column{3}{@{}>{\hspre}l<{\hspost}@{}}%
\column{E}{@{}>{\hspre}l<{\hspost}@{}}%
\>[3]{}\lambda\!\!\mathbin{>}\mathcal{S}_{\mathbf{clair}}\denot{\varid{read}\;\text{\ttfamily \char34 let~f~=~λx.x~in~let~g~=~λy.f~in~g~g\char34}}_{\varcolor{\varepsilon}}\mathbin{::}\conid{T}\;(\conid{Value}\;(\conid{Clairvoyant}\;\conid{T})){}\<[E]%
\ColumnHook
\end{hscode}\resethooks
$
\LetOT\xhookrightarrow{\hspace{1.1ex}}\LetIT\xhookrightarrow{\hspace{1.1ex}}\LetOT\xhookrightarrow{\hspace{1.1ex}}\LetIT\xhookrightarrow{\hspace{1.1ex}}\AppIT\xhookrightarrow{\hspace{1.1ex}}\LookupT(g)\xhookrightarrow{\hspace{1.1ex}}\AppET\xhookrightarrow{\hspace{1.1ex}}\LookupT(f)\xhookrightarrow{\hspace{1.1ex}}\langle \lambda\rangle 
$
\\[\belowdisplayskip]
\noindent
The first example discards $f$, but the second needs it, so the trace starts
with an additional $\LetOT$ event.
Similar to \ensuremath{\conid{ByValue}}, the interpreter is not total so it is unfit as a
denotational semantics without a complicated domain theoretic judgment.
Furthermore, the decision whether or not a $\LetOT$ step is needed can be
delayed for an infinite amount of time, as exemplified by
\begin{hscode}\SaveRestoreHook
\column{B}{@{}>{\hspre}l<{\hspost}@{}}%
\column{3}{@{}>{\hspre}l<{\hspost}@{}}%
\column{E}{@{}>{\hspre}l<{\hspost}@{}}%
\>[3]{}\lambda\!\!\mathbin{>}\mathcal{S}_{\mathbf{clair}}\denot{\varid{read}\;\text{\ttfamily \char34 let~i~=~Z()~in~let~w~=~λy.y~y~in~w~w\char34}}_{\varcolor{\varepsilon}}\mathbin{::}\conid{T}\;(\conid{Value}\;(\conid{Clairvoyant}\;\conid{T})){}\<[E]%
\ColumnHook
\end{hscode}\resethooks
\texttt{\textasciicircum{}CInterrupted}
\\[\belowdisplayskip]
\noindent
The program diverges without producing even a prefix of a trace because the
binding for $i$ might be needed at an unknown point in the future
(a \emph{liveness property} and hence impossible to verify at runtime).
This renders Clairvoyant call-by-value inadequate for verifying properties of
infinite executions.
\end{toappendix}

The examples so far suggest that \ensuremath{\mathcal{S}_{\mathbf{need}}\denot{\wild}_{\wild}} agrees with the LK machine on
\emph{many} programs.
The next section proves that \ensuremath{\mathcal{S}_{\mathbf{need}}\denot{\wild}_{\wild}} agrees with the LK machine on
\emph{all} programs, including ones that diverge.

\section{Totality and Semantic Adequacy}

In this section, we prove that \ensuremath{\mathcal{S}_{\mathbf{need}}\denot{\wild}_{\wild}} produces small-step traces of the
lazy Krivine machine and is indeed a \emph{denotational semantics}.%
\footnote{Similar results for \ensuremath{\mathcal{S}_{\mathbf{name}}\denot{\wild}_{\wild}} should be derivative.}
Excitingly, to our knowledge, \ensuremath{\mathcal{S}_{\mathbf{need}}\denot{\wild}_{\wild}} is the first denotational call-by-need
semantics that was proven so!
Specifically, denotational semantics must be total and adequate.
\emph{Totality} says that the interpreter is well-defined for every input expression and \emph{adequacy} says that the interpreter produces similar traces as the reference semantics.
This is an important result because it allows us to switch between operational reference semantics and denotational interpreter as needed, thus guaranteeing compatibility
of definitions such as absence in \Cref{defn:absence}.
As before, all the (pen-and-paper) proofs can be found in the Appendix.

\begin{toappendix}
\subsection{Guarded Type Theory}
\label{sec:guarded-types}

For our proofs, we will make use of guarded type theory.
Specifically, \Cref{sec:totality-detail} encodes the interpreter in Guarded
Cubical Agda, which implements a variant of guarded type theory called Ticked
Cubical Type Theory~\citep{tctt}.
Our pen and paper proofs in \Cref{sec:adequacy-detail} and
\Cref{sec:soundness-detail} respect the coinductive nature of this encoding.
Thus, this subsection shall serve as a short introduction.

Whereas traditional theories of coinduction require syntactic productivity
checks~\citep{Coquand:94}, imposing tiresome constraints on the form of guarded
recursive functions, the appeal of guarded type theories is that productivity
is instead proven semantically, in the type system.
Compared to the alternative of \emph{sized types}~\citep{Hughes:96}, guarded
types don't require complicated algebraic manipulations of size parameters;
however perhaps sized types would work just as well.
Any fuel-based (or step-indexed) approach is equivalent to our use of guarded
type theory, but we find that the latter is a more direct (and thus preferable)
encoding.

The fundamental innovation of guarded recursive type theory is the integration
of the ``later'' modality $\later$ which allows to define coinductive data
types with negative recursive occurrences such as in the data constructor \ensuremath{\conid{Fun}\mathbin{::}(\highlight{\conid{D}\;\varcolor{\tau}}\to \conid{D}\;\varcolor{\tau})\to \conid{Value}\;\varcolor{\tau}} (recall that \ensuremath{\conid{D}\;\varcolor{\tau}\mathrel{=}\varcolor{\tau}\;(\highlight{\conid{Value}}\;\varcolor{\tau})}), as
first realised by \citet{Nakano:00}.
The way that is achieved is roughly as follows: The type $\later T$
represents data of type $T$ that will become available after a finite amount
of computation, such as unrolling one layer of a fixpoint definition.
It comes with a general fixpoint combinator $\fix : \forall A.\ (\later A \to
A) \to A$ that can be used to define both coinductive \emph{types} (via guarded
recursive functions on the universe of types~\citep{BirkedalMogelbergEjlers:13})
as well as guarded recursive \emph{terms} inhabiting said types.
The classic example is that of infinite streams:
\[
  Str = ℕ \times \later Str \qquad ones = \fix (r : \later Str).\ (1,r),
\]
where $ones : Str$ is the constant stream of $1$.
In particular, $Str$ is the fixpoint of a locally contractive functor $F(X) =
ℕ \times \later X$.
According to \citet{BirkedalMogelbergEjlers:13}, any type expression in simply
typed lambda calculus defines a locally contractive functor as long as any
occurrence of $X$ is under a $\later$.
The most exciting consequence is that changing the \ensuremath{\conid{Fun}} data constructor to
\ensuremath{\conid{Fun}\mathbin{::}(\later\!\;(\conid{D}\;\varcolor{\tau})\to \conid{D}\;\varcolor{\tau})\to \conid{Value}\;\varcolor{\tau}} makes \ensuremath{\conid{Value}\;\varcolor{\tau}} a well-defined
coinductive data type,%
\footnote{The reason why the positive occurrence of \ensuremath{\conid{D}\;\varcolor{\tau}} does not need to be
guarded is that the type of \ensuremath{\conid{Fun}} can more formally be encoded by a mixed
inductive-coinductive type, \eg
$\ensuremath{\conid{Value}\;\varcolor{\tau}} = \fix X.\ \lfp Y.\ ...~|~\ensuremath{\conid{Fun}}~(X \to Y)~|~...$ }
whereas syntactic approaches to coinduction reject any negative recursive
occurrence.

As a type constructor, $\later$ is an applicative
functor~\citep{McBridePaterson:08} via functions
\[
  \purelater : \forall A.\ A \to \later A \qquad \wild \aplater \wild : \forall A,B.\ \later (A \to B) \to \later A \to \later B,
\]
allowing us to apply a familiar framework of reasoning around $\later$.
In order not to obscure our work with pointless symbol pushing, we will often
omit the idiom brackets~\citep{McBridePaterson:08} $\idiom{\wild}$ to indicate
where the $\later$ ``effects'' happen.

\subsection{Total Encoding in Guarded Cubical Agda}
\label{sec:totality-detail}

Building on \Cref{sec:guarded-types}, we will now outline the changes necessary to encode \ensuremath{\mathcal{S}\denot{\wild}_{\wild}} in Guarded Cubical
Agda, a system implementing Ticked Cubical Type Theory~\citep{tctt}, as well
as the concrete instances \ensuremath{\conid{D}\;(\conid{ByName}\;\conid{T})} and \ensuremath{\conid{D}\;(\conid{ByNeed}\;\conid{T})} from
\Cref{fig:trace-instances,fig:by-need}.
The full, type-checked development is available in the Supplement.
\begin{itemize}
  \item We need to delay in \ensuremath{\varid{step}}; thus its definition in \ensuremath{\conid{Trace}} changes to
    \ensuremath{\varid{step}\mathbin{::}\conid{Event}\to \later\!\;\varid{d}\to \varid{d}}.
  \item
    All \ensuremath{\conid{D}}s that will be passed to lambdas, put into the environment or
    stored in fields need to have the form \ensuremath{\varid{step}\;(\conid{Look}\;\varid{x})\;\varid{d}} for some
    \ensuremath{\varid{x}\mathbin{::}\conid{Name}} and a delayed \ensuremath{\varid{d}\mathbin{::}\later\!\;(\conid{D}\;\varcolor{\tau})}.
    This is enforced as follows:
    \begin{enumerate}
      \item
        The \ensuremath{\conid{Domain}} type class gains an additional predicate parameter \ensuremath{\varid{p}\mathbin{::}\conid{D}\to \conid{Set}}
        that will be instantiated by the semantics to a predicate that checks
        that the \ensuremath{\conid{D}} has the required form \ensuremath{\varid{step}\;(\conid{Look}\;\varid{x})\;\varid{d}} for some
        \ensuremath{\varid{x}\mathbin{::}\conid{Name}}, \ensuremath{\varid{d}\mathbin{::}\later\!\;(\conid{D}\;\varcolor{\tau})}.
      \item
        Then the method types of \ensuremath{\conid{Domain}} use a Sigma type to encode conformance
        to \ensuremath{\varid{p}}.
        For example, the type of \ensuremath{\conid{Fun}} changes to \ensuremath{(\conid{\mathrm{\Sigma}}\;\conid{D}\;\varid{p}\to \conid{D})\to \conid{D}}.
      \item
        The reason why we need to encode this fact is that the guarded recursive
        data type \ensuremath{\conid{Value}} has a constructor the type of which amounts to
        \ensuremath{\conid{Fun}\mathbin{::}(\conid{Name}\;\times\;\later\!\;(\conid{D}\;\varcolor{\tau})\to \conid{D}\;\varcolor{\tau})\to \conid{Value}\;\varcolor{\tau}}, breaking the
        previously discussed negative recursive cycle by a $\later$, and
        expecting \ensuremath{\varid{x}\mathbin{::}\conid{Name}}, \ensuremath{\varid{d}\mathbin{::}\later\!\;(\conid{D}\;\varcolor{\tau})} such that the original \ensuremath{\conid{D}\;\varcolor{\tau}} can
        be recovered as \ensuremath{\varid{step}\;(\conid{Look}\;\varid{x})\;\varid{d}}.
        This is in contrast to the original definition \ensuremath{\conid{Fun}\mathbin{::}(\conid{D}\;\varcolor{\tau}\to \conid{D}\;\varcolor{\tau})\to \conid{Value}\;\varcolor{\tau}} which would \emph{not} type-check.
        One can understand \ensuremath{\conid{Fun}} as carrying the ``closure'' resulting from
        \emph{defunctionalising}~\citep{Reynolds:72} a \ensuremath{\conid{\mathrm{\Sigma}}\;\conid{D}\;\varid{p}}, and that this
        defunctionalisation is presently necessary in Agda to eliminate negative
        cycles.
    \end{enumerate}
  \item
    Expectedly, \ensuremath{\conid{HasBind}} becomes more complicated because it encodes the
    fixpoint combinator.
    We settled on \ensuremath{\varid{bind}\mathbin{::}\later\!\;(\later\!\;\conid{D}\mathbin{→}\conid{D})\mathbin{→}(\later\!\;\conid{D}\mathbin{→}\conid{D})\mathbin{→}\conid{D}}.
    We tried rolling up \ensuremath{\varid{step}\;(\conid{Look}\;\varid{x})\;\anonymous } in the definition of \ensuremath{\mathcal{S}\denot{\wild}_{\wild}}
    to get a simpler type \ensuremath{\varid{bind}\mathbin{::}(\conid{\mathrm{\Sigma}}\;\conid{D}\;\varid{p}\mathbin{→}\conid{D})\mathbin{→}(\conid{\mathrm{\Sigma}}\;\conid{D}\;\varid{p}\mathbin{→}\conid{D})\mathbin{→}\conid{D}},
    but then had trouble defining \ensuremath{\conid{ByNeed}} heaps independently of the concrete
    predicate \ensuremath{\varid{p}}.
  \item
    Higher-order mutable state is among the classic motivating examples for
    guarded recursive types.
    As such it is no surprise that the state-passing of the mutable \ensuremath{\conid{Heap}} in
    the implementation of \ensuremath{\conid{ByNeed}} requires breaking of a recursive cycle
    by delaying heap entries, \ensuremath{\conid{Heap}\;\varcolor{\tau}\mathrel{=}\conid{Addr}\mathbin{:\rightharpoonup}\later\!\;(\conid{D}\;\varcolor{\tau})}.
  \item
    We need to pass around \ensuremath{\conid{Tick}} binders in \ensuremath{\mathcal{S}\denot{\wild}_{\wild}} in a way that the type
    checker is satisfied; a simple exercise.
    We find it remarkable how non-invasive these adjustment are!
\end{itemize}

Thus we have proven that \ensuremath{\mathcal{S}\denot{\wild}_{\wild}} is a total, mathematical function, and
fast and loose equational reasoning about \ensuremath{\mathcal{S}\denot{\wild}_{\wild}} is not only \emph{morally}
correct~\citep{Danielsson:06}, but simply \emph{correct}.
Furthermore, since evaluation order doesn't matter in Agda and hence for \ensuremath{\mathcal{S}\denot{\wild}_{\wild}},
we could have defined it in a strict language (lowering \ensuremath{\later\!\;\varid{a}} as \ensuremath{()\to \varid{a}})
just as well.
\end{toappendix}

\subsection{Adequacy of \ensuremath{\mathcal{S}_{\mathbf{need}}\denot{\wild}_{\wild}}}
\label{sec:adequacy}

\begin{figure}
\[\ruleform{\begin{array}{c}
  α_\Events : \States \to \ensuremath{\conid{Event}}
  \qquad
  α_\Environments : \Environments \times \Heaps \to (\ensuremath{\conid{Name}\mathbin{:\rightharpoonup}\conid{D}_{\mathbf{ne}}})
  \qquad
  α_\Heaps : \Heaps \to \ensuremath{\conid{Heap}_{\mathbf{ne}}}
  \\
  α_{\STraces} : \STraces \times \Continuations \to \ensuremath{\conid{T}\;(\conid{Value}_{\mathbf{ne}},\conid{Heap}_{\mathbf{ne}})}
  \qquad
  α_{\Values} : \States \times \Continuations \to \ensuremath{\conid{Value}_{\mathbf{ne}}}
\end{array}}\]
\arraycolsep=2pt
\[\begin{array}{lcl}
  α_\Events(σ) & = & \begin{cases}
    \ensuremath{\conid{Let}_{1}} & \text{if }σ = (\Let{\px}{\wild}{\wild},\wild,\wild,\wild) \\
    \ensuremath{\conid{App}_{1}} & \text{if }σ = (\pe~\px,\wild,\wild,\wild) \\
    \ensuremath{\conid{Case}_{1}} & \text{if }σ = (\Case{\wild}{\wild},\wild,\wild, \wild)\\
    \ensuremath{\conid{Look}\;\varid{y}} & \text{if }σ = (\px,ρ,μ,\wild), μ(ρ(\px)) = (\py,\wild,\wild) \\
    \ensuremath{\conid{App}_{2}} & \text{if }σ = (\Lam{\wild}{\wild},\wild,\wild,\ApplyF(\wild) \pushF \wild) \\
    \ensuremath{\conid{Case}_{2}} & \text{if }σ = (K~\wild, \wild, \wild, \SelF(\wild,\wild) \pushF \wild) \\
    \ensuremath{\conid{Upd}} & \text{if }σ = (\pv,\wild,\wild,\UpdateF(\wild) \pushF \wild) \\
  \end{cases} \\
  \\[-0.75em]
  α_\Environments([\many{\px ↦ \pa}], μ) & = & [\many{\ensuremath{\varid{x}} ↦ \ensuremath{\conid{Step}\;(\conid{Look}\;\varid{y})\;(\varid{fetch}\;\varid{a})} \mid μ(\pa) = (\py,\wild,\wild)}] \\
  \\[-0.75em]
  α_\Heaps([\many{\pa ↦ (\wild,ρ,\pe)}]) & = & [\many{\ensuremath{\varid{a}} ↦ \ensuremath{\varid{memo}\;\varid{a}\;(\mathcal{S}\denot{\varid{e}}_{α_\Environments(\varcolor{\rho},\varcolor{\mu})})}}] \\
  \\[-0.75em]
  α_\Values(σ,κ_0) & = & \begin{cases}
    \ensuremath{\conid{Fun}\;(\lambda \varid{d}\to \conid{Step}\;\conid{App}_{2}\;(\mathcal{S}\denot{\varid{e}}_{(α_\Environments(\varcolor{\rho},\varcolor{\mu}))[\varid{x}\mapsto\varid{d}]}))} & \text{if } σ = (\Lam{\px}{\pe},ρ,μ,κ_0) \\
    \ensuremath{\conid{Con}\;\varid{k}\;(\varid{map}\;(α_\Environments(\varcolor{\rho},\varcolor{\mu})\mathop{!})\;\varid{xs})}                         & \text{if } σ = (K~\overline{\px},ρ,μ,κ_0) \\
    \ensuremath{\conid{Stuck}}                                               & \text{otherwise} \\
  \end{cases} \\
  \\[-0.75em]
  α_{\STraces}((σ_i)_{i∈\overline{n}},κ_0) & = & \begin{cases}
    \ensuremath{\conid{Step}\;( α_\Events(σ_0) )\;\idiom{α_{\STraces}((σ_{i+1})_{i∈\overline{n-1}},\kappa_0)}} & \text{if }n > 0 \\
    \ensuremath{\conid{Ret}\;( α_\Values(σ_0,κ_0) ,α_\Heaps(\varcolor{\mu}))} & \text{where }(\wild,\wild, μ, \wild) = σ_0
  \end{cases}
\end{array}\]
\vspace{-1em}
\caption{Abstraction function $α_{\STraces}$ from LK machine to \ensuremath{\mathcal{S}_{\mathbf{need}}\denot{\wild}_{\wild}}}
  \label{fig:eval-correctness}
\end{figure}

For proving adequacy of \ensuremath{\mathcal{S}_{\mathbf{need}}\denot{\wild}_{\wild}}, we give an abstraction function
$α_{\STraces}$ from small-step traces $\STraces$ in the lazy Krivine machine
(\Cref{fig:lk-semantics}) to denotational traces \ensuremath{\conid{T}}, with \ensuremath{\conid{Events}} and all,
such that
\[
  α_{\STraces}(\init(\pe) \smallstep ..., \StopF) = \ensuremath{\mathcal{S}_{\mathbf{need}}\denot{\varid{e}}_{\varcolor{\varepsilon}}(\varcolor{\varepsilon})},
\]
where $\init(\pe) \smallstep ...$ denotes the \emph{maximal} (\ie longest
possible) LK trace evaluating the closed expression $\pe$.
For example, for the LK trace \labelcref{ex:trace2}, $α_{\STraces}$ produces
the trace at the end of
\hyperlink{ex:eval-need-trace2}{\Cref*{sec:by-need-instance}}.

It turns out that function $α_{\STraces}$, defined in
\Cref{fig:eval-correctness}, preserves a number of important observable
properties, such as termination behavior (\ie stuck, diverging, or balanced
execution~\citep{Sestoft:97}), length of the trace and transition events, as
expressed in the following theorem:

\begin{toappendix}
\subsection{Adequacy}
\label{sec:adequacy-detail}
\end{toappendix}

\begin{theoremrep}[Strong Adequacy]
  \label{thm:need-adequate-strong}
  Let \ensuremath{\varid{e}} be a closed expression, \ensuremath{\varcolor{\tau}\triangleq\mathcal{S}_{\mathbf{need}}\denot{\varid{e}}_{\varcolor{\varepsilon}}(\varcolor{\varepsilon})} the
  denotational by-need trace and $\init(\pe) \smallstep ...$ the maximal lazy
  Krivine trace.
  Then
  \begin{itemize}
    \item \ensuremath{\varcolor{\tau}} preserves the observable termination properties of $\init(\pe) \smallstep ...$
      in the above sense.
    \item \ensuremath{\varcolor{\tau}} preserves the length of $\init(\pe) \smallstep ...$ (\ie number of \ensuremath{\conid{Step}}s equals number of $\smallstep$).
    \item every \ensuremath{\varid{ev}\mathbin{::}\conid{Event}} in \ensuremath{\varcolor{\tau}\mathrel{=}\many{\conid{Step}\;\varid{ev}\;\mathbin{...}}} corresponds to the
      transition rule taken in $\init(\pe) \smallstep ...$.
  \end{itemize}
\end{theoremrep}
\begin{proofsketch}
  Generalise $α_{\STraces}(\init(\pe) \smallstep ..., \StopF) = \ensuremath{\mathcal{S}_{\mathbf{need}}\denot{\varid{e}}_{\varcolor{\varepsilon}}(\varcolor{\varepsilon})}$ to
  open configurations and prove it by Löb induction.
  Then it suffices to prove that $α_{\STraces}$ preserves the observable properties of
  interest.
  The full proof for a rigorous reformulation of the proposition can be found in \Cref*{sec:adequacy-detail}.
\end{proofsketch}
\begin{proof}
  $\ensuremath{\mathcal{S}_{\mathbf{need}}\denot{\varid{e}}_{\varcolor{\varepsilon}}(\varcolor{\varepsilon})} = α(\init(\pe) \smallstep ..., \StopF)$ follows directly
  from \Cref{thm:need-abstracts-lk}.
  The preservation results in are a consequence of \Cref{thm:abs-length} and \Cref{thm:need-adequate};
  function $α_\Events$ in \Cref{fig:eval-correctness} encodes the intuition in
  which LK transitions abstract into \ensuremath{\conid{Event}}s.
\end{proof}

\begin{toappendix}
To formalise the main adequacy result, we must characterise the maximal traces in
the LK transition system and relate them to the trace produced by \ensuremath{\mathcal{S}_{\mathbf{need}}\denot{\wild}_{\wild}}
via the abstraction function in \Cref{fig:eval-correctness} and its associated
correctness relation.

\subsubsection{Maximal Lazy Krivine Traces}

Formally, an LK trace is a trace in $(\smallstep)$ from
\Cref{fig:lk-semantics}, \ie a non-empty and potentially infinite sequence of
LK states $(σ_i)_{i∈\overline{n}}$, such that $σ_i \smallstep σ_{i+1}$ for
$i,(i+1)∈\overline{n}$.
The \emph{length} of $(σ_i)_{i∈\overline{n}}$ is the potentially infinite
number of $\smallstep$ transitions $n$, where infinity is expressed by the first
limit ordinal $ω$.
The notation $\overline{n}$ means $\{ m ∈ ℕ \mid m \leq n \}$ when $n∈ℕ$
is finite (note that $0 ∈ ℕ$), and $ℕ$ when $n = ω$ is infinite.

The \emph{source} state $σ_0$ exists for finite and infinite traces, while the
\emph{target} state $σ_n$ is only defined when $n \not= ω$ is finite.
When the control expression of a state $σ$ (selected via $\ctrl(σ)$) is a value
$\pv$, we call $σ$ a \emph{return} state and say that the continuation (selected
via $\cont(σ)$) drives evaluation.
Otherwise, $σ$ is an \emph{evaluation} state and $\ctrl(σ)$ drives evaluation.

An important kind of trace is an \emph{interior trace}, one that never leaves
the evaluation context of its source state:

\begin{definition}[Deep]
  An LK trace $(σ_i)_{i∈\overline{n}}$ is
  \emph{$κ$-deep} if every intermediate continuation $κ_i \triangleq
  \cont(σ_i)$ extends $κ$ (so $κ_i = κ$ or $κ_i = \wild \pushF κ$,
  abbreviated $κ_i = ...κ$).
\end{definition}

\begin{definition}[Interior]
  A trace $(σ_i)_{i∈\overline{n}}$ is called \emph{interior} (notated as
  $\interior{(σ_i)_{i∈\overline{n}}}$) if it is $\cont(σ_0)$-deep.
\end{definition}

A \emph{balanced trace}~\citep{Sestoft:97} is an interior trace that is about to
return from the initial evaluation context; it corresponds to a big-step
evaluation of the initial focus expression:

\begin{definition}[Balanced]
  An interior trace $(σ_i)_{i∈\overline{n}}$ is
  \emph{balanced} if the target state exists and is a return
  state with continuation $\cont(σ_0)$.
\end{definition}

In the following we give an example for interior and balanced traces.
We will omit the first component of heap entries in the examples because they
bear no semantic significance apart from instrumenting $\LookupT$ transitions,
and it is confusing when the heap-bound expression is a variable $x$,
\eg $(y,ρ,x)$.
Of course, the abstraction function in \Cref{fig:eval-correctness} will need to
look at the first component.
\begin{example}
  Let $ρ=[x↦\pa_1],μ=[\pa_1↦(\wild,[], \Lam{y}{y})]$ and $κ$ an arbitrary
  continuation. The trace
  \[
     (x, ρ, μ, κ) \smallstep (\Lam{y}{y}, ρ, μ, \UpdateF(\pa_1) \pushF κ) \smallstep (\Lam{y}{y}, ρ, μ, κ)
  \]
  is interior and balanced. Its proper prefixes are interior but not balanced.
  The suffix
  \[
     (\Lam{y}{y}, ρ, μ, \UpdateF(\pa_1) \pushF κ) \smallstep (\Lam{y}{y}, ρ, μ, κ)
  \]
  is neither interior nor balanced.
\end{example}

As shown by \citet{Sestoft:97}, a balanced trace starting at a control
expression $\pe$ and ending with $\pv$ corresponds to a derivation of $\pe
\Downarrow \pv$ in a big-step semantics or a non-$⊥$ result in a Scott-style
denotational semantics.
It is when a derivation in a big-step semantics does \emph{not} exist that a
small-step semantics shows finesse.
In this case, a small-step semantics differentiates two different kinds
of \emph{maximally interior} (or, just \emph{maximal}) traces, namely
\emph{diverging} and \emph{stuck} traces:

\begin{definition}[Maximal]
  An LK trace $(σ_i)_{i∈\overline{n}}$ is \emph{maximal} if and only if it is
  interior and there is no $σ_{n+1}$ such that $(σ_i)_{i∈\overline{n+1}}$ is
  interior.
  More formally,
  \[
    \maxtrace{(σ_i)_{i∈\overline{n}}} \triangleq \interior{(σ_i)_{i∈\overline{n}}} \wedge (\not\exists σ_{n+1}.\ σ_n \smallstep σ_{n+1} \wedge \cont(σ_{n+1}) = ...\cont(σ_0)).
  \]
\end{definition}

\begin{definition}[Diverging]
  An infinite and interior trace is called \emph{diverging}.
\end{definition}

\begin{definition}[Stuck]
  A finite, maximal and unbalanced trace is called \emph{stuck}.
\end{definition}

Usually, stuckness is associated with a state of a transition
system rather than a trace.
That is not possible in our framework; the following example clarifies.

\begin{example}[Stuck and diverging traces]
\label{ex:stuck-div}
Consider the interior trace
\[
             (\ttrue~x, [x↦\pa_1], [\pa_1↦...], κ)
  \smallstep (\ttrue, [x↦\pa_1], [\pa_1↦...], \ApplyF(\pa_1) \pushF κ),
\]
where $\ttrue$ is a data constructor.
It is stuck, but its singleton suffix is balanced.
An example for a diverging trace, where $ρ=[x↦\pa_1]$ and $μ=[\pa_1↦(\wild,ρ,x)]$, is
\[
  (\Let{x}{x}{x}, [], [], κ) \smallstep (x, ρ, μ, κ) \smallstep (x, ρ, μ, \UpdateF(\pa_1) \pushF κ) \smallstep ...
\]
\end{example}

Note that balanced traces are maximal traces as well.
In fact, balanced, diverging and stuck traces are the \emph{only} three kinds of
maximal traces, as the following lemma formalises:

\begin{lemmarep}[Characterisation of maximal traces]
  An LK trace $(σ_i)_{i∈\overline{n}}$ is maximal if and only if it is balanced,
  diverging or stuck.
\end{lemmarep}
\begin{proof} $ $
\begin{itemize}
\item[$\Rightarrow$:]
  Let $(σ_i)_{i∈\overline{n}}$ be maximal.
  Then it is interior by definition.
  Thus, if $n=ω$ is infinite, then it is diverging.
  Otherwise, $(σ_i)_{i∈\overline{n}}$ is finite.
  If it is \emph{not} balanced, then it is still maximal and finite, hence stuck.
  Otherwise, it is balanced.

\item[$\Leftarrow$:]
  Let $(σ_i)_{i∈\overline{n}}$ be balanced, diverging or stuck. \\
  If $(σ_i)_{i∈\overline{n}}$ is balanced, then it is interior, and $σ_n$ is
  a return state with continuation $\cont(σ_0)$.
  Now suppose there would exist $σ_{n+1}$ such that
  $(σ_i)_{i∈\overline{n+1}}$ was interior.
  Then the transition $σ_n \smallstep σ_{n+1}$ must be one of the
  \emph{returning transition rules} $\UpdateT$, $\AppET$ and $\CaseET$, which
  are the only ones in which $σ_n$ is a return state (\ie $\ctrl(σ_n)$ is a
  value $\pv$).
  But all return transitions leave $\cont(σ_0)$, which is in contradiction to
  interiority of $(σ_i)_{i∈\overline{n+1}}$.
  Thus, $(σ_i)_{i∈\overline{n}}$ is maximal. \\
  If $(σ_i)_{i∈\overline{n}}$ is diverging, it is interior and infinite,
  hence $n=ω$.
  Indeed $(σ_i)_{i∈\overline{ω}}$ is maximal, because the expression $σ_{ω+1}$
  is undefined and hence does not exist. \\
  If $(σ_i)_{i∈\overline{n}}$ is stuck, then it is maximal by definition.
\end{itemize}
\end{proof}

Interiority guarantees that the particular initial stack $\cont(σ_0)$ of a
maximal trace is irrelevant to execution, so maximal traces that differ only in
the initial stack are bisimilar.
This is a consequence of the fact that the semantics of a called function may
not depend on the contents of the call stack; this fact is encoded implicitly in
big-step derivations.

\subsubsection{Abstraction preserves Termination Observable}

One class of maximal traces is of particular interest:
the maximal trace starting in $\init(\pe)$!
Whether it is infinite, stuck or balanced is the semantically defining
\emph{termination observable} of $\pe$.
If we can show that \ensuremath{\mathcal{S}_{\mathbf{need}}\denot{\varid{e}}_{\varcolor{\varepsilon}}(\varcolor{\varepsilon})} distinguishes these behaviors of
$\init(\pe) \smallstep ...$, we have proven \ensuremath{\mathcal{S}_{\mathbf{need}}\denot{\wild}_{\wild}} \emph{adequate}, and may
use it as a replacement for the LK transition system.

In order to show that \ensuremath{\mathcal{S}_{\mathbf{need}}\denot{\wild}_{\wild}} preserves the termination observable of \ensuremath{\varid{e}},
\begin{itemize}
\setlength{\itemsep}{0pt}
\item
  we define a family of abstraction functions $α$ from LK traces to by-need
  traces, formally in \ensuremath{\conid{T}\;(\conid{Value}_{\mathbf{ne}},\conid{Heap}_{\mathbf{ne}})} (\Cref{fig:eval-correctness}),
\item
  we show that this function $α$ preserves the termination observable of a given
  LK trace $\init(\pe) \smallstep ...$ (\Cref{thm:abs-max-trace}), and
\item
  we show that running \ensuremath{\mathcal{S}_{\mathbf{need}}\denot{\varid{e}}_{\varcolor{\varepsilon}}(\varcolor{\varepsilon})} is the same as mapping $α$ over
  the LK trace $\init(\pe) \smallstep ...$, hence the termination behavior is
  observable (\Cref{thm:need-abstracts-lk}).
\end{itemize}

In the following, we will omit repeated wrapping and unwrapping of the \ensuremath{\conid{ByNeed}}
newtype wrapper when constructing and taking apart elements in \ensuremath{\conid{D}_{\mathbf{ne}}}.
Furthermore, we will sometimes need to disambiguate the clashing definitions from
\Cref{sec:interp} and \Cref{sec:problem} by adorning ``Haskell objects'' with a
tilde, in which case \ensuremath{\varcolor{\tm}\triangleqα_\Heaps(\varcolor{\mu})\mathbin{::}\conid{Heap}_{\mathbf{ne}}} denotes a semantic
by-need heap, defined as an abstraction of a syntactic LK heap $μ ∈ \Heaps$.

Now consider the trace abstraction function $α_{\STraces}$ from
\Cref{fig:eval-correctness}.
It maps syntactic entities in the transition system to the definable entities
in the denotational by-need trace domain \ensuremath{\conid{T}\;(\conid{Value}_{\mathbf{ne}},\conid{Heap}_{\mathbf{ne}}}), henceforth
abbreviated as \ensuremath{\conid{T}} because it is the only use of the type constructor \ensuremath{\conid{T}} in
this subsection.

$α_{\STraces}$ is defined by guarded recursion over the LK trace, in the
following sense:
We regard $(σ_i)_{i∈\overline{n}}$ as a dependent pair type $\STraces
\triangleq ∃n∈ℕ_ω.\ \overline{n} \to \States$, where $ℕ_ω$ is defined
by guarded recursion as \ensuremath{\keyword{data}\;ℕ_ω\mathrel{=}\conid{Z}\mid \conid{S}\;(\later\!\;ℕ_ω)}.
Now $ℕ_ω$ contains all natural numbers (where $n$ is encoded as
\ensuremath{(\conid{S}\hsdot{\circ }{.\ }\varid{next}\;\!)^n~\conid{Z}}) and the transfinite limit ordinal
\ensuremath{\varcolor{\omega}\mathrel{=}\varid{fix}\;(\conid{S}\hsdot{\circ }{.\ }\varid{next})}.
We will assume that addition and subtraction are defined as on Peano numbers,
and \ensuremath{\varcolor{\omega}\mathbin{+}\anonymous \mathrel{=}\anonymous \mathbin{+}\varcolor{\omega}\mathrel{=}\varcolor{\omega}}.
When $(σ_i)_{i∈\overline{n}} ∈ \STraces$ is an LK trace and $n > 1$, then
$(σ_{i+1})_{i∈\overline{n-1}} ∈ \later \STraces$ is the guarded tail of the
trace.

As such, the expression $\idiom{α_{\STraces}((σ_{i+1})_{i∈\overline{n-1}},κ_0)}$
has type \ensuremath{\later\!\;\conid{T}}, where the $\later$ in the type of
$(σ_{i+1})_{i∈\overline{n-1}}$ maps through $α_{\STraces}$ via the idiom
brackets.
Definitional equality $=$ on \ensuremath{\conid{T}} is defined in the obvious structural way by
guarded recursion (as it would be if it was a finite, inductive type).

Function $α_{\STraces}$ expects an LK trace as well as a continuation parameter
$κ_0$ that remains constant; it is initialised to the continuation of the
source state $\cont(σ_0)$ in order to tell apart stuck from balanced traces
when $α_{\Values}$ is ultimately called on the target state.
To that end, the first two equations of $α_{\Values}$ will not match unless the
final continuation is the same as the initial continuation $\cont(σ_0)$,
indicated by the $(κ = κ_0)$ test at the end of the line.

The event abstraction function $α_\Events(σ)$ encodes how intensional
information from small-step transitions is retained as \ensuremath{\conid{Event}}s.
Its implementation does not influence the adequacy result, and we imagine
that this function is tweaked on an as-needed basis to retain more or less
intensional detail, depending on the particular trace property one is interested
in observing.
In our example, we focus on \ensuremath{\conid{Look}\;\varid{y}} events that carry with them the \ensuremath{\varid{y}\mathbin{::}\conid{Name}} of the let binding that allocated the heap entry.
This event corresponds precisely to a $\LookupT(\py)$ transition, so $α_\Events(σ)$
maps $σ$ to \ensuremath{\conid{Look}\;\varid{y}} when $σ$ is about to make a $\LookupT(\py)$ transition.
In that case, the focus expression must be $\px$, and $\py$ is the first
component of the heap entry $μ(ρ(\px))$.
The other cases are similar.

Our first goal is to establish a few auxiliary lemmas showing what kind of
properties of LK traces are preserved by $α_{\STraces}$, and in which way.
Let us warm up by defining a length function on traces:
\begin{hscode}\SaveRestoreHook
\column{B}{@{}>{\hspre}l<{\hspost}@{}}%
\column{3}{@{}>{\hspre}l<{\hspost}@{}}%
\column{19}{@{}>{\hspre}l<{\hspost}@{}}%
\column{E}{@{}>{\hspre}l<{\hspost}@{}}%
\>[3]{}\varid{len}\mathbin{::}\conid{T}\;\varid{a}\to ℕ_ω{}\<[E]%
\\
\>[3]{}\varid{len}\;(\conid{Ret}\;\anonymous ){}\<[19]%
\>[19]{}\mathrel{=}\conid{Z}{}\<[E]%
\\
\>[3]{}\varid{len}\;(\conid{Step}\;\anonymous \;\varcolor{\tau^{\later}})\mathrel{=}\conid{S}\;\idiom{\varid{len}\;\varcolor{\tau^{\later}}}{}\<[E]%
\ColumnHook
\end{hscode}\resethooks
The length is an important property of LK traces that is preserved by $α$.
\begin{lemma}[Abstraction preserves length]
  \label{thm:abs-length}
  Let $(σ_i)_{i∈\overline{n}}$ be a trace. Then
  \[ \ensuremath{\varid{len}\;( α_{\STraces}((σ_i)_{i∈\overline{n}},\cont(σ_0)) )} = n. \]
\end{lemma}
\begin{proof}
  This is quite simple to see and hence a good opportunity to familiarise
  ourselves with Löb induction, the induction principle of guarded recursion.
  Löb induction arises simply from applying the guarded recursive fixpoint
  combinator to a proposition:
  \[
    \textsf{löb} = \fix : \forall P.\ (\later P \Longrightarrow P) \Longrightarrow P
  \]
  That is, we assume that our proposition holds \emph{later}, \ie
  \[
    IH ∈ (\later P_{\ensuremath{\varid{len}}} \triangleq \later (
        \forall (σ_i)_{i∈\overline{n}}.\
        \ensuremath{\varid{len}\;( α_{\STraces}((σ_i)_{i∈\overline{n}},\cont(σ_0)) )} = n
      )),
  \]
  and use $IH$ to prove $P_{\ensuremath{\varid{len}}}$.

  To that end, let $(σ_i)_{i∈\overline{n}}$ be an LK trace and define
  $\ensuremath{\varcolor{\tau}} \triangleq α_{\STraces}((σ_i)_{i∈\overline{n}},\cont(σ_0))$.
  Now proceed by case analysis on $n$:
  \begin{itemize}
    \setlength{\itemsep}{0pt}
    \item \textbf{Case \ensuremath{\conid{Z}}}: Then we have either
      \ensuremath{\varcolor{\tau}\mathrel{=}\conid{Ret}\;\anonymous } which maps to \ensuremath{\conid{Z}} under \ensuremath{\varid{len}}.
    \item \textbf{Case \ensuremath{\conid{S}\;\idiom{\varid{m}}}}: Then
      \ensuremath{\varcolor{\tau}\mathrel{=}\conid{Step}\;\anonymous \;\idiom{α_{\STraces}((σ_{i+1})_{i∈\overline{m}},\cont(σ_0))}},
      where $(σ_{i+1})_{i∈\overline{m}} ∈ \later \STraces$ is the guarded
      tail of the LK trace $(σ_i)_{i∈\overline{n}}$.
      Now we apply the inductive hypothesis, as follows:
      \[
        (IH \aplater (σ_{i+1})_{i∈\overline{m}}) \in \later (\ensuremath{\varid{len}\;( α_{\STraces}((σ_{i+1})_{i∈\overline{m}},\cont(σ_0)) )} = m).
      \]
      We use this fact and congruence to prove
      \[\begin{array}{lclcl}
        n & = & \ensuremath{\conid{S}\;\idiom{\varid{m}}} & = & \ensuremath{\conid{S}\;\idiom{\varid{len}\;( α_{\STraces}((σ_{i+1})_{i∈\overline{m}},\cont(σ_0)) )}} \\
          &   &               & = & \ensuremath{\varid{len}\;( α_{\STraces}((σ_{i})_{i∈\overline{n}},\cont(σ_0)) )}.
      \end{array}\]
  \end{itemize}
\end{proof}

It is rewarding to study the use of Löb induction in the proof above in detail,
because many proofs in this subsection as well as \Cref{sec:soundness} will make
good use of it.

The next step is to prove that $α_{\STraces}$ preserves the termination
observable; then all that is left to do is to show that \ensuremath{\mathcal{S}_{\mathbf{need}}\denot{\wild}_{\wild}} abstracts
LK traces via $α_{\STraces}$.
The preservation property is formally expressed as follows:

\begin{lemmarep}[Abstraction preserves termination observable]
  \label{thm:abs-max-trace}
  Let $(σ_i)_{i∈\overline{n}}$ be a maximal trace.
  Then $α_{\STraces}((σ_i)_{i∈\overline{n}}, cont(σ_0))$
  \begin{itemize}
  \setlength{\itemsep}{0pt}
    \item ends in \ensuremath{\conid{Ret}\;(\conid{Fun}\;\anonymous ,\anonymous )} or \ensuremath{\conid{Ret}\;(\conid{Con}\;\anonymous \;\anonymous ,\anonymous )} if and only
      if $(σ_i)_{i∈\overline{n}}$ is balanced.
    \item is infinite if and only if $(σ_i)_{i∈\overline{n}}$ is diverging.
    \item ends in \ensuremath{\conid{Ret}\;(\conid{Stuck},\anonymous )} if and only if $(σ_i)_{i∈\overline{n}}$ is stuck.
  \end{itemize}
\end{lemmarep}
\begin{proof}
  The second point follows by a similar inductive argument as in \Cref{thm:abs-length}.

  In the other cases, we may assume that $n$ is finite.
  If $(σ_i)_{i∈\overline{n}}$ is balanced, then $σ_n$ is a return state with
  continuation $\cont(σ_0)$, so its control expression is a value.
  Then $α_{\STraces}$ will conclude with \ensuremath{\conid{Ret}\;(\varid{αValue}\;\anonymous \;\anonymous )}, and the latter is
  never \ensuremath{\conid{Ret}\;(\conid{Stuck},\anonymous )}.
  Conversely, if the trace ended with \ensuremath{\conid{Ret}\;(\conid{Fun}\;\anonymous )} or \ensuremath{\conid{Ret}\;(\conid{Con}\;\anonymous \;\anonymous )},
  then $\cont(σ_n) = \cont(σ_0)$ and $\ctrl(σ_n)$ is a value, so
  $(σ_i)_{i∈\overline{n}}$ forms a balanced trace.
  The stuck case is similar.
\end{proof}

The previous lemma allows us to apply the classifying terminology of maximal LK
traces to a \ensuremath{\varcolor{\tau}\mathbin{::}\conid{T}} in the range of $α_{\STraces}$.
For such a maximal \ensuremath{\varcolor{\tau}} we will say that it is balanced when it ends with
\ensuremath{\conid{Ret}\;(\varid{v},\varcolor{\mu})} for a \ensuremath{\varid{v}\mathrel{\clipbox{0pt 1.5pt 3pt 1.5pt}{\slantbox[.3]{$\models$}}\hspace{0.5pt}\clipbox{0pt 0pt 2.5pt 0pt}{$=$}}\conid{Stuck}}, stuck if ending in \ensuremath{\conid{Ret}\;(\conid{Stuck},\varcolor{\mu})} and diverging if
infinite.

The final remaining step is to prove that \ensuremath{\mathcal{S}_{\mathbf{need}}\denot{\wild}_{\wild}} produces an abstraction of
traces in the LK machine:

\begin{theorem}[\ensuremath{\mathcal{S}_{\mathbf{need}}\denot{\wild}_{\wild}} abstracts LK machine]
  \label{thm:need-abstracts-lk}
  Let $(σ_i)_{i∈\overline{n}}$ be a maximal LK trace with source
  state $(\pe,ρ,μ,κ)$.
  Then $α_{\STraces}((σ_i)_{i∈\overline{n}},κ) = \ensuremath{\mathcal{S}_{\mathbf{need}}\denot{\varid{e}}_{α_\Environments(\varcolor{\rho},\varcolor{\mu})}(α_\Heaps(\varcolor{\mu}))}$,
  where $α_{\STraces}$ is the abstraction function defined in
  \Cref{fig:eval-correctness}.
\end{theorem}
\begin{proof}
Let us abbreviate the proposed correctness relation as
\[\begin{array}{c}
  P_{α}((σ_i)_{i∈\overline{n}})
  \triangleq
  \maxtrace{(σ_i)_{i∈\overline{n}}}
  \Longrightarrow
  α_{\STraces}((σ_i)_{i∈\overline{n}},κ) = \ensuremath{\mathcal{S}_{\mathbf{need}}\denot{\varid{e}}_{α_\Environments(\varcolor{\rho},\varcolor{\mu})}(α_\Heaps(\varcolor{\mu}))} \\
  \hspace{12em}\text{where }(\pe,ρ,μ,κ) = σ_0
\end{array}\]
We prove it by Löb induction, with $IH ∈ \later P_{α}$ as the inductive hypothesis.

Now let $(σ_i)_{i∈\overline{n}}$ be a maximal LK trace with source state
$σ_0=(\pe,ρ,μ,κ)$ and let \ensuremath{\varcolor{\tau}\triangleq\mathcal{S}_{\mathbf{need}}\denot{\varid{e}}_{α_\Environments(\varcolor{\rho},\varcolor{\mu})}(α_\Heaps(\varcolor{\mu}))}.
Then the goal is to show $α_{\STraces}((σ_i)_{i∈\overline{n}},κ) = \ensuremath{\varcolor{\tau}}$.
We do so by cases over $\pe$, abbreviating \ensuremath{\varcolor{\tm}\triangleqα_\Heaps(\varcolor{\mu})} and \ensuremath{\varcolor{\tr}\triangleqα_\Environments(\varcolor{\rho},\varcolor{\mu})}:
\begin{itemize}
  \item \textbf{Case $\px$}:
    Note first that $\LookupT$ is the only applicable transition rule according
    to rule inversion on $\ctrl(σ_0) = \px$.

    In case that $n = 0$, $(σ_i)_{i∈\overline{n}}$ is stuck because
    $\ctrl(σ_0)$ is not a value, hence $α_{\STraces}$ returns \ensuremath{\conid{Ret}\;(\conid{Stuck},\anonymous )}.
    Since $\LookupT$ does not apply (otherwise $n > 0$), we must have $\px
    \not∈ \dom(ρ)$, and hence \ensuremath{\varcolor{\tau}\mathrel{=}\conid{Ret}\;(\conid{Stuck},\anonymous )} by calculation as well.

    Otherwise, $σ_1 \triangleq (\pe', ρ_1, μ, \UpdateF(\pa) \pushF κ), σ_0 \smallstep σ_1$
    via $\LookupT(\py)$, and $ρ(\px) = \pa, μ(\pa) = (\py, ρ_1, \pe')$.
    This matches \ensuremath{\varcolor{\tr}\mathop{!}\varid{x}\mathrel{=}\varid{step}\;(\conid{Look}\;\varid{y})\;(\varid{fetch}\;\varid{a})} in the interpreter.

    It suffices to show that the tails equate \emph{later}.

    We can infer that \ensuremath{\varcolor{\tm}\mathop{!}\varid{a}\mathrel{=}\varid{memo}\;\varid{a}\;(\mathcal{S}_{\mathbf{need}}\denot{\varid{e}'}_{\varcolor{\tr}})} by definition of
    $α_\Heaps$, so
    \begin{hscode}\SaveRestoreHook
\column{B}{@{}>{\hspre}l<{\hspost}@{}}%
\column{7}{@{}>{\hspre}l<{\hspost}@{}}%
\column{9}{@{}>{\hspre}l<{\hspost}@{}}%
\column{18}{@{}>{\hspre}l<{\hspost}@{}}%
\column{E}{@{}>{\hspre}l<{\hspost}@{}}%
\>[7]{}\varid{fetch}\;\varid{a}\;\varcolor{\tm}\mathrel{=}(\varcolor{\tm}\mathop{!}\varid{a})\;\varcolor{\tm}\mathrel{=}\mathcal{S}_{\mathbf{need}}\denot{\varid{e}'}_{\varcolor{\tr}}(\varcolor{\tm})\bind \lambda \keyword{case}{}\<[E]%
\\
\>[7]{}\hsindent{2}{}\<[9]%
\>[9]{}(\conid{Stuck},{}\<[18]%
\>[18]{}\varcolor{\tm})\to \conid{Ret}\;(\conid{Stuck},\varcolor{\tm}){}\<[E]%
\\
\>[7]{}\hsindent{2}{}\<[9]%
\>[9]{}(\varid{val},{}\<[18]%
\>[18]{}\varcolor{\tm})\to \conid{Step}\;\conid{Upd}\;(\conid{Ret}\;(\varid{val},\varcolor{\tm}[\varid{a}\mapsto\varid{memo}\;\varid{a}\;(\varid{return}\;\varid{val})])){}\<[E]%
\ColumnHook
\end{hscode}\resethooks

    Let us define \ensuremath{\varcolor{\tau^{\later}}\triangleq\idiom{\mathcal{S}_{\mathbf{need}}\denot{\varid{e}'}_{\varcolor{\tr}}(\varcolor{\tm})}} and apply the induction
    hypothesis $IH$ to the maximal trace starting at $σ_1$.
    This yields an equality
    \[
      IH \aplater (σ_{i+1})_{i∈\overline{m}} ∈ \idiom{α_{\STraces}((σ_{i+1})_{i∈\overline{m}},\UpdateF(\pa) \pushF κ) = τ^{\later}}
    \]
    Any \ensuremath{\conid{Step}}s in \ensuremath{\varcolor{\tau^{\later}}} match the transitions of
    $(σ_{i+1})_{i∈\overline{m}}$ per $IH$, and \ensuremath{\bind } simply forwards these
    \ensuremath{\conid{Step}}s.
    What remains to be shown is that the continuation passed to \ensuremath{\bind }
    operates correctly.

    If \ensuremath{\varcolor{\tau^{\later}}} is infinite, we are done, because the continuation is never called.
    If \ensuremath{\varcolor{\tau^{\later}}} ends in \ensuremath{\conid{Ret}\;(\conid{Stuck},\varcolor{\tm})}, then the continuation
    will return \ensuremath{\conid{Ret}\;(\conid{Stuck},\varcolor{\tm})}, indicating by \Cref{thm:abs-length} and
    \Cref{thm:abs-max-trace} that $(σ_{i+1})_{i∈\overline{n-1}}$ is stuck and
    hence $(σ_i)_{i∈\overline{n}}$ is stuck as well with the compatible heap
    from $σ_{n-1}$.

    Otherwise \ensuremath{\varcolor{\tau^{\later}}} ends after $m-1$ \ensuremath{\conid{Step}}s with \ensuremath{\conid{Ret}\;(\varid{val},\varcolor{\tm_m})} and
    by \Cref{thm:abs-max-trace} $(σ_{i+1})_{i∈\overline{m}}$ is balanced; hence
    $\cont(σ_m) = \UpdateF(\pa) \pushF κ$ and $\ctrl(σ_m)$ is a value.
    So $σ_m = (\pv,ρ_m,μ_m,\UpdateF(\pa) \pushF κ)$ and the
    $\UpdateT$ transition fires, reaching $(\pv,ρ_m,μ_m[\pa ↦ (\py, ρ_m, \pv)],κ)$
    and this must be the target state $σ_n$ (so $m = n-2$), because it remains
    a return state and has continuation $κ$, so $(σ_i)_{i∈\overline{n}}$ is
    balanced.
    Likewise, the continuation argument of \ensuremath{\bind } does a \ensuremath{\conid{Step}\;\conid{Upd}} on
    \ensuremath{\conid{Ret}\;(\varid{val},\varcolor{\tm_m})}, updating the heap.
    By cases on $\pv$ and the \ensuremath{\conid{Domain}\;(\conid{D}_{\mathbf{ne}})} instance we can see that
    \begin{hscode}\SaveRestoreHook
\column{B}{@{}>{\hspre}l<{\hspost}@{}}%
\column{7}{@{}>{\hspre}c<{\hspost}@{}}%
\column{7E}{@{}l@{}}%
\column{10}{@{}>{\hspre}l<{\hspost}@{}}%
\column{E}{@{}>{\hspre}l<{\hspost}@{}}%
\>[10]{}\conid{Ret}\;(\varid{val},\varcolor{\tm_m}[\varid{a}\mapsto\varid{memo}\;\varid{a}\;(\varid{return}\;\varid{val})]){}\<[E]%
\\
\>[7]{}\mathrel{=}{}\<[7E]%
\>[10]{}\conid{Ret}\;(\varid{val},\varcolor{\tm_m}[\varid{a}\mapsto\varid{memo}\;\varid{a}\;(\mathcal{S}_{\mathbf{need}}\denot{\varid{v}}_{\varcolor{\tr_m}})]){}\<[E]%
\\
\>[7]{}\mathrel{=}{}\<[7E]%
\>[10]{}\conid{Ret}\;(\varid{αValue}\;\sigma_n\;\kappa, α_\Heaps(μ_m[\pa ↦ (\py, ρ_m, \pv)]) ){}\<[E]%
\ColumnHook
\end{hscode}\resethooks
    and this equality concludes the proof, because the heap in $σ_n$ is
    exactly $μ_m[\pa ↦ (\py, ρ_m, \pv)]$.

  \item \textbf{Case $\pe~\px$}:
    The cases where $τ$ gets stuck or diverges before finishing evaluation of
    $\pe$ are similar to the variable case.
    So let us focus on the situation when \ensuremath{\varcolor{\tau^{\later}}\triangleq\idiom{\mathcal{S}_{\mathbf{need}}\denot{\varid{e}}_{\varcolor{\tr}}(\varcolor{\tm})}}
    returns and let $σ_m$ be LK state at the end of the balanced trace
    $(σ_{i+1})_{i∈\overline{m-1}}$ through $\pe$ starting in stack
    $\ApplyF(\pa) \pushF κ$.

    Now, either there exists a transition $σ_m \smallstep σ_{m+1}$, or it does
    not.
    When the transition exists, it must must leave the stack $\ApplyF(\pa)
    \pushF κ$ due to maximality, necessarily by an $\AppET$ transition.
    That in turn means that the value in $\ctrl(σ_m)$ must be a lambda
    $\Lam{\py}{\pe'}$, and $σ_{m+1} = (\pe',ρ_m[\py ↦ ρ(\px)],μ_m,κ)$.

    Likewise, \ensuremath{\varcolor{\tau^{\later}}} ends in
    \hfuzz=1em
    \begin{hscode}\SaveRestoreHook
\column{B}{@{}>{\hspre}l<{\hspost}@{}}%
\column{7}{@{}>{\hspre}l<{\hspost}@{}}%
\column{E}{@{}>{\hspre}l<{\hspost}@{}}%
\>[7]{} α_\Values(σ_m, \ApplyF(\pa) \pushF κ) \mathrel{=}\conid{Fun}\;(\lambda \varid{d}\to \varid{step}\;\conid{App}_{2}\;(\mathcal{S}_{\mathbf{need}}\denot{\varid{e}'}_{\varcolor{\tr_m}[\varid{y}\mapsto\varid{d}]})){}\<[E]%
\ColumnHook
\end{hscode}\resethooks
    (where \ensuremath{\varcolor{\tr_m}} corresponds to the environment in $σ_m$ in the usual way,
    similarly for \ensuremath{\varcolor{\tm_m}}).
    The \ensuremath{\varid{apply}} implementation of \ensuremath{\conid{Domain}\;(\conid{D}_{\mathbf{ne}})} applies the \ensuremath{\conid{Fun}} value
    to the argument denotation \ensuremath{\varcolor{\tr}\mathop{!}\varid{x}}, hence it remains to be shown that
    \ensuremath{\mathcal{S}_{\mathbf{need}}\denot{\varid{e}'}_{\varcolor{\tr_m}[\varid{y}\mapsto\varcolor{\tr}\;\varid{x}]}(\varcolor{\tm_m})} is equal to
    $α_{\STraces}((σ_{i+m+1})_{i∈\overline{k}}, κ)$ \emph{later},
    where $(σ_{i+m+1})_{i∈\overline{k}}$ is the maximal trace starting at
    $σ_{m+1}$.

    We can once again apply the induction hypothesis to this situation.
    From this and our earlier equalities, we get
    $α_{\STraces}((σ_i)_{i∈\overline{n}},κ) = \ensuremath{\varcolor{\tau}}$, concluding the proof of
    the case where there exists a transition $σ_m \smallstep σ_{m+1}$.

    When $σ_m \not\smallstep$, then $\ctrl(σ_m)$ is not a lambda, otherwise
    $\AppET$ would apply.
    In this case, \ensuremath{\varid{apply}} gets to see a \ensuremath{\conid{Stuck}} or \ensuremath{\conid{Con}} value, for which it
    is \ensuremath{\conid{Stuck}} as well.

  \item \textbf{Case $\Case{\pe_s}{\Sel[r]}$}:
    Similar to the application and lookup case.

  \item \textbf{Cases $\Lam{\px}{\pe}$, $K~\many{\px}$}:
    The length of both traces is $n = 0$ and the goal follows by simple calculation.

  \item \textbf{Case $\Let{\px}{\pe_1}{\pe_2}$}:
    Let $σ_0 = (\Let{\px}{\pe_1}{\pe_2},ρ,μ,κ)$.
    Then $σ_1 = (\pe_2, ρ_1, μ',κ)$ by $\LetIT$, where $ρ_1 = ρ[\px↦\pa_{\px,i}],
    μ' = μ[\pa_{\px,i}↦(\px,ρ_1,\pe_1)]$.
    Since the stack does not grow, maximality from the tail $(σ_{i+1})_{i∈\overline{n-1}}$
    transfers to $(σ_{i})_{i∈\overline{n}}$.
    Straightforward application of the induction hypothesis to
    $(σ_{i+1})_{i∈\overline{n-1}}$ yields the equality for the tail (after a bit
    of calculation for the updated environment and heap), which concludes the
    proof.
\end{itemize}
\end{proof}

\Cref{thm:need-abstracts-lk} and \Cref{thm:abs-max-trace} are the key to
proving the following theorem of adequacy, which formalises the intuitive
notion of adequacy from before.

(A state $σ$ is \emph{final} when $\ctrl(σ)$ is a value and $\cont(σ) = \StopF$.)

\begin{theorem}[Adequacy of \ensuremath{\mathcal{S}_{\mathbf{need}}\denot{\wild}_{\wild}}]
  \label{thm:need-adequate}
  Let \ensuremath{\varcolor{\tau}\triangleq\mathcal{S}_{\mathbf{need}}\denot{\varid{e}}_{\varcolor{\varepsilon}}(\varcolor{\varepsilon})}.
  \begin{itemize}
    \item
      \ensuremath{\varcolor{\tau}} ends with \ensuremath{\conid{Ret}\;(\conid{Fun}\;\anonymous ,\anonymous )} or \ensuremath{\conid{Ret}\;(\conid{Con}\;\anonymous \;\anonymous ,\anonymous )} (is balanced) iff there
      exists a final state $σ$ such that $\init(\pe) \smallstep^* σ$.
    \item
      \ensuremath{\varcolor{\tau}} ends with \ensuremath{\conid{Ret}\;(\conid{Stuck},\anonymous )} (is stuck) iff there exists a non-final
      state $σ$ such that $\init(\pe) \smallstep^* σ$ and there exists no $σ'$
      such that $σ \smallstep σ'$.
    \item
      \ensuremath{\varcolor{\tau}} is infinite \ensuremath{\conid{Step}\;\anonymous \;(\conid{Step}\;\anonymous \;\mathbin{...})} (is diverging) iff for all $σ$ with
      $\init(\pe) \smallstep^* σ$ there exists $σ'$ with $σ \smallstep σ'$.
    \item
      The \ensuremath{\varid{e}\mathbin{::}\conid{Event}} in every \ensuremath{\conid{Step}\;\varid{e}\;\mathbin{...}} occurrence in \ensuremath{\varcolor{\tau}} corresponds in
      the intuitive way to the matching small-step transition rule that was
      taken.
  \end{itemize}
\end{theorem}
\begin{proof}
  There exists a maximal trace $(σ_i)_{i∈\overline{n}}$ starting
  from $σ_0 = \init(\pe)$, and by \Cref{thm:need-abstracts-lk} we have
  $α_{\STraces}((σ_i)_{i∈\overline{n}},\StopF) = τ$.
  The correctness of \ensuremath{\conid{Event}}s emitted follows directly from $α_\Events$.
  \begin{itemize}
    \item[$\Rightarrow$]
      \begin{itemize}
        \item
          If $(σ_i)_{i∈\overline{n}}$ is balanced, its target state $σ_n$
          is a return state that must also have the empty continuation, hence it
          is a final state.
        \item
          If $(σ_i)_{i∈\overline{n}}$ is stuck, it is finite and maximal, but not balanced, so its
          target state $σ_n$ cannot be a return state;
          otherwise maximality implies $σ_n$ has an (initial) empty continuation
          and the trace would be balanced. On the other hand, the only returning
          transitions apply to return states, so maximality implies there is no
          $σ'$ such that $σ \smallstep σ'$ whatsoever.
        \item
          If $(σ_i)_{i∈\overline{n}}$ is diverging, $n=ω$ and for every $σ$ with
          $\init(\pe) \smallstep^* σ$ there exists an $i$ such that $σ = σ_i$ by
          determinism.
      \end{itemize}

    \item[$\Leftarrow$]
      \begin{itemize}
        \item
          If $σ_n$ is a final state, it has $\cont(σ) = \cont(\init(\pe)) = []$,
          so the trace is balanced.
        \item
          If $σ$ is not a final state, $τ'$ is not balanced. Since there is no
          $σ'$ such that $σ \smallstep^* σ'$, it is still maximal; hence it must
          be stuck.
        \item
          Suppose that $n∈ℕ_ω$ was finite.
          Then, if for every choice of $σ$ there exists $σ'$ such that $σ
          \smallstep σ'$, then there must be $σ_{n+1}$ with $σ_n \smallstep
          σ_{n+1}$, violating maximality of the trace.
          Hence it must be infinite.
          It is also interior, because every stack extends the empty stack,
          hence it is diverging.
      \end{itemize}
  \end{itemize}
\end{proof}
\end{toappendix}

\subsection{Totality of \ensuremath{\mathcal{S}_{\mathbf{name}}\denot{\wild}_{\wild}} and \ensuremath{\mathcal{S}_{\mathbf{need}}\denot{\wild}_{\wild}}}
\label{sec:totality}

\begin{theorem}[Totality]
The interpreters \ensuremath{\mathcal{S}_{\mathbf{name}}\denot{\varid{e}}_{\varcolor{\rho}}} and \ensuremath{\mathcal{S}_{\mathbf{need}}\denot{\varid{e}}_{\varcolor{\rho}}(\varcolor{\mu})} are defined for every
\ensuremath{\varid{e}}, \ensuremath{\varcolor{\rho}}, \ensuremath{\varcolor{\mu}}.
\end{theorem}
\begin{proofsketch}
In the Supplement, we provide an implementation of the generic interpreter
\ensuremath{\mathcal{S}\denot{\wild}_{\wild}} and its instances at \ensuremath{\conid{ByName}} and \ensuremath{\conid{ByNeed}} in Guarded Cubical Agda,
which offers a total type theory with \emph{guarded recursive
types}~\citet{tctt}.
Agda enforces that all encodable functions are total, therefore \ensuremath{\mathcal{S}_{\mathbf{name}}\denot{\wild}_{\wild}} and
\ensuremath{\mathcal{S}_{\mathbf{need}}\denot{\wild}_{\wild}} must be total as well.

The essential idea of the totality proof is that \emph{there is only a finite
number of transitions between every $\LookupT$ transition}.
In other words, if every environment lookup produces a \ensuremath{\conid{Step}} constructor, then
our semantics is total by coinduction.
Such an argument is quite natural to encode in guarded recursive types, hence
our use of Guarded Cubical Agda is appealing.
See \Cref*{sec:totality-detail} for the details of the encoding in Agda.
\end{proofsketch}


\section{Static Analysis}
\label{sec:abstraction}

So far, our semantic domains have all been \emph{infinite}, simply because the
dynamic traces they express are potentially infinite as well.
However, by instantiating the generic denotational interpreter with
a semantic domain in which every element is \emph{finite data}, we can run the
interpreter on the program statically, at compile time, to yield a \emph{finite}
abstraction of the dynamic behavior.
This gives us a \emph{static program analysis}.


We can get a wide range of static analyses by choosing appropriate semantic domains.
For example, we have successfully realised the following analyses as
denotational interpreters:
\begin{itemize}
  \item
    \Cref{sec:usage-analysis} defines a summary-based \emph{usage analysis},
    the running example of this work.
    We prove that usage analysis correctly infers absence in \Cref{sec:soundness}.

  \item
    \Cref{sec:type-analysis} defines a \emph{type analysis} with
    let generalisation that implements \citeauthor{Milner:78}'s Algorithm~J,
    inferring polytypes such as $\forall α_3.\ \mathtt{option}\;(α_3
    \rightarrow α_3)$ that act as summaries.

  \item
    \Cref*{sec:0cfa} defines 0CFA \emph{control-flow analysis}~\citep{Shivers:91},
    a non-modular analysis lacking a finite summary mechanism, simply as a
    proof of concept.

  \item
    To demonstrate that our framework scales to real-world compilers,
    we have refactored relevant parts of \emph{Demand Analysis} in the Glasgow
    Haskell Compiler into an abstract denotational interpreter as an artefact.
    The resulting compiler bootstraps and passes the testsuite.%
    \footnote{There is a small caveat: we did not try to optimise for compiler
    performance in our proof of concept and hence it regresses in a few
    compiler performance test cases.
    None of the runtime performance test cases regress and the inferred
    demand signatures stay unchanged.}
    Demand Analysis is the real-world implementation of the cardinality analysis
    work of \citet{Sergey:14}, generalising usage analysis and implementing
    strictness analysis as well.
    For a report of this case study, we defer to \Cref*{sec:demand-analysis}.

  \item
    Static compiler analyses such as Demand Analysis usually drive a subsequent
    optimisation, for which a single denotation for the entire program is
    insufficient.
    Rather, we need one for every sub-expression, or at least every binding.
    \Cref*{sec:annotations} proposes a very slight generalisation of the
    \ensuremath{\conid{Domain}} type class that lifts a stateless analysis into a stateful
    analysis writing out annotations for let bindings in a separate, global map.
    As a substantial bonus, we can use another stateful map to cache the
    results of fixpoint iterations.

%
%

\end{itemize}

In the following, we discuss usage analysis (\Cref{sec:usage-analysis}) and type
analysis (\Cref{sec:type-analysis}) in detail.

\subsection{Usage Analysis}
\label{sec:usage-analysis}

In this subsection, we give a detailed account of \emph{usage analysis} as
an instance of the denotational interpreter.
Usage analysis generalises the summary-based absence analysis from
\Cref{sec:problem}.
It is a compelling example because it illustrates that our framework is suitable
to infer \emph{operational properties}, such as an upper bound on the number of
variable lookups.

\subsubsection{Trace Abstraction in \ensuremath{\conid{Trace}\;\concolor{\mathsf{T_U}}}}
\label{sec:usage-trace-abstraction}

\begin{figure}
\begin{minipage}{0.4\textwidth}
\begin{hscode}\SaveRestoreHook
\column{B}{@{}>{\hspre}l<{\hspost}@{}}%
\column{3}{@{}>{\hspre}l<{\hspost}@{}}%
\column{8}{@{}>{\hspre}l<{\hspost}@{}}%
\column{E}{@{}>{\hspre}l<{\hspost}@{}}%
\>[B]{}\keyword{data}\;\conid{U}\mathrel{=}\concolor{\mathsf{U_0}}\mid \concolor{\mathsf{U_1}}\mid \concolor{\mathsf{U_\omega}}{}\<[E]%
\\
\>[B]{}\keyword{type}\;\conid{Uses}\mathrel{=}\conid{Name}\mathbin{:\rightharpoonup}\conid{U}{}\<[E]%
\\
\>[B]{}\keyword{class}\;\conid{UVec}\;\varid{a}\;\keyword{where}{}\<[E]%
\\
\>[B]{}\hsindent{3}{}\<[3]%
\>[3]{}(\mathbin{+}){}\<[8]%
\>[8]{}\mathbin{::}\varid{a}\to \varid{a}\to \varid{a}{}\<[E]%
\\
\>[B]{}\hsindent{3}{}\<[3]%
\>[3]{}(\mathbin{*}){}\<[8]%
\>[8]{}\mathbin{::}\conid{U}\to \varid{a}\to \varid{a}{}\<[E]%
\\
\>[B]{}\keyword{instance}\;\conid{UVec}\;\conid{U}\;\keyword{where}\; ... {}\<[E]%
\\
\>[B]{}\keyword{instance}\;\conid{UVec}\;\conid{Uses}\;\keyword{where}\; ... {}\<[E]%
\ColumnHook
\end{hscode}\resethooks
\end{minipage}%
\begin{minipage}{0.6\textwidth}
\begin{hscode}\SaveRestoreHook
\column{B}{@{}>{\hspre}l<{\hspost}@{}}%
\column{3}{@{}>{\hspre}l<{\hspost}@{}}%
\column{18}{@{}>{\hspre}l<{\hspost}@{}}%
\column{30}{@{}>{\hspre}l<{\hspost}@{}}%
\column{E}{@{}>{\hspre}l<{\hspost}@{}}%
\>[B]{}\keyword{data}\;\concolor{\mathsf{T_U}}\;\varid{v}\mathrel{=}\langle \conid{Uses}, \varid{v} \rangle{}\<[E]%
\\
\>[B]{}\keyword{instance}\;\conid{Trace}\;(\concolor{\mathsf{T_U}}\;\varid{v})\;\keyword{where}{}\<[E]%
\\
\>[B]{}\hsindent{3}{}\<[3]%
\>[3]{}\varid{step}\;(\conid{Look}\;\varid{x})\;{}\<[18]%
\>[18]{}\langle \varcolor{\varphi}, \varid{v} \rangle{}\<[30]%
\>[30]{}\mathrel{=}\langle [\varid{x}\mapsto\concolor{\mathsf{U_1}}]\mathbin{+}\varcolor{\varphi}, \varid{v} \rangle{}\<[E]%
\\
\>[B]{}\hsindent{3}{}\<[3]%
\>[3]{}\varid{step}\;\anonymous \;{}\<[18]%
\>[18]{}\varcolor{\tau}{}\<[30]%
\>[30]{}\mathrel{=}\varcolor{\tau}{}\<[E]%
\\
\>[B]{}\keyword{instance}\;\conid{Monad}\;\concolor{\mathsf{T_U}}\;\keyword{where}{}\<[E]%
\\
\>[B]{}\hsindent{3}{}\<[3]%
\>[3]{}\varid{return}\;\varid{a}\mathrel{=}\langle \varcolor{\varepsilon}, \varid{a} \rangle{}\<[E]%
\\
\>[B]{}\hsindent{3}{}\<[3]%
\>[3]{}\langle \varcolor{\varphi}_{1}, \varid{a} \rangle\bind \varid{k}\mathrel{=}\keyword{let}\;\langle \varcolor{\varphi}_{2}, \varid{b} \rangle\mathrel{=}\varid{k}\;\varid{a}\;\keyword{in}\;\langle \varcolor{\varphi}_{1}\mathbin{+}\varcolor{\varphi}_{2}, \varid{b} \rangle{}\<[E]%
\ColumnHook
\end{hscode}\resethooks
\end{minipage}
\\
\begin{minipage}{0.63\textwidth}
\begin{hscode}\SaveRestoreHook
\column{B}{@{}>{\hspre}l<{\hspost}@{}}%
\column{3}{@{}>{\hspre}l<{\hspost}@{}}%
\column{5}{@{}>{\hspre}l<{\hspost}@{}}%
\column{15}{@{}>{\hspre}c<{\hspost}@{}}%
\column{15E}{@{}l@{}}%
\column{18}{@{}>{\hspre}l<{\hspost}@{}}%
\column{42}{@{}>{\hspre}l<{\hspost}@{}}%
\column{E}{@{}>{\hspre}l<{\hspost}@{}}%
\>[B]{}\mathcal{S}_{\mathbf{usage}}\denot{\varid{e}}_{\varcolor{\rho}}\mathrel{=}\mathcal{S}\denot{\varid{e}}_{\varcolor{\rho}}\mathbin{::}\concolor{\mathsf{D_U}}{}\<[E]%
\\[\blanklineskip]%
\>[B]{}\keyword{instance}\;\conid{Domain}\;\concolor{\mathsf{D_U}}\;\keyword{where}{}\<[E]%
\\
\>[B]{}\hsindent{3}{}\<[3]%
\>[3]{}\varid{stuck}{}\<[42]%
\>[42]{}\mathrel{=}\bot{}\<[E]%
\\
\>[B]{}\hsindent{3}{}\<[3]%
\>[3]{}\varid{fun}\;\varid{x}\; \varid{f}{}\<[42]%
\>[42]{}\mathrel{=}\keyword{case}\;\varid{f}\;\langle [\varid{x}\mapsto\concolor{\mathsf{U_1}}], \conid{Rep}\;\concolor{\mathsf{U_\omega}} \rangle\;\keyword{of}{}\<[E]%
\\
\>[3]{}\hsindent{2}{}\<[5]%
\>[5]{}\langle \varcolor{\varphi}, \varid{v} \rangle\to \langle \varcolor{\varphi}[\varid{x}\mapsto\concolor{\mathsf{U_0}}], \varcolor{\varphi}\mathbin{!?}\varid{x} \argcons \varid{v} \rangle{}\<[E]%
\\
\>[B]{}\hsindent{3}{}\<[3]%
\>[3]{}\varid{apply}\;\langle \varcolor{\varphi}_{1}, \varid{v}_{1} \rangle\;\langle \varcolor{\varphi}_{2}, \anonymous  \rangle{}\<[42]%
\>[42]{}\mathrel{=}\keyword{case}\;\varid{peel}\;\varid{v}_{1}\;\keyword{of}{}\<[E]%
\\
\>[3]{}\hsindent{2}{}\<[5]%
\>[5]{}(\varid{u},\varid{v}_{2})\to \langle \varcolor{\varphi}_{1}\mathbin{+}\varid{u}\mathbin{*}\varcolor{\varphi}_{2}, \varid{v}_{2} \rangle{}\<[E]%
\\
\>[B]{}\hsindent{3}{}\<[3]%
\>[3]{}\varid{con}\; \anonymous \;\varid{ds}{}\<[42]%
\>[42]{}\mathrel{=}\varid{foldl}\;\varid{apply}\;\langle \varcolor{\varepsilon}, \conid{Rep}\;\concolor{\mathsf{U_\omega}} \rangle\;\varid{ds}{}\<[E]%
\\
\>[B]{}\hsindent{3}{}\<[3]%
\>[3]{}\varid{select}\;\varid{d}\;\varid{fs}{}\<[42]%
\>[42]{}\mathrel{=}{}\<[E]%
\\
\>[3]{}\hsindent{2}{}\<[5]%
\>[5]{}\varid{d}\sequ \varid{lub}\;{}\<[15]%
\>[15]{}[\mskip1.5mu {}\<[15E]%
\>[18]{}\varid{f}\;(\varid{replicate}\;(\varid{conArity}\;\varid{k})\;\langle \varcolor{\varepsilon}, \conid{Rep}\;\concolor{\mathsf{U_\omega}} \rangle){}\<[E]%
\\
\>[15]{}\mid {}\<[15E]%
\>[18]{}(\varid{k},\varid{f})\leftarrow \varid{assocs}\;\varid{fs}\mskip1.5mu]{}\<[E]%
\\[\blanklineskip]%
\>[B]{}\keyword{instance}\;\conid{HasBind}\;\concolor{\mathsf{D_U}}\;\keyword{where}{}\<[E]%
\\
\>[B]{}\hsindent{3}{}\<[3]%
\>[3]{}\varid{bind}\;\varid{rhs}\;\varid{body}\mathrel{=}\varid{body}\;(\varid{kleeneFix}\;\varid{rhs}){}\<[E]%
\ColumnHook
\end{hscode}\resethooks
\end{minipage}%
\begin{minipage}{0.3\textwidth}
\begin{hscode}\SaveRestoreHook
\column{B}{@{}>{\hspre}l<{\hspost}@{}}%
\column{3}{@{}>{\hspre}l<{\hspost}@{}}%
\column{7}{@{}>{\hspre}c<{\hspost}@{}}%
\column{7E}{@{}l@{}}%
\column{9}{@{}>{\hspre}l<{\hspost}@{}}%
\column{10}{@{}>{\hspre}l<{\hspost}@{}}%
\column{14}{@{}>{\hspre}l<{\hspost}@{}}%
\column{19}{@{}>{\hspre}l<{\hspost}@{}}%
\column{22}{@{}>{\hspre}l<{\hspost}@{}}%
\column{E}{@{}>{\hspre}l<{\hspost}@{}}%
\>[B]{}\keyword{data}\;\concolor{\mathsf{Value_U}}\mathrel{=}\conid{U} \argcons \concolor{\mathsf{Value_U}}\mid \conid{Rep}\;\conid{U}{}\<[E]%
\\
\>[B]{}\keyword{type}\;\concolor{\mathsf{D_U}}\mathrel{=}\concolor{\mathsf{T_U}}\;\concolor{\mathsf{Value_U}}{}\<[E]%
\\[\blanklineskip]%
\>[B]{}\keyword{instance}\;\conid{Lat}\;\conid{U}\;\keyword{where}\; ... {}\<[E]%
\\
\>[B]{}\keyword{instance}\;\conid{Lat}\;\conid{Uses}\;\keyword{where}\mathbin{...}{}\<[E]%
\\
\>[B]{}\keyword{instance}\;\conid{Lat}\;\concolor{\mathsf{Value_U}}\;\keyword{where}\; ... {}\<[E]%
\\
\>[B]{}\keyword{instance}\;\conid{Lat}\;\concolor{\mathsf{D_U}}\;\keyword{where}\; ... {}\<[E]%
\\[\blanklineskip]%
\>[B]{}\varid{peel}\mathbin{::}\concolor{\mathsf{Value_U}}\to (\conid{U},\concolor{\mathsf{Value_U}}){}\<[E]%
\\
\>[B]{}\varid{peel}\;(\conid{Rep}\;\varid{u}){}\<[19]%
\>[19]{}\mathrel{=}(\varid{u},(\conid{Rep}\;\varid{u})){}\<[E]%
\\
\>[B]{}\varid{peel}\;(\varid{u} \argcons \varid{v}){}\<[19]%
\>[19]{}\mathrel{=}(\varid{u},\varid{v}){}\<[E]%
\\[\blanklineskip]%
\>[B]{}(\mathbin{!?})\mathbin{::}\conid{Uses}\to \conid{Name}\to \conid{U}{}\<[E]%
\\
\>[B]{}\varid{m}\mathbin{!?}\varid{x}{}\<[9]%
\>[9]{}\mid \varid{x}\in \varid{dom}\;\varid{m}{}\<[22]%
\>[22]{}\mathrel{=}\varid{m}\mathop{!}\varid{x}{}\<[E]%
\\
\>[9]{}\mid \varid{otherwise}{}\<[22]%
\>[22]{}\mathrel{=}\concolor{\mathsf{U_0}}{}\<[E]%
\ColumnHook
\end{hscode}\resethooks
\end{minipage}
\\[-1em]
\caption{Summary-based usage analysis}
\label{fig:usage-analysis}
\end{figure}

In order to recover usage analysis as an instance of our generic interpreter, we
must define its finitely represented semantic domain \ensuremath{\concolor{\mathsf{D_U}}}.
A good first step is to replace the potentially
infinite traces \ensuremath{\conid{T}} in dynamic semantic domains such as \ensuremath{\conid{D}_{\mathbf{na}}} with finite data
such as \ensuremath{\concolor{\mathsf{T_U}}} in \Cref{fig:usage-analysis}.
A \emph{usage trace} \ensuremath{\langle \varcolor{\varphi}, \varid{val} \rangle\mathbin{::}\concolor{\mathsf{T_U}}\;\varid{v}} is a pair of a value \ensuremath{\varid{val}\mathbin{::}\varid{v}}
and a finite map \ensuremath{\varcolor{\varphi}\mathbin{::}\conid{Uses}}, mapping variables to a \emph{usage} \ensuremath{\conid{U}}.
The usage \ensuremath{\varcolor{\varphi}\mathbin{!?}\varid{x}} assigned to \ensuremath{\varid{x}} is meant to approximate the number of \ensuremath{\conid{Look}\;\varid{x}}
events; \ensuremath{\concolor{\mathsf{U_0}}} means ``at most 0 times'', \ensuremath{\concolor{\mathsf{U_1}}} means ``at most 1 times'',
and \ensuremath{\concolor{\mathsf{U_\omega}}} means ``an unknown number of times''.
In this way, \ensuremath{\concolor{\mathsf{T_U}}} is an \emph{abstraction} of \ensuremath{\conid{T}}: it squashes all \ensuremath{\conid{Look}\;\varid{x}}
events into a single entry \ensuremath{\varcolor{\varphi}\mathbin{!?}\varid{x}\mathbin{::}\conid{U}} and discards all other events.

Consider as an example the by-name trace evaluating $\pe \triangleq
\Let{i}{\Lam{x}{x}}{\Let{j}{\Lam{y}{y}}{i~j~j}}$:
\[
\LetIT\xhookrightarrow{\hspace{1.1ex}}\LetIT\xhookrightarrow{\hspace{1.1ex}}\AppIT\xhookrightarrow{\hspace{1.1ex}}\AppIT\xhookrightarrow{\hspace{1.1ex}}\LookupT(i)\xhookrightarrow{\hspace{1.1ex}}\AppET\xhookrightarrow{\hspace{1.1ex}}\LookupT(j)\xhookrightarrow{\hspace{1.1ex}}\AppET\xhookrightarrow{\hspace{1.1ex}}\LookupT(j)\xhookrightarrow{\hspace{1.1ex}}\langle \lambda\rangle 
\]
\noindent
We would like to abstract this trace into \ensuremath{\langle [\mskip1.5mu \varid{i}\mapsto\concolor{\mathsf{U_1}},\varid{j}\mapsto\concolor{\mathsf{U_\omega}}\mskip1.5mu], ... \rangle}.
One plausible way to achieve this is to replace every \ensuremath{\conid{Step}\;(\conid{Look}\;\varid{x})\;...}
in the by-name trace with a call to \ensuremath{\varid{step}\;(\conid{Look}\;\varid{x})\;...} from the \ensuremath{\conid{Trace}\;\concolor{\mathsf{T_U}}}
instance in \Cref{fig:usage-analysis}, essentially folding over the trace.
The \ensuremath{\varid{step}} implementation increments the usage of \ensuremath{\varid{x}} whenever a \ensuremath{\conid{Look}\;\varid{x}}
event occurs.
The addition operation used to carry out incrementation is defined in type class
instances \ensuremath{\conid{UVec}\;\conid{U}} and \ensuremath{\conid{UVec}\;\conid{Uses}}, together with scalar multiplication.
For example, \ensuremath{\concolor{\mathsf{U_0}}\mathbin{+}\varid{u}\mathrel{=}\varid{u}} and \ensuremath{\concolor{\mathsf{U_1}}\mathbin{+}\concolor{\mathsf{U_1}}\mathrel{=}\concolor{\mathsf{U_\omega}}} in \ensuremath{\conid{U}}, as well as \ensuremath{\concolor{\mathsf{U_0}}\mathbin{*}\varid{u}\mathrel{=}\concolor{\mathsf{U_0}}},
\ensuremath{\concolor{\mathsf{U_\omega}}\mathbin{*}\concolor{\mathsf{U_1}}\mathrel{=}\concolor{\mathsf{U_\omega}}}.
These operations lift to \ensuremath{\conid{Uses}} pointwise, \eg
\ensuremath{[\mskip1.5mu \varid{i}\mapsto\concolor{\mathsf{U_1}}\mskip1.5mu]\mathbin{+}(\concolor{\mathsf{U_\omega}}\mathbin{*}[\mskip1.5mu \varid{j}\mapsto\concolor{\mathsf{U_1}}\mskip1.5mu])\mathrel{=}[\mskip1.5mu \varid{i}\mapsto\concolor{\mathsf{U_1}},\varid{j}\mapsto\concolor{\mathsf{U_\omega}}\mskip1.5mu]}.

Abstracting \ensuremath{\conid{T}} into \ensuremath{\concolor{\mathsf{T_U}}} but keeping the concrete semantic \ensuremath{\conid{Value}} definition
amounts to what \citet{adi} call a \emph{collecting semantics}.
To recover such an analysis-specific collecting semantics,
it is sufficient to define a \ensuremath{\conid{Monad}} instance on \ensuremath{\concolor{\mathsf{T_U}}} mirroring trace
concatenation and then running our interpreter at, \eg $\ensuremath{\conid{D}\;(\conid{ByName}\;\concolor{\mathsf{T_U}})} \cong
\ensuremath{\concolor{\mathsf{T_U}}\;(\conid{Value}\;\concolor{\mathsf{T_U}})}$ on expression $\pe$ from earlier:
\[
  \ensuremath{\mathcal{S}\denot{(\Let{i}{\Lam{x}{x}}{\Let{j}{\Lam{y}{y}}{i~j~j}})}_{\varcolor{\varepsilon}}} = 
\langle [\mathit{i} \! \mapsto \! \concolor{\mathsf{U_1}},\mathit{j} \! \mapsto \! \concolor{\mathsf{U_ω}}], \lambda \rangle
\ensuremath{\mathbin{::}\conid{D}\;(\conid{ByName}\;\concolor{\mathsf{T_U}})}.
\]
\noindent
It is nice to explore whether the \ensuremath{\conid{Trace}} instance encodes the desired
operational property in this way, but of little practical relevance because
this interpreter instance will diverge whenever the input expression diverges.
We will now fix this by introducing a finitely represented \ensuremath{\concolor{\mathsf{Value_U}}} to replace
\ensuremath{\conid{Value}\;\concolor{\mathsf{T_U}}}.

\subsubsection{Value Abstraction \ensuremath{\concolor{\mathsf{Value_U}}} and Summarisation in \ensuremath{\conid{Domain}\;\concolor{\mathsf{D_U}}}}

We complement the trace type \ensuremath{\concolor{\mathsf{T_U}}} with an abstract value type \ensuremath{\concolor{\mathsf{Value_U}}}
to get the finitely represented semantic domain \ensuremath{\concolor{\mathsf{D_U}}\mathrel{=}\concolor{\mathsf{T_U}}\;\concolor{\mathsf{Value_U}}} in
\Cref{fig:usage-analysis}, and thus a \emph{static usage analysis} \ensuremath{\mathcal{S}_{\mathbf{usage}}\denot{\wild}_{\wild}}
when we instantiate \ensuremath{\mathcal{S}\denot{\wild}_{\wild}} at \ensuremath{\concolor{\mathsf{D_U}}}.


The definition of \ensuremath{\concolor{\mathsf{Value_U}}} is just a copy of $π ∈ \Args$ in
\Cref{fig:absence} that lists argument usage \ensuremath{\conid{U}} instead of $\Absence$ flags;
the entire intuition transfers.
For example, the \ensuremath{\concolor{\mathsf{Value_U}}} abstracting $\Lam{y}{\Lam{z}{y}}$ is
\ensuremath{\concolor{\mathsf{U_1}} \argcons \concolor{\mathsf{U_0}} \argcons \conid{Rep}\;\concolor{\mathsf{U_\omega}}}, because the first argument is used once while
the second is used 0 times.
What we previously called absence types $θ ∈ \AbsTy$ in \Cref{fig:absence} is
now the abstract semantic domain \ensuremath{\concolor{\mathsf{D_U}}}.
It is now evident that usage analysis is a modest generalisation of absence
analysis in \Cref{fig:absence}:
a variable is absent ($\aA$) when it has usage \ensuremath{\concolor{\mathsf{U_0}}}, otherwise it is used
($\aU$).

Consider
$\ensuremath{\mathcal{S}_{\mathbf{usage}}\denot{( \Let{k}{\Lam{y}{\Lam{z}{y}}}{k~x_1~x_2} )}_{\varcolor{\rho}_e}}
 = 
\langle [\mathit{k} \! \mapsto \! \concolor{\mathsf{U_1}},\mathit{x_1} \! \mapsto \! \concolor{\mathsf{U_1}}], \conid{Rep}\;\concolor{\mathsf{U_ω}} \rangle
$,
analysing the example expression from \Cref{sec:problem}.
Usage analysis successfully infers that $x_1$ is used at most once and that
$x_2$ is absent, because it does not occur in the reported \ensuremath{\conid{Uses}}.


On the other hand,
$\ensuremath{\mathcal{S}_{\mathbf{usage}}\denot{( \Let{i}{\Lam{x}{x}}{\Let{j}{\Lam{y}{y}}{i~i~j}} )}_{\varcolor{\varepsilon}}}
 = 
\langle [\mathit{i} \! \mapsto \! \concolor{\mathsf{U_ω}},\mathit{j} \! \mapsto \! \concolor{\mathsf{U_ω}}], \conid{Rep}\;\concolor{\mathsf{U_ω}} \rangle
$
demonstrates the limitations of the first-order summary mechanism.
While the program trace would only have one lookup for $j$, the analysis is
unable to reason through the indirect call and conservatively reports that $j$
may be used many times.

The \ensuremath{\conid{Domain}} instance is responsible for implementing the summary mechanism.
While \ensuremath{\varid{stuck}} expressions do not evaluate anything and hence are denoted by
\ensuremath{\bot\mathrel{=}\langle \varcolor{\varepsilon}, \conid{Rep}\;\concolor{\mathsf{U_0}} \rangle}, the \ensuremath{\varid{fun}} and \ensuremath{\varid{apply}} functions play exactly the same
roles as $\mathit{fun}_\px$ and $\mathit{app}$ in \Cref{fig:absence}.
Let us briefly review how the summary for the right-hand side $\Lam{x}{x}$ of
$i$ in the previous example is computed:
\begin{hscode}\SaveRestoreHook
\column{B}{@{}>{\hspre}c<{\hspost}@{}}%
\column{BE}{@{}l@{}}%
\column{4}{@{}>{\hspre}l<{\hspost}@{}}%
\column{30}{@{}>{\hspre}l<{\hspost}@{}}%
\column{76}{@{}>{\hspre}l<{\hspost}@{}}%
\column{E}{@{}>{\hspre}l<{\hspost}@{}}%
\>[4]{}\mathcal{S}\denot{\conid{Lam}\;\varid{x}\;(\conid{Var}\;\varid{x})}_{\varcolor{\rho}}\mathrel{=}{}\<[30]%
\>[30]{}\varid{fun}\;\varid{x}\;(\lambda \varid{d}\to \varid{step}\;\conid{App}_{2}\;(\mathcal{S}\denot{\conid{Var}\;\varid{x}}_{\varcolor{\rho}[\varid{x}\mapsto\varid{d}]})){}\<[E]%
\\
\>[B]{}\mathrel{=}{}\<[BE]%
\>[4]{}\keyword{case}\;\varid{step}\;\conid{App}_{2}\;(\mathcal{S}\denot{\conid{Var}\;\varid{x}}_{\varcolor{\rho}[\varid{x}\mapsto\langle [\varid{x}\mapsto\concolor{\mathsf{U_1}}], \conid{Rep}\;\concolor{\mathsf{U_\omega}} \rangle]})\;{}\<[76]%
\>[76]{}\keyword{of}\;\langle \varcolor{\varphi}, \varid{v} \rangle\to \langle \varcolor{\varphi}[\varid{x}\mapsto\concolor{\mathsf{U_0}}], \varcolor{\varphi}\mathbin{!?}\varid{x} \argcons \conid{Rep}\;\concolor{\mathsf{U_\omega}} \rangle{}\<[E]%
\\
\>[B]{}\mathrel{=}{}\<[BE]%
\>[4]{}\keyword{case}\;\langle [\varid{x}\mapsto\concolor{\mathsf{U_1}}], \conid{Rep}\;\concolor{\mathsf{U_\omega}} \rangle\;{}\<[76]%
\>[76]{}\keyword{of}\;\langle \varcolor{\varphi}, \varid{v} \rangle\to \langle \varcolor{\varphi}[\varid{x}\mapsto\concolor{\mathsf{U_0}}], \varcolor{\varphi}\mathbin{!?}\varid{x} \argcons \conid{Rep}\;\concolor{\mathsf{U_\omega}} \rangle{}\<[E]%
\\
\>[B]{}\mathrel{=}{}\<[BE]%
\>[4]{}\langle \varcolor{\varepsilon}, \concolor{\mathsf{U_1}} \argcons \conid{Rep}\;\concolor{\mathsf{U_\omega}} \rangle{}\<[E]%
\ColumnHook
\end{hscode}\resethooks
The definition of \ensuremath{\varid{fun}\;\varid{x}} applies the lambda body to a \emph{proxy} \ensuremath{\langle [\varid{x}\mapsto\concolor{\mathsf{U_1}}], \conid{Rep}\;\concolor{\mathsf{U_\omega}} \rangle}
to summarise how the body uses its argument by way of looking at how it uses \ensuremath{\varid{x}}.%
\footnote{As before, the exact identity of \ensuremath{\varid{x}} is exchangeable; we use it as a
De Bruijn level.}
Every use of \ensuremath{\varid{x}}'s proxy will contribute a usage of \ensuremath{\concolor{\mathsf{U_1}}} on \ensuremath{\varid{x}}, and multiple
uses in the lambda body would accumulate to a usage of \ensuremath{\concolor{\mathsf{U_\omega}}}.
In this case there is only a single use of \ensuremath{\varid{x}} and the final usage \ensuremath{\varcolor{\varphi}\mathbin{!?}\varid{x}\mathrel{=}\concolor{\mathsf{U_1}}} from the lambda body will be prepended to the value abstraction.
Occurrences of \ensuremath{\varid{x}} unleash the uninformative top value (\ensuremath{\conid{Rep}\;\concolor{\mathsf{U_\omega}}}) from \ensuremath{\varid{x}}'s
proxy for lack of knowing the actual argument at call sites.

The definition of \ensuremath{\varid{apply}} to apply such summaries to an argument is nearly the
same as in \Cref{fig:absence}, except for the use of \ensuremath{\mathbin{+}} instead of $⊔$ to
carry over \ensuremath{\concolor{\mathsf{U_1}}\mathbin{+}\concolor{\mathsf{U_1}}\mathrel{=}\concolor{\mathsf{U_\omega}}}, and an explicit \ensuremath{\varid{peel}} to view a \ensuremath{\concolor{\mathsf{Value_U}}} in terms
of $\argcons$ (it is $\ensuremath{\conid{Rep}\;\varid{u}} \equiv \ensuremath{\varid{u} \argcons \conid{Rep}\;\varid{u}}$).
The usage \ensuremath{\varid{u}} thus pelt from the value determines how often the actual
argument was evaluated, and multiplying the uses of the argument \ensuremath{\varcolor{\varphi}_{2}} with \ensuremath{\varid{u}}
accounts for that.

The example
$\ensuremath{\mathcal{S}_{\mathbf{usage}}\denot{( \Let{z}{Z()}{\Case{S(z)}{S(n) → n}} )}_{\varcolor{\varepsilon}}}
 = 
\langle [\mathit{z} \! \mapsto \! \concolor{\mathsf{U_ω}}], \conid{Rep}\;\concolor{\mathsf{U_ω}} \rangle
$
illustrates the summary mechanism for data types.
Our analysis imprecisely infers that \ensuremath{\varid{z}} might be used many times when it is
only used once.%
\footnote{Following \citet{Sergey:14} we could model \emph{demand} as
a property of evaluation contexts and propagate uses of field binders to the
scrutinee's fields to do better.}
This is achieved in \ensuremath{\varid{con}} by repeatedly \ensuremath{\varid{apply}}ing to the top value \ensuremath{(\conid{Rep}\;\concolor{\mathsf{U_\omega}})},
as if a data constructor was a lambda-bound variable.
Dually, \ensuremath{\varid{select}} does not need to track how fields are used and can pass \ensuremath{\langle \varcolor{\varepsilon}, \conid{Rep}\;\concolor{\mathsf{U_\omega}} \rangle} as proxies for field denotations.
The result uses anything the scrutinee expression used, plus the upper bound of
uses in case alternatives, one of which will be taken.

Note that the finite representation of the type \ensuremath{\concolor{\mathsf{D_U}}} rules out injective
implementations of \ensuremath{\varid{fun}\;\varid{x}\mathbin{::}(\concolor{\mathsf{D_U}}\to \concolor{\mathsf{D_U}})\to \concolor{\mathsf{D_U}}} and thus requires the
aforementioned \emph{approximate} summary mechanism.
There is another potential source of approximation: the \ensuremath{\conid{HasBind}}
instance discussed next.

\begin{figure}
\begin{hscode}\SaveRestoreHook
\column{B}{@{}>{\hspre}l<{\hspost}@{}}%
\column{E}{@{}>{\hspre}l<{\hspost}@{}}%
\>[B]{}\keyword{class}\;\conid{Eq}\;\varid{a}\Rightarrow \conid{Lat}\;\varid{a}\;\keyword{where}\;\bot\mathbin{::}\varid{a};(\mathbin{⊔})\mathbin{::}\varid{a}\to \varid{a}\to \varid{a};{}\<[E]%
\\
\>[B]{}\varid{kleeneFix}\mathbin{::}\conid{Lat}\;\varid{a}\Rightarrow (\varid{a}\to \varid{a})\to \varid{a};\qquad\;\varid{lub}\mathbin{::}\conid{Lat}\;\varid{a}\Rightarrow [\mskip1.5mu \varid{a}\mskip1.5mu]\to \varid{a}{}\<[E]%
\\
\>[B]{}\varid{kleeneFix}\;\varid{f}\mathrel{=}\varid{go}\;\bot\;\keyword{where}\;\varid{go}\;\varid{x}\mathrel{=}\keyword{let}\;\varid{x'}\mathrel{=}\varid{f}\;\varid{x}\;\keyword{in}\;\keyword{if}\;\varid{x'}\mathbin{⊑}\varid{x}\;\keyword{then}\;\varid{x'}\;\keyword{else}\;\varid{go}\;\varid{x'}{}\<[E]%
\ColumnHook
\end{hscode}\resethooks
\\[-1em]
\caption{Order theory and Kleene iteration}
\label{fig:lat}
\end{figure}

\subsubsection{Finite Fixpoint Strategy in \ensuremath{\conid{HasBind}\;\concolor{\mathsf{D_U}}} and Totality}
\label{sec:usage-fixpoint}

The third and last ingredient to recover a static analysis is the fixpoint
strategy in \ensuremath{\conid{HasBind}\;\concolor{\mathsf{D_U}}}, to be used for recursive let bindings.

For the dynamic semantics in \Cref{sec:interp} we made liberal use of
\emph{guarded fixpoints}, that is, recursively defined values such as \ensuremath{\keyword{let}\;\varid{d}\mathrel{=}\varid{rhs}\;\varid{d}\;\keyword{in}\;\varid{body}\;\varid{d}} in \ensuremath{\conid{HasBind}\;\conid{D}_{\mathbf{na}}} (\Cref{fig:eval}).
At least for \ensuremath{\mathcal{S}_{\mathbf{name}}\denot{\wild}_{\wild}} and \ensuremath{\mathcal{S}_{\mathbf{need}}\denot{\wild}_{\wild}}, we have proved in \Cref{sec:adequacy}
that these fixpoints always exist by a coinductive argument.
Alas, among other things this argument relies on the \ensuremath{\conid{Step}} constructor --- and
thus the \ensuremath{\varid{step}} method --- of the trace type \ensuremath{\conid{T}} being \emph{lazy} in
the tail of the trace!

When we replaced \ensuremath{\conid{T}} in favor of the finite data type \ensuremath{\concolor{\mathsf{T_U}}} in
\Cref{sec:usage-trace-abstraction} to get a collecting semantics \ensuremath{\conid{D}\;(\conid{ByName}\;\concolor{\mathsf{T_U}})}, we got a partial interpreter.
That was because the \ensuremath{\varid{step}} implementation of \ensuremath{\concolor{\mathsf{T_U}}} is \emph{not} lazy, and hence
the guarded fixpoint \ensuremath{\keyword{let}\;\varid{d}\mathrel{=}\varid{rhs}\;\varid{d}\;\keyword{in}\;\varid{body}\;\varid{d}} is not guaranteed to exist.

In general, finite data trace types cannot have a lazy \ensuremath{\varid{step}}
implementation, so finite data domains such as \ensuremath{\concolor{\mathsf{D_U}}} require a different fixpoint
strategy to ensure termination.
Depending on the abstract domain, different fixpoint strategies can be employed.
For an unusual example, in our type analysis \Cref{sec:type-analysis}, we
generate and solve a constraint system via unification to define fixpoints.
In case of \ensuremath{\concolor{\mathsf{D_U}}}, we compute least fixpoints by Kleene iteration \ensuremath{\varid{kleeneFix}}
in \Cref{fig:lat}.
\ensuremath{\varid{kleeneFix}} requires us to define an order on \ensuremath{\concolor{\mathsf{D_U}}}, which is induced
by \ensuremath{\concolor{\mathsf{U_0}}\mathbin{⊏}\concolor{\mathsf{U_1}}\mathbin{⊏}\concolor{\mathsf{U_\omega}}} in the same way that the order
on $\AbsTy$ in \Cref{sec:absence} was induced from the order $\aA ⊏ \aU$
on $\Absence$ flags.

The iteration procedure terminates whenever the type class instances of \ensuremath{\concolor{\mathsf{D_U}}} are
monotone and there are no infinite ascending chains in \ensuremath{\concolor{\mathsf{D_U}}}.
Alas, our \ensuremath{\concolor{\mathsf{Value_U}}} indeed contains such infinite chains, for example, \ensuremath{\concolor{\mathsf{U_1}} \argcons \concolor{\mathsf{U_1}} \argcons ... \argcons \conid{Rep}\;\concolor{\mathsf{U_0}}}!
This is easily worked around in practice by employing appropriate monotone
widening measures such as trimming any \ensuremath{\concolor{\mathsf{Value_U}}} at depth 10 to flat \ensuremath{\conid{Rep}\;\concolor{\mathsf{U_\omega}}}.
The resulting definition of \ensuremath{\conid{HasBind}} is safe for by-name and by-need semantics.%

\subsection{Type Analysis: Algorithm J}
\label{sec:type-analysis}

\begin{figure}
\belowdisplayskip0pt
\begin{hscode}\SaveRestoreHook
\column{B}{@{}>{\hspre}l<{\hspost}@{}}%
\column{3}{@{}>{\hspre}l<{\hspost}@{}}%
\column{5}{@{}>{\hspre}l<{\hspost}@{}}%
\column{7}{@{}>{\hspre}l<{\hspost}@{}}%
\column{9}{@{}>{\hspre}l<{\hspost}@{}}%
\column{11}{@{}>{\hspre}l<{\hspost}@{}}%
\column{17}{@{}>{\hspre}l<{\hspost}@{}}%
\column{19}{@{}>{\hspre}l<{\hspost}@{}}%
\column{22}{@{}>{\hspre}l<{\hspost}@{}}%
\column{37}{@{}>{\hspre}l<{\hspost}@{}}%
\column{41}{@{}>{\hspre}c<{\hspost}@{}}%
\column{41E}{@{}l@{}}%
\column{44}{@{}>{\hspre}l<{\hspost}@{}}%
\column{51}{@{}>{\hspre}l<{\hspost}@{}}%
\column{56}{@{}>{\hspre}l<{\hspost}@{}}%
\column{61}{@{}>{\hspre}l<{\hspost}@{}}%
\column{71}{@{}>{\hspre}l<{\hspost}@{}}%
\column{E}{@{}>{\hspre}l<{\hspost}@{}}%
\>[B]{}\keyword{data}\;\conid{TyCon}\mathrel{=}\conid{BoolTyCon}\mid \conid{NatTyCon}\mid \conid{OptionTyCon}\mid \conid{PairTyCon}{}\<[E]%
\\
\>[B]{}\keyword{data}\;\conid{Type}\mathrel{=}\conid{Type}\mathbin{:\rightarrow:}\conid{Type}\mid \conid{TyConApp}\;\conid{TyCon}\;[\mskip1.5mu \conid{Type}\mskip1.5mu]\mid \conid{TyVar}\;\conid{Name}\mid \conid{Wrong}{}\<[E]%
\\
\>[B]{}\keyword{data}\;\conid{PolyType}\mathrel{=}\conid{PT}\;[\mskip1.5mu \conid{Name}\mskip1.5mu]\;\conid{Type}{}\<[E]%
\\[\blanklineskip]%
\>[B]{}\mathcal{S}_{\mathbf{type}}\denot{\varid{e}}\mathrel{=}\varid{closedType}\;(\mathcal{S}\denot{\varid{e}}_{\varcolor{\varepsilon}})\mathbin{::}\conid{PolyType}{}\<[E]%
\\[\blanklineskip]%
\>[B]{}\keyword{type}\;\conid{Subst}\mathrel{=}\conid{Name}\mathbin{:\rightharpoonup}\conid{Type};\keyword{type}\;\conid{Constraint}\mathrel{=}(\conid{Type},\conid{Type}){}\<[E]%
\\
\>[B]{}\keyword{newtype}\;\conid{J}\;\varid{a}\mathrel{=}\conid{J}\;(\conid{StateT}\;(\conid{Set}\;\conid{Name},\conid{Subst})\;\conid{Maybe}\;\varid{a}){}\<[E]%
\\
\>[B]{}\varid{freshTyVar}\mathbin{::}\conid{J}\;\conid{Type};\qquad\;{}\<[37]%
\>[37]{}\varid{instantiatePolyTy}{}\<[56]%
\>[56]{}\mathbin{::}\conid{PolyType}\to \conid{J}\;\conid{Type}{}\<[E]%
\\
\>[B]{}\varid{unify}\mathbin{::}\conid{Constraint}\to \conid{J}\;();\qquad\;{}\<[37]%
\>[37]{}\varid{generaliseTy}{}\<[56]%
\>[56]{}\mathbin{::}\conid{J}\;\conid{Type}\to \conid{J}\;\conid{PolyType}{}\<[E]%
\\
\>[B]{}\varid{closedType}\mathbin{::}\conid{J}\;\conid{Type}\to \conid{PolyType} ;\qquad\;{}\<[37]%
\>[37]{}\keyword{instance}\;\conid{Trace}\;(\conid{J}\;\varid{v})\;\keyword{where}\;\varid{step}\;\anonymous \mathrel{=}\varid{id}{}\<[E]%
\\
\>[B]{}\keyword{instance}\;\conid{Domain}\;(\conid{J}\;\conid{Type})\;\keyword{where}{}\<[E]%
\\
\>[B]{}\hsindent{3}{}\<[3]%
\>[3]{}\varid{stuck}{}\<[41]%
\>[41]{}\mathrel{=}{}\<[41E]%
\>[44]{}\varid{return}\;\conid{Wrong}{}\<[E]%
\\
\>[B]{}\hsindent{3}{}\<[3]%
\>[3]{}\varid{fun}\;\anonymous \; \varid{f}{}\<[41]%
\>[41]{}\mathrel{=}{}\<[41E]%
\>[44]{}\keyword{do}\;\varcolor{\theta}_\varcolor{\alpha}\leftarrow \varid{freshTyVar};\varcolor{\theta}\leftarrow \varid{f}\;(\varid{return}\;\varcolor{\theta}_\varcolor{\alpha});\varid{return}\;(\varcolor{\theta}_\varcolor{\alpha}\mathbin{:\rightarrow:}\varcolor{\theta}){}\<[E]%
\\
\>[B]{}\hsindent{3}{}\<[3]%
\>[3]{}\varid{apply}\;\varid{v}\;\varid{a}{}\<[41]%
\>[41]{}\mathrel{=}{}\<[41E]%
\>[44]{}\keyword{do}\;\varcolor{\theta}_{1}{}\<[51]%
\>[51]{}\leftarrow \varid{v};\varcolor{\theta}_{2}{}\<[61]%
\>[61]{}\leftarrow \varid{a};\varcolor{\theta}_\varcolor{\alpha}{}\<[71]%
\>[71]{}\leftarrow \varid{freshTyVar};\varid{unify}\;(\varcolor{\theta}_{1},\varcolor{\theta}_{2}\mathbin{:\rightarrow:}\varcolor{\theta}_\varcolor{\alpha});\varid{return}\;\varcolor{\theta}_\varcolor{\alpha} {}\<[E]%
\\
\>[B]{}\varid{uniFix}\mathbin{::}(\conid{J}\;\conid{Type}\to \conid{J}\;\conid{Type})\to \conid{J}\;\conid{Type}{}\<[E]%
\\
\>[B]{}\varid{uniFix}\;\varid{rhs}\mathrel{=}\keyword{do}\;\varcolor{\theta}_\varcolor{\alpha}\leftarrow \varid{freshTyVar};\varcolor{\theta}\leftarrow \varid{rhs}\;(\varid{return}\;\varcolor{\theta}_\varcolor{\alpha});\varid{unify}\;(\varcolor{\theta}_\varcolor{\alpha},\varcolor{\theta});\varid{return}\;\varcolor{\theta}_\varcolor{\alpha}{}\<[E]%
\\
\>[B]{}\keyword{instance}\;\conid{HasBind}\;(\conid{J}\;\conid{Type})\;\keyword{where}{}\<[E]%
\\
\>[B]{}\hsindent{3}{}\<[3]%
\>[3]{}\varid{bind}\;\varid{rhs}\;\varid{body}\mathrel{=}\keyword{do}\;\varcolor{\sigma}\leftarrow \varid{generaliseTy}\;(\varid{uniFix}\;\varid{rhs});\varid{body}\;(\varid{instantiatePolyTy}\;\varcolor{\sigma}){}\<[E]%
\ColumnHook
\end{hscode}\resethooks
\vspace{-1em}
\caption{Type analysis with let generalisation (Algorithm J)}
\label{fig:type-analysis}
\end{figure}

Computing least fixpoints is common practice in static program analysis.
However, some abstract domains employ quite different fixpoint strategies.
The abstract domain of the type analysis we sketch in this subsection is
an interesting example:
Type analysis --- specifically, \citeauthor{Milner:78}'s Algorithm~J ---
computes fixpoints by generating and solving a constraint system via
unification.
Furthermore, since the domain is familiar, it is a good one to study in the
context of denotational interpreters.

\Cref{fig:type-analysis} outlines the abstract domain \ensuremath{\conid{J}\;\conid{Type}} at which the
generic denotational interpreter can be instantiated to perform Type analysis.
We omit implementational details that are derivative of Milner's description of
Algorithm~J.
The full implementation can be found in the extract generated from this
document, but the provided code is sufficiently exemplary of the approach.
The decompressed form of this subsection can be found in
\Cref*{sec:type-analysis-detail}, including many examples.

Type analysis \ensuremath{\mathcal{S}_{\mathbf{type}}\denot{\wild}} infers the most general Hindley-Milner polytype of an expression, \eg
\[\ensuremath{\mathcal{S}_{\mathbf{type}}\denot{( \Let{i}{\Lam{x}{x}}{\Let{o}{\mathit{Some}(i)}{o}} )}}
  = 
\forall\alpha_{6}.\ \texttt{option}\;(\alpha_{6} \rightarrow \alpha_{6})
.\]
Key to the analysis is its abstract trace type \ensuremath{\conid{J}}, offering means to invoke
unification (\ensuremath{\varid{unify}}), fresh name generation (\ensuremath{\varid{freshTyVar}}, \ensuremath{\varid{instantiatePolyTy}})
and let generalisation (\ensuremath{\varid{generaliseTy}}).
Type \ensuremath{\conid{J}} implements these effects by maintaining a unifying substitution and
a set of used names via the standard monad transformer \ensuremath{\conid{StateT}}.
Unification failure is signalled by returning \ensuremath{\conid{Nothing}} in the base monad
\ensuremath{\conid{Maybe}}, and function \ensuremath{\varid{closedType}} for handling \ensuremath{\conid{J}} effects will return \ensuremath{\conid{Wrong}}
when that happens:
\[\ensuremath{\mathcal{S}_{\mathbf{type}}\denot{( \Let{x}{\mathit{None}()}{x~x} )}}
  = 
\textbf{wrong}
.\]
Throughout the analysis, the invariant is maintained that the \ensuremath{\conid{J}\;\conid{Type}}
summaries of let-bound variables in the interpreter environment \ensuremath{\varcolor{\rho}} are of the
form \ensuremath{\varid{instantiatePolyTy}\;\varcolor{\sigma}} for a polytype \ensuremath{\varcolor{\sigma}}, while lambda- and field-bound
variables are denoted by \ensuremath{\varid{return}\;\varcolor{\theta}}, yielding the same monotype \ensuremath{\varcolor{\theta}} at all use
sites.
(Use of the term ``summary'' is justified because both \ensuremath{\varcolor{\sigma}} and \ensuremath{\varcolor{\theta}} are data,
and it would be easy to defunctionalise \ensuremath{\varid{instantiatePolyTy}\;\varcolor{\sigma}} and \ensuremath{\varid{return}\;\varcolor{\theta}}
to be data as well.)
Thus, let-bound denotations instantiate polytypes on-the-fly at occurrence
sites, just as in Algorithm~J.

The \ensuremath{\conid{Domain}\;(\conid{J}\;\conid{Type})} instance bears no surprises:
\ensuremath{\varid{stuck}} terms are denoted by the monotype \ensuremath{\conid{Wrong}} and the definition of \ensuremath{\varid{fun}}
and \ensuremath{\varid{apply}} are literal translations of Algorithm~J.

The generalisation machinery comes to bear in the implementation
of \ensuremath{\varid{bind}}, which implements a combination of the $\mathit{fix}$ and $\mathit{let}$
cases in Algorithm~J, computing fixpoints by unification (\ensuremath{\varid{uniFix}}).

\begin{toappendix}
\subsection{Type Analysis: Algorithm J}
\label{sec:type-analysis-detail}

Computing least fixpoints is common practice in static program analysis.
However, some abstract domains employ quite different fixpoint strategies.
The abstract domain of the type analysis we define in this subsection is
an interesting example:
Type analysis --- specifically, \citeauthor{Milner:78}'s Algorithm~J ---
computes fixpoints by generating and solving a constraint system via
unification.

\Cref{fig:type-analysis} outlines the abstract domain \ensuremath{\conid{J}\;\conid{Type}} at which the
generic denotational interpreter can be instantiated to perform Type analysis.
We omit implementational details that are derivative of Milner's description of
Algorithm~J.
The full implementation can be found in the extract generated from this
document, but the provided code is sufficiently exemplary of the approach.

Type analysis \ensuremath{\mathcal{S}_{\mathbf{type}}\denot{\wild}} infers the most general type of an expression, \eg
\[\ensuremath{\mathcal{S}_{\mathbf{type}}\denot{( \Let{f}{\Lam{g}{\Lam{x}{g~x}}}{f} )}}
  = 
\forall\alpha_{4}, \alpha_{5}.\ (\alpha_{4} \rightarrow \alpha_{5}) \rightarrow \alpha_{4} \rightarrow \alpha_{5}
.\]
The most general type can be \emph{polymorphic} when it universally quantifies
over \emph{generic} type variables such as $α_4$ and $α_5$ above.
In general, such a \ensuremath{\conid{PolyType}} is of the form $\forall \many{\alpha}.\ θ$,
where $θ$ ranges over a monomorphic \ensuremath{\conid{Type}} that can be either a type variable
\ensuremath{\conid{TyVar}\;\alpha} (we will use \ensuremath{\varcolor{\theta}_\varcolor{\alpha}} as meta variable for this form), a function type
\ensuremath{\varcolor{\theta}_{1}\mathbin{:\rightarrow:}\varcolor{\theta}_{2}}, or a type constructor application \ensuremath{\conid{TyConApp}}, where
\ensuremath{\conid{TyConApp}\;\conid{OptionTyCon}\;[\mskip1.5mu \varcolor{\theta}_{1}\mskip1.5mu]} is printed as $\mathtt{option}~θ_1$.
The \ensuremath{\conid{Wrong}} type indicates a type error and is printed as $\textbf{wrong}$.

Key to the analysis is its abstract trace type \ensuremath{\conid{J}}, the name of which refers to the ambient
effects of Milner's Algorithm~J, offering means to invoke unification (\ensuremath{\varid{unify}}),
fresh name generation (\ensuremath{\varid{freshTyVar}}, \ensuremath{\varid{instantiatePolyTy}}) and let
generalisation (\ensuremath{\varid{generaliseTy}}).
Our type \ensuremath{\conid{J}} implements these effects by maintaining two pieces of state via the
standard monad transformer \ensuremath{\conid{StateT}}:
\begin{enumerate}
  \item
    a consistent set of type constraints as a unifying substitution \ensuremath{\conid{Subst}}.
  \item
    the set of used names as a \ensuremath{\conid{Set}\;\conid{Name}}.
    This is to supply fresh names in \ensuremath{\varid{freshTyVar}}
    and to instantiate a polytype $\forall α. α \to α$ to a monotype $α_1
    \to α_1$ for fresh $α_1$ as done by \ensuremath{\varid{instantiatePolyTy}}, but also to
    identify the type variables which are \emph{generic}~\citep{Milner:78} in
    the ambient type context and hence may be generalised by \ensuremath{\varid{generaliseTy}}.
\end{enumerate}
Unification failure is signalled by returning \ensuremath{\conid{Nothing}} in the base monad
\ensuremath{\conid{Maybe}}, and function \ensuremath{\varid{closedType}} for handling \ensuremath{\conid{J}} effects will return \ensuremath{\conid{Wrong}}
when that happens:
\[\ensuremath{\mathcal{S}_{\mathbf{type}}\denot{( \Let{x}{\mathit{None}()}{x~x} )}}
  = 
\textbf{wrong}
\]
The operational detail offered by \ensuremath{\conid{Trace}} is ignored by \ensuremath{\conid{J}}, but the \ensuremath{\conid{Domain}}
and \ensuremath{\conid{HasBind}} instances for the abstract semantic domain \ensuremath{\conid{J}\;\conid{Type}} are quite
interesting.
Throughout the analysis, the invariant is maintained that the \ensuremath{\conid{J}\;\conid{Type}} denotation
of let-bound variables in the interpreter environment \ensuremath{\varcolor{\rho}} is of the form
\ensuremath{\varid{instantiatePolyTy}\;\varcolor{\sigma}} for a polytype \ensuremath{\varcolor{\sigma}}, while lambda- and field-bound
variables are denoted by \ensuremath{\varid{return}\;\varcolor{\theta}}, yielding the same monotype \ensuremath{\varcolor{\theta}} at all use
sites.
Thus, let-bound denotations instantiate polytypes on-the-fly at occurrence
sites, just as in Algorithm~J.

As expected, \ensuremath{\varid{stuck}} terms are denoted by the monotype \ensuremath{\conid{Wrong}}.
The definition of \ensuremath{\varid{fun}} resembles the abstraction rule of Algorithm~J,
in that it draws a fresh variable type \ensuremath{\varcolor{\theta}_\varcolor{\alpha}\mathbin{::}\conid{Type}} (of the form \ensuremath{\conid{TyVar}\;\alpha})
to stand for the type of the argument.
This type is passed as a monotype \ensuremath{\varid{return}\;\varcolor{\theta}_\varcolor{\alpha}} to the body denotation
\ensuremath{\varid{f}}, where it will be added to the environment (\cf \Cref{fig:eval}) in order to
compute the result type \ensuremath{\varcolor{\theta}} of the function body.
The type for the whole function is then \ensuremath{\varcolor{\theta}_\varcolor{\alpha}\mathbin{:\rightarrow:}\varcolor{\theta}}.
The definition for \ensuremath{\varid{apply}} is a literal translation of Algorithm~J as well.
The cases for \ensuremath{\varid{con}} and \ensuremath{\varid{select}} are omitted as their implementation follows
a similar routine.

The generalisation and instantiation machinery comes to bear in the implementation
of \ensuremath{\varid{bind}}, which implements a combination of the $\mathit{fix}$ and $\mathit{let}$
cases in Algorithm~J, computing fixpoints by unification (\ensuremath{\varid{uniFix}}).
It is best understood by tracing the right-hand side of $o$ in the following
example:
\[\ensuremath{\mathcal{S}_{\mathbf{type}}\denot{( \Let{i}{\Lam{x}{x}}{\Let{o}{\mathit{Some}(i)}{o}} )}}
  = 
\forall\alpha_{6}.\ \texttt{option}\;(\alpha_{6} \rightarrow \alpha_{6})
\]
The recursive knot is tied in the \ensuremath{\keyword{do}} block passed to \ensuremath{\varid{generaliseTy}}.
It works by calling the iteratee \ensuremath{\varid{rhs}} (corresponding to $\mathit{Some}(i)$)
with a fresh unification variable type \ensuremath{\varcolor{\theta}_\varcolor{\alpha}}, for example $α_1$.
The result of the call to \ensuremath{\varid{rhs}} in turn is a monotype \ensuremath{\varcolor{\theta}},
for example $\mathtt{option}\;(α_3 \rightarrow α_3)$ for \emph{generic}
$α_3$, meaning that $α_3$ is a fresh name introduced in the right-hand side
while instantiating the polymorphic identity function $i$.
Then \ensuremath{\varcolor{\theta}_\varcolor{\alpha}} is unified with \ensuremath{\varcolor{\theta}}, substituting $α_1$ with
$\mathtt{option}\;(α_3 \rightarrow α_3)$.
This concludes the implementation of Milner's $\mathit{fix}$ case.

For Milner's $\mathit{let}$ case, the type \ensuremath{\varcolor{\theta}_\varcolor{\alpha}} is
generalised to $\forall α_3.\ \mathtt{option}\;(α_3 \rightarrow α_3)$
by universally quantifying the generic variable $α_3$.
It is easy for \ensuremath{\varid{generaliseTy}} to deduce that $α_3$ must be generic \wrt the
right-hand side, because $α_3$ does not occur in the set of used \ensuremath{\conid{Name}}s prior
to the call to \ensuremath{\varid{rhs}}.
The generalised polytype \ensuremath{\varcolor{\sigma}} is then instantiated afresh via \ensuremath{\varid{instantiatePolyTy}\;\varcolor{\sigma}} at every use site of $o$ in the let body, implementing polymorphic
instantiation.

\subsection{Control-flow Analysis}
\label{sec:0cfa}

\begin{figure}
\belowdisplayskip=0pt
\begin{hscode}\SaveRestoreHook
\column{B}{@{}>{\hspre}l<{\hspost}@{}}%
\column{3}{@{}>{\hspre}l<{\hspost}@{}}%
\column{5}{@{}>{\hspre}l<{\hspost}@{}}%
\column{7}{@{}>{\hspre}l<{\hspost}@{}}%
\column{9}{@{}>{\hspre}l<{\hspost}@{}}%
\column{23}{@{}>{\hspre}c<{\hspost}@{}}%
\column{23E}{@{}l@{}}%
\column{26}{@{}>{\hspre}l<{\hspost}@{}}%
\column{E}{@{}>{\hspre}l<{\hspost}@{}}%
\>[B]{}\mathcal{S}_{\mathbf{cfa}}\denot{\varid{e}}\mathrel{=}\varid{runCFA}\;(\mathcal{S}\denot{\varid{e}}_{\varcolor{\varepsilon}});\;\varid{runCFA}\mathbin{::}\concolor{\mathsf{D_C}}\to \conid{Labels}{}\<[E]%
\\
\>[B]{}\keyword{newtype}\;\conid{Labels}\mathrel{=}\conid{Lbls}\;(\conid{Set}\;\conid{Label}){}\<[E]%
\\
\>[B]{}\keyword{type}\;\concolor{\mathsf{D_C}}\mathrel{=}\conid{State}\;\conid{Cache}\;\conid{Labels}{}\<[E]%
\\
\>[B]{}\keyword{data}\;\conid{Cache}\mathrel{=}\conid{Cache}\;(\conid{Label}\mathbin{:\rightharpoonup}\conid{ConCache})\;(\conid{Label}\mathbin{:\rightharpoonup}\conid{FunCache}){}\<[E]%
\\
\>[B]{}\keyword{type}\;\conid{ConCache}\mathrel{=}(\conid{Tag},[\mskip1.5mu \conid{Labels}\mskip1.5mu]){}\<[E]%
\\
\>[B]{}\keyword{data}\;\conid{FunCache}\mathrel{=}\conid{FC}\;(\conid{Maybe}\;(\conid{Labels},\conid{Labels}))\;(\concolor{\mathsf{D_C}}\to \concolor{\mathsf{D_C}}){}\<[E]%
\\[\blanklineskip]%
\>[B]{}\varid{updConCache}\mathbin{::}\conid{Label}\to \conid{Tag}\to [\mskip1.5mu \conid{Labels}\mskip1.5mu]\to \conid{State}\;\conid{Cache}\;(){}\<[E]%
\\
\>[B]{}\varid{updFunCache}\mathbin{::}\conid{Label}\to (\concolor{\mathsf{D_C}}\to \concolor{\mathsf{D_C}})\to \conid{State}\;\conid{Cache}\;(){}\<[E]%
\\
\>[B]{}\varid{cachedCall}\mathbin{::}\conid{Labels}\to \conid{Labels}\to \concolor{\mathsf{D_C}}{}\<[E]%
\\
\>[B]{}\varid{cachedCons}\mathbin{::}\conid{Labels}\to \conid{State}\;\conid{Cache}\;(\conid{Tag}\mathbin{:\rightharpoonup}[\mskip1.5mu \conid{Labels}\mskip1.5mu]){}\<[E]%
\\[\blanklineskip]%
\>[B]{}\keyword{instance}\;\conid{HasBind}\;\concolor{\mathsf{D_C}}\;\keyword{where}\; ... {}\<[E]%
\\
\>[B]{}\keyword{instance}\;\conid{Trace}\;\concolor{\mathsf{D_C}}\;\keyword{where}\;\varid{step}\;\anonymous \mathrel{=}\varid{id}{}\<[E]%
\\
\>[B]{}\keyword{instance}\;\conid{Domain}\;\concolor{\mathsf{D_C}}\;\keyword{where}{}\<[E]%
\\
\>[B]{}\hsindent{3}{}\<[3]%
\>[3]{}\varid{stuck}\mathrel{=}\varid{return}\;\bot{}\<[E]%
\\
\>[B]{}\hsindent{3}{}\<[3]%
\>[3]{}\varid{fun}\;\anonymous \;\varcolor{\ell}\;\varid{f}\mathrel{=}\keyword{do}{}\<[E]%
\\
\>[3]{}\hsindent{2}{}\<[5]%
\>[5]{}\varid{updFunCache}\;\varcolor{\ell}\;\varid{f}{}\<[E]%
\\
\>[3]{}\hsindent{2}{}\<[5]%
\>[5]{}\varid{return}\;(\conid{Lbls}\;(\varid{\conid{Set}.singleton}\;\varcolor{\ell})){}\<[E]%
\\
\>[B]{}\hsindent{3}{}\<[3]%
\>[3]{}\varid{apply}\;\varid{dv}\;\varid{da}\mathrel{=}\keyword{do}{}\<[E]%
\\
\>[3]{}\hsindent{2}{}\<[5]%
\>[5]{}\varid{v}\leftarrow \varid{dv}{}\<[E]%
\\
\>[3]{}\hsindent{2}{}\<[5]%
\>[5]{}\varid{a}\leftarrow \varid{da}{}\<[E]%
\\
\>[3]{}\hsindent{2}{}\<[5]%
\>[5]{}\varid{cachedCall}\;\varid{v}\;\varid{a}{}\<[E]%
\\
\>[B]{}\hsindent{3}{}\<[3]%
\>[3]{}\varid{con}\;\varcolor{\ell}\;\varid{k}\;\varid{ds}\mathrel{=}\keyword{do}{}\<[E]%
\\
\>[3]{}\hsindent{2}{}\<[5]%
\>[5]{}\varid{lbls}\leftarrow \varid{sequence}\;\varid{ds}{}\<[E]%
\\
\>[3]{}\hsindent{2}{}\<[5]%
\>[5]{}\varid{updConCache}\;\varcolor{\ell}\;\varid{k}\;\varid{lbls}{}\<[E]%
\\
\>[3]{}\hsindent{2}{}\<[5]%
\>[5]{}\varid{return}\;(\conid{Lbls}\;(\varid{\conid{Set}.singleton}\;\varcolor{\ell})){}\<[E]%
\\
\>[B]{}\hsindent{3}{}\<[3]%
\>[3]{}\varid{select}\;\varid{dv}\;\varid{fs}\mathrel{=}\keyword{do}{}\<[E]%
\\
\>[3]{}\hsindent{2}{}\<[5]%
\>[5]{}\varid{v}\leftarrow \varid{dv}{}\<[E]%
\\
\>[3]{}\hsindent{2}{}\<[5]%
\>[5]{}\varid{tag2flds}\leftarrow \varid{cachedCons}\;\varid{v}{}\<[E]%
\\
\>[3]{}\hsindent{2}{}\<[5]%
\>[5]{}\varid{lub}\mathbin{<\mspace{-6mu}\$\mspace{-6mu}>}\varid{sequence}\;{}\<[23]%
\>[23]{}[\mskip1.5mu {}\<[23E]%
\>[26]{}\varid{f}\;(\varid{map}\;\varid{return}\;(\varid{tag2flds}\mathop{!}\varid{k})){}\<[E]%
\\
\>[23]{}\mid {}\<[23E]%
\>[26]{}(\varid{k},\varid{f})\leftarrow \varid{\conid{Map}.assocs}\;\varid{fs},\varid{k}\in \varid{dom}\;\varid{tag2flds}\mskip1.5mu]{}\<[E]%
\ColumnHook
\end{hscode}\resethooks

\caption{Domain \ensuremath{\concolor{\mathsf{D_C}}} for 0CFA control-flow analysis}
\label{fig:cfa}
\end{figure}

Traditionally, control-flow analysis (CFA)~\citep{Shivers:91} is an important
instance of higher-order abstract interpreters~\citep{aam,adi}.
Although one of the main advantages of denotational interpreters is that
summary-based analyses can be derived as instances, this subsection demonstrates
that a call-strings-based CFA can be derived as an instance from the generic
denotational interpreter in \Cref{fig:eval} as well.

CFA overapproximates the set of syntactic values an expression evaluates to,
so as to narrow down the possible control-flow edges at application sites.
The resulting control-flow graph conservatively approximates the control-flow of
the whole program and can be used to apply classic intraprocedural analyses such
as interval analysis or constant propagation in an interprocedural setting.

\Cref{fig:cfa} implements the 0CFA variant of control-flow analysis~\citep{Shivers:91}.
For a given expression, it reports a set of \emph{program labels} --- textual
representations of positions in the program ---
that the expression might evaluate to:
\begin{align}\ensuremath{\mathcal{S}_{\mathbf{cfa}}\denot{( \Let{i}{\Lam{x}{x}}{\Let{j}{\Lam{y}{y}}{i~j~j}} )}}
  = 
\{\lambda y..\}
 \label{ex:cfa}\end{align}
Here, 0CFA infers that the example expression will evaluate to the lambda
expression bound to $j$.
This lambda is uniquely identified by the reported label $λy..$ per the unique
binder assumption in \Cref{sec:lang}.
Furthermore, the analysis determined that the expression cannot evaluate to the
lambda expression bound to $i$, hence its label $λx..$ is \emph{not} included
in the set.

By contrast, when $i$ is dynamically called both with $i$ and with $j$, the
result becomes approximate because 0CFA joins together the information from the
two call sites:
\[\ensuremath{\mathcal{S}_{\mathbf{cfa}}\denot{( \Let{i}{\Lam{x}{x}}{\Let{j}{\Lam{y}{y}}{i~\highlight{i}~j}} )}}
  = 
\{\lambda x..,  \lambda y..\}
\]

Labels for constructor applications simply print their syntax, \eg
\begin{equation}\thickmuskip=4mu\ensuremath{\mathcal{S}_{\mathbf{cfa}}\denot{( \Let{x}{\Let{y}{S(x)}{S(y)}}{\Case{x}{\{ Z() \to x; S(z) \to z \}}} )}}
  = 
\{S(x)\}
.\label{ex:cfa2} \end{equation}
Note that in this example, 0CFA discovers that $x$ evaluates to $S(y)$ and hence
is able to conclude that the $Z()$ branch of the case expression is dead.
In doing so, 0CFA rules out that the expression evaluates to $S(y)$,
reporting $S(x)$ as the only value of the expression.

In general, the label (\ie string) $S(y)$ does not uniquely determine a position
in the program because the expression may occur multiple times.
However, eliminating such common subexpressions is semantics-preserving, so
We argue that this poor man's notion of program labels is good enough for the
purpose of this demonstration.

To facilitate 0CFA as an instance of the generic denotational interpreter, we
need to slightly revise the \ensuremath{\conid{Domain}} class to pass the syntactic label to \ensuremath{\varid{fun}}
and \ensuremath{\varid{con}}:
\begin{hscode}\SaveRestoreHook
\column{B}{@{}>{\hspre}l<{\hspost}@{}}%
\column{3}{@{}>{\hspre}l<{\hspost}@{}}%
\column{8}{@{}>{\hspre}l<{\hspost}@{}}%
\column{E}{@{}>{\hspre}l<{\hspost}@{}}%
\>[B]{}\keyword{type}\;\conid{Label}\mathrel{=}\conid{String}{}\<[E]%
\\
\>[B]{}\keyword{class}\;\conid{Domain}\;\varid{d}\;\keyword{where}{}\<[E]%
\\
\>[B]{}\hsindent{3}{}\<[3]%
\>[3]{}\varid{con}{}\<[8]%
\>[8]{}\mathbin{::}\highlight{\conid{Label}}\to \conid{Tag}\to [\mskip1.5mu \varid{d}\mskip1.5mu]\to \varid{d}{}\<[E]%
\\
\>[B]{}\hsindent{3}{}\<[3]%
\>[3]{}\varid{fun}{}\<[8]%
\>[8]{}\mathbin{::}\conid{Name}\to \highlight{\conid{Label}}\to (\varid{d}\to \varid{d})\to \varid{d}{}\<[E]%
\ColumnHook
\end{hscode}\resethooks
\noindent
Constructing and forwarding labels appropriately in \ensuremath{\mathcal{S}\denot{\wild}_{\wild}} and adjusting
previous \ensuremath{\conid{Domain}} instances is routine and hence omitted.

\Cref{fig:cfa} represents sets of labels with the type \ensuremath{\conid{Labels}}, the
type of abstract values of the analysis.
The abstract domain \ensuremath{\concolor{\mathsf{D_C}}} of 0CFA is simply a stateful computation returning \ensuremath{\conid{Labels}}.
To this end, we define \ensuremath{\concolor{\mathsf{D_C}}} in terms of the standard \ensuremath{\conid{State}} monad to mutate a
\ensuremath{\conid{Cache}}, an abstraction of the heap discussed next.

\medskip

Recall that each \ensuremath{\conid{Label}} determines a syntactic value in the program.
The \ensuremath{\conid{Cache}} maintains, for every labelled value encountered thus far, an
approximation of its action on \ensuremath{\conid{Labels}}.

For example, the denotational interpreter evaluates the constructor application
$S(y)$ in the right-hand side of $x$ in \Cref{ex:cfa2} by calling
the \ensuremath{\conid{Domain}} method \ensuremath{\varid{con}}.
This call is implemented by updating the \ensuremath{\conid{ConCache}} field under the label $S(y)$
so that it carries the constructor tag $S$ as well as the \ensuremath{\conid{Labels}} that its
field $y$ evaluates to. In our example, $y$ evaluates to the set $\{S(x)\}$,
so the \ensuremath{\conid{ConCache}} entry at label $S(y)$ is updated to $(S,[\{S(x)\}])$.
This information is then available when evaluating the $\mathbf{case}$ expression
in \Cref{ex:cfa2} with \ensuremath{\varid{select}}, where the scrutinee $x$ returns $\ensuremath{\varid{v}} \triangleq \{S(y)\}$.
Function \ensuremath{\varid{cachedCons}} looks up for each label in \ensuremath{\varid{v}} the respective \ensuremath{\conid{ConCache}}
entry and merges these entries into an environment
\ensuremath{\varid{tag2flds}\mathbin{::}\conid{Tag}\mathbin{:\rightharpoonup}[\mskip1.5mu \conid{Labels}\mskip1.5mu]}, representing all the possible shapes the
scrutinee can take.
In our case, \ensuremath{\varid{tag2flds}} is just a singleton environment $[S ↦ [\{S(x)\}]]$.
This environment is subsequently joined with the alternatives of the case expression.
The only alternative that matches is $S(z) \to z$, where $z$ is bound to $\{S(x)\}$
from the information in the \ensuremath{\conid{ConCache}}.
The alternative $Z() \to x$ is dead because the case scrutinee $x$ does not
evaluate to shape $Z()$.

For another example involving the \ensuremath{\conid{FunCache}}, consider the example \Cref{ex:cfa}.
When the lambda expression $\Lam{x}{x}$ in the right-hand side of $i$ is
evaluated through \ensuremath{\varid{fun}}, it updates the \ensuremath{\conid{FunCache}} at label $λx..$ with
the corresponding abstract transformer \ensuremath{(\lambda \varid{x}\to \varid{x})\mathbin{::}\concolor{\mathsf{D_C}}\to \concolor{\mathsf{D_C}}}, registering it
for call sites.
Later, the application site $i~j$ in \Cref{ex:cfa} is evaluated to a call to
the \ensuremath{\conid{Domain}} method \ensuremath{\varid{apply}} with the denotations of $i$ and $j$.
The denotation for $i$ is bound to \ensuremath{\varid{dv}} and returns a set $\ensuremath{\varid{v}} \triangleq \{λx..\}$, while
the denotation for $j$ is bound to \ensuremath{\varid{da}} and returns a set $\ensuremath{\varid{a}} \triangleq \{λy..\}$.
These sets are passed to \ensuremath{\varid{cachedCall}} which iterates over the labels in the
callee \ensuremath{\varid{v}}.
For each such label, it looks up the abstract transformer in \ensuremath{\conid{FunCache}}, applies
it to the set of labels \ensuremath{\varid{a}} (taking approximative measures described below) and
joins together the labels returned from each call.
In our example, there is just a single callee label $λx..$, the abstract transformer
of which is the identity function \ensuremath{(\lambda \varid{x}\to \varid{x})\mathbin{::}\concolor{\mathsf{D_C}}\to \concolor{\mathsf{D_C}}}.
Applying the identity transformer to the set of labels $\{λy..\}$ from the
denotation of the argument $j$ returns this same set; the result of the
application $i~j$.

The above description calls a function label's abstract transformer anew at
every call site.
This yields the exact control-flow semantics of the original control-flow
analysis work~\citep[Section 3.4]{Shivers:91}, which is potentially diverging.
The way 0CFA (and our implementation of it) becomes finite is by maintaining only
a single point approximation of each abstract transformer's graph ($k$-CFA would
maintain one point per contour).
This single point approximation can be seen as the transformer's summary, but
this summary is \emph{call-site sensitive}:
Since the single point must be applicable at all call sites, the function body
must be reanalysed as the inputs from call sites increase.
Maintaining the single point approximation is the purpose of the \ensuremath{\conid{Maybe}\;(\conid{Labels},\conid{Labels})} field of the \ensuremath{\conid{FunCache}} and is a standard, if somewhat delicate hassle
in control-flow analyses.

Note that the given formulation of 0CFA is not modular; that is, the single point
approximation for a function \ensuremath{\varid{\conid{A}.f}} is not generally applicable at a call site
in module \ensuremath{\conid{B}} such as \ensuremath{\varid{\conid{A}.f}\;\varid{\conid{B}.x}} because the labels that \ensuremath{\varid{\conid{B}.x}} evaluate to
might not be known when compiling module \ensuremath{\conid{A}}.
\citet[Section 3.8.2]{Shivers:91} proposes a solution to this problem that
we chose not to implement for the simple proof of concept here.


Note that a \ensuremath{\concolor{\mathsf{D_C}}} transitively (through \ensuremath{\conid{Cache}}) recurses into \ensuremath{\concolor{\mathsf{D_C}}\to \concolor{\mathsf{D_C}}},
rendering the encoding non-inductive due to the negative occurrence.
This highlights a common challenge with instances of CFA:
the obligation to prove that the analysis actually terminates on all inputs; an
obligation that we will gloss over in this short demonstration.

\subsection{Stateful Analysis and Annotations}
\label{sec:annotations}

\begin{figure}
\hfuzz=2em
\belowdisplayskip=0pt
\begin{hscode}\SaveRestoreHook
\column{B}{@{}>{\hspre}l<{\hspost}@{}}%
\column{3}{@{}>{\hspre}l<{\hspost}@{}}%
\column{5}{@{}>{\hspre}l<{\hspost}@{}}%
\column{11}{@{}>{\hspre}l<{\hspost}@{}}%
\column{13}{@{}>{\hspre}l<{\hspost}@{}}%
\column{16}{@{}>{\hspre}l<{\hspost}@{}}%
\column{23}{@{}>{\hspre}c<{\hspost}@{}}%
\column{23E}{@{}l@{}}%
\column{26}{@{}>{\hspre}l<{\hspost}@{}}%
\column{E}{@{}>{\hspre}l<{\hspost}@{}}%
\>[B]{}\keyword{class}\;\conid{Domain}\;\varid{d}\Rightarrow \conid{StaticDomain}\;\varid{d}\;\keyword{where}{}\<[E]%
\\
\>[B]{}\hsindent{3}{}\<[3]%
\>[3]{}\keyword{type}\;\conid{Ann}\;\varid{d}{}\<[16]%
\>[16]{}\mathbin{::}\mathbin{*}{}\<[E]%
\\
\>[B]{}\hsindent{3}{}\<[3]%
\>[3]{}\varid{extractAnn}{}\<[16]%
\>[16]{}\mathbin{::}\conid{Name}\to \varid{d}\to (\varid{d},\conid{Ann}\;\varid{d}){}\<[E]%
\\
\>[B]{}\hsindent{3}{}\<[3]%
\>[3]{}\varid{funS}{}\<[16]%
\>[16]{}\mathbin{::}\conid{Monad}\;\varid{m}\Rightarrow \conid{Name}\to  (\varid{m}\;\varid{d}\to \varid{m}\;\varid{d})\to \varid{m}\;\varid{d}{}\<[E]%
\\
\>[B]{}\hsindent{3}{}\<[3]%
\>[3]{}\varid{selectS}{}\<[16]%
\>[16]{}\mathbin{::}\conid{Monad}\;\varid{m}\Rightarrow \varid{m}\;\varid{d}\to (\conid{Tag}\mathbin{:\rightharpoonup}([\mskip1.5mu \varid{m}\;\varid{d}\mskip1.5mu]\to \varid{m}\;\varid{d}))\to \varid{m}\;\varid{d}{}\<[E]%
\\
\>[B]{}\hsindent{3}{}\<[3]%
\>[3]{}\varid{bindS}{}\<[16]%
\>[16]{}\mathbin{::}\conid{Monad}\;\varid{m}\Rightarrow \conid{Name}\to \varid{d}\to (\varid{d}\to \varid{m}\;\varid{d})\to (\varid{d}\to \varid{m}\;\varid{d})\to \varid{m}\;\varid{d}{}\<[E]%
\\[\blanklineskip]%
\>[B]{}\keyword{instance}\;\conid{StaticDomain}\;\concolor{\mathsf{D_U}}\;\keyword{where}{}\<[E]%
\\
\>[B]{}\hsindent{3}{}\<[3]%
\>[3]{}\keyword{type}\;\conid{Ann}\;\concolor{\mathsf{D_U}}\mathrel{=}\conid{U}{}\<[E]%
\\
\>[B]{}\hsindent{3}{}\<[3]%
\>[3]{}\varid{extractAnn}\;\varid{x}\;\langle \varcolor{\varphi}, \varid{v} \rangle\mathrel{=}(\langle \varid{\conid{Map}.delete}\;\varid{x}\;\varcolor{\varphi}, \varid{v} \rangle,\varcolor{\varphi}\mathbin{!?}\varid{x}){}\<[E]%
\\
\>[B]{}\hsindent{3}{}\<[3]%
\>[3]{}\varid{funS}\;\varid{x}\;\varid{f}\mathrel{=}\keyword{do}{}\<[E]%
\\
\>[3]{}\hsindent{2}{}\<[5]%
\>[5]{}\langle \varcolor{\varphi}, \varid{v} \rangle\leftarrow \varid{f}\;(\varid{return}\;\langle [\varid{x}\mapsto\concolor{\mathsf{U_1}}], \conid{Rep}\;\concolor{\mathsf{U_\omega}} \rangle){}\<[E]%
\\
\>[3]{}\hsindent{2}{}\<[5]%
\>[5]{}\varid{return}\;\langle \varcolor{\varphi}[\varid{x}\mapsto\concolor{\mathsf{U_0}}], \varcolor{\varphi}\mathbin{!?}\varid{x} \argcons \varid{v} \rangle{}\<[E]%
\\
\>[B]{}\hsindent{3}{}\<[3]%
\>[3]{}\varid{selectS}\;\varid{md}\;\varid{mfs}\mathrel{=}\keyword{do}{}\<[E]%
\\
\>[3]{}\hsindent{2}{}\<[5]%
\>[5]{}\varid{d}\leftarrow \varid{md}{}\<[E]%
\\
\>[3]{}\hsindent{2}{}\<[5]%
\>[5]{}\varid{alts}\leftarrow \varid{sequence}\;{}\<[23]%
\>[23]{}[\mskip1.5mu {}\<[23E]%
\>[26]{}\varid{f}\;(\varid{replicate}\;(\varid{conArity}\;\varid{k})\;(\varid{return}\;\langle \varcolor{\varepsilon}, \conid{Rep}\;\concolor{\mathsf{U_\omega}} \rangle)){}\<[E]%
\\
\>[23]{}\mid {}\<[23E]%
\>[26]{}(\varid{k},\varid{f})\leftarrow \varid{\conid{Map}.assocs}\;\varid{mfs}\mskip1.5mu]{}\<[E]%
\\
\>[3]{}\hsindent{2}{}\<[5]%
\>[5]{}\varid{return}\;(\varid{d}\sequ \varid{lub}\;\varid{alts}){}\<[E]%
\\
\>[B]{}\hsindent{3}{}\<[3]%
\>[3]{}\varid{bindS}\;\anonymous \;\varid{init}\;\varid{rhs}\;\varid{body}\mathrel{=}\varid{kleeneFixAboveM}\;\varid{init}\;\varid{rhs}\bind \varid{body}{}\<[E]%
\\[\blanklineskip]%
\>[B]{}\varid{kleeneFixAboveM}\mathbin{::}(\conid{Monad}\;\varid{m},\conid{Lat}\;\varid{a})\Rightarrow \varid{a}\to (\varid{a}\to \varid{m}\;\varid{a})\to \varid{m}\;\varid{a}{}\<[E]%
\\[\blanklineskip]%
\>[B]{}\mathcal{S}_{\mathbf{usage}\leadsto}\denot{\varid{e}}_{\varcolor{\rho}}\mathrel{=}\varid{runAnn}\;(\mathcal{S}\denot{\varid{e}}_{\varid{return}\mathbin{\lhd}\varcolor{\rho}})\mathbin{::}(\concolor{\mathsf{D_U}},\conid{Name}\mathbin{:\rightharpoonup}\conid{U}){}\<[E]%
\\[\blanklineskip]%
\>[B]{}\keyword{data}\;\conid{Refs}\;\varid{s}\;\varid{d}\mathrel{=}\conid{Refs}\;(\conid{STRef}\;\varid{s}\;(\conid{Name}\mathbin{:\rightharpoonup}\varid{d}))\;(\conid{STRef}\;\varid{s}\;(\conid{Name}\mathbin{:\rightharpoonup}\conid{Ann}\;\varid{d})){}\<[E]%
\\
\>[B]{}\keyword{newtype}\;\conid{AnnT}\;\varid{s}\;\varid{d}\;\varid{a}\mathrel{=}\conid{AnnT}\;(\conid{Refs}\;\varid{s}\;\varid{d}\to \conid{ST}\;\varid{s}\;\varid{a});\keyword{type}\;\conid{AnnD}\;\varid{s}\;\varid{d}\mathrel{=}\conid{AnnT}\;\varid{s}\;\varid{d}\;\varid{d}{}\<[E]%
\\
\>[B]{}\varid{runAnn}{}\<[11]%
\>[11]{}\mathbin{::}(\keyword{\forall}\!\! \hsforall \;\varid{s}\hsdot{\circ }{.\ }\conid{AnnD}\;\varid{s}\;\varid{d})\to (\varid{d},\conid{Name}\mathbin{:\rightharpoonup}\conid{Ann}\;\varid{d}){}\<[E]%
\\[\blanklineskip]%
\>[B]{}\keyword{instance}\;\conid{Monad}\;(\conid{AnnT}\;\varid{s}\;\varid{d})\;\keyword{where}\mathbin{...}{}\<[E]%
\\[\blanklineskip]%
\>[B]{}\keyword{instance}\;\conid{Trace}\;\varid{d}\Rightarrow \conid{Trace}\;(\conid{AnnD}\;\varid{s}\;\varid{d})\;\keyword{where}{}\<[E]%
\\
\>[B]{}\hsindent{3}{}\<[3]%
\>[3]{}\varid{step}\;\varid{ev}\;(\conid{AnnT}\;\varid{f})\mathrel{=}\conid{AnnT}\;(\lambda \varid{refs}\to \varid{step}\;\varid{ev}\mathbin{<\mspace{-6mu}\$\mspace{-6mu}>}\varid{f}\;\varid{refs}){}\<[E]%
\\[\blanklineskip]%
\>[B]{}\keyword{instance}\;\conid{StaticDomain}\;\varid{d}\Rightarrow \conid{Domain}\;(\conid{AnnD}\;\varid{s}\;\varid{d})\;\keyword{where}\; ... {}\<[E]%
\\[\blanklineskip]%
\>[B]{}\varid{readCache}{}\<[13]%
\>[13]{}\mathbin{::}\conid{Lat}\;\varid{d}\Rightarrow \conid{Name}\to \conid{AnnD}\;\varid{s}\;\varid{d}{}\<[E]%
\\
\>[B]{}\varid{writeCache}{}\<[13]%
\>[13]{}\mathbin{::}\conid{Name}\to \varid{d}\to \conid{AnnT}\;\varid{s}\;\varid{d}\;(){}\<[E]%
\\
\>[B]{}\varid{annotate}{}\<[13]%
\>[13]{}\mathbin{::}\conid{StaticDomain}\;\varid{d}\Rightarrow \conid{Name}\to \conid{AnnD}\;\varid{s}\;\varid{d}\to \conid{AnnD}\;\varid{s}\;\varid{d}{}\<[E]%
\\[\blanklineskip]%
\>[B]{}\keyword{instance}\;(\conid{Lat}\;\varid{d},\conid{StaticDomain}\;\varid{d})\Rightarrow \conid{HasBind}\;(\conid{AnnD}\;\varid{s}\;\varid{d})\;\keyword{where}{}\<[E]%
\\
\>[B]{}\hsindent{3}{}\<[3]%
\>[3]{}\varid{bind}\;\varid{x}\;\varid{rhs}\;\varid{body}\mathrel{=}\keyword{do}{}\<[E]%
\\
\>[3]{}\hsindent{2}{}\<[5]%
\>[5]{}\varid{init}\leftarrow \varid{readCache}\;\varid{x}{}\<[E]%
\\
\>[3]{}\hsindent{2}{}\<[5]%
\>[5]{}\keyword{let}\;\varid{rhs'}\;\varid{d}_{1}\mathrel{=}\keyword{do}\;\varid{d}_{2}\leftarrow \varid{rhs}\;(\varid{return}\;\varid{d}_{1});\varid{writeCache}\;\varid{x}\;\varid{d}_{2};\varid{return}\;\varid{d}_{2}{}\<[E]%
\\
\>[3]{}\hsindent{2}{}\<[5]%
\>[5]{}\varid{annotate}\;\varid{x}\;(\varid{bindS}\;\varid{x}\;\varid{init}\;\varid{rhs'}\;(\varid{body}\hsdot{\circ }{.\ }\varid{return})){}\<[E]%
\ColumnHook
\end{hscode}\resethooks
\caption{Trace transformer \ensuremath{\conid{AnnT}} for recording annotations and caching of fixpoints}
\label{fig:annotations}
\end{figure}

Thus far, the static analyses derived from the generic denotational interpreter
produce a single abstract denotation \ensuremath{\varid{d}\triangleq\mathcal{S}\denot{\varid{e}}_{\varcolor{\varepsilon}}} for the program
expression \ensuremath{\varid{e}}.
However, in practice static compiler analyses such as Demand Analysis usually
drive a subsequent optimisation, for which a single denotation for the entire
program is insufficient.
Rather, we need one for every sub-expression, or at least every binding.

If we are interested in analysis results for variables \emph{bound} in
\ensuremath{\varid{e}}, then either the analysis must collect these results in \ensuremath{\varid{d}}, or we must
redundantly re-run the analysis for subexpressions.

In this subsection we show how to lift a pure, \emph{single-result} analysis into a
\emph{stateful} analysis that gives results for every binder, such that
\begin{itemize}
  \item it collects analysis results for bound variables in a separate, global map, and
  \item it caches fixpoints in yet another global map, so that nested fixpoint
    iteration can be sped up by starting from a previous approximation.
\end{itemize}
It is a common pattern for analyses to be stateful in this
way~\citep{Sergey:14}; GHC's Demand Analysis is a good real-world example.
The following demonstration targets usage analysis, but the technique should be
easy to adapt for other analyses discussed in this section.

\subsubsection{The Need for Isolating Bound Variable Usage}

For a concrete example, let us compare the results of usage analysis
from \Cref{sec:usage-analysis} on the expression $\pe_1 \triangleq
\Let{i}{\Lam{x}{\Let{j}{\Lam{y}{y}}{j}}}{i~i~i}$ and its subexpression
$\pe_2 \triangleq \Let{j}{\Lam{y}{y}}{j}$:
\[\begin{array}{lcl}
\ensuremath{\mathcal{S}_{\mathbf{usage}}\denot{( \Let{i}{\Lam{x}{\Let{j}{\Lam{y}{y}}{j}}}{i~i~i} )}_{\varcolor{\varepsilon}}}
 & = & 
\langle [\mathit{i} \! \mapsto \! \concolor{\mathsf{U_ω}},\mathit{j} \! \mapsto \! \concolor{\mathsf{U_ω}}], \conid{Rep}\;\concolor{\mathsf{U_ω}} \rangle
 \\
\ensuremath{\mathcal{S}_{\mathbf{usage}}\denot{( \Let{j}{\Lam{y}{y}}{j} )}_{\varcolor{\varepsilon}}}
 & = & 
\langle [\mathit{j} \! \mapsto \! \concolor{\mathsf{U_1}}], \concolor{\mathsf{U_1}} \argcons \conid{Rep}\;\concolor{\mathsf{U_ω}} \rangle

\end{array}\]
The analysis reports a different usage \ensuremath{\concolor{\mathsf{U_1}}} for the bound variable $j$ in the
subexpression $\pe_2$ versus \ensuremath{\concolor{\mathsf{U_\omega}}} in the containing expression $\pe_1$.
This is because in order for single-result usage analysis to report information
on \emph{bound} variable $j$ at all, it treats $j$ like a \emph{free} variable
of $i$, adding a use on $j$ for every call of $i$.
While this treatment reflects that multiple $\LookupT(j)$ events
will be observed when evaluating $\pe_1$, each event associates to a
\emph{different} activation (\ie heap entry) of the let binding $j$.
The result \ensuremath{\concolor{\mathsf{U_1}}} reported for $j$ in subexpression $\pe_2$ is more useful
because it reflects that every \emph{activation} of the binding
$j$ is looked up at most once over its lifetime, which is indeed the formal
property of interest in \Cref{sec:soundness}.

Rather than to re-run the analysis for every let binding such as $j$, we will
now demonstrate a way to write out an \emph{annotation} for $j$, just before
analysis leaves the $\mathbf{let}$ that binds $j$.
Annotations for bound variables constitute analysis state that will be
maintained separately from information on free variables.

\subsubsection{Maintaining Annotations by Implementing \ensuremath{\conid{StaticDomain}}}

\Cref{fig:annotations} lifts the existing definition for single-result usage
analysis \ensuremath{\mathcal{S}_{\mathbf{usage}}\denot{\wild}_{\wild}} into a stateful analysis \ensuremath{\mathcal{S}_{\mathbf{usage}\leadsto}\denot{\wild}_{\wild}} that writes out usage
information on bound variables into a separate map.
Consider the result on the example expression $\pe_1$ from above, where the pair
$(d, \mathit{anns})$ returned by \ensuremath{\mathcal{S}_{\mathbf{usage}\leadsto}\denot{\wild}_{\wild}} is printed as $d \leadsto \mathit{anns}$:
\[\thickmuskip=4mu\ensuremath{\mathcal{S}_{\mathbf{usage}\leadsto}\denot{( \Let{i}{\Lam{x}{\Let{j}{\Lam{y}{y}}{j}}}{i~i~i} )}_{\varcolor{\varepsilon}}}
 = 
\langle [], \conid{Rep}\;\concolor{\mathsf{U_ω}} \rangle \leadsto [\mathit{i} \! \mapsto \! \concolor{\mathsf{U_ω}},\mathit{j} \! \mapsto \! \concolor{\mathsf{U_1}}]
 \]
The annotations for both bound variables $i$ and $j$ are returned in an
annotation environment separate from the empty abstract free variable
environment \ensuremath{\varcolor{\varepsilon}\mathbin{::}\conid{Uses}} of the expression.
Furthermore, the use \ensuremath{\concolor{\mathsf{U_1}}} reported for $j$ is exactly as when analysing the
subexpression $\pe_2$ in isolation, as required.

Lifting the single-result analysis \ensuremath{\mathcal{S}_{\mathbf{usage}}\denot{\wild}_{\wild}} defined on \ensuremath{\concolor{\mathsf{D_U}}} to a stateful
analysis \ensuremath{\mathcal{S}_{\mathbf{usage}\leadsto}\denot{\wild}_{\wild}} requires very little extra code, implementing a type class instance \ensuremath{\conid{StaticDomain}\;\concolor{\mathsf{D_U}}}.
Before going into detail about how this lifting is implemented in terms of type
\ensuremath{\conid{AnnT}}, let us review its type class interface.
The type class \ensuremath{\conid{StaticDomain}} defines the associated type \ensuremath{\conid{Ann}} of annotations
in the particular static domain, along with a function \ensuremath{\varid{extractAnn}\;\varid{x}\;\varid{d}} for
extracting information on a let-bound \ensuremath{\varid{x}} from the denotation \ensuremath{\varid{d}}.
The instance for \ensuremath{\concolor{\mathsf{D_U}}} instantiates \ensuremath{\conid{Ann}\;\concolor{\mathsf{D_U}}} to bound variable use \ensuremath{\conid{U}}, and
\ensuremath{\varid{extractAnn}\;\varid{x}\;\langle \varcolor{\varphi}, \varid{v} \rangle} isolates the free variable use \ensuremath{\varcolor{\varphi}\mathop{!}\varid{x}} as annotation.
The remaining type class methods \ensuremath{\varid{funS}}, \ensuremath{\varid{selectS}} and \ensuremath{\varid{bindS}} are
simple monadic generalisations of their counterparts in \ensuremath{\conid{Domain}} and \ensuremath{\conid{HasBind}}.

The implementation of \ensuremath{\conid{StaticDomain}} requires very little extra code to
maintain, because the original definitions of \ensuremath{\varid{fun}}, \ensuremath{\varid{select}} and \ensuremath{\varid{bind}} can be
recovered in terms of the generalised definitions via the standard \ensuremath{\conid{Identity}}
monad as follows, where \ensuremath{\varid{coerce}} denotes a safe zero-cost coercion function
provided by GHC~\citep{Breitner:14}:
\begin{hscode}\SaveRestoreHook
\column{B}{@{}>{\hspre}l<{\hspost}@{}}%
\column{E}{@{}>{\hspre}l<{\hspost}@{}}%
\>[B]{}\keyword{newtype}\;\conid{Identity}\;\varid{a}\mathrel{=}\conid{Identity}\;\{\mskip1.5mu \varid{runIdentity}\mathbin{::}\varid{a}\mskip1.5mu\}{}\<[E]%
\\[\blanklineskip]%
\>[B]{}\varid{fun'}\mathbin{::}\conid{StaticDomain}\;\varid{d}\Rightarrow \conid{Name}\to \conid{Label}\to (\varid{d}\to \varid{d})\to \varid{d}{}\<[E]%
\\
\>[B]{}\varid{fun'}\;\varid{x}\;\varid{f}\mathrel{=}\varid{runIdentity}\;(\varid{funS}\;\varid{x}\;(\varid{coerce}\;\varid{f})){}\<[E]%
\\
\>[B]{}\varid{select'}\mathbin{::}\conid{StaticDomain}\;\varid{d}\Rightarrow \varid{d}\to (\conid{Tag}\mathbin{:\rightharpoonup}([\mskip1.5mu \varid{d}\mskip1.5mu]\to \varid{d}))\to \varid{d}{}\<[E]%
\\
\>[B]{}\varid{select'}\;\varid{d}\;\varid{fs}\mathrel{=}\varid{runIdentity}\;(\varid{selectS}\;(\conid{Identity}\;\varid{d})\;(\varid{coerce}\;\varid{fs})){}\<[E]%
\\
\>[B]{}\varid{bind'}\mathbin{::}(\conid{Lat}\;\varid{d},\conid{StaticDomain}\;\varid{d})\Rightarrow \conid{Name}\to (\varid{d}\to \varid{d})\to (\varid{d}\to \varid{d})\to \varid{d}{}\<[E]%
\\
\>[B]{}\varid{bind'}\;\varid{x}\;\varid{rhs}\;\varid{body}\mathrel{=}\varid{runIdentity}\;(\varid{bindS}\;\varid{x}\;\bot\;(\varid{coerce}\;\varid{rhs})\;(\varid{coerce}\;\varid{body})){}\<[E]%
\ColumnHook
\end{hscode}\resethooks
Any reasonable instance of \ensuremath{\conid{StaticDomain}} must satisfy the laws \ensuremath{\varid{fun}\mathrel{=}\varid{fun'}},
\ensuremath{\varid{select}\mathrel{=}\varid{select'}} and \ensuremath{\varid{bind}\mathrel{=}\varid{bind'}}.
(As can be seen in \Cref{fig:annotations} and above, we needed to slightly revise
the \ensuremath{\conid{HasBind}} type class in order to pass the name \ensuremath{\varid{x}} of the let-bound variable
to \ensuremath{\varid{bind}} and \ensuremath{\varid{bindS}}, similar as for \ensuremath{\varid{fun}}.)

Let us now look at how \ensuremath{\conid{AnnT}} extends the pure, single-result usage analysis
into a stateful one that maintains annotations.

\subsubsection{Trace Transformer \ensuremath{\conid{AnnT}} for Stateful Analysis}

Every instance \ensuremath{\conid{StaticDomain}\;\varid{d}} induces an instance \ensuremath{\conid{Domain}\;(\conid{AnnD}\;\varid{s}\;\varid{d})},
where the type \ensuremath{\conid{AnnD}\;\varid{s}\;\varid{d}} is another example of a \emph{trace transformer}:
It transforms the \ensuremath{\conid{Trace}} instance on type \ensuremath{\varid{d}} into a \ensuremath{\conid{Trace}} instance for \ensuremath{\conid{AnnD}\;\varid{s}\;\varid{d}}. The abstract domain \ensuremath{\conid{AnnD}} is defined in terms of the abstract trace
type \ensuremath{\conid{AnnT}}, which is a standard \ensuremath{\conid{ST}} monad utilising efficient and pure mutable
state threads~\citep{Launchbury:94}, stacked below a \ensuremath{\conid{Refs}} environment that
carries the mutable ref cells.
A stateful analysis computation \ensuremath{\keyword{\forall}\!\! \hsforall \;\varid{s}\hsdot{\circ }{.\ }\conid{AnnD}\;\varid{s}\;\concolor{\mathsf{D_U}}} is then run with \ensuremath{\varid{runAnn}},
initialising \ensuremath{\conid{Refs}} with ref cells pointing to empty environments.
(The universal quantification over \ensuremath{\varid{s}} in the type of \ensuremath{\varid{runAnn}} ensures that no
mutable \ensuremath{\conid{STRef}} from \ensuremath{\conid{Refs}} escapes the functional state thread of the
underlying \ensuremath{\conid{ST}} computation~\citep{Launchbury:94}.)

The induced instance \ensuremath{\conid{Domain}\;(\conid{AnnD}\;\varid{s}\;\varid{d})} is implemented
by lifting operations \ensuremath{\varid{stuck}}, \ensuremath{\varid{apply}} and \ensuremath{\varid{con}} into monadic \ensuremath{\conid{AnnT}\;\varid{s}\;\varid{d}} context
and by calling \ensuremath{\varid{funS}} and \ensuremath{\varid{selectS}}.
Finally, the stateful nature of the domain \ensuremath{\conid{AnnD}\;\varid{s}\;\varid{d}} is exploited in the
\ensuremath{\conid{HasBind}\;(\conid{AnnD}\;\varid{s}\;\varid{d})} instance, in two ways:

\begin{itemize}
  \item
    The call to \ensuremath{\varid{annotate}} writes out the annotation on the let-bound variable
    \ensuremath{\varid{x}} that is extracted from the denotation returned by the call to \ensuremath{\varid{bindS}}.
    The omitted definition of \ensuremath{\varid{annotate}} is just a thin wrapper around
    \ensuremath{\varid{extractAnn}} to store the extracted annotation in the \ensuremath{\conid{Name}\mathbin{:\rightharpoonup}\conid{Ann}\;\varid{d}}
    ref cell of \ensuremath{\conid{Refs}}, the contents of which are returned from \ensuremath{\varid{runAnn}}.
  \item
    The calls to \ensuremath{\varid{readCache}} and \ensuremath{\varid{writeCache}} read from and write to the
    \ensuremath{\conid{Name}\mathbin{:\rightharpoonup}\varid{d}} ref cell of \ensuremath{\conid{Refs}} in order to provide the initial value \ensuremath{\varid{init}}
    for fixpoint iteration.
    To this end, \ensuremath{\varid{kleeneFix}} is generalised to \ensuremath{\varid{kleeneFixAboveM}\;\varid{init}\;\varid{f}} which
    iterates the monadic function \ensuremath{\varid{f}} starting from \ensuremath{\varid{init}} until a reductive
    point of \ensuremath{\varid{f}} is found (\ie a \ensuremath{\varid{d}} such that \ensuremath{\varid{f}\;\varid{d}\mathbin{⊑}\varid{return}\;\varid{d}}).
    When fixpoint iteration is first started, there is no cached value, in which
    case \ensuremath{\varid{readCache}} returns \ensuremath{\bot} to be used as the initial value, just as
    for the single-result analysis.
    However, after every iteration of \ensuremath{\varid{rhs}}, the call to \ensuremath{\varid{writeCache}} persists
    the current iterate, which will be the initial value of the fixpoint
    iteration for any future calls to \ensuremath{\varid{bind}} for the same let binding \ensuremath{\varid{x}}.
\end{itemize}
The caching technique is important because naïve fixpoint iteration in
single-result analysis can be exponentially slow for nested let bindings, such
as in
\[
  \Lam{z}{\Let{x_1}{(\Let{x_2}{...(\Let{x_n}{z}{x_n})...}{x_2})}{x_1}}.
\]
Naïvely, every let binding needs two iterations per one iteration of its
enclosing binding: the first iteration assuming \ensuremath{\bot} as the initial value
for $x_i$ and the next assuming the fixpoint \ensuremath{\langle [\varid{z}\mapsto\concolor{\mathsf{U_1}}], \conid{Rep}\;\concolor{\mathsf{U_\omega}} \rangle}.
Ultimately, $z$ is used in the denotation of $x_n$, ..., $x_1$, totalling to
$2^n$ iterations for $x_n$ during stateless analysis.

Stateful caching of the previous fixpoint improves this drastically.
The right-hand side of $x_n = z$ is only iterated $n+1$ times in total:
once with \ensuremath{\bot} as the initial value for $x_n$, once more to confirm the
fixpoint \ensuremath{\langle [\varid{z}\mapsto\concolor{\mathsf{U_1}}], \conid{Rep}\;\concolor{\mathsf{U_\omega}} \rangle} and then $n-1$ more times to confirm
the fixpoints of $x_{n-1}, ..., x_1$.

It is possible to improve the number of iterations for $x_n$ to a constant, by
employing classic chaotic iteration and worklist techniques.
These techniques require a decoupling of iteration order from the lexical
nesting imposed by the syntax tree, instead choosing the next iteratee by
examining the graph of data flow dependencies.
Crucially, such sophisticated and stateful data-flow frameworks can be developed
and improved without complicating the analysis domain, which is often very
complicated in its own right.

\subsection{Case Study: GHC's Demand Analyser}
\label{sec:demand-analysis}

To test how well our denotational interpreter framework scales to real-world
applications, we applied the design pattern to GHC's existing Demand Analyser
and will reproduce the salient points here.
GHC's Demand Analyser infers nested usage~\citep{Sergey:14},
strictness~\citep{SPJ:06} and boxity information.
These analysis results thus fuel a number of optimisations, such
as dead code elimination and unboxing through the worker/wrapper
transformation~\citep{Gill:09}, update avoidance~\citep{Gustavsson:98},
η-expansion and η-reduction, and inlining under lambdas, to name a few.

Concretely, the refactoring entailed
\begin{itemize}
  \item
    identifying which parts of the analyser need to be part of the \ensuremath{\conid{Domain}} interface,
  \item
    writing an abstract denotational interpreter for GHC Core, the typed
    intermediate representation of GHC, thereby identifying
  \item
    validating the usefulness of this interpreter by instantiating it at the GHC
    Core-specific analogue of the concrete by-need domain \ensuremath{\conid{D}\;(\conid{ByNeed}\;\conid{T})}, and finally
  \item
    defining the abstract \ensuremath{\conid{Domain}} instance for Demand Analysis, to replace
    its compositional analysis function on expressions by a call to the
    denotational interpreter.
\end{itemize}
The resulting compiler bootstraps and passes the testsuite.

\subsubsection{GHC Core}

\begin{figure}
\begin{hscode}\SaveRestoreHook
\column{B}{@{}>{\hspre}l<{\hspost}@{}}%
\column{3}{@{}>{\hspre}c<{\hspost}@{}}%
\column{3E}{@{}l@{}}%
\column{6}{@{}>{\hspre}l<{\hspost}@{}}%
\column{16}{@{}>{\hspre}l<{\hspost}@{}}%
\column{E}{@{}>{\hspre}l<{\hspost}@{}}%
\>[B]{}\keyword{data}\;\conid{Expr}{}\<[E]%
\\
\>[B]{}\hsindent{3}{}\<[3]%
\>[3]{}\mathrel{=}{}\<[3E]%
\>[6]{}\conid{Var}\;{}\<[16]%
\>[16]{}\conid{Id}{}\<[E]%
\\
\>[B]{}\hsindent{3}{}\<[3]%
\>[3]{}\mid {}\<[3E]%
\>[6]{}\conid{Lit}\;{}\<[16]%
\>[16]{}\conid{Literal}{}\<[E]%
\\
\>[B]{}\hsindent{3}{}\<[3]%
\>[3]{}\mid {}\<[3E]%
\>[6]{}\conid{App}\;{}\<[16]%
\>[16]{}\conid{Expr}\;\conid{Expr}{}\<[E]%
\\
\>[B]{}\hsindent{3}{}\<[3]%
\>[3]{}\mid {}\<[3E]%
\>[6]{}\conid{Lam}\;{}\<[16]%
\>[16]{}\conid{Var}\;\conid{Expr}{}\<[E]%
\\
\>[B]{}\hsindent{3}{}\<[3]%
\>[3]{}\mid {}\<[3E]%
\>[6]{}\conid{Let}\;{}\<[16]%
\>[16]{}\conid{Bind}\;\conid{Expr}{}\<[E]%
\\
\>[B]{}\hsindent{3}{}\<[3]%
\>[3]{}\mid {}\<[3E]%
\>[6]{}\conid{Case}\;{}\<[16]%
\>[16]{}\conid{Expr}\;\conid{Id}\;\conid{Type}\;\conid{Alt}{}\<[E]%
\\
\>[B]{}\hsindent{3}{}\<[3]%
\>[3]{}\mid {}\<[3E]%
\>[6]{}\conid{Cast}\;{}\<[16]%
\>[16]{}\conid{Expr}\;\conid{Coercion}{}\<[E]%
\\
\>[B]{}\hsindent{3}{}\<[3]%
\>[3]{}\mid {}\<[3E]%
\>[6]{}\conid{Tick}\;{}\<[16]%
\>[16]{}\conid{Tickish}\;\conid{Expr}{}\<[E]%
\\
\>[B]{}\hsindent{3}{}\<[3]%
\>[3]{}\mid {}\<[3E]%
\>[6]{}\conid{Type}\;{}\<[16]%
\>[16]{}\conid{Type}{}\<[E]%
\\
\>[B]{}\hsindent{3}{}\<[3]%
\>[3]{}\mid {}\<[3E]%
\>[6]{}\conid{Coercion}\;{}\<[16]%
\>[16]{}\conid{Coercion}{}\<[E]%
\\
\>[B]{}\keyword{data}\;\conid{Var}\mathrel{=}\conid{Id}\;\mathbin{...}\mid \conid{TyVar}\;\mathbin{...}\mid \conid{CoVar}\mathbin{...}{}\<[E]%
\\
\>[B]{}\keyword{type}\;\conid{Id}\mathrel{=}\conid{Var}\mbox{\onelinecomment  always a term-level Id}{}\<[E]%
\\
\>[B]{}\keyword{data}\;\conid{Literal}\mathrel{=}\conid{LitNumber}\;\mathbin{...}\mid \conid{LitFloat}\;\mathbin{...}\mid \conid{LitString}\mathbin{...}{}\<[E]%
\\
\>[B]{}\keyword{type}\;\conid{Alt}\mathrel{=}(\conid{AltCon},[\mskip1.5mu \conid{Var}\mskip1.5mu],\conid{Expr}){}\<[E]%
\\
\>[B]{}\keyword{data}\;\conid{AltCon}\mathrel{=}\conid{LitAlt}\;\conid{Literal}\mid \conid{DataAlt}\;\conid{DataCon}\mid \conid{DEFAULT}{}\<[E]%
\\
\>[B]{}\keyword{data}\;\conid{Bind}\mathrel{=}\conid{NonRec}\;\conid{Id}\;\conid{Expr}\mid \conid{Rec}\;[\mskip1.5mu (\conid{Id},\conid{Expr})\mskip1.5mu]{}\<[E]%
\\
\>[B]{}\keyword{data}\;\conid{Type}{}\<[16]%
\>[16]{}\mathrel{=}\mathbin{...}{}\<[E]%
\\
\>[B]{}\keyword{data}\;\conid{Coercion}{}\<[16]%
\>[16]{}\mathrel{=}\mathbin{...}{}\<[E]%
\ColumnHook
\end{hscode}\resethooks
\caption{GHC Core}
\label{fig:core}
\end{figure}

GHC Core implements a variant of the polymorphic lambda calculus System $F_ω$
called System $F_C$~\citep{Sulzmann:07}.
Its definition in GHC is given in \Cref{fig:core} and includes explicit
type applications as well as witnesses of type equality constraints called
\emph{coercions}.

GHC Core's \ensuremath{\conid{Expr}} has a lot in common with the untyped object language \ensuremath{\conid{Exp}}
introduced in \Cref{sec:lang}.
For example, there are constructors for \ensuremath{\conid{Var}}, \ensuremath{\conid{App}}, \ensuremath{\conid{Lam}}, \ensuremath{\conid{Let}} and \ensuremath{\conid{Case}}.
There are a number of differences, however:
\begin{itemize}
  \item
    GHC Core allows non-variable arguments in applications.
    This has implications on the denotational interpreter, which must let-bind
    non-variable arguments to establish A-normal form on-the-fly.
  \item
    There is no distinguished \ensuremath{\conid{ConApp}} form. That is because data constructors
    are just special kinds of \ensuremath{\conid{Id}}s and may be unsaturated; the interpreter
    must eta-expand such data constructor applications on-the-fly.
  \item
    \ensuremath{\conid{Case}} alternatives allow matching on literals (\ensuremath{\conid{LitAlt}}) as well as data
    constructors (\ensuremath{\conid{DataAlt}}), and include a default alternative (\ensuremath{\conid{DEFAULT}}) that
    matches any case not matched by other alternatives.
    Furthermore, after \ensuremath{\conid{Case}} evaluates the scrutinee, its value is bound to
    a designated \ensuremath{\conid{Id}} called the \emph{case binder} that scopes over all case
    alternatives.
  \item
    Beyond data constructors, there are other distinguished \ensuremath{\conid{Id}}s without a
    local binding, such as ``global'' identifiers imported from a different
    module, class method seelctors and primitive operations defined by the
    runtime system.
  \item
    \ensuremath{\conid{Let}} bindings are either explicitly non-recursive (\ensuremath{\conid{NonRec}}) or a mutually
    recursive group with potentially many bindings (\ensuremath{\conid{Rec}}).
  \item
    Not shown in \Cref{fig:core} is GHC's support for \emph{inline unfoldings}
    attached to let-bound \ensuremath{\conid{Id}}s as well as \emph{rewrite rules} declared by
    \texttt{RULES} pragmas.
    Each give rise to additional right-hand sides which must be handled with
    conservative care.
    Mistreatment of these subtle constructs in the Demand Analyser has caused
    numerous bugs over the years.
\end{itemize}
Beyond these differences, GHC Core includes forms for embedding \ensuremath{\conid{Literal}}s,
\ensuremath{\conid{Type}}s and \ensuremath{\conid{Coercion}}s in select expression forms.
Type abstraction and application use regular \ensuremath{\conid{Lam}} and \ensuremath{\conid{App}} constructors,
whereas rewriting an expression's type along a \ensuremath{\conid{Coercion}} happens through \ensuremath{\conid{Cast}}s.
The constructor \ensuremath{\conid{Tick}} annotates debugging and profiling information and can be
ignored.

\subsubsection{A Semantic \ensuremath{\conid{Domain}} for GHC Core}

\begin{figure}
\begin{hscode}\SaveRestoreHook
\column{B}{@{}>{\hspre}l<{\hspost}@{}}%
\column{3}{@{}>{\hspre}l<{\hspost}@{}}%
\column{13}{@{}>{\hspre}c<{\hspost}@{}}%
\column{13E}{@{}l@{}}%
\column{16}{@{}>{\hspre}l<{\hspost}@{}}%
\column{E}{@{}>{\hspre}l<{\hspost}@{}}%
\>[B]{}\keyword{data}\;\conid{Event}{}\<[13]%
\>[13]{}\mathrel{=}{}\<[13E]%
\>[16]{}\conid{Look}\;\conid{Id}\mid \conid{LookArg}\;\conid{CoreExpr}\mid \conid{Update}{}\<[E]%
\\
\>[13]{}\mid {}\<[13E]%
\>[16]{}\conid{App}_{1}\mid \conid{App}_{2}\mid \conid{Case}_{1}\mid \conid{Case}_{2}\mid \conid{Let}_{1}{}\<[E]%
\\
\>[B]{}\keyword{class}\;\conid{Trace}\;\varid{d}\;\keyword{where}\;\varid{step}\mathbin{::}\conid{Event}\to \varid{d}\to \varid{d}{}\<[E]%
\\[\blanklineskip]%
\>[B]{}\keyword{class}\;\conid{Domain}\;\varid{d}\;\keyword{where}{}\<[E]%
\\
\>[B]{}\hsindent{3}{}\<[3]%
\>[3]{}\varid{stuck}\mathbin{::}\varid{d}{}\<[E]%
\\
\>[B]{}\hsindent{3}{}\<[3]%
\>[3]{}\varid{lit}\mathbin{::}\conid{Literal}\to \varid{d}{}\<[E]%
\\
\>[B]{}\hsindent{3}{}\<[3]%
\>[3]{}\varid{global}\mathbin{::}\conid{Id}\to \varid{d}{}\<[E]%
\\
\>[B]{}\hsindent{3}{}\<[3]%
\>[3]{}\varid{classOp}\mathbin{::}\conid{Id}\to \conid{Class}\to \varid{d}{}\<[E]%
\\
\>[B]{}\hsindent{3}{}\<[3]%
\>[3]{}\varid{primOp}\mathbin{::}\conid{Id}\to \conid{PrimOp}\to \varid{d}{}\<[E]%
\\
\>[B]{}\hsindent{3}{}\<[3]%
\>[3]{}\varid{fun}\mathbin{::}\conid{Id}\to (\varid{d}\to \varid{d})\to \varid{d}{}\<[E]%
\\
\>[B]{}\hsindent{3}{}\<[3]%
\>[3]{}\varid{con}\mathbin{::}\conid{DataCon}\to [\mskip1.5mu \varid{d}\mskip1.5mu]\to \varid{d}{}\<[E]%
\\
\>[B]{}\hsindent{3}{}\<[3]%
\>[3]{}\varid{apply}\mathbin{::}\varid{d}\to (\conid{Bool},\varid{d})\to \varid{d}{}\<[E]%
\\
\>[B]{}\hsindent{3}{}\<[3]%
\>[3]{}\varid{select}\mathbin{::}\varid{d}\to \conid{CoreExpr}\to \conid{Id}\to [\mskip1.5mu \conid{DAlt}\;\varid{d}\mskip1.5mu]\to \varid{d}{}\<[E]%
\\
\>[B]{}\hsindent{3}{}\<[3]%
\>[3]{}\varid{erased}\mathbin{::}\varid{d}{}\<[E]%
\\
\>[B]{}\hsindent{3}{}\<[3]%
\>[3]{}\varid{keepAlive}\mathbin{::}[\mskip1.5mu \varid{d}\mskip1.5mu]\to \varid{d}\to \varid{d}{}\<[E]%
\\
\>[B]{}\keyword{type}\;\conid{DAlt}\;\varid{d}\mathrel{=}(\conid{AltCon},[\mskip1.5mu \conid{Id}\mskip1.5mu],\varid{d}\to [\mskip1.5mu \varid{d}\mskip1.5mu]\to \varid{d}){}\<[E]%
\\[\blanklineskip]%
\>[B]{}\keyword{data}\;\conid{BindHint}\mathrel{=}\conid{BindArg}\;\conid{Id}\mid \conid{BindLet}\;\conid{Bind}{}\<[E]%
\\
\>[B]{}\keyword{class}\;\conid{HasBind}\;\varid{d}\;\keyword{where}{}\<[E]%
\\
\>[B]{}\hsindent{3}{}\<[3]%
\>[3]{}\varid{bind}\mathbin{::}\conid{BindHint}\to [\mskip1.5mu [\mskip1.5mu \varid{d}\mskip1.5mu]\to \varid{d}\mskip1.5mu]\to ([\mskip1.5mu \varid{d}\mskip1.5mu]\to \varid{d})\to \varid{d}{}\<[E]%
\ColumnHook
\end{hscode}\resethooks
\caption{A \ensuremath{\conid{Domain}} interface for GHC Core}
\label{fig:core-domain}
\end{figure}

\Cref{fig:core-domain} defines the semantic domain abstraction for which
we implemented both a concrete \ensuremath{\conid{ByNeed}} instance as well as an abstract
instance for Demand Analysis.
Its design was inspired by the domain definition in \Cref{fig:eval}, but
ultimately driven by the hands-on desire to accommodate both \ensuremath{\conid{ByNeed}} and Demand
Analysis as instances.

The \ensuremath{\varid{stuck}}, \ensuremath{\varid{con}}, \ensuremath{\varid{fun}}, \ensuremath{\varid{apply}} and \ensuremath{\varid{select}} methods serve the exact same
purpose as in prior sections, generalised to deal with the Core expressions
they are modelled after.
Method \ensuremath{\varid{apply}} receives an additional \ensuremath{\conid{Bool}} to tell whether it is a
runtime-irrelevant type application.
Unsurprisingly, there is a method \ensuremath{\varid{lit}} for embedding \ensuremath{\conid{Literal}}s, similar to
\ensuremath{\varid{con}}.
Demand Analysis assigns special meaning to primitive operations (\ensuremath{\varid{primOp}}),
class method selectors (\ensuremath{\varid{classOp}}) and imported \ensuremath{\conid{Id}}s (\ensuremath{\varid{global}}), so each
get their own \ensuremath{\conid{Domain}} method.

Types and coercions are erased at runtime, represented by method \ensuremath{\varid{erased}}.
Coercion expressions, inline unfoldings and rewrite \ensuremath{\conid{RULES}} keep alive
their free variables (\ensuremath{\varid{keepAlive}}).

The \ensuremath{\conid{HasBind}} type class accommodates both non-recursive as well as mutually
recursive let bindings.
The \ensuremath{\conid{BindHint}} is used to communicate whether such a binding comes from
the on-the-fly ANF-isation pass of the interpreter (\ensuremath{\conid{BindArg}}) or whether it was
a manifest let binding in the Core program (\ensuremath{\conid{BindLet}}).


\subsubsection{The Glasgow Haskell Denotational Interpreter (GHDi)}

\begin{figure}
\begin{hscode}\SaveRestoreHook
\column{B}{@{}>{\hspre}l<{\hspost}@{}}%
\column{3}{@{}>{\hspre}l<{\hspost}@{}}%
\column{4}{@{}>{\hspre}c<{\hspost}@{}}%
\column{4E}{@{}l@{}}%
\column{5}{@{}>{\hspre}l<{\hspost}@{}}%
\column{7}{@{}>{\hspre}l<{\hspost}@{}}%
\column{8}{@{}>{\hspre}l<{\hspost}@{}}%
\column{9}{@{}>{\hspre}l<{\hspost}@{}}%
\column{13}{@{}>{\hspre}c<{\hspost}@{}}%
\column{13E}{@{}l@{}}%
\column{17}{@{}>{\hspre}l<{\hspost}@{}}%
\column{18}{@{}>{\hspre}l<{\hspost}@{}}%
\column{33}{@{}>{\hspre}l<{\hspost}@{}}%
\column{E}{@{}>{\hspre}l<{\hspost}@{}}%
\>[B]{}\keyword{type}\;\conid{D}\;\varid{d}\mathrel{=}(\conid{Trace}\;\varid{d},\conid{Domain}\;\varid{d},\conid{HasBind}\;\varid{d}){}\<[E]%
\\
\>[B]{}\varid{anfise}{}\<[13]%
\>[13]{}\mathbin{::}{}\<[13E]%
\>[17]{}\conid{D}\;\varid{d}\Rightarrow [\mskip1.5mu \conid{Expr}\mskip1.5mu]\to (\conid{Name}\mathbin{:\rightharpoonup}\varid{d})\to ([\mskip1.5mu \varid{d}\mskip1.5mu]\to \varid{d})\to \varid{d}{}\<[E]%
\\
\>[B]{}\varid{evalConApp}{}\<[13]%
\>[13]{}\mathbin{::}{}\<[13E]%
\>[17]{}\conid{D}\;\varid{d}\Rightarrow \conid{DataCon}\to [\mskip1.5mu \varid{d}\mskip1.5mu]\to \varid{d}{}\<[E]%
\\[\blanklineskip]%
\>[B]{}\mathcal{S}\denot{\wild}_{\wild}{}\<[13]%
\>[13]{}\mathbin{::}{}\<[13E]%
\>[17]{}\conid{D}\;\varid{d}\Rightarrow \conid{Expr}\to (\conid{Name}\mathbin{:\rightharpoonup}\varid{d})\to \varid{d}{}\<[E]%
\\
\>[B]{}\mathcal{S}\denot{\conid{Type}\;\anonymous }_{\varcolor{\rho}}{}\<[33]%
\>[33]{}\mathrel{=}\varid{erased}{}\<[E]%
\\
\>[B]{}\mathcal{S}\denot{\conid{Lit}\;\varid{l}}_{\varcolor{\rho}}{}\<[33]%
\>[33]{}\mathrel{=}\varid{lit}\;\varid{l}{}\<[E]%
\\
\>[B]{}\mathcal{S}\denot{\conid{Var}\;\varid{x}}_{\varcolor{\rho}}{}\<[18]%
\>[18]{}\mid \varid{not}\;\varid{special}{}\<[33]%
\>[33]{}\mathrel{=}\varcolor{\rho}\mathop{!}\varid{x}{}\<[E]%
\\
\>[18]{}\mid \varid{otherwise}{}\<[33]%
\>[33]{}\mathrel{=}\mathbin{...}{}\<[E]%
\\
\>[B]{}\mathcal{S}\denot{\conid{Lam}\;\varid{x}\;\varid{e}}_{\varcolor{\rho}}{}\<[33]%
\>[33]{}\mathrel{=}\varid{fun}\;\varid{x}\;(\lambda \varid{d}\to \varid{step}\;\conid{App}_{2}\;(\mathcal{S}\denot{\varid{e}}_{\varcolor{\rho}[\varid{x}\mapsto\varid{d}]})){}\<[E]%
\\
\>[B]{}\mathcal{S}\denot{\varid{e}\mathord{@}\conid{App}\;\{\mskip1.5mu \mskip1.5mu\}}_{\varcolor{\rho}}{}\<[E]%
\\
\>[B]{}\hsindent{4}{}\<[4]%
\>[4]{}\mid {}\<[4E]%
\>[7]{}\conid{Var}\;\varid{v}\leftarrow \varid{f},\conid{Just}\;\varid{dc}\leftarrow \varid{isDataConWorkId\char95 maybe}\;\varid{v}{}\<[E]%
\\
\>[B]{}\hsindent{4}{}\<[4]%
\>[4]{}\mathrel{=}{}\<[4E]%
\>[7]{}\varid{anfise}\;\varid{as}\;\varcolor{\rho}\;(\varid{evalConApp}\;\varid{dc}){}\<[E]%
\\
\>[B]{}\hsindent{4}{}\<[4]%
\>[4]{}\mid {}\<[4E]%
\>[7]{}\varid{otherwise}{}\<[E]%
\\
\>[B]{}\hsindent{4}{}\<[4]%
\>[4]{}\mathrel{=}{}\<[4E]%
\>[7]{}\varid{anfise}\;(\varid{f}\mathbin{:}\varid{as})\;\varcolor{\rho}\mathbin{\$}\lambda (\varid{df}\mathbin{:}\varid{das})\to {}\<[E]%
\\
\>[7]{}\hsindent{2}{}\<[9]%
\>[9]{}\varid{go}\;\varid{df}\;(\varid{zipWith}\;(\lambda \varid{d}\;\varid{a}\to (\varid{d},\varid{isTypeArg}\;\varid{a}))\;\varid{das}\;\varid{as}){}\<[E]%
\\
\>[B]{}\hsindent{3}{}\<[3]%
\>[3]{}\keyword{where}{}\<[E]%
\\
\>[3]{}\hsindent{2}{}\<[5]%
\>[5]{}(\varid{f},\varid{as})\mathrel{=}\varid{collectArgs}\;\varid{e}{}\<[E]%
\\
\>[3]{}\hsindent{2}{}\<[5]%
\>[5]{}\varid{go}\;\varid{df}\;[\mskip1.5mu \mskip1.5mu]\mathrel{=}\varid{df}{}\<[E]%
\\
\>[3]{}\hsindent{2}{}\<[5]%
\>[5]{}\varid{go}\;\varid{df}\;((\varid{d},\varid{is\char95 ty})\mathbin{:}\varid{ds})\mathrel{=}\varid{go}\;(\varid{step}\;\conid{App}_{1}\mathbin{\$}\varid{apply}\;\varid{df}\;(\varid{is\char95 ty},\varid{d}))\;\varid{ds}{}\<[E]%
\\
\>[B]{}\mathcal{S}\denot{\conid{Let}\;\varid{b}\mathord{@}(\conid{NonRec}\;\varid{x}\;\varid{rhs})\;\varid{body}}_{\varcolor{\rho}}\mathrel{=}{}\<[E]%
\\
\>[B]{}\hsindent{3}{}\<[3]%
\>[3]{}\varid{bind}\;(\conid{BindLet}\;\varid{b})\;{}\<[E]%
\\
\>[3]{}\hsindent{5}{}\<[8]%
\>[8]{}[\mskip1.5mu \lambda \varid{ds}\to \varid{keepAliveUnfRules}\;\varid{x}\;\varcolor{\rho}\;(\mathcal{S}\denot{\varid{rhs}}_{\varcolor{\rho}})\mskip1.5mu]\;{}\<[E]%
\\
\>[3]{}\hsindent{5}{}\<[8]%
\>[8]{}(\lambda \varid{ds}\to \varid{step}\;\conid{Let}_{1}\;(\mathcal{S}\denot{\varid{body}}_{\varcolor{\rho}[\varid{x}\mapsto\varid{step}\;(\conid{Lookup}\;\varid{x})\;(\varid{only}\;\varid{ds})]})){}\<[E]%
\\
\>[B]{}\mathbin{...}{}\<[E]%
\ColumnHook
\end{hscode}\resethooks
\caption{A glimpse of the Glasgow Haskell Denotational Interpreter (GHDi)}
\label{fig:core-eval}
\end{figure}

\Cref{fig:core-eval} shows a slightly adjusted and abridged version of the
denotational interpreter.
The actual definition takes around 100 lines of Haskell.
Its highlights include erasure of types, a new case for literals, on-the-fly
ANF-isation in the application case and picking out data constructor application
from regular function application in order to η-expand accordingly in
\ensuremath{\varid{evalConApp}}.
Whenever an ANF-ised argument is looked up, a \ensuremath{\conid{LookArg}} event is emitted;
this is simply for a lack of a globally unique \ensuremath{\conid{Id}}.
In the \ensuremath{\conid{Let}} case, the call \ensuremath{\varid{keepAliveUnfRules}} makes sure to keep alive
the free variables of inline unfoldings and rewrite rules attached to \ensuremath{\varid{x}}.

The \ensuremath{\conid{Domain}} and \ensuremath{\conid{HasBind}} instance for the concrete semantics \ensuremath{\conid{D}\;(\conid{ByNeed}\;\conid{T})}
is routine.
The resulting denotational interpreter can execute GHC Core expressions.
To demonstrate this, we wrote a small REPL around it:
\begin{Verbatim}
$ ./ghdi $(ghc --print-libdir)
prompt> let f x = x*42 :: Int; {-# NOINLINE f #-} in even $ f 3
Above expression as (optimised) Core:
  join {
    f_sZe [InlPrag=NOINLINE, Dmd=1C(1,L)] :: Int -> Bool
    [LclId[JoinId(1)(Just [!])], Arity=1, Str=<1L>]
    f_sZe (x_aYj [OS=OneShot] :: Int)
      = case x_aYj of { I# x1_aHU ->
        case remInt# (*# x1_aHU 42#) 2# of {
          __DEFAULT -> False;
          0# -> True
        }
        } } in
  jump f_sZe (I# 3#)
Trace of denotational interpreter:
Let1->App1->Lookup(f_sZe)->Update->App2->Case1->
  LookupArg(I# 3#)->Update->Case2->Case1->App1->
  App1->App2->App2->LookupArg(*# x1_aHU 42#)->App1->App1->
  App2->App2->Update->Case2-><(True, [0↦_, 1↦_, 2↦_])>
\end{Verbatim}

\subsubsection{Demand Analysis as Denotational Interpreter}

\begin{figure}
\begin{hscode}\SaveRestoreHook
\column{B}{@{}>{\hspre}l<{\hspost}@{}}%
\column{3}{@{}>{\hspre}l<{\hspost}@{}}%
\column{5}{@{}>{\hspre}l<{\hspost}@{}}%
\column{E}{@{}>{\hspre}l<{\hspost}@{}}%
\>[B]{}\keyword{type}\;\conid{DmdT}\;\varid{s}\;\varid{v}\mathrel{=}\conid{AnalEnv}\to \conid{SubDemand}\to \conid{AnalM}\;\varid{s}\;(\varid{v},\conid{Uses}){}\<[E]%
\\
\>[B]{}\keyword{type}\;\conid{DmdVal}\mathrel{=}[\mskip1.5mu \conid{Demand}\mskip1.5mu]{}\<[E]%
\\
\>[B]{}\keyword{type}\;\conid{DmdD}\;\varid{s}\mathrel{=}\conid{DmdT}\;\varid{s}\;\conid{DmdVal}{}\<[E]%
\\[\blanklineskip]%
\>[B]{}\keyword{instance}\;\conid{Trace}\;(\conid{DmdD}\;\varid{s})\;\keyword{where}{}\<[E]%
\\
\>[B]{}\hsindent{3}{}\<[3]%
\>[3]{}\varid{step}\;(\conid{Look}\;\varid{x})\;\varid{d}\mathrel{=}\lambda \varid{env}\;\varid{sd}\to \keyword{do}{}\<[E]%
\\
\>[3]{}\hsindent{2}{}\<[5]%
\>[5]{}(\varid{v},\varcolor{\varphi})\leftarrow \varid{d}\;\varid{env}\;\varid{sd}{}\<[E]%
\\
\>[3]{}\hsindent{2}{}\<[5]%
\>[5]{}\keyword{if}\;\varid{isBoundAtTopLvl}\;\varid{env}\;\varid{x}\;\keyword{then}\mathbin{...}\keyword{else}\;\varid{pure}\;(\varid{v},\varcolor{\varphi}\mathbin{+}[\varid{x}\mapsto\conid{C\char95 11}\mathbin{:*:}\varid{sd}]){}\<[E]%
\\
\>[B]{}\hsindent{3}{}\<[3]%
\>[3]{}\varid{step}\;\anonymous \;\varid{d}\mathrel{=}\varid{d}{}\<[E]%
\\
\>[B]{}\keyword{instance}\;\conid{Domain}\;(\conid{DmdD}\;\varid{s})\;\keyword{where}\mathbin{...}{}\<[E]%
\\
\>[B]{}\keyword{instance}\;\conid{HasBind}\;(\conid{DmdD}\;\varid{s})\;\keyword{where}\mathbin{...}{}\<[E]%
\ColumnHook
\end{hscode}\resethooks
\caption{A rough outline of the semantic domain of Demand Analysis}
\label{fig:dmd-domain}
\end{figure}

\Cref{fig:dmd-domain} gives a rough sketch of the semantic domain definition for
Demand Analysis.
The abstract trace type \ensuremath{\conid{DmdT}} produces some value \ensuremath{\varid{v}} as well as a \ensuremath{\conid{Uses}},
just as for \ensuremath{\concolor{\mathsf{T_U}}} in \Cref{sec:usage-analysis}.
However, it does so in a rather deep nest of types:
\begin{itemize}
  \item
    \ensuremath{\conid{AnalM}\;\varid{s}} plays the role of \ensuremath{\conid{AnnT}\;\varid{s}} in \Cref{sec:annotations}, maintaining
    annotations and speeding up fixpoint iteration.
  \item
    The analysis result is \emph{indexed} by a \ensuremath{\conid{SubDemand}}; a description of how
    deep the expression is to be evaluated.
    A \ensuremath{\conid{SubDemand}} is best understood as an abstraction of evaluation contexts.
    The more precise this abstraction, the more accurate are the \ensuremath{\conid{Uses}} returned
    for that expression.
  \item
    Furthermore, an \ensuremath{\conid{AnalEnv}} carries global state such as optimisation flags,
    means for reducing types and further syntactic information about bindings,
    such as whether a variable is bound at the top-level.
\end{itemize}
An abstract domain defined as a function sounds antithetical to our mantra in
\Cref{sec:abstraction} that abstract domains are finitely represented.
However, Demand Analysis only ever maintains one particular point of the indexed
domain, that is, every expression is analysed under one particular \ensuremath{\conid{SubDemand}}.
This \ensuremath{\conid{SubDemand}} may increase during fixpoint iteration, though, causing another
round of analysis.
We apply the typical widening measures in \ensuremath{\conid{HasBind}}, so in practice Demand
Analysis has not run into infinite loops for a couple of years.

Type \ensuremath{\conid{DmdVal}} is the similar to \ensuremath{\concolor{\mathsf{Value_U}}}, except that it lists full \ensuremath{\conid{Demand}}s
instead of flat usage cardinalities \ensuremath{\conid{U}}.
Such a demand \ensuremath{\varid{n}\mathbin{:*:}\varid{sd}} describes how often (\ensuremath{\varid{sd}}) and how deep (\ensuremath{\varid{sd}}) a
variable is used; it is an abstraction of its contexts of use.

The \ensuremath{\conid{Trace}} instance is very similar to the one for \ensuremath{\concolor{\mathsf{D_U}}}, it is just a little
bit more complicated because of special code for top-level bindings and
the fact that bindings get annotated with demands instead of simple usage cardinalities.
The demand \ensuremath{\conid{C\char95 11}\mathbin{:*:}\varid{sd}} describes a single, strict use of the variable in the
evaluation context described by sub-demand \ensuremath{\varid{sd}}.

The resulting analysis is sufficient to bootstrap the compiler and passes the
testsuite.
However, the compiler performance takes a serious hit due to the implementation of
\ensuremath{\varid{bind}\mathbin{::}\conid{BindHint}\to [\mskip1.5mu [\mskip1.5mu \varid{d}\mskip1.5mu]\to \varid{d}\mskip1.5mu]\to ([\mskip1.5mu \varid{d}\mskip1.5mu]\to \varid{d})\to \varid{d}}.
The way fixpoint iteration updates one binding \ensuremath{\varid{d}} in mutually recursive groups
\ensuremath{[\mskip1.5mu \varid{d}\mskip1.5mu]} at a time is very inefficient for the linked list representation, also
because every \ensuremath{[\mskip1.5mu \varid{d}\mskip1.5mu]} ultimately turn into as many updates of the \ensuremath{\conid{Name}\mathbin{:\rightharpoonup}\varid{d}}
mapping.
It would be far preferable to operate on the \ensuremath{\conid{Name}\mathbin{:\rightharpoonup}\varid{d}} environment directly.
Finding a good abstraction that achieves this without exposing the whole
environment is left for future work.
\end{toappendix}

\smallskip
\noindent
It is nice that usage and type analysis fit into the same
framework as the call-by-need semantics.
Another important benefit is that correctness proofs become simpler, as we will
see next.

\section{Generic Abstract By-Name and By-Need Interpretation}
\label{sec:soundness}

\begin{toappendix}
\label{sec:soundness-detail}

So far, we have seen how to \emph{use} the abstraction
\Cref{thm:abstract-by-need}, but not proved it.
Proving this theorem correct is the goal of this section,
and we approach the problem from the bottom up.

We will discuss why it is safe to approximate guarded fixpoints with
least fixpoints and why the definition of the Galois connection in
\Cref{fig:abstract-name-need} as a fold over the trace is well-defined
in \Cref{sec:safety-extension}.
Then we will go on to prove sound abstract by-name interpretation in
\Cref{sec:by-name-soundness}, and finally sound abstract by-need interpretation
in \Cref{sec:by-need-soundness}.

\subsection{Guarded Fixpoints, Safety Properties and Safety Extension of a Galois Connection}
\label{sec:safety-extension}

\Cref{fig:abstract-name-need} describes a semantic trace property as a ``fold'', in
terms of a \ensuremath{\conid{Trace}} instance.
Of course such a fold (an inductive elimination procedure) has no meaning when
the trace is infinite!
Yet it is always clear what we mean: When the trace is infinite and described by
a guarded fixpoint, we consider the meaning of the fold as the limit (\ie least
fixpoint) of its finite prefixes.
In this subsection, we discuss when and why this meaning is correct.

Suppose that we were only interested in the trace component of our
semantic domain, thus effectively restricting ourselves to
$\Traces \triangleq \ensuremath{\conid{T}\;()}$, and that we were to approximate properties $P ∈
\pow{\Traces}$ about such traces by a Galois connection
$\ensuremath{\alpha} : (\pow{\Traces},⊆) \rightleftarrows (\ensuremath{\widehat{\conid{D}}}, ⊑) : \ensuremath{\gamma}$.
Alas, although the abstraction function \ensuremath{\alpha} is well-defined as a mathematical
function, it most certainly is \emph{not} computable at infinite inputs (in
$\Traces^{\infty}$), for example at
\ensuremath{\varid{fix}\;(\conid{Step}\;(\conid{Look}\;\varid{x}))\mathrel{=}\conid{Step}\;(\conid{Look}\;\varid{x})\;(\conid{Step}\;(\conid{Look}\;\varid{x})\mathbin{...})}!

The whole point about \emph{static} analyses is that they approximate program
behavior in finite data.
As we have discussed in \Cref{sec:usage-fixpoint}, this rules out use of
\emph{guarded} fixpoints \ensuremath{\varid{fix}} for usage analysis, so it uses computes the
\emph{least} fixpoint \ensuremath{\varid{lfp}} instead.
Concretely, static analyses often approximate the abstraction of the guarded
fixpoint by the least fixpoint of the abstracted iteratee, assuming the
following approximation relationship:
\[
\ensuremath{\alpha\;(\{\varid{fix}\;(\conid{Step}\;(\conid{Look}\;\varid{x}))\})} ⊑ \ensuremath{\varid{lfp}\;(\alpha\hsdot{\circ }{.\ }{(\conid{Step}\;(\conid{Look}\;\varid{x}))}^{*}\hsdot{\circ }{.\ }\gamma)}.
\]
This inequality does not hold for \emph{all} trace properties, but we will
show that it holds for \emph{safety} properties~\citep{Lamport:77}:

\begin{definition}[Safety property]
A trace property $P ⊆ \Traces$ is a \emph{safety property} iff,
whenever $\ensuremath{\varcolor{\tau}_{1}}∈\Traces^{\infty}$ violates $P$ (so $\ensuremath{\varcolor{\tau}_{1}} \not∈ P$), then there exists some proper
prefix $\ensuremath{\varcolor{\tau}_{2}}∈\Traces^{*}$ (written $\ensuremath{\varcolor{\tau}_{2}} \lessdot \ensuremath{\varcolor{\tau}_{1}}$) such that $\ensuremath{\varcolor{\tau}_{2}} \not∈ P$.
\end{definition}

Note that both well-typedness (``\ensuremath{\varcolor{\tau}} does not go wrong'') and usage cardinality
abstract safety properties.
Conveniently, guarded recursive predicates (on traces) always describe safety
properties~\citep{Spies:21,iris-lecture-notes}.

The contraposition of the above definition is
\[
  \forall \ensuremath{\varcolor{\tau}_{1}}∈\Traces^{\infty}.\ (\forall \ensuremath{\varcolor{\tau}_{2}}∈\Traces^{*}.\ \ensuremath{\varcolor{\tau}_{2}} \lessdot \ensuremath{\varcolor{\tau}_{1}} \Longrightarrow \ensuremath{\varcolor{\tau}_{2}} ∈ P) \Longrightarrow \ensuremath{\varcolor{\tau}_{1}} ∈ P,
\]
and we can exploit safety to extend a finitary Galois connection, such as
$α_{\mathcal{S}}$ in \Cref{fig:abstract-name-need} defined by a fold over the
trace, to infinite inputs:
\begin{lemma}[Safety extension]
\label{thm:safety-extension}
Let \ensuremath{\widehat{\conid{D}}} be a domain with an instance for \ensuremath{\conid{Lat}},
$\ensuremath{\alpha} : (\pow{\Traces^{*}},⊆) \rightleftarrows (\ensuremath{\widehat{\conid{D}}}, ⊑) : \ensuremath{\gamma}$ a Galois
connection and $P ∈ \pow{\Traces}$ a safety property.
Then any domain element \ensuremath{\widehat{\varid{d}}} that soundly approximates $P$ via \ensuremath{\gamma} on finite
traces soundly approximates $P$ on infinite traces as well:
\[
  \forall \ensuremath{\widehat{\varid{d}}}.\ P ∩ \Traces^{*} ⊆ \ensuremath{\gamma}(\ensuremath{\widehat{\varid{d}}}) \Longrightarrow P ∩ \Traces^{\infty} ⊆ \ensuremath{\gamma^{\infty}}(\ensuremath{\widehat{\varid{d}}}),
\]
where the \emph{extension} $\ensuremath{\alpha^{\infty}} : (\pow{\Traces^{*}},⊆) \rightleftarrows (\ensuremath{\widehat{\conid{D}}}, ⊑) : \ensuremath{\gamma^{\infty}}$ of
$\ensuremath{\alpha}\rightleftarrows\ensuremath{\gamma}$ is defined by the following abstraction function:
\[
  \ensuremath{\alpha^{\infty}}(P) \triangleq \ensuremath{\alpha}(\{ \ensuremath{\varcolor{\tau}_{2}} \mid \exists \ensuremath{\varcolor{\tau}_{1}}∈P.\ \ensuremath{\varcolor{\tau}_{2}} \lessdot \ensuremath{\varcolor{\tau}_{1}} \})
\]
\end{lemma}
\begin{proof}
First note that \ensuremath{\alpha^{\infty}} uniquely determines the Galois connection by the
representation function~\citep[Section 4.3]{Nielson:99}
\[
  \ensuremath{\beta^{\infty}}(\ensuremath{\varcolor{\tau}_{1}}) \triangleq \ensuremath{\alpha}({\textstyle \bigcup} \{ \ensuremath{\varcolor{\tau}_{2}} \mid \ensuremath{\varcolor{\tau}_{2}} \lessdot \ensuremath{\varcolor{\tau}_{1}} \}).
\]
Now let $\ensuremath{\varcolor{\tau}} ∈ P ∩ \Traces^{\infty}$.
The goal is to show that $\ensuremath{\varcolor{\tau}} ∈ \ensuremath{\gamma^{\infty}}(\ensuremath{\widehat{\varid{d}}})$, which we rewrite as follows:
\begin{hscode}\SaveRestoreHook
\column{B}{@{}>{\hspre}c<{\hspost}@{}}%
\column{BE}{@{}l@{}}%
\column{7}{@{}>{\hspre}l<{\hspost}@{}}%
\column{E}{@{}>{\hspre}l<{\hspost}@{}}%
\>[7]{}\varcolor{\tau}\in \gamma^{\infty}\;\widehat{\varid{d}}{}\<[E]%
\\
\>[B]{}\Longleftrightarrow{}\<[BE]%
\>[7]{}\mbox{\commentbegin  Galois  \commentend}{}\<[E]%
\\
\>[7]{}\beta^{\infty}\;\varcolor{\tau}\mathbin{⊑}\widehat{\varid{d}}{}\<[E]%
\\
\>[B]{}\Longleftrightarrow{}\<[BE]%
\>[7]{}\mbox{\commentbegin  Definition of \ensuremath{\beta^{\infty}}  \commentend}{}\<[E]%
\\
\>[7]{}\alpha\;\bigcup\{\varcolor{\tau}_{2}\mid \varcolor{\tau}_{2}\lessdot\varcolor{\tau}_{1}\}\mathbin{⊑}\widehat{\varid{d}}{}\<[E]%
\\
\>[B]{}\Longleftrightarrow{}\<[BE]%
\>[7]{}\mbox{\commentbegin  Galois  \commentend}{}\<[E]%
\\
\>[7]{}\bigcup\{\varcolor{\tau}_{2}\mid \varcolor{\tau}_{2}\lessdot\varcolor{\tau}_{1}\}\mathbin{⊆}\gamma\;\widehat{\varid{d}}{}\<[E]%
\\
\>[B]{}\Longleftrightarrow{}\<[BE]%
\>[7]{}\mbox{\commentbegin  Definition of Union  \commentend}{}\<[E]%
\\
\>[7]{}\keyword{\forall}\!\! \hsforall \;\varcolor{\tau}_{2}\hsdot{\circ }{.\ }\varcolor{\tau}_{2}\lessdot\varcolor{\tau}\implies\varcolor{\tau}_{2}\in \gamma\;\widehat{\varid{d}}{}\<[E]%
\ColumnHook
\end{hscode}\resethooks
Now, $P$ is a safety property and $\ensuremath{\varcolor{\tau}} ∈ P$, so for any
prefix \ensuremath{\varcolor{\tau}_{2}} of \ensuremath{\varcolor{\tau}} we have $\ensuremath{\varcolor{\tau}_{2}} ∈ P ∩ \Traces^{*}$.
Hence the goal follows by assumption that $P ∩ \Traces^{*} ⊆ \ensuremath{\gamma}(\ensuremath{\widehat{\varid{d}}})$.
\end{proof}

From now on, we tacitly assume that all trace properties of interest are safety
properties, and that any Galois connection defined in Haskell has been extended
to infinite traces via \Cref{thm:safety-extension}.
Any such Galois connection can be used to approximate guarded fixpoints via
least fixpoints:

\begin{lemma}[Guarded fixpoint abstraction for safety extensions]
\label{thm:guarded-fixpoint-abstraction}
Let \ensuremath{\widehat{\conid{D}}} be a domain with an instance for
\ensuremath{\conid{Lat}}, and let $\ensuremath{\alpha} : (\pow{\Traces},⊆) \rightleftarrows (\ensuremath{\widehat{\conid{D}}}, ⊑) : \ensuremath{\gamma}$ a Galois
connection extended to infinite traces via \Cref{thm:safety-extension}.
Then, for any iteratee \ensuremath{\varid{f}\mathbin{::}\Traces\to \Traces},
\[
  \ensuremath{\alpha}(\{ \ensuremath{\varid{fix}\;\varid{f}} \}) ⊑ \ensuremath{\varid{lfp}\;(\alpha\hsdot{\circ }{.\ }{\varid{f}}^{*}\hsdot{\circ }{.\ }\gamma)},
\]
where \ensuremath{\varid{lfp}\;\widehat{\varid{f}}} denotes the least fixpoint of \ensuremath{\widehat{\varid{f}}} and \ensuremath{{\varid{f}}^{*}\mathbin{::}\pow{\Traces}\to \pow{\Traces}} is the lifting of \ensuremath{\varid{f}} to powersets.
\end{lemma}
\begin{proof}
We should note that we are sloppy in the treatment of the \emph{later} modality
$\later$ here.
Since we have proven totality of all expressions worth considering in
\Cref{sec:totality}, the utility of being explicit in \ensuremath{\varid{next}} is rather low (much
more so since a pen and paper proof is not type checked) and we will admit
ourselves this kind of sloppiness from now on.

Let us assume that \ensuremath{\varcolor{\tau}\mathrel{=}\varid{fix}\;\varid{f}} is finite and proceed by Löb induction.
\begin{hscode}\SaveRestoreHook
\column{B}{@{}>{\hspre}c<{\hspost}@{}}%
\column{BE}{@{}l@{}}%
\column{5}{@{}>{\hspre}l<{\hspost}@{}}%
\column{E}{@{}>{\hspre}l<{\hspost}@{}}%
\>[5]{}\alpha\;\{\varid{fix}\;\varid{f}\}\mathbin{⊑}\varid{lfp}\;(\alpha\hsdot{\circ }{.\ }{\varid{f}}^{*}\hsdot{\circ }{.\ }\gamma){}\<[E]%
\\
\>[B]{}\mathrel{=}{}\<[BE]%
\>[5]{}\mbox{\commentbegin  \ensuremath{\varid{fix}\;\varid{f}\mathrel{=}\varid{f}\;(\varid{fix}\;\varid{f})}  \commentend}{}\<[E]%
\\
\>[5]{}\alpha\;\{\varid{f}\;(\varid{fix}\;\varid{f})\}{}\<[E]%
\\
\>[B]{}\mathrel{=}{}\<[BE]%
\>[5]{}\mbox{\commentbegin  Commute \ensuremath{\varid{f}} and \ensuremath{\{\wild\}}  \commentend}{}\<[E]%
\\
\>[5]{}\alpha\;({\varid{f}}^{*}\;\{\varid{fix}\;\varid{f}\}){}\<[E]%
\\
\>[B]{}\mathbin{⊑}{}\<[BE]%
\>[5]{}\mbox{\commentbegin  \ensuremath{\varid{id}\mathbin{⊑}\gamma\hsdot{\circ }{.\ }\alpha}  \commentend}{}\<[E]%
\\
\>[5]{}\alpha\;({\varid{f}}^{*}\;(\gamma\;(\alpha\;\{\varid{fix}\;\varid{f}\}))){}\<[E]%
\\
\>[B]{}\mathbin{⊑}{}\<[BE]%
\>[5]{}\mbox{\commentbegin  Induction hypothesis  \commentend}{}\<[E]%
\\
\>[5]{}\alpha\;({\varid{f}}^{*}\;(\gamma\;(\varid{lfp}\;(\alpha\hsdot{\circ }{.\ }{\varid{f}}^{*}\hsdot{\circ }{.\ }\gamma)))){}\<[E]%
\\
\>[B]{}\mathbin{⊑}{}\<[BE]%
\>[5]{}\mbox{\commentbegin  \ensuremath{\varid{lfp}\;\widehat{\varid{f}}\mathrel{=}\widehat{\varid{f}}\;(\varid{lfp}\;\widehat{\varid{f}})}  \commentend}{}\<[E]%
\\
\>[5]{}\varid{lfp}\;(\alpha\hsdot{\circ }{.\ }{\varid{f}}^{*}\hsdot{\circ }{.\ }\gamma){}\<[E]%
\ColumnHook
\end{hscode}\resethooks
When \ensuremath{\varcolor{\tau}} is infinite, the result follows by \Cref{thm:safety-extension}
and the fact that all properties of interest are safety properties.
\end{proof}

\subsection{Abstract By-name Soundness, in Detail}
\label{sec:by-name-soundness}

\begin{figure}
\[\ruleform{\begin{array}{c}
  α_{\mathcal{S}} : (\ensuremath{(\conid{Name}\mathbin{:\rightharpoonup}\conid{D}_{\mathbf{na}})\to \conid{D}_{\mathbf{na}}}) \rightleftarrows (\ensuremath{(\conid{Name}\mathbin{:\rightharpoonup}\widehat{\conid{D}})\to \widehat{\conid{D}}}) : γ_{\mathcal{S}}
  \\
  α_{\Environments} : \pow{\ensuremath{\conid{Name}\mathbin{:\rightharpoonup}\conid{D}_{\mathbf{na}}}} \rightleftarrows (\ensuremath{\conid{Name}\mathbin{:\rightharpoonup}\widehat{\conid{D}}}) : γ_{\Environments}
  \\
  α_{\Domain{}} : \pow{\ensuremath{\conid{D}_{\mathbf{na}}}} \rightleftarrows \ensuremath{\widehat{\conid{D}}} : γ_{\Domain{}}
  \qquad
  β_\Traces : \ensuremath{\conid{T}\;(\conid{Value}\;(\conid{ByName}\;\conid{T}))} \to \ensuremath{\widehat{\conid{D}}}
  \qquad
\end{array}}\]
\belowdisplayskip=0pt
\arraycolsep=2pt
\[\begin{array}{lcl}
α_{\mathcal{S}}(S)(\widehat{ρ}) & = & α_\Traces(\{\  S(ρ) \mid ρ ∈ γ_{\Environments}(\widehat{ρ}) \ \}) \\
α_{\Environments}(R)(x) & = & \Lub \{\  α_{\Domain{}}(\{ρ(x)\}) \mid ρ ∈ R \ \} \\
α_{\Domain{}}(D) & = & \Lub \{\  β_\Traces(d) \mid d ∈ D \ \}  \\
\\[-0.75em]
β_\Traces(\ensuremath{\varcolor{\tau}}) & = & \begin{cases}
  \ensuremath{\varid{step}\;\varid{e}\;(\beta_\Traces (\varcolor{\tau}))}               & \text{if \ensuremath{\varcolor{\tau}\mathrel{=}\conid{Step}\;\varid{e}\;\varcolor{\tau}'}} \\
  \ensuremath{\varid{stuck}}                         & \text{if \ensuremath{\varcolor{\tau}\mathrel{=}\conid{Ret}\;\conid{Stuck}}} \\
  \ensuremath{\varid{fun}\;(\alpha_\Domain\hsdot{\circ }{.\ }{\varid{f}}^{*}\hsdot{\circ }{.\ }\gamma_\Domain)}      & \text{if \ensuremath{\varcolor{\tau}\mathrel{=}\conid{Ret}\;(\conid{Fun}\;\varid{f})}} \\
  \ensuremath{\varid{con}\;\varid{k}\;(\varid{map}\;(\alpha_\Domain\hsdot{\circ }{.\ }\{\wild\})\;\varid{ds})}     & \text{if \ensuremath{\varcolor{\tau}\mathrel{=}\conid{Ret}\;(\conid{Con}\;\varid{k}\;\varid{ds})}} \\
  \end{cases} \\
\\[-0.75em]
\end{array}\]
\caption{Galois connection $α_{\mathcal{S}}$ for by-name abstraction derived from \ensuremath{\conid{Trace}}, \ensuremath{\conid{Domain}} and \ensuremath{\conid{Lat}} instances on \ensuremath{\widehat{\conid{D}}}}
\label{fig:abstract-name}
\end{figure}

We will now prove that the by-name abstraction laws in \Cref{fig:abstraction-laws}
induce an abstract interpretation of by-name evaluation via \ensuremath{\alpha_\mathcal{S}} defined in
\Cref{fig:abstract-name}.
The Galois connection and the corresponding proofs are very similar, yet
somewhat simpler than for by-need because no heap update is involved.

We write $\ensuremath{{\varid{f}}^{*}} : \pow{A} \to \pow{B}$ to lift a function $\ensuremath{\varid{f}} : A \to B$
to powersets, and write $\ensuremath{\{\wild\}} : A \to \pow{A}$ to construct a singleton set in
pointfree style.
Note that we will omit \ensuremath{\conid{ByName}} newtype wrappers, as in many other preceding
sections, as well as the \ensuremath{\conid{Name}} passed to \ensuremath{\varid{fun}} as a poor man's De Bruijn level.

Compared to the by-need trace abstraction in \Cref{fig:abstract-name-need}, the
by-name trace abstraction function in \Cref{fig:abstract-name} is rather
straightforward because no heap is involved.

Note that the recursion in \ensuremath{\beta_\Traces} is defined in terms of the least fixpoint;
we discussed in \Cref{sec:safety-extension} why this is a natural choice.

We will now prove sound by-name interpretation by appealing to parametricity.

Following the semi-formal style of \citet[Section 3]{Wadler:89},
we apply the abstraction theorem to the System $F$ encoding
of the type of \ensuremath{\mathcal{S}\denot{\wild}_{\wild}}
\[
  \ensuremath{\mathcal{S}\denot{\wild}_{\wild}} : \forall X.\ \mathsf{Dict}(X) → \mathsf{Exp} → (\mathsf{Name} \pfun X) → X
\]
where $\mathsf{Dict}(\ensuremath{\varid{d}})$ is the type of the type class
dictionary for \ensuremath{(\conid{Trace}\;\varid{d},\conid{Domain}\;\varid{d},\conid{HasBind}\;\varid{d})}.
The abstraction theorem yields the following assertion about relations:
\[
  (\ensuremath{\mathcal{S}\denot{\wild}_{\wild}}, \ensuremath{\mathcal{S}\denot{\wild}_{\wild}}) ∈ \forall \mathcal{X}.\ \mathsf{Dict}(\mathcal{X}) → \mathsf{Exp} → (\mathsf{Name} \pfun \mathcal{X}) → \mathcal{X}
\]
Wadler overloads the type formers with a suitable relational interpretation, which translates to
\begin{align}
  &\forall A, B.\
  \forall R ⊆ A \times B.\
  \forall (\mathit{inst_1}, \mathit{inst_2}) ∈ \mathsf{Dict}(R).\
  \forall \ensuremath{\varid{e}} ∈ \mathsf{Exp}.\
  \forall (ρ_1, ρ_2) ∈ (\mathsf{Name} \pfun R). \notag \\
  &(\ensuremath{\mathcal{S}_{A }\denot{\varid{e}}}(\mathit{inst_1})(ρ_1), \ensuremath{\mathcal{S}_{B }\denot{\varid{e}}}(\mathit{inst_2})(ρ_2)) ∈ R \label{eqn:name-abs}
\end{align}
and in the following proof, we will instantiate $R(\ensuremath{\varid{d}},\ensuremath{\widehat{\varid{d}}}) \triangleq \ensuremath{\alpha_\Domain (\{\varid{d}\})\mathbin{⊑}\widehat{\varid{d}}}$ to show the abstraction relationship.

We will need the following auxiliary lemma for the \ensuremath{\varid{apply}} and \ensuremath{\varid{select}} cases:
\begin{lemma}[By-name bind]
\label{thm:by-name-bind}
It is \ensuremath{\beta_\Traces (\varid{d}\bind \varid{f})\mathbin{⊑}\widehat{\varid{f}}\;\widehat{\varid{d}}} if
\begin{enumerate}
  \item \ensuremath{\beta_\Traces (\varid{d})\mathbin{⊑}\widehat{\varid{d}}}, and
  \item for all events \ensuremath{\varid{ev}} and elements \ensuremath{\widehat{\varid{d'}}}, \ensuremath{\widehat{\varid{step}}\;\varid{ev}\;(\widehat{\varid{f}}\;\widehat{\varid{d'}})\mathbin{⊑}\widehat{\varid{f}}\;(\widehat{\varid{step}}\;\varid{ev}\;\widehat{\varid{d'}})}, and
  \item for all values \ensuremath{\varid{v}}, \ensuremath{\beta_\Traces (\varid{f}\;\varid{v})\mathbin{⊑}\widehat{\varid{f}}\;(\beta_\Traces (\conid{Ret}\;\varid{v}))}.
\end{enumerate}
\end{lemma}
\begin{proof}
By Löb induction.

If \ensuremath{\varid{d}\mathrel{=}\conid{Step}\;\varid{ev}\;\varid{d'}}, define \ensuremath{\widehat{\varid{d'}}\triangleq\beta_\Traces (\varid{d'}}.
We) get
\begin{hscode}\SaveRestoreHook
\column{B}{@{}>{\hspre}l<{\hspost}@{}}%
\column{3}{@{}>{\hspre}l<{\hspost}@{}}%
\column{4}{@{}>{\hspre}l<{\hspost}@{}}%
\column{5}{@{}>{\hspre}l<{\hspost}@{}}%
\column{E}{@{}>{\hspre}l<{\hspost}@{}}%
\>[3]{}\beta_\Traces (\varid{d}\bind \varid{f})\mathrel{=}\beta_\Traces (\conid{Step}\;\varid{ev}\;\varid{d'}\bind \varid{f})\mathrel{=}\widehat{\varid{step}}\;\varid{ev}\;(\beta_\Traces (\varid{d'}\bind \varid{f})){}\<[E]%
\\
\>[B]{}\mathbin{⊑}{}\<[4]%
\>[4]{}\mbox{\commentbegin  Induction hypothesis at \ensuremath{\beta_\Traces (\varid{d'})\mathrel{=}\widehat{\varid{d'}}}, Monotonicity of \ensuremath{\widehat{\varid{step}}}  \commentend}{}\<[E]%
\\
\>[B]{}\hsindent{3}{}\<[3]%
\>[3]{}\widehat{\varid{step}}\;\varid{ev}\;(\widehat{\varid{f}}\;(\beta_\Traces (\varid{d'})){}\<[E]%
\\
\>[B]{}\mathbin{⊑}){}\<[5]%
\>[5]{}\mbox{\commentbegin  Assumption (2)  \commentend}{}\<[E]%
\\
\>[B]{}\hsindent{3}{}\<[3]%
\>[3]{}\widehat{\varid{f}}\;(\widehat{\varid{step}}\;\varid{ev}\;(\beta_\Traces (\varid{d'})))\mathrel{=}\widehat{\varid{f}}\;(\beta_\Traces (\varid{d}){}\<[E]%
\\
\>[B]{}\mathbin{⊑}){}\<[5]%
\>[5]{}\mbox{\commentbegin  Assumption (1)  \commentend}{}\<[E]%
\\
\>[B]{}\hsindent{3}{}\<[3]%
\>[3]{}\widehat{\varid{f}}\;\widehat{\varid{d}}{}\<[E]%
\ColumnHook
\end{hscode}\resethooks

Otherwise, \ensuremath{\varid{d}\mathrel{=}\conid{Ret}\;\varid{v}} for some \ensuremath{\varid{v}\mathbin{::}\conid{Value}}.
\begin{hscode}\SaveRestoreHook
\column{B}{@{}>{\hspre}l<{\hspost}@{}}%
\column{3}{@{}>{\hspre}l<{\hspost}@{}}%
\column{4}{@{}>{\hspre}l<{\hspost}@{}}%
\column{5}{@{}>{\hspre}l<{\hspost}@{}}%
\column{E}{@{}>{\hspre}l<{\hspost}@{}}%
\>[3]{}\beta_\Traces (\conid{Ret}\;\varid{v}\bind \varid{f})\mathrel{=}\beta_\Traces (\varid{f}\;\varid{v}){}\<[E]%
\\
\>[B]{}\mathbin{⊑}{}\<[4]%
\>[4]{}\mbox{\commentbegin  Assumption (3)  \commentend}{}\<[E]%
\\
\>[B]{}\hsindent{3}{}\<[3]%
\>[3]{}\widehat{\varid{f}}\;(\beta_\Traces (\conid{Ret}\;\varid{v}))\mathrel{=}\widehat{\varid{f}}\;(\beta_\Traces (\varid{d}){}\<[E]%
\\
\>[B]{}\mathbin{⊑}){}\<[5]%
\>[5]{}\mbox{\commentbegin  Assumption (1)  \commentend}{}\<[E]%
\\
\>[B]{}\hsindent{3}{}\<[3]%
\>[3]{}\widehat{\varid{f}}\;\widehat{\varid{d}}{}\<[E]%
\ColumnHook
\end{hscode}\resethooks
\end{proof}

What follows is the sound abstraction proof by parametricity.
Note that its statement fixes the interpreter to \ensuremath{\mathcal{S}\denot{\wild}_{\wild}}, however the proof would
still work if generalised to \emph{any} definition with the same type as \ensuremath{\mathcal{S}\denot{\wild}_{\wild}}!

\begin{theorem}[Abstract By-name Interpretation]
\label{thm:abstract-by-name}
Let \ensuremath{\varid{e}} be an expression, \ensuremath{\widehat{\conid{D}}} a domain with instances for \ensuremath{\conid{Trace}}, \ensuremath{\conid{Domain}}, \ensuremath{\conid{HasBind}} and
\ensuremath{\conid{Lat}}, and let $α_{\mathcal{S}}$ be the abstraction function from \Cref{fig:abstract-name}.
If the by-name abstraction laws in \Cref{fig:abstraction-laws} hold,
then \ensuremath{\mathcal{S}_{\widehat{\conid{D}}}\denot{\wild}} is an abstract interpreter that is sound \wrt $α_{\mathcal{S}}$,
\[
  α_{\mathcal{S}}(\ensuremath{\mathcal{S}_{\mathbf{name}}\denot{\varid{e}}}) ⊑ \ensuremath{\mathcal{S}_{\widehat{\conid{D}}}\denot{\varid{e}}}.
\]
\end{theorem}
\begin{proof}
Let $\ensuremath{\varid{inst}} : \mathsf{Dict}(\ensuremath{\conid{D}_{\mathbf{na}}})$, $\ensuremath{\widehat{\varid{inst}}} : \mathsf{Dict}(\ensuremath{\widehat{\conid{D}}})$ the
canonical dictionaries from the type class instance definitions.
Instantiate the free theorem \labelcref{eqn:name-abs} above as follows:
\[\begin{array}{c}
A \triangleq \ensuremath{\conid{D}_{\mathbf{na}}}, B \triangleq \ensuremath{\widehat{\conid{D}}}, R(\ensuremath{\varid{d}}, \ensuremath{\widehat{\varid{d}}}) \triangleq \ensuremath{\alpha_\Domain (\{\varid{d}\})\mathbin{⊑}\widehat{\varid{d}}},
\mathit{inst_1} \triangleq \ensuremath{\varid{inst}}, \mathit{inst_2} \triangleq \ensuremath{\widehat{\varid{inst}}}, \ensuremath{\varid{e}\triangleq\varid{e}}
\end{array}\]
Note that $(\ensuremath{\varcolor{\rho}},\ensuremath{\widehat{\varcolor{\rho}}}) ∈ (\mathsf{Name} \pfun R) \Longleftrightarrow \ensuremath{\alpha_\Environments (\{\varcolor{\rho}\})\mathbin{⊑}\widehat{\varcolor{\rho}}\Longleftrightarrow\varcolor{\rho}\in \gamma_\Environments (\widehat{\varcolor{\rho}})}$ by simple calculation.

The above instantiation yields, in Haskell,
\[
  \inferrule
    {(\ensuremath{\varid{inst}}, \ensuremath{\widehat{\varid{inst}}}) ∈ \mathsf{Dict}(R) \\ \ensuremath{\varcolor{\rho}\in \gamma_\Environments (\widehat{\varcolor{\rho}})}}
    {\ensuremath{\alpha_\Domain (\{\mathcal{S}_{\mathbf{name}}\denot{\varid{e}}_{\varcolor{\rho}}\})\mathbin{⊑}\mathcal{S}_{\widehat{\conid{D}}}\denot{\varid{e}}_{\widehat{\varcolor{\rho}}}}}
\]
and since \ensuremath{\varcolor{\rho}} and \ensuremath{\widehat{\varcolor{\rho}}} can be chosen arbitrarily, this can be reformulated as
\[
  \inferrule
    {(\ensuremath{\varid{inst}}, \ensuremath{\widehat{\varid{inst}}}) ∈ \mathsf{Dict}(R)}
    {α_{\mathcal{S}}(\ensuremath{\mathcal{S}_{\mathbf{name}}\denot{\varid{e}}}) ⊑ \ensuremath{\mathcal{S}_{\widehat{\conid{D}}}\denot{\varid{e}}}}
\]
Hence, in order to show the goal, it suffices to prove $(\ensuremath{\varid{inst}}, \ensuremath{\widehat{\varid{inst}}}) ∈ \mathsf{Dict}(R)$.
By the relational interpretation of products, we get one subgoal per instance method.
Note that $R(\ensuremath{\varid{d}}, \ensuremath{\widehat{\varid{d}}}) \Longleftrightarrow \ensuremath{\beta_\Traces (\varid{d})\mathbin{⊑}\widehat{\varid{d}}}$ and it is more
direct to argue in terms of the latter.
\begin{itemize}
  \item \textbf{Case \ensuremath{\varid{step}}}.
    Goal: $\inferrule{(\ensuremath{\varid{d}},\ensuremath{\widehat{\varid{d}}}) ∈ R}{(\ensuremath{\varid{step}\;\varid{ev}\;\varid{d}}, \ensuremath{\widehat{\varid{step}}\;\varid{ev}\;\widehat{\varid{d}}}) ∈ R}$. \\
    Then \ensuremath{\beta_\Traces (\conid{Step}\;\varid{ev}\;\varid{d})\mathrel{=}\widehat{\varid{step}}\;\varid{ev}\;(\beta_\Traces (\varid{d}))\mathbin{⊑}\widehat{\varid{step}}\;\varid{ev}\;\widehat{\varid{d}}} by assumption and monotonicity.

  \item \textbf{Case \ensuremath{\varid{stuck}}}.
    Goal: $(\ensuremath{\varid{stuck}}, \ensuremath{\widehat{\varid{stuck}}}) ∈ R$. \\
    Then \ensuremath{\beta_\Traces (\varid{stuck})\mathrel{=}\beta_\Traces (\conid{Ret}\;\conid{Stuck})\mathrel{=}\widehat{\varid{stuck}}}.

  \item \textbf{Case \ensuremath{\varid{fun}}}.
    Goal: $\inferrule{\forall (\ensuremath{\varid{d}},\ensuremath{\widehat{\varid{d}}}) ∈ R.\ (\ensuremath{\varid{f}\;\varid{d}}, \ensuremath{\widehat{\varid{f}}\;\widehat{\varid{d}}}) ∈ R}{(\ensuremath{\varid{fun}\;\varid{f}}, \ensuremath{\widehat{\varid{fun}}\;\widehat{\varid{f}}}) ∈ R}$. \\
    Then \ensuremath{\beta_\Traces (\varid{fun}\;\varid{f})\mathrel{=}\beta_\Traces (\conid{Ret}\;(\conid{Fun}\;\varid{f}))\mathrel{=}\widehat{\varid{fun}}\;(\alpha_\Domain\hsdot{\circ }{.\ }{\varid{f}}^{*}\hsdot{\circ }{.\ }\gamma_\Domain)} and
    it suffices to show that \ensuremath{\alpha_\Domain\hsdot{\circ }{.\ }{\varid{f}}^{*}\hsdot{\circ }{.\ }\gamma_\Domain\mathbin{⊑}\widehat{\varid{f}}} by monotonicity of \ensuremath{\widehat{\varid{fun}}}.
    \begin{hscode}\SaveRestoreHook
\column{B}{@{}>{\hspre}l<{\hspost}@{}}%
\column{5}{@{}>{\hspre}c<{\hspost}@{}}%
\column{5E}{@{}l@{}}%
\column{7}{@{}>{\hspre}l<{\hspost}@{}}%
\column{8}{@{}>{\hspre}l<{\hspost}@{}}%
\column{E}{@{}>{\hspre}l<{\hspost}@{}}%
\>[7]{}(\alpha_\Domain\hsdot{\circ }{.\ }{\varid{f}}^{*}\hsdot{\circ }{.\ }\gamma_\Domain)\;\widehat{\varid{d}}{}\<[E]%
\\
\>[5]{}\mathrel{=}{}\<[5E]%
\>[8]{}\mbox{\commentbegin  Unfold \ensuremath{{\wild}^{*}}, \ensuremath{\alpha_\Domain}, simplify  \commentend}{}\<[E]%
\\
\>[5]{}\hsindent{2}{}\<[7]%
\>[7]{}\Lub\{\beta_\Traces (\varid{f}\;\varid{d})\mid \varid{d}\in \gamma_\Domain\;\widehat{\varid{d}}\}{}\<[E]%
\\
\>[5]{}\mathbin{⊑}{}\<[5E]%
\>[8]{}\mbox{\commentbegin  Apply premise to \ensuremath{\beta_\Traces (\varid{d})\mathbin{⊑}\widehat{\varid{d}}}  \commentend}{}\<[E]%
\\
\>[5]{}\hsindent{2}{}\<[7]%
\>[7]{}\widehat{\varid{f}}\;\widehat{\varid{d}}{}\<[E]%
\ColumnHook
\end{hscode}\resethooks

  \item \textbf{Case \ensuremath{\varid{apply}}}.
    Goal: $\inferrule{(\ensuremath{\varid{d}},\ensuremath{\widehat{\varid{d}}}) ∈ R \\ (\ensuremath{\varid{a}},\ensuremath{\widehat{\varid{a}}}) ∈ R}{(\ensuremath{\varid{apply}\;\varid{d}\;\varid{a}}, \ensuremath{\widehat{\varid{apply}}\;\widehat{\varid{d}}\;\widehat{\varid{a}}}) ∈ R}$. \\
    \ensuremath{\varid{apply}\;\varid{d}\;\varid{a}} is defined as \ensuremath{\varid{d}\bind \lambda \varid{v}\to \keyword{case}\;\varid{v}\;\keyword{of}\;\conid{Fun}\;\varid{g}\to \varid{g}\;\varid{a};\anonymous \to \varid{stuck}}.
    In order to show the goal, we need to apply \Cref{thm:by-name-bind} at \ensuremath{\widehat{\varid{f}}\;\widehat{\varid{d}}\triangleq\widehat{\varid{apply}}\;\widehat{\varid{d}}\;\widehat{\varid{a}}}.
    We get three subgoals for the premises of \Cref{thm:by-name-bind}:
    \begin{itemize}
      \item \ensuremath{\beta_\Traces (\varid{d})\mathbin{⊑}\widehat{\varid{d}}}: By assumption $(\ensuremath{\varid{d}},\ensuremath{\widehat{\varid{d}}}) ∈ R$.
      \item \ensuremath{\keyword{\forall}\!\! \hsforall \;\varid{ev}\;\widehat{\varid{d'}}\hsdot{\circ }{.\ }\widehat{\varid{step}}\;\varid{ev}\;(\widehat{\varid{apply}}\;\widehat{\varid{d'}}\;\widehat{\varid{a}})\mathbin{⊑}\widehat{\varid{apply}}\;(\widehat{\varid{step}}\;\varid{ev}\;\widehat{\varid{d'}})\;\widehat{\varid{a}}}: By assumption \textsc{Step-App}.
      \item \ensuremath{\keyword{\forall}\!\! \hsforall \;\varid{v}\hsdot{\circ }{.\ }\beta_\Traces (\keyword{case}\;\varid{v}\;\keyword{of}\;\conid{Fun}\;\varid{g}\to \varid{g}\;\varid{a};\anonymous \to \varid{stuck})\mathbin{⊑}\widehat{\varid{apply}}\;(\beta_\Traces (\conid{Ret}\;\varid{v}))\;\widehat{\varid{a}}}: \\
        By cases over \ensuremath{\varid{v}}.
        \begin{itemize}
          \item \textbf{Case \ensuremath{\varid{v}\mathrel{=}\conid{Stuck}}}:
            Then \ensuremath{\beta_\Traces (\varid{stuck})\mathrel{=}\widehat{\varid{stuck}}\mathbin{⊑}\widehat{\varid{apply}}\;\widehat{\varid{stuck}}\;\widehat{\varid{a}}} by assumption \textsc{Stuck-App}.
          \item \textbf{Case \ensuremath{\varid{v}\mathrel{=}\conid{Con}\;\varid{k}\;\varid{ds}}}:
            Then \ensuremath{\beta_\Traces (\varid{stuck})\mathrel{=}\widehat{\varid{stuck}}\mathbin{⊑}\widehat{\varid{apply}}\;(\widehat{\varid{con}}\;\varid{k}\;\widehat{\varid{ds}})\;\widehat{\varid{a}}} by assumption \textsc{Stuck-App}, for the suitable \ensuremath{\widehat{\varid{ds}}}.
          \item \textbf{Case \ensuremath{\varid{v}\mathrel{=}\conid{Fun}\;\varid{g}}}:
            Then
            \begin{hscode}\SaveRestoreHook
\column{B}{@{}>{\hspre}l<{\hspost}@{}}%
\column{15}{@{}>{\hspre}c<{\hspost}@{}}%
\column{15E}{@{}l@{}}%
\column{17}{@{}>{\hspre}l<{\hspost}@{}}%
\column{18}{@{}>{\hspre}l<{\hspost}@{}}%
\column{E}{@{}>{\hspre}l<{\hspost}@{}}%
\>[17]{}\beta_\Traces (\varid{g}\;\varid{a}){}\<[E]%
\\
\>[15]{}\mathbin{⊑}{}\<[15E]%
\>[18]{}\mbox{\commentbegin  \ensuremath{\varid{id}\mathbin{⊑}\gamma_\Domain\hsdot{\circ }{.\ }\alpha_\Domain}, rearrange  \commentend}{}\<[E]%
\\
\>[15]{}\hsindent{2}{}\<[17]%
\>[17]{}(\alpha_\Domain\hsdot{\circ }{.\ }{\varid{g}}^{*}\hsdot{\circ }{.\ }\gamma_\Domain)\;(\alpha_\Domain\;\varid{a}){}\<[E]%
\\
\>[15]{}\mathbin{⊑}{}\<[15E]%
\>[18]{}\mbox{\commentbegin  Assumption \ensuremath{\beta_\Traces (\varid{a})\mathbin{⊑}\widehat{\varid{a}}}  \commentend}{}\<[E]%
\\
\>[15]{}\hsindent{2}{}\<[17]%
\>[17]{}(\alpha_\Domain\hsdot{\circ }{.\ }{\varid{g}}^{*}\hsdot{\circ }{.\ }\gamma_\Domain)\;\widehat{\varid{a}}{}\<[E]%
\\
\>[15]{}\mathbin{⊑}{}\<[15E]%
\>[18]{}\mbox{\commentbegin  Assumption \textsc{Beta-App}  \commentend}{}\<[E]%
\\
\>[15]{}\hsindent{2}{}\<[17]%
\>[17]{}\widehat{\varid{apply}}\;(\widehat{\varid{fun}}\;(\alpha_\Domain\hsdot{\circ }{.\ }{\varid{g}}^{*}\hsdot{\circ }{.\ }\gamma_\Domain))\;\widehat{\varid{a}}{}\<[E]%
\\
\>[15]{}\mathrel{=}{}\<[15E]%
\>[18]{}\mbox{\commentbegin  Definition of \ensuremath{\beta_\Traces}, \ensuremath{\varid{v}}  \commentend}{}\<[E]%
\\
\>[15]{}\hsindent{2}{}\<[17]%
\>[17]{}\widehat{\varid{apply}}\;(\beta_\Traces (\conid{Ret}\;\varid{v}))\;\widehat{\varid{a}}{}\<[E]%
\ColumnHook
\end{hscode}\resethooks
        \end{itemize}
    \end{itemize}

  \item \textbf{Case \ensuremath{\varid{con}}}.
    Goal: $\inferrule{(\ensuremath{\varid{ds}},\ensuremath{\widehat{\varid{ds}}}) ∈ \ensuremath{[\mskip1.5mu R \mskip1.5mu]}}{(\ensuremath{\varid{con}\;\varid{k}\;\varid{ds}}, \ensuremath{\widehat{\varid{con}}\;\varid{k}\;\widehat{\varid{ds}}}) ∈ R}$. \\
    Then \ensuremath{\beta_\Traces (\varid{con}\;\varid{k}\;\varid{ds})\mathrel{=}\beta_\Traces (\conid{Ret}\;(\conid{Con}\;\varid{k}\;\varid{ds}))\mathrel{=}\widehat{\varid{con}}\;\varid{k}\;(\varid{map}\;(\alpha_\Domain\hsdot{\circ }{.\ }\{\wild\})\;\varid{ds})}.
    The assumption $(\ensuremath{\varid{ds}},\ensuremath{\widehat{\varid{ds}}}) ∈ \ensuremath{[\mskip1.5mu R \mskip1.5mu]}$ implies \ensuremath{\varid{map}\;(\alpha_\Domain\hsdot{\circ }{.\ }\{\wild\})\;\varid{ds}\mathbin{⊑}\widehat{\varid{ds}}} and
    the goal follows by monotonicity of \ensuremath{\widehat{\varid{con}}}.

  \item \textbf{Case \ensuremath{\varid{select}}}.
    Goal: $\inferrule{(\ensuremath{\varid{d}},\ensuremath{\widehat{\varid{d}}}) ∈ R \\ (\ensuremath{\varid{alts}},\ensuremath{\widehat{\varid{alts}}}) ∈ \ensuremath{\conid{Tag}\mathbin{:\rightharpoonup}([\mskip1.5mu R \mskip1.5mu]\to R )}}{(\ensuremath{\varid{select}\;\varid{d}\;\varid{alts}}, \ensuremath{\widehat{\varid{select}}\;\widehat{\varid{d}}\;\widehat{\varid{alts}}}) ∈ R}$. \\
    \ensuremath{\varid{select}\;\varid{d}\;\varid{fs}} is defined as \ensuremath{\varid{d}\bind \lambda \varid{v}\to \keyword{case}\;\varid{v}\;\keyword{of}\;\conid{Con}\;\varid{k}\;\varid{ds}\mid \varid{k}\in \varid{dom}\;\varid{alts}\to (\varid{alts}\mathop{!}\varid{k})\;\varid{ds};\anonymous \to \varid{stuck}}.
    In order to show the goal, we need to apply \Cref{thm:by-name-bind} at \ensuremath{\widehat{\varid{f}}\;\widehat{\varid{d}}\triangleq\widehat{\varid{select}}\;\widehat{\varid{d}}\;\widehat{\varid{alts}}}.
    We get three subgoals for the premises of \Cref{thm:by-name-bind}:
    \begin{itemize}
      \item \ensuremath{\beta_\Traces (\varid{d})\mathbin{⊑}\widehat{\varid{d}}}: By assumption $(\ensuremath{\varid{d}},\ensuremath{\widehat{\varid{d}}}) ∈ R$.
      \item \ensuremath{\keyword{\forall}\!\! \hsforall \;\varid{ev}\;\widehat{\varid{d'}}\hsdot{\circ }{.\ }\widehat{\varid{step}}\;\varid{ev}\;(\widehat{\varid{select}}\;\widehat{\varid{d'}}\;\widehat{\varid{alts}})\mathbin{⊑}\widehat{\varid{select}}\;(\widehat{\varid{step}}\;\varid{ev}\;\widehat{\varid{d'}})\;\widehat{\varid{alts}}}: By assumption \textsc{Step-Sel}.
      \item \ensuremath{\keyword{\forall}\!\! \hsforall \;\varid{v}\hsdot{\circ }{.\ }\beta_\Traces (\keyword{case}\;\varid{v}\;\keyword{of}\;\conid{Con}\;\varid{k}\;\varid{ds}\mid \varid{k}\in \varid{dom}\;\varid{alts}\to (\varid{alts}\mathop{!}\varid{k})\;\varid{ds};\anonymous \to \varid{stuck})\mathbin{⊑}\widehat{\varid{select}}\;(\beta_\Traces (\conid{Ret}\;\varid{v}))\;\widehat{\varid{alts}}}: \\
        By cases over \ensuremath{\varid{v}}. The first three all correspond to when the continuation of \ensuremath{\varid{select}} gets stuck.
        \begin{itemize}
          \item \textbf{Case \ensuremath{\varid{v}\mathrel{=}\conid{Stuck}}}:
            Then \ensuremath{\beta_\Traces (\varid{stuck})\mathrel{=}\widehat{\varid{stuck}}\mathbin{⊑}\widehat{\varid{select}}\;\widehat{\varid{stuck}}\;\widehat{\varid{alts}}} by assumption \textsc{Stuck-Sel}.
          \item \textbf{Case \ensuremath{\varid{v}\mathrel{=}\conid{Fun}\;\varid{f}}}:
            Then \ensuremath{\beta_\Traces (\varid{stuck})\mathrel{=}\widehat{\varid{stuck}}\mathbin{⊑}\widehat{\varid{select}}\;(\widehat{\varid{fun}}\;\widehat{\varid{f}})\;\widehat{\varid{alts}}} by assumption \textsc{Stuck-Sel}, for the suitable \ensuremath{\widehat{\varid{f}}}.
          \item \textbf{Case \ensuremath{\varid{v}\mathrel{=}\conid{Con}\;\varid{k}\;\varid{ds}}, $\ensuremath{\varid{k}} \not∈ \ensuremath{\varid{dom}\;\varid{alts}}$}:
            Then \ensuremath{\beta_\Traces (\varid{stuck})\mathrel{=}\widehat{\varid{stuck}}\mathbin{⊑}\widehat{\varid{select}}\;(\widehat{\varid{con}}\;\varid{k}\;\widehat{\varid{ds}})\;\widehat{\varid{alts}}} by assumption \textsc{Stuck-Sel}, for the suitable \ensuremath{\widehat{\varid{ds}}}.
          \item \textbf{Case \ensuremath{\varid{v}\mathrel{=}\conid{Con}\;\varid{k}\;\varid{ds}}, $\ensuremath{\varid{k}} ∈ \ensuremath{\varid{dom}\;\varid{alts}}$}:
            Then
            \begin{hscode}\SaveRestoreHook
\column{B}{@{}>{\hspre}l<{\hspost}@{}}%
\column{15}{@{}>{\hspre}c<{\hspost}@{}}%
\column{15E}{@{}l@{}}%
\column{17}{@{}>{\hspre}l<{\hspost}@{}}%
\column{18}{@{}>{\hspre}l<{\hspost}@{}}%
\column{E}{@{}>{\hspre}l<{\hspost}@{}}%
\>[17]{}\beta_\Traces ((\varid{alts}\mathop{!}\varid{k})\;\varid{ds}){}\<[E]%
\\
\>[15]{}\mathbin{⊑}{}\<[15E]%
\>[18]{}\mbox{\commentbegin  \ensuremath{\varid{id}\mathbin{⊑}\gamma_\Domain\hsdot{\circ }{.\ }\alpha_\Domain}, rearrange  \commentend}{}\<[E]%
\\
\>[15]{}\hsindent{2}{}\<[17]%
\>[17]{}(\alpha_\Domain\hsdot{\circ }{.\ }{(\varid{alts}\mathop{!}\varid{k})}^{*}\hsdot{\circ }{.\ }\varid{map}\;\gamma_\Domain)\;(\varid{map}\;(\alpha_\Domain\hsdot{\circ }{.\ }\{\wild\})\;\varid{ds}){}\<[E]%
\\
\>[15]{}\mathbin{⊑}{}\<[15E]%
\>[18]{}\mbox{\commentbegin  Assumption $(\ensuremath{\varid{alts}},\ensuremath{\widehat{\varid{alts}}}) ∈ \ensuremath{\conid{Tag}\mathbin{:\rightharpoonup}([\mskip1.5mu R \mskip1.5mu]\to R )}$  \commentend}{}\<[E]%
\\
\>[15]{}\hsindent{2}{}\<[17]%
\>[17]{}(\widehat{\varid{alts}}\mathop{!}\varid{k})\;(\varid{map}\;(\alpha_\Domain\hsdot{\circ }{.\ }\{\wild\})\;\varid{ds}){}\<[E]%
\\
\>[15]{}\mathbin{⊑}{}\<[15E]%
\>[18]{}\mbox{\commentbegin  Assumption \textsc{Beta-Sel}  \commentend}{}\<[E]%
\\
\>[15]{}\hsindent{2}{}\<[17]%
\>[17]{}\widehat{\varid{select}}\;(\widehat{\varid{con}}\;\varid{k}\;(\varid{map}\;(\alpha_\Domain\hsdot{\circ }{.\ }\{\wild\})\;\varid{ds}))\;\widehat{\varid{alts}}{}\<[E]%
\\
\>[15]{}\mathrel{=}{}\<[15E]%
\>[18]{}\mbox{\commentbegin  Definition of \ensuremath{\beta_\Traces}, \ensuremath{\varid{v}}  \commentend}{}\<[E]%
\\
\>[15]{}\hsindent{2}{}\<[17]%
\>[17]{}\widehat{\varid{select}}\;(\beta_\Traces (\conid{Ret}\;\varid{v}))\;\widehat{\varid{alts}}{}\<[E]%
\ColumnHook
\end{hscode}\resethooks
        \end{itemize}
    \end{itemize}

  \item \textbf{Case \ensuremath{\varid{bind}}}.
    Goal: $\inferrule{(\forall (\ensuremath{\varid{d}},\ensuremath{\widehat{\varid{d}}}) ∈ R.\ (\ensuremath{\varid{rhs}\;\varid{d}}, \ensuremath{\widehat{\varid{rhs}}\;\widehat{\varid{d}}}) ∈ R) \\ (\forall (\ensuremath{\varid{d}},\ensuremath{\widehat{\varid{d}}}) ∈ R.\ (\ensuremath{\varid{body}\;\varid{d}}, \ensuremath{\widehat{\varid{body}}\;\widehat{\varid{d}}}) ∈ R)}
                     {(\ensuremath{\varid{bind}\;\varid{rhs}\;\varid{body}}, \ensuremath{\widehat{\varid{bind}}\;\widehat{\varid{rhs}}\;\widehat{\varid{body}}}) ∈ R}$. \\
    It is \ensuremath{\varid{bind}\;\varid{rhs}\;\varid{body}\mathrel{=}\varid{body}\;(\varid{fix}\;\varid{rhs})} and \ensuremath{\widehat{\varid{body}}\;(\varid{lfp}\;\widehat{\varid{rhs}})\mathbin{⊑}\widehat{\varid{bind}}\;\widehat{\varid{rhs}}\;\widehat{\varid{body}}} by Assumption \textsc{ByName-Bind}.
    Let us first establish that $(\ensuremath{\varid{fix}\;\varid{rhs}}, \ensuremath{\varid{lfp}\;\widehat{\varid{rhs}}}) ∈ R$, leaning on
    our theory about safety extension in \Cref{sec:safety-extension}:
    \begin{hscode}\SaveRestoreHook
\column{B}{@{}>{\hspre}l<{\hspost}@{}}%
\column{5}{@{}>{\hspre}c<{\hspost}@{}}%
\column{5E}{@{}l@{}}%
\column{7}{@{}>{\hspre}l<{\hspost}@{}}%
\column{8}{@{}>{\hspre}l<{\hspost}@{}}%
\column{E}{@{}>{\hspre}l<{\hspost}@{}}%
\>[7]{}\alpha_\Domain (\{\varid{fix}\;\varid{rhs}\}){}\<[E]%
\\
\>[5]{}\mathbin{⊑}{}\<[5E]%
\>[8]{}\mbox{\commentbegin  By \Cref{thm:safety-extension}  \commentend}{}\<[E]%
\\
\>[5]{}\hsindent{2}{}\<[7]%
\>[7]{}\varid{lfp}\;(\alpha_\Domain\hsdot{\circ }{.\ }{\varid{rhs}}^{*}\hsdot{\circ }{.\ }\gamma_\Domain){}\<[E]%
\\
\>[5]{}\mathrel{=}{}\<[5E]%
\>[8]{}\mbox{\commentbegin  Unfolding \ensuremath{{\wild}^{*}}, \ensuremath{\alpha_\Domain}  \commentend}{}\<[E]%
\\
\>[5]{}\hsindent{2}{}\<[7]%
\>[7]{}\varid{lfp}\;(\lambda \widehat{\varid{d}}\to \Lub\{\beta_\Traces (\varid{rhs}\;\varid{d})\mid \varid{d}\in \gamma_\Domain\;\widehat{\varid{d}}\}{}\<[E]%
\\
\>[5]{}\mathbin{⊑}{}\<[5E]%
\>[8]{}\mbox{\commentbegin  Apply assumption to $\ensuremath{\alpha_\Domain (\{\varid{d}\})\mathbin{⊑}\alpha_\Domain\;(\gamma_\Domain\;\widehat{\varid{d}})\mathbin{⊑}\widehat{\varid{d}}\Longleftrightarrow} (\ensuremath{\varid{d}},\ensuremath{\widehat{\varid{d}}}) ∈ R$  \commentend}{}\<[E]%
\\
\>[5]{}\hsindent{2}{}\<[7]%
\>[7]{}\varid{lfp}\;\widehat{\varid{rhs}}{}\<[E]%
\ColumnHook
\end{hscode}\resethooks
    Applying this fact to the second assumption proves
    $(\ensuremath{\varid{body}\;(\varid{fix}\;\varid{rhs})}, \ensuremath{\widehat{\varid{body}}\;(\varid{lfp}\;\widehat{\varid{rhs}})}) ∈ R$ and thus the goal.
\end{itemize}
\end{proof}

\subsection{Abstract By-need Soundness, in Detail}
\label{sec:by-need-soundness}

The goal of this section is to prove \Cref{thm:abstract-by-need} correct,
which is applicable for analyses that are sound both \wrt to by-name
as well as by-need, such as usage analysis or perhaps type analysis in
\Cref{sec:type-analysis} (we have however not proven it so).

The setup in \Cref{sec:by-name-soundness} with its Galois connection in
\Cref{fig:abstract-name} is somewhat similar to the Galois connection in
\Cref{fig:abstract-name-need}, however for by-need abstraction the Galois
connection for domain elements \ensuremath{\varid{d}\mathbin{::}\conid{D}_{\mathbf{ne}}} is indexed by a heap \wrt to which the
element is abstracted.
We will later see how this indexing yields a Kripke-style logical relation
as the soundness condition, and that, sadly, such a relation cannot easily be
proven by appealing to parametricity.

The reason we need to index correctness relations by a heap is as follows:
Although in \Cref{sec:evaluation-strategies} we considered an element \ensuremath{\varid{d}}
as an atomic denotation, such a denotation actually only carries meaning when it
travels together with a heap \ensuremath{\varcolor{\mu}} that ties the addresses that \ensuremath{\varid{d}} references to
actual meaning.

There are \emph{many} elements (functions!) \ensuremath{\varid{d}\mathbin{::}\conid{D}_{\mathbf{ne}}}, many with very
surprising behavior, but we are only interested in elements \emph{definable}
by our interpreter:

\begin{definition}[Definable by-need entities]
  \label{defn:definable}
  We write \ensuremath{{\mathop{\vdash_{\Domain{}}} \varid{d}}}, \ensuremath{{\mathop{\vdash_{\Environments}} \varcolor{\rho}}} or \ensuremath{{\mathop{\vdash_{\Heaps}} \varcolor{\mu}}} to say that the by-need
  element \ensuremath{\varid{d}}, environment \ensuremath{\varcolor{\rho}} or heap \ensuremath{\varcolor{\mu}} is \emph{definable}, defined as
  \begin{itemize}
    \item \ensuremath{{\mathop{\vdash_{\Environments}} \varcolor{\rho}}\triangleq\keyword{\forall}\!\! \hsforall \;\varid{x}\in \varid{dom}\;\varcolor{\rho}\hsdot{\circ }{.\ }\keyword{\exists}\!\! \hsexists \;\varid{x}\;\varid{a}\hsdot{\circ }{.\ }\varcolor{\rho}\mathop{!}\varid{x}\mathrel{=}\varid{step}\;(\conid{Look}\;\varid{y})\;(\varid{fetch}\;\varid{a})}.
    \item \ensuremath{\varid{adom}\;\varcolor{\rho}\triangleq\{\varid{a}\mid \varcolor{\rho}\mathop{!}\varid{x}\mathrel{=}\varid{step}\;(\conid{Look}\;\varid{y})\;(\varid{fetch}\;\varid{a})\}}.
    \item \ensuremath{{\mathop{\vdash_{\Domain{}}} \varid{d}}\triangleq\keyword{\exists}\!\! \hsexists \;\varid{e}\;\varcolor{\rho}\hsdot{\circ }{.\ }{\mathop{\vdash_{\Environments}} \varcolor{\rho}}\land\varid{d}\mathrel{=}\mathcal{S}_{\mathbf{need}}\denot{\varid{e}}_{\varcolor{\rho}}}.
    \item \ensuremath{\varid{adom}\;\varid{d}\triangleq\varid{adom}\;\varcolor{\rho}} where \ensuremath{\varid{d}\mathrel{=}\mathcal{S}_{\mathbf{need}}\denot{\varid{e}}_{\varcolor{\rho}}}.
    \item \ensuremath{{\mathop{\vdash_{\Heaps}} \varcolor{\mu}}\triangleq\keyword{\forall}\!\! \hsforall \;\varid{a}\hsdot{\circ }{.\ }\keyword{\exists}\!\! \hsexists \;\varid{d}\hsdot{\circ }{.\ }\varcolor{\mu}\mathop{!}\varid{a}\mathrel{=}\varid{memo}\;\varid{a}\;\varid{d}\land\later\!\;({\mathop{\vdash_{\Domain{}}} \varid{d}}\land\varid{adom}\;\varid{d}\mathbin{⊆}\varid{dom}\;\varcolor{\mu})}.
  \end{itemize}
  We refer to \ensuremath{\varid{adom}\;\varid{d}} (resp. \ensuremath{\varid{adom}\;\varcolor{\rho}}) as the \emph{address domain} of \ensuremath{\varid{d}} (resp. \ensuremath{\varcolor{\rho}}).
\end{definition}

Henceforth, we assume that all elements \ensuremath{\varid{d}}, environments \ensuremath{\varcolor{\rho}} and heaps \ensuremath{\varcolor{\mu}} of
interest are definable in this sense.
It is easy to see that definability is preserved by \ensuremath{\mathcal{S}_{\mathbf{need}}\denot{\wild}_{\wild}} whenever the
environment or heap is extended; the important case is the implementation of
\ensuremath{\varid{bind}}.


The indexed family of abstraction functions improves whenever the heap with
which we index is ``more evaluated'' --- the binary relation \ensuremath{(\progressto)} (``progresses
to'') on heaps in \Cref{fig:heap-progression} captures this progression.
It is defined as the smallest preorder (rules \progresstorefl, \progresstotrans)
that also contains rules \progresstoext and \progresstomemo.
The former corresponds to extending the heap in the \ensuremath{\conid{Let}} case.
The latter corresponds to memoising a heap entry after it was evaluated in the
\ensuremath{\conid{Var}} case.

\begin{figure}
  \[\begin{array}{c}
    \ruleform{ μ_1 \progressto μ_2 }
    \\ \\[-0.5em]
    \inferrule[\progresstorefl]{\ensuremath{{\mathop{\vdash_{\Heaps}} \varcolor{\mu}}}}{\ensuremath{\varcolor{\mu}\progressto\varcolor{\mu}}}
    \qquad
    \inferrule[\progresstotrans]{\ensuremath{\varcolor{\mu}_{1}\progressto\varcolor{\mu}_{2}} \quad \ensuremath{\varcolor{\mu}_{2}\progressto\varcolor{\mu}_{3}}}{\ensuremath{\varcolor{\mu}_{1}\progressto\varcolor{\mu}_{3}}}
    \qquad
    \inferrule[\progresstoext]
      {\ensuremath{\varid{a}} \not∈ \ensuremath{\varid{dom}\;\varcolor{\mu}} \quad \ensuremath{\varid{adom}\;\varcolor{\rho}\mathbin{⊆}\varid{dom}\;\varcolor{\mu}\mathbin{∪}\{\varid{a}\}}}
      {\ensuremath{\varcolor{\mu}\progressto\varcolor{\mu}[\varid{a}\mapsto\varid{memo}\;\varid{a}\;(\mathcal{S}_{\mathbf{need}}\denot{\varid{e}}_{\varcolor{\rho}})]}}
    \\ \\[-0.5em]
    \inferrule[\progresstomemo]
      {\ensuremath{\varcolor{\mu}_{1}\mathop{!}\varid{a}\mathrel{=}\varid{memo}\;\varid{a}\;(\mathcal{S}_{\mathbf{need}}\denot{\varid{e}}_{\varcolor{\rho}_{1}})} \quad \ensuremath{\later\!\;(\mathcal{S}_{\mathbf{need}}\denot{\varid{e}}_{\varcolor{\rho}_{1}}(\varcolor{\mu}_{1})\mathrel{=}\many{\conid{Step}\;\varid{ev}}\;(\mathcal{S}_{\mathbf{need}}\denot{\varid{v}}_{\varcolor{\rho}_{2}}(\varcolor{\mu}_{2})))}}
      {\ensuremath{\varcolor{\mu}_{1}\progressto\varcolor{\mu}_{2}[\varid{a}\mapsto\varid{memo}\;\varid{a}\;(\mathcal{S}_{\mathbf{need}}\denot{\varid{v}}_{\varcolor{\rho}_{2}})]}}
    \\[-0.5em]
  \end{array}\]
  \caption{Heap progression relation}
  \label{fig:heap-progression}
\end{figure}

Heap progression is the primary mechanism by which we can reason about the
meaning of programs:
If \ensuremath{\varcolor{\mu}_{1}} progresses to \ensuremath{\varcolor{\mu}_{2}} (\ie \ensuremath{\varcolor{\mu}_{1}\progressto\varcolor{\mu}_{2}}), and \ensuremath{\varid{adom}\;\varid{d}\mathbin{⊆}\varid{dom}\;\varcolor{\mu}_{1}}, then
\ensuremath{\varid{d}\;\varcolor{\mu}_{1}} has the same by-name semantics as \ensuremath{\varid{d}\;\varcolor{\mu}_{2}}, with the latter potentially
terminating in fewer steps.
We will exploit this observation in the abstract in
\Cref{thm:heap-progress-mono}, and now work towards proof.

To that end, it is important to build witnesses of \ensuremath{\varcolor{\mu}_{1}\progressto\varcolor{\mu}_{2}} in the first
place:

\begin{lemma}[Evaluation progresses the heap]
\label{thm:eval-progression}
If \ensuremath{\mathcal{S}_{\mathbf{need}}\denot{\varid{e}}_{\varcolor{\rho}_{1}}(\varcolor{\mu}_{1})\mathrel{=}\many{\conid{Step}\;\varid{ev}}\;(\mathcal{S}_{\mathbf{need}}\denot{\varid{v}}_{\varcolor{\rho}_{2}}(\varcolor{\mu}_{2}))}, then \ensuremath{\varcolor{\mu}_{1}\progressto\varcolor{\mu}_{2}}.
\end{lemma}
\begin{proof}
By Löb induction and cases on \ensuremath{\varid{e}}.
\begin{itemize}
  \item \textbf{Case} \ensuremath{\conid{Var}\;\varid{x}}:
    Let \ensuremath{\many{\varid{ev}_{1}}\triangleq\varid{tail}\;(\varid{init}\;(\many{\varid{ev}}))}.
    \begin{hscode}\SaveRestoreHook
\column{B}{@{}>{\hspre}l<{\hspost}@{}}%
\column{5}{@{}>{\hspre}c<{\hspost}@{}}%
\column{5E}{@{}l@{}}%
\column{9}{@{}>{\hspre}l<{\hspost}@{}}%
\column{11}{@{}>{\hspre}l<{\hspost}@{}}%
\column{E}{@{}>{\hspre}l<{\hspost}@{}}%
\>[9]{}(\varcolor{\rho}_{1}\mathop{!}\varid{x})\;\varcolor{\mu}_{1}{}\<[E]%
\\
\>[5]{}\mathrel{=}{}\<[5E]%
\>[9]{}\mbox{\commentbegin  \ensuremath{{\mathop{\vdash_{\Environments}} \varcolor{\rho}_{1}}}, some \ensuremath{\varid{y}}, \ensuremath{\varid{a}}  \commentend}{}\<[E]%
\\
\>[9]{}\conid{Step}\;(\conid{Look}\;\varid{y})\;(\varid{fetch}\;\varid{a}\;\varcolor{\mu}_{1}){}\<[E]%
\\
\>[5]{}\mathrel{=}{}\<[5E]%
\>[9]{}\mbox{\commentbegin  Unfold \ensuremath{\varid{fetch}}  \commentend}{}\<[E]%
\\
\>[9]{}\conid{Step}\;(\conid{Look}\;\varid{y})\;((\varcolor{\mu}_{1}\mathop{!}\varid{a})\;\varcolor{\mu}_{1}){}\<[E]%
\\
\>[5]{}\mathrel{=}{}\<[5E]%
\>[9]{}\mbox{\commentbegin  \ensuremath{{\mathop{\vdash_{\Heaps}} \varcolor{\mu}}}, some \ensuremath{\varid{e}}, \ensuremath{\varcolor{\rho}_{3}}  \commentend}{}\<[E]%
\\
\>[9]{}\conid{Step}\;(\conid{Look}\;\varid{y})\;(\varid{memo}\;\varid{a}\;(\mathcal{S}_{\mathbf{need}}\denot{\varid{e}}_{\varcolor{\rho}_{3}}(\varcolor{\mu}_{1}))){}\<[E]%
\\
\>[5]{}\mathrel{=}{}\<[5E]%
\>[9]{}\mbox{\commentbegin  Unfold \ensuremath{\varid{memo}}  \commentend}{}\<[E]%
\\
\>[9]{}\conid{Step}\;(\conid{Look}\;\varid{y})\;(\mathcal{S}_{\mathbf{need}}\denot{\varid{e}}_{\varcolor{\rho}_{3}}(\varcolor{\mu}_{1})\bind \varid{upd}){}\<[E]%
\\
\>[5]{}\mathrel{=}{}\<[5E]%
\>[9]{}\mbox{\commentbegin  \ensuremath{\mathcal{S}_{\mathbf{need}}\denot{\varid{e}}_{\varcolor{\rho}_{3}}(\varcolor{\mu}_{1})\mathrel{=}\many{\conid{Step}\;\varid{ev}_{1}}\;(\mathcal{S}_{\mathbf{need}}\denot{\varid{v}}_{\varcolor{\rho}_{2}}(\varcolor{\mu}_{3}))} for some \ensuremath{\varcolor{\mu}_{3}}, unfold \ensuremath{\bind }, \ensuremath{\varid{upd}}  \commentend}{}\<[E]%
\\
\>[9]{}\conid{Step}\;(\conid{Look}\;\varid{y})\;(\many{\conid{Step}\;\varid{ev}_{1}}\;(\mathcal{S}_{\mathbf{need}}\denot{\varid{v}}_{\varcolor{\rho}_{2}}(\varcolor{\mu}_{3})\bind \lambda \varid{v}\;\varcolor{\mu}_{3}\to {}\<[E]%
\\
\>[9]{}\hsindent{2}{}\<[11]%
\>[11]{}\conid{Step}\;\conid{Upd}\;(\conid{Ret}\;(\varid{v},\varcolor{\mu}_{3}[\varid{a}\mapsto\varid{memo}\;\varid{a}\;(\varid{return}\;\varid{v})])))){}\<[E]%
\ColumnHook
\end{hscode}\resethooks
    Now let \ensuremath{\varid{sv}\mathbin{::}\conid{Value}\;(\conid{ByNeed}\;\conid{T})} be the semantic value such that \ensuremath{\mathcal{S}_{\mathbf{need}}\denot{\varid{v}}_{\varcolor{\rho}_{2}}(\varcolor{\mu}_{3})\mathrel{=}\conid{Ret}\;(\varid{sv},\varcolor{\mu}_{3})}.
    \begin{hscode}\SaveRestoreHook
\column{B}{@{}>{\hspre}l<{\hspost}@{}}%
\column{5}{@{}>{\hspre}c<{\hspost}@{}}%
\column{5E}{@{}l@{}}%
\column{9}{@{}>{\hspre}l<{\hspost}@{}}%
\column{E}{@{}>{\hspre}l<{\hspost}@{}}%
\>[5]{}\mathrel{=}{}\<[5E]%
\>[9]{}\mbox{\commentbegin  \ensuremath{\mathcal{S}_{\mathbf{need}}\denot{\varid{v}}_{\varcolor{\rho}_{2}}(\varcolor{\mu}_{3})\mathrel{=}\conid{Ret}\;(\varid{sv},\varcolor{\mu}_{3})}  \commentend}{}\<[E]%
\\
\>[9]{}\conid{Step}\;(\conid{Look}\;\varid{y})\;(\many{\conid{Step}\;\varid{ev}_{1}}\;(\conid{Step}\;\conid{Upd}\;(\conid{Ret}\;(\varid{sv},\varcolor{\mu}_{3}[\varid{a}\mapsto\varid{memo}\;\varid{a}\;(\varid{return}\;\varid{sv})])))){}\<[E]%
\\
\>[5]{}\mathrel{=}{}\<[5E]%
\>[9]{}\mbox{\commentbegin  Refold \ensuremath{\mathcal{S}_{\mathbf{need}}\denot{\varid{v}}_{\varcolor{\rho}_{2}}(\wild)}, \ensuremath{\many{\varid{ev}}\mathrel{=}[\mskip1.5mu \conid{Look}\;\varid{y}\mskip1.5mu]\plus \many{\varid{ev}_{1}}\plus [\mskip1.5mu \conid{Upd}\mskip1.5mu]}  \commentend}{}\<[E]%
\\
\>[9]{}\many{\conid{Step}\;\varid{ev}}\;(\mathcal{S}_{\mathbf{need}}\denot{\varid{v}}_{\varcolor{\rho}_{2}}(\varcolor{\mu}_{3}[\varid{a}\mapsto\varid{memo}\;\varid{a}\;(\mathcal{S}_{\mathbf{need}}\denot{\varid{v}}_{\varcolor{\rho}_{2}})])){}\<[E]%
\\
\>[5]{}\mathrel{=}{}\<[5E]%
\>[9]{}\mbox{\commentbegin  Determinism of \ensuremath{\mathcal{S}_{\mathbf{need}}\denot{\wild}_{\wild}}, assumption  \commentend}{}\<[E]%
\\
\>[9]{}\many{\conid{Step}\;\varid{ev}}\;(\mathcal{S}_{\mathbf{need}}\denot{\varid{v}}_{\varcolor{\rho}_{2}}(\varcolor{\mu}_{2})){}\<[E]%
\ColumnHook
\end{hscode}\resethooks
    We have
    \begin{align}
      & \ensuremath{\varcolor{\mu}_{1}\mathop{!}\varid{a}\mathrel{=}\varid{memo}\;\varid{a}\;(\mathcal{S}_{\mathbf{need}}\denot{\varid{e}}_{\varcolor{\rho}_{3}})} \label{eqn:eval-progression-memo} \\
      & \ensuremath{\later\!\;(\mathcal{S}_{\mathbf{need}}\denot{\varid{e}}_{\varcolor{\rho}_{3}}(\varcolor{\mu}_{1})\mathrel{=}\many{\conid{Step}\;\varid{ev}_{1}}\;(\mathcal{S}_{\mathbf{need}}\denot{\varid{v}}_{\varcolor{\rho}_{2}}(\varcolor{\mu}_{3})))} \label{eqn:eval-progression-eval} \\
      & \ensuremath{\varcolor{\mu}_{2}\mathrel{=}\varcolor{\mu}_{3}[\varid{a}\mapsto\varid{memo}\;\varid{a}\;(\mathcal{S}_{\mathbf{need}}\denot{\varid{v}}_{\varcolor{\rho}_{2}})]} \label{eqn:eval-progression-heaps}
    \end{align}
    We can apply rule \progresstomemo to \Cref{eqn:eval-progression-memo} and \Cref{eqn:eval-progression-eval}
    to get \ensuremath{\varcolor{\mu}_{1}\progressto\varcolor{\mu}_{3}[\varid{a}\mapsto\varid{memo}\;\varid{a}\;(\mathcal{S}_{\mathbf{need}}\denot{\varid{v}}_{\varcolor{\rho}_{2}})]}, and rewriting along
    \Cref{eqn:eval-progression-heaps} proves the goal.
  \item \textbf{Case} \ensuremath{\conid{Lam}\;\varid{x}\;\varid{body}}, \ensuremath{\conid{ConApp}\;\varid{k}\;\varid{xs}}:
    Then \ensuremath{\varcolor{\mu}_{1}\mathrel{=}\varcolor{\mu}_{2}} and the goal follows by \progresstorefl.
  \item \textbf{Case} \ensuremath{\conid{App}\;\varid{e}_{1}\;\varid{x}}:
    Let us assume that \ensuremath{\mathcal{S}_{\mathbf{need}}\denot{\varid{e}_{1}}_{\varcolor{\rho}_{1}}(\varcolor{\mu}_{1})\mathrel{=}\many{\conid{Step}\;\varid{ev}_{1}}\;(\mathcal{S}_{\mathbf{need}}\denot{\conid{Lam}\;\varid{y}\;\varid{e}_{2}}_{\varcolor{\rho}_{3}}(\varcolor{\mu}_{3}))} and
    \ensuremath{\mathcal{S}_{\mathbf{need}}\denot{\varid{e}_{2}}_{\varcolor{\rho}_{3}[\varid{y}\mapsto\varcolor{\rho}\mathop{!}\varid{x}]}(\varcolor{\mu}_{3})\mathrel{=}\many{\conid{Step}\;\varid{ev}_{2}}\;(\mathcal{S}_{\mathbf{need}}\denot{\varid{v}}_{\varcolor{\rho}_{2}}(\varcolor{\mu}_{2}))}, so that
    \ensuremath{\varcolor{\mu}_{1}\progressto\varcolor{\mu}_{3}}, \ensuremath{\varcolor{\mu}_{3}\progressto\varcolor{\mu}_{2}} by the induction hypothesis.
    The goal follows by \progresstotrans, because
    \ensuremath{\many{\varid{ev}}\mathrel{=}[\mskip1.5mu \conid{App}_{1}\mskip1.5mu]\plus \many{\varid{ev}_{1}}\plus [\mskip1.5mu \conid{App}_{2}\mskip1.5mu]\plus \many{\varid{ev}_{2}}}.
  \item \textbf{Case} \ensuremath{\conid{Case}\;\varid{e}_{1}\;\varid{alts}}:
    Similar to \ensuremath{\conid{App}\;\varid{e}_{1}\;\varid{x}}.
  \item \textbf{Case} \ensuremath{\conid{Let}\;\varid{x}\;\varid{e}_{1}\;\varid{e}_{2}}:
    \begin{hscode}\SaveRestoreHook
\column{B}{@{}>{\hspre}l<{\hspost}@{}}%
\column{5}{@{}>{\hspre}c<{\hspost}@{}}%
\column{5E}{@{}l@{}}%
\column{9}{@{}>{\hspre}l<{\hspost}@{}}%
\column{15}{@{}>{\hspre}l<{\hspost}@{}}%
\column{32}{@{}>{\hspre}l<{\hspost}@{}}%
\column{E}{@{}>{\hspre}l<{\hspost}@{}}%
\>[9]{}\mathcal{S}_{\mathbf{need}}\denot{\conid{Let}\;\varid{x}\;\varid{e}_{1}\;\varid{e}_{2}}_{\varcolor{\rho}_{1}}(\varcolor{\mu}_{1}){}\<[E]%
\\
\>[5]{}\mathrel{=}{}\<[5E]%
\>[9]{}\mbox{\commentbegin  Unfold \ensuremath{\mathcal{S}_{\mathbf{need}}\denot{\wild}_{\wild}}  \commentend}{}\<[E]%
\\
\>[9]{}\varid{bind}\;{}\<[15]%
\>[15]{}(\lambda \varid{d}_{1}\to \mathcal{S}_{\mathbf{need}}\denot{\varid{e}_{1}}_{\varcolor{\rho}_{1}[\varid{x}\mapsto\varid{step}\;(\conid{Look}\;\varid{x})\;\varid{d}_{1}]}(\wild))\;{}\<[E]%
\\
\>[15]{}(\lambda \varid{d}_{1}\to \varid{step}\;\conid{Let}_{1}\;(\mathcal{S}_{\mathbf{need}}\denot{\varid{e}_{2}}_{\varcolor{\rho}_{1}[\varid{x}\mapsto\varid{step}\;(\conid{Look}\;\varid{x})\;\varid{d}_{1}]}(\wild)))\;{}\<[E]%
\\
\>[15]{}\varcolor{\mu}_{1}{}\<[E]%
\\
\>[5]{}\mathrel{=}{}\<[5E]%
\>[9]{}\mbox{\commentbegin  Unfold \ensuremath{\varid{bind}}, \ensuremath{\varid{a}\triangleq\varid{nextFree}\;\varcolor{\mu}} with $\ensuremath{\varid{a}} \not\in \ensuremath{\varid{dom}\;\varcolor{\mu}}$  \commentend}{}\<[E]%
\\
\>[9]{}\varid{step}\;\conid{Let}_{1}\;(\mathcal{S}_{\mathbf{need}}\denot{\varid{e}_{2}}_{\varcolor{\rho}_{1}[\varid{x}\mapsto\varid{step}\;(\conid{Look}\;\varid{x})\;(\varid{fetch}\;\varid{a})]}({}\<[E]%
\\
\>[9]{}\hsindent{23}{}\<[32]%
\>[32]{}\varcolor{\mu}_{1}[\varid{a}\mapsto\varid{memo}\;\varid{a}\;(\mathcal{S}_{\mathbf{need}}\denot{\varid{e}_{1}}_{\varcolor{\rho}_{1}[\varid{x}\mapsto\varid{step}\;(\conid{Look}\;\varid{x})\;(\varid{fetch}\;\varid{a})]})])){}\<[E]%
\ColumnHook
\end{hscode}\resethooks
    At this point, we can apply the induction hypothesis to \ensuremath{\mathcal{S}_{\mathbf{need}}\denot{\varid{e}_{2}}_{\varcolor{\rho}_{1}[\varid{x}\mapsto\varid{step}\;(\conid{Look}\;\varid{x})\;(\varid{fetch}\;\varid{a})]}} to conclude that
    \ensuremath{\varcolor{\mu}_{1}[\varid{a}\mapsto\varid{memo}\;\varid{a}\;(\mathcal{S}_{\mathbf{need}}\denot{\varid{e}_{1}}_{\varcolor{\rho}_{1}[\varid{x}\mapsto\varid{step}\;(\conid{Look}\;\varid{x})\;(\varid{fetch}\;\varid{a})]})]\progressto\varcolor{\mu}_{2}}.

    On the other hand, we have
    \ensuremath{\varcolor{\mu}_{1}\progressto\varcolor{\mu}_{1}[\varid{a}\mapsto\varid{memo}\;\varid{a}\;(\mathcal{S}_{\mathbf{need}}\denot{\varid{e}_{1}}_{\varcolor{\rho}_{1}[\varid{x}\mapsto\varid{step}\;(\conid{Look}\;\varid{x})\;(\varid{fetch}\;\varid{a})]})]}
    by rule \progresstoext (note that $\ensuremath{\varid{a}} \not∈ \ensuremath{\varid{dom}\;\varcolor{\mu}}$), so the goal follows
    by \progresstotrans.
\end{itemize}
\end{proof}


It is often necessary, but non-trivial to cope with equality of elements \ensuremath{\varid{d}}
modulo readdressing.
Fortunately, we only need to consider equality in the abstract, that is,
modulo \ensuremath{\beta_\Domain}, where \ensuremath{\beta_\Domain (\varcolor{\mu}) (\varid{d})\triangleq\alpha_\Domain (\varcolor{\mu}) (\{\varid{d}\})} is the representation
function of \ensuremath{\alpha_\Domain}.

\begin{lemma}[Abstract readdressing]
\label{thm:abstract-readdressing}
If \ensuremath{\varid{a}_{1}\in \varid{dom}\;\varcolor{\mu}_{1}}, but $\ensuremath{\varid{a}_{2}} \not∈ \ensuremath{\varid{dom}\;\varcolor{\mu}_{1}}$,
then \ensuremath{\beta_\Domain (\varcolor{\mu}_{1}) (\mathcal{S}_{\mathbf{need}}\denot{\varid{e}}_{\varcolor{\rho}_{1}})\mathrel{=}\beta_\Domain (\varcolor{\mu}_{2}) (\mathcal{S}_{\mathbf{need}}\denot{\varid{e}}_{\varcolor{\rho}_{2}})},
where \ensuremath{\varcolor{\rho}_{2}} and \ensuremath{\varcolor{\mu}_{2}} are renamings of \ensuremath{\varcolor{\rho}_{1}} and \ensuremath{\varcolor{\mu}_{1}} defined as follows:
\begin{itemize}
  \item $\ensuremath{\varcolor{\rho}_{2}\mathop{!}\varid{x}} = \begin{cases} \ensuremath{\varid{step}\;(\conid{Look}\;\varid{y})\;(\varid{fetch}\;\varid{a}_{2})} & \textup{if }\ensuremath{\varcolor{\rho}_{1}\mathop{!}\varid{x}\mathrel{=}\varid{step}\;(\conid{Look}\;\varid{y})\;(\varid{fetch}\;\varid{a}_{1})} \\ \ensuremath{\varcolor{\rho}_{1}\mathop{!}\varid{x}} & \textup{otherwise} \end{cases}$
  \item $\ensuremath{\varid{a}_{1}} \not∈ \ensuremath{\varid{dom}\;\varcolor{\mu}_{2}}$
  \item $\ensuremath{\varcolor{\mu}_{2}\mathop{!}\varid{a}} = \begin{cases} \ensuremath{\varid{memo}\;\varid{a}_{2}\;(\mathcal{S}_{\mathbf{need}}\denot{\varid{e}_{1}}_{\varcolor{\rho}_{4}})} & \textup{if \ensuremath{\varid{a}\mathrel{=}\varid{a}_{2}}, \ensuremath{\varcolor{\mu}_{1}\mathop{!}\varid{a}_{1}\mathrel{=}\varid{memo}\;\varid{a}_{1}\;(\mathcal{S}_{\mathbf{need}}\denot{\varid{e}_{1}}_{\varcolor{\rho}_{3}})}, \ensuremath{\varcolor{\rho}_{4}} renaming of \ensuremath{\varcolor{\rho}_{3}}} \\ \ensuremath{\varid{memo}\;\varid{a}\;(\mathcal{S}_{\mathbf{need}}\denot{\varid{e}_{1}}_{\varcolor{\rho}_{4}})} & \textup{where \ensuremath{\varcolor{\mu}_{1}\mathop{!}\varid{a}\mathrel{=}\varid{memo}\;\varid{a}\;(\mathcal{S}_{\mathbf{need}}\denot{\varid{e}_{1}}_{\varcolor{\rho}_{3}})}, \ensuremath{\varcolor{\rho}_{4}} renaming of \ensuremath{\varcolor{\rho}_{3}}} \end{cases}$
\end{itemize}
\end{lemma}
\begin{proof}
Simple proof by Löb induction and cases on \ensuremath{\varid{e}}.
\end{proof}

Readdressing allows us to prove an abstract pendant of the venerable \emph{frame
rule} of separation logic:

\begin{lemma}[Abstract frame rule]
\label{thm:abstract-frame}
If \ensuremath{\varid{adom}\;\varcolor{\rho}\mathbin{⊆}\varid{dom}\;\varcolor{\mu}} and $\ensuremath{\varid{a}} \not∈ \ensuremath{\varid{dom}\;\varcolor{\mu}}$,
then
\[\ensuremath{\beta_\Domain (\varcolor{\mu}) (\mathcal{S}_{\mathbf{need}}\denot{\varid{e}}_{\varcolor{\rho}})\mathrel{=}\beta_\Domain (\varcolor{\mu}[\varid{a}\mapsto\varid{memo}\;\varid{a}\;\varid{d}]) (\mathcal{S}_{\mathbf{need}}\denot{\varid{e}}_{\varcolor{\rho}})}.\]
\end{lemma}
\begin{proof}
By Löb induction and cases on \ensuremath{\varid{e}}.
Only the cases that access the heap are interesting.
\begin{itemize}
  \item \textbf{Case} \ensuremath{\conid{Var}\;\varid{x}}:
    We never fetch \ensuremath{\varid{a}}, because $\ensuremath{\varid{a}} \not∈ \ensuremath{\varid{adom}\;\varcolor{\rho}}$.
    Furthermore, the environment \ensuremath{\varcolor{\rho}_{1}} of the heap entry \ensuremath{\mathcal{S}_{\mathbf{need}}\denot{\varid{e}_{1}}_{\varcolor{\rho}_{1}}} thus
    fetched satisfies \ensuremath{\varid{adom}\;\varcolor{\rho}_{1}\mathbin{⊆}\varid{dom}\;\varcolor{\mu}} so that we may apply the induction
    hypothesis.
  \item \textbf{Case} \ensuremath{\conid{Let}\;\varid{x}\;\varid{e}_{1}\;\varid{e}_{2}}:
    Follows by the induction hypothesis after readdressing the extended heap
    (\Cref{thm:abstract-readdressing}) so that the induction hypothesis can be applied.
\end{itemize}
\end{proof}

The frame rule in turn is important to show that heap progression preserves the
results of the abstraction function:

\begin{lemma}[Heap progression preserves abstraction]
  \label{thm:heap-progress-mono}
  Let \ensuremath{\widehat{\conid{D}}} be a domain with instances for \ensuremath{\conid{Trace}}, \ensuremath{\conid{Domain}}, \ensuremath{\conid{HasBind}} and
  \ensuremath{\conid{Lat}}, satisfying \textsc{Beta-App}, \textsc{Beta-Sel}, \textsc{ByName-Bind}
  and \textsc{Step-Inc} from \Cref{fig:abstraction-laws}.
  \[ \ensuremath{\varcolor{\mu}_{1}\progressto\varcolor{\mu}_{2}\land\varid{adom}\;\varid{d}\mathbin{⊆}\varid{dom}\;\varcolor{\mu}_{1}\implies\beta_\Domain (\varcolor{\mu}_{2}) (\varid{d})\mathbin{⊑}\beta_\Domain (\varcolor{\mu}_{1}) (\varid{d})}. \]
\end{lemma}
\begin{proof}
By Löb induction.
Element \ensuremath{\varid{d}} is definable of the form \ensuremath{\varid{d}\mathrel{=}\mathcal{S}_{\mathbf{need}}\denot{\varid{e}}_{\varcolor{\rho}}} for definable \ensuremath{\varcolor{\rho}}.
Proceed by cases on \ensuremath{\varid{e}}. Only the \ensuremath{\conid{Var}} and \ensuremath{\conid{Let}} cases are interesting.
\begin{itemize}
  \item \textbf{Case} \ensuremath{\conid{Let}\;\varid{x}\;\varid{e}_{1}\;\varid{e}_{2}}:
    We need to readdress the extended heaps with \Cref{thm:abstract-readdressing} so that
    \ensuremath{\varcolor{\mu}_{1}[\varid{a}_{1}\mapsto\varid{memo}\;\varid{a}_{1}\;\varid{d}_{1}]\progressto\varcolor{\mu}_{2}[\varid{a}_{1}\mapsto\varid{memo}\;\varid{a}_{1}\;\varid{d}_{1}]} is preserved, in which case the goal follows by applying the
    induction hypothesis.
  \item \textbf{Case} \ensuremath{\conid{Var}\;\varid{x}}:
    Let us assume that \ensuremath{\varcolor{\mu}_{1}\progressto\varcolor{\mu}_{2}} and \ensuremath{\varid{adom}\;\varid{d}\mathbin{⊆}\varid{dom}\;\varcolor{\mu}_{1}}.
    We get \ensuremath{\varid{d}\mathrel{=}\varid{step}\;(\conid{Look}\;\varid{y})\;(\varid{fetch}\;\varid{a})} for some \ensuremath{\varid{y}} and \ensuremath{\varid{a}}.
    Furthermore, let us abbreviate \ensuremath{\varid{memo}\;\varid{a}\;(\mathcal{S}_{\mathbf{need}}\denot{\varcolor{e}_i}_{\varcolor{\rho}_i})\triangleq\varcolor{\mu}_i\mathop{!}\varid{a}}.
    The goal is to show
    \[
      \ensuremath{\varid{step}\;(\conid{Look}\;\varid{y})\;(\beta_\Domain (\varcolor{\mu}_{2}) (\varid{memo}\;\varid{a}\;(\mathcal{S}_{\mathbf{need}}\denot{\varid{e}_{2}}_{\varcolor{\rho}_{2}})))\mathbin{⊑}\varid{step}\;(\conid{Look}\;\varid{y})\;(\beta_\Domain (\varcolor{\mu}_{1}) (\varid{memo}\;\varid{a}\;(\mathcal{S}_{\mathbf{need}}\denot{\varid{e}_{1}}_{\varcolor{\rho}_{1}})))}.
    \]
    Monotonicity allows us to drop the \ensuremath{\varid{step}\;(\conid{Look}\;\varid{y})} context
    \[
      \ensuremath{\later\!\;(\beta_\Domain (\varcolor{\mu}_{2}) (\varid{memo}\;\varid{a}\;(\mathcal{S}_{\mathbf{need}}\denot{\varid{e}_{2}}_{\varcolor{\rho}_{2}}))\mathbin{⊑}\beta_\Domain (\varcolor{\mu}_{1}) (\varid{memo}\;\varid{a}\;(\mathcal{S}_{\mathbf{need}}\denot{\varid{e}_{1}}_{\varcolor{\rho}_{1}})))}.
    \]
    Now we proceed by induction on \ensuremath{\varcolor{\mu}_{1}\progressto\varcolor{\mu}_{2}}, which we only use to prove correct the
    reflexive and transitive closure in \progresstorefl and \progresstotrans.
    \begin{itemize}
      \item \textbf{Case} \progresstorefl:
        Then \ensuremath{\varcolor{\mu}_{1}\mathrel{=}\varcolor{\mu}_{2}} and hence \ensuremath{\beta_\Domain (\varcolor{\mu}_{1})\mathrel{=}\beta_\Domain (\varcolor{\mu}_{2})}.
      \item \textbf{Case} \progresstotrans:
        Apply the induction hypothesis to the sub-derivations and apply transitivity
        of \ensuremath{\mathbin{⊑}}.

      \item \textbf{Case} $\inferrule*[vcenter,left=\progresstoext]{\ensuremath{\varid{a}_{1}} \not∈ \ensuremath{\varid{dom}\;\varcolor{\mu}_{1}} \quad \ensuremath{\varid{adom}\;\varcolor{\rho}\mathbin{⊆}\varid{dom}\;\varcolor{\mu}_{1}\mathbin{∪}\{\varid{a}_{1}\}}}{\ensuremath{\varcolor{\mu}_{1}\progressto\varcolor{\mu}_{1}[\varid{a}_{1}\mapsto\varid{memo}\;\varid{a}_{1}\;(\mathcal{S}_{\mathbf{need}}\denot{\varid{e}}_{\varcolor{\rho}})]}}$:\\
        We get to refine \ensuremath{\varcolor{\mu}_{2}\mathrel{=}\varcolor{\mu}_{1}[\varid{a}_{1}\mapsto\varid{memo}\;\varid{a}_{1}\;(\mathcal{S}_{\mathbf{need}}\denot{\varid{e}}_{\varcolor{\rho}})]}.
        Since \ensuremath{\varid{a}\in \varid{dom}\;\varcolor{\mu}_{1}},
        we have $\ensuremath{\varid{a}_{1}} \not= \ensuremath{\varid{a}}$
        and thus \ensuremath{\varcolor{\mu}_{1}\mathop{!}\varid{a}\mathrel{=}\varcolor{\mu}_{2}\mathop{!}\varid{a}}, thus \ensuremath{\varid{e}_{1}\mathrel{=}\varid{e}_{2}}, \ensuremath{\varcolor{\rho}_{1}\mathrel{=}\varcolor{\rho}_{2}}.
        We can exploit monotonicity of \ensuremath{\later\!} and simplify the goal to
        \begin{hscode}\SaveRestoreHook
\column{B}{@{}>{\hspre}l<{\hspost}@{}}%
\column{11}{@{}>{\hspre}l<{\hspost}@{}}%
\column{E}{@{}>{\hspre}l<{\hspost}@{}}%
\>[11]{}\beta_\Domain (\varcolor{\mu}_{1}[\varid{a}_{1}\mapsto\varid{memo}\;\varid{a}_{1}\;(\mathcal{S}_{\mathbf{need}}\denot{\varid{e}}_{\varcolor{\rho}})]) (\varid{memo}\;\varid{a}\;(\mathcal{S}_{\mathbf{need}}\denot{\varid{e}_{1}}_{\varcolor{\rho}_{1}}))\mathbin{⊑}\beta_\Domain (\varcolor{\mu}_{1}) (\varid{memo}\;\varid{a}\;(\mathcal{S}_{\mathbf{need}}\denot{\varid{e}_{1}}_{\varcolor{\rho}_{1}})){}\<[E]%
\ColumnHook
\end{hscode}\resethooks
        This follows by applying the abstract frame rule (\Cref{thm:abstract-frame}).

      \item \textbf{Case} $\inferrule*[vcenter,left=\progresstomemo]{\ensuremath{\varcolor{\mu}_{1}\mathop{!}\varid{a}_{1}\mathrel{=}\varid{memo}\;\varid{a}_{1}\;(\mathcal{S}_{\mathbf{need}}\denot{\varid{e}}_{\varcolor{\rho}_{3}})} \quad \ensuremath{\later\!\;(\mathcal{S}_{\mathbf{need}}\denot{\varid{e}}_{\varcolor{\rho}_{3}}(\varcolor{\mu}_{1})\mathrel{=}\many{\conid{Step}\;\varid{ev}}\;(\mathcal{S}_{\mathbf{need}}\denot{\varid{v}}_{\varcolor{\rho}_{4}}(\varcolor{\mu}_{3})))}}{\ensuremath{\varcolor{\mu}_{1}\progressto\varcolor{\mu}_{3}[\varid{a}_{1}\mapsto\varid{memo}\;\varid{a}_{1}\;(\mathcal{S}_{\mathbf{need}}\denot{\varid{v}}_{\varcolor{\rho}_{4}})]}}$:\\
        We get to refine \ensuremath{\varcolor{\mu}_{2}\mathrel{=}\varcolor{\mu}_{3}[\varid{a}_{1}\mapsto\varid{memo}\;\varid{a}_{1}\;(\mathcal{S}_{\mathbf{need}}\denot{\varid{v}}_{\varcolor{\rho}_{2}})]}.
        When $\ensuremath{\varid{a}_{1}} \not= \ensuremath{\varid{a}}$, we have \ensuremath{\varcolor{\mu}_{1}\mathop{!}\varid{a}\mathrel{=}\varcolor{\mu}_{2}\mathop{!}\varid{a}} and the goal follows as in the \progresstoext case.
        Otherwise, \ensuremath{\varid{a}\mathrel{=}\varid{a}_{1}}, \ensuremath{\varid{e}_{1}\mathrel{=}\varid{e}}, \ensuremath{\varcolor{\rho}_{3}\mathrel{=}\varcolor{\rho}_{1}}, \ensuremath{\varcolor{\rho}_{4}\mathrel{=}\varcolor{\rho}_{2}}, \ensuremath{\varid{e}_{2}\mathrel{=}\varid{v}}.

        The goal can be simplified to
        \ensuremath{\later\!\;(\beta_\Domain (\varcolor{\mu}_{2}) (\varid{memo}\;\varid{a}\;(\mathcal{S}_{\mathbf{need}}\denot{\varid{v}}_{\varcolor{\rho}_{2}}))\mathbin{⊑}\beta_\Domain (\varcolor{\mu}_{1}) (\varid{memo}\;\varid{a}\;(\mathcal{S}_{\mathbf{need}}\denot{\varid{e}_{1}}_{\varcolor{\rho}_{1}})))}.
        We reason under \ensuremath{\later\!} as follows
        \begin{hscode}\SaveRestoreHook
\column{B}{@{}>{\hspre}l<{\hspost}@{}}%
\column{9}{@{}>{\hspre}l<{\hspost}@{}}%
\column{11}{@{}>{\hspre}l<{\hspost}@{}}%
\column{13}{@{}>{\hspre}l<{\hspost}@{}}%
\column{E}{@{}>{\hspre}l<{\hspost}@{}}%
\>[11]{}\beta_\Domain (\varcolor{\mu}_{2}) (\varid{memo}\;\varid{a}\;(\mathcal{S}_{\mathbf{need}}\denot{\varid{v}}_{\varcolor{\rho}_{2}})){}\<[E]%
\\
\>[9]{}\mathrel{=}\mbox{\commentbegin  \ensuremath{\varcolor{\mu}_{2}\mathop{!}\varid{a}\mathrel{=}\varid{memo}\;\varid{a}\;(\mathcal{S}_{\mathbf{need}}\denot{\varid{v}}_{\varcolor{\rho}_{2}})}  \commentend}{}\<[E]%
\\
\>[9]{}\hsindent{2}{}\<[11]%
\>[11]{}\beta_\Traces (\conid{Step}\;\conid{Update}\;(\mathcal{S}_{\mathbf{need}}\denot{\varid{v}}_{\varcolor{\rho}_{2}}(\varcolor{\mu}_{2}))){}\<[E]%
\\
\>[9]{}\mathrel{=}\mbox{\commentbegin  \ensuremath{\varcolor{\mu}_{2}\mathrel{=}\varcolor{\mu}_{3}[\varid{a}\mapsto\varid{memo}\;\varid{a}\;(\mathcal{S}_{\mathbf{need}}\denot{\varid{v}}_{\varcolor{\rho}_{2}})]}  \commentend}{}\<[E]%
\\
\>[9]{}\hsindent{2}{}\<[11]%
\>[11]{}\beta_\Domain (\varcolor{\mu}_{3}) (\varid{memo}\;\varid{a}\;(\mathcal{S}_{\mathbf{need}}\denot{\varid{v}}_{\varcolor{\rho}_{2}})){}\<[E]%
\\
\>[9]{}\mathbin{⊑}{}\<[13]%
\>[13]{}\mbox{\commentbegin  \ensuremath{\mathcal{S}_{\mathbf{need}}\denot{\varid{e}_{1}}_{\varcolor{\rho}_{1}}(\varcolor{\mu}_{1})\mathrel{=}\many{\conid{Step}\;\varid{ev}}\;(\mathcal{S}_{\mathbf{need}}\denot{\varid{v}}_{\varcolor{\rho}_{2}}(\varcolor{\mu}_{3}))}, assumption \textsc{Step-Inc}  \commentend}{}\<[E]%
\\
\>[9]{}\hsindent{2}{}\<[11]%
\>[11]{}\beta_\Domain (\varcolor{\mu}_{1}) (\varid{memo}\;\varid{a}\;(\mathcal{S}_{\mathbf{need}}\denot{\varid{e}_{1}}_{\varcolor{\rho}_{1}})){}\<[E]%
\ColumnHook
\end{hscode}\resethooks
    \end{itemize}
\end{itemize}
\end{proof}

The preceding lemma corresponds to the $\UpdateT$ step of the preservation
\Cref{thm:preserve-absent} where we (and \citet{Sergey:14}) resorted to
hand-waving.
Here, we hand-wave no more!

In order to prove the main soundness \Cref{thm:abstract-by-need},
we need two more auxiliary lemmas about \ensuremath{(\bind )} and environment
access.

\begin{lemma}[By-need bind]
\label{thm:by-need-bind}
It is \ensuremath{\beta_\Traces ((\varid{d}\bind \varid{f})\;\varcolor{\mu}_{1})\mathbin{⊑}\widehat{\varid{f}}\;\widehat{\varid{d}}} if
\begin{enumerate}
  \item \ensuremath{\beta_\Traces (\varid{d}\;\varcolor{\mu}_{1})\mathbin{⊑}\widehat{\varid{d}}}, and
  \item for all events \ensuremath{\varid{ev}} and elements \ensuremath{\widehat{\varid{d'}}}, \ensuremath{\widehat{\varid{step}}\;\varid{ev}\;(\widehat{\varid{f}}\;\widehat{\varid{d'}})\mathbin{⊑}\widehat{\varid{f}}\;(\widehat{\varid{step}}\;\varid{ev}\;\widehat{\varid{d'}})}, and
  \item for all values \ensuremath{\varid{v}} and heaps \ensuremath{\varcolor{\mu}_{2}} such that \ensuremath{\varcolor{\mu}_{1}\progressto\varcolor{\mu}_{2}}, \ensuremath{\beta_\Traces (\varid{f}\;\varid{v}\;\varcolor{\mu}_{2})\mathbin{⊑}\widehat{\varid{f}}\;(\beta_\Traces (\conid{Ret}\;(\varid{v},\varcolor{\mu}_{2})))}.
\end{enumerate}
\end{lemma}
\begin{proof}
By Löb induction.

If \ensuremath{\varid{d}\;\varcolor{\mu}_{1}\mathrel{=}\conid{Step}\;\varid{ev}\;(\varid{d'}\;\varcolor{\mu}_{1}')}, define \ensuremath{\widehat{\varid{d'}}\triangleq\beta_\Traces (\varid{d'}\;\varcolor{\mu}_{1}')} and note that
\ensuremath{\varcolor{\mu}_{1}\progressto\varcolor{\mu}_{1}'} for the definable \ensuremath{\varid{d}} as in \Cref{thm:eval-progression} (we will
always instantiate the original \ensuremath{\varid{d}} to \ensuremath{\mathcal{S}_{\mathbf{need}}\denot{\varid{e}}_{\varcolor{\rho}}}).
We get
\begin{hscode}\SaveRestoreHook
\column{B}{@{}>{\hspre}c<{\hspost}@{}}%
\column{BE}{@{}l@{}}%
\column{3}{@{}>{\hspre}l<{\hspost}@{}}%
\column{4}{@{}>{\hspre}l<{\hspost}@{}}%
\column{E}{@{}>{\hspre}l<{\hspost}@{}}%
\>[3]{}\beta_\Traces ((\varid{d}\bind \varid{f})\;\varcolor{\mu}_{1})\mathrel{=}\beta_\Traces (\conid{Step}\;\varid{ev}\;((\varid{d'}\bind \varid{f})\;\varcolor{\mu}_{1}'))\mathrel{=}\widehat{\varid{step}}\;\varid{ev}\;(\beta_\Traces ((\varid{d'}\bind \varid{f})\;\varcolor{\mu}_{1}')){}\<[E]%
\\
\>[B]{}\mathbin{⊑}{}\<[BE]%
\>[4]{}\mbox{\commentbegin  Induction hypothesis at \ensuremath{\beta_\Traces (\varid{d'}\;\varcolor{\mu}_{1}')\mathrel{=}\widehat{\varid{d'}}}, Monotonicity of \ensuremath{\widehat{\varid{step}}}  \commentend}{}\<[E]%
\\
\>[B]{}\hsindent{3}{}\<[3]%
\>[3]{}\widehat{\varid{step}}\;\varid{ev}\;(\widehat{\varid{f}}\;(\beta_\Traces (\varid{d'}\;\varcolor{\mu}_{1}'))){}\<[E]%
\\
\>[B]{}\mathbin{⊑}{}\<[BE]%
\>[4]{}\mbox{\commentbegin  Assumption (2)  \commentend}{}\<[E]%
\\
\>[B]{}\hsindent{3}{}\<[3]%
\>[3]{}\widehat{\varid{f}}\;(\widehat{\varid{step}}\;\varid{ev}\;(\beta_\Traces (\varid{d'}\;\varcolor{\mu}_{1}')))\mathrel{=}\widehat{\varid{f}}\;(\beta_\Traces (\varid{d}\;\varcolor{\mu}_{1})){}\<[E]%
\\
\>[B]{}\mathbin{⊑}{}\<[BE]%
\>[4]{}\mbox{\commentbegin  Assumption (1)  \commentend}{}\<[E]%
\\
\>[B]{}\hsindent{3}{}\<[3]%
\>[3]{}\widehat{\varid{f}}\;\widehat{\varid{d}}{}\<[E]%
\ColumnHook
\end{hscode}\resethooks
Note that in order to apply the induction hypothesis at \ensuremath{\varcolor{\mu}_{1}'} above, we need
refine assumption (3) to apply at any \ensuremath{\varcolor{\mu}_{2}} such that \ensuremath{\varcolor{\mu}_{1}'\progressto\varcolor{\mu}_{2}}.
This would not be possible without generalising for any such \ensuremath{\varcolor{\mu}_{2}} in the first
place.

Otherwise, \ensuremath{\varid{d}\mathrel{=}\varid{return}\;\varid{v}} for some \ensuremath{\varid{v}\mathbin{::}\conid{Value}}.
\begin{hscode}\SaveRestoreHook
\column{B}{@{}>{\hspre}c<{\hspost}@{}}%
\column{BE}{@{}l@{}}%
\column{3}{@{}>{\hspre}l<{\hspost}@{}}%
\column{4}{@{}>{\hspre}l<{\hspost}@{}}%
\column{E}{@{}>{\hspre}l<{\hspost}@{}}%
\>[3]{}\beta_\Traces ((\varid{return}\;\varid{v}\bind \varid{f})\;\varcolor{\mu}_{1})\mathrel{=}\beta_\Traces (\varid{f}\;\varid{v}\;\varcolor{\mu}_{1}){}\<[E]%
\\
\>[B]{}\mathbin{⊑}{}\<[BE]%
\>[4]{}\mbox{\commentbegin  Assumption (3)  \commentend}{}\<[E]%
\\
\>[B]{}\hsindent{3}{}\<[3]%
\>[3]{}\widehat{\varid{f}}\;(\beta_\Traces (\conid{Ret}\;\varid{v},\varcolor{\mu}_{1}))\mathrel{=}\widehat{\varid{f}}\;(\beta_\Traces (\varid{d}\;\varcolor{\mu}_{1})){}\<[E]%
\\
\>[B]{}\mathbin{⊑}{}\<[BE]%
\>[4]{}\mbox{\commentbegin  Assumption (1)  \commentend}{}\<[E]%
\\
\>[B]{}\hsindent{3}{}\<[3]%
\>[3]{}\widehat{\varid{f}}\;\widehat{\varid{d}}{}\<[E]%
\ColumnHook
\end{hscode}\resethooks
\end{proof}

\begin{lemma}[By-need environment unrolling]
\label{thm:by-need-env-unroll}
Let \ensuremath{\widehat{\conid{D}}} a domain with instances for \ensuremath{\conid{Trace}}, \ensuremath{\conid{Domain}}, \ensuremath{\conid{HasBind}} and \ensuremath{\conid{Lat}},
satisfying $\textsc{Update}$ from \Cref{fig:abstraction-laws},
and let \ensuremath{\varcolor{\mu}_{1}\triangleq\varcolor{\mu}[\varid{a}\mapsto\varid{memo}\;\varid{a}\;(\mathcal{S}_{\mathbf{need}}\denot{\varid{e}_{1}}_{\varcolor{\rho}_{1}}(\wild))],\varcolor{\rho}_{1}\triangleq\varcolor{\rho}[\varid{x}\mapsto\varid{step}\;(\conid{Look}\;\varid{x})\;(\varid{fetch}\;\varid{a})]}. \\
If \ensuremath{\later\!\;(\keyword{\forall}\!\! \hsforall \;\varid{e}\;\varcolor{\rho}\;\varcolor{\mu}\hsdot{\circ }{.\ }\beta_\Traces (\mathcal{S}_{\mathbf{need}}\denot{\varid{e}}_{\varcolor{\rho}}(\varcolor{\mu}))\mathbin{⊑}(\mathcal{S}_{\widehat{\conid{D}}}\denot{\varid{e}}_{\beta_\Domain (\varcolor{\mu})\mathbin{\lhd}\varcolor{\rho}}))},
then \ensuremath{\beta_\Domain (\varcolor{\mu}_{1}) (\varcolor{\rho}_{1}\mathop{!}\varid{x})\mathbin{⊑}\varid{step}\;(\conid{Look}\;\varid{x})\;(\mathcal{S}_{\widehat{\conid{D}}}\denot{\varid{e}_{1}}_{\beta_\Domain (\varcolor{\mu}_{1})\mathbin{\lhd}\varcolor{\rho}_{1}})}.
\end{lemma}
\begin{proof} $ $
\begin{hscode}\SaveRestoreHook
\column{B}{@{}>{\hspre}c<{\hspost}@{}}%
\column{BE}{@{}l@{}}%
\column{3}{@{}>{\hspre}l<{\hspost}@{}}%
\column{5}{@{}>{\hspre}l<{\hspost}@{}}%
\column{E}{@{}>{\hspre}l<{\hspost}@{}}%
\>[3]{}\beta_\Domain (\varcolor{\mu}_{1}) (\varcolor{\rho}_{1}\mathop{!}\varid{x}){}\<[E]%
\\
\>[B]{}\mathrel{=}{}\<[BE]%
\>[5]{}\mbox{\commentbegin  \ensuremath{{\mathop{\vdash_{\Environments}} \varcolor{\rho}_{1}}}, \ensuremath{{\mathop{\vdash_{\Heaps}} \varcolor{\mu}_{1}}}, unfold \ensuremath{\beta_\Domain} and \ensuremath{\varid{fetch}\;\varid{a}}  \commentend}{}\<[E]%
\\
\>[B]{}\hsindent{3}{}\<[3]%
\>[3]{}\varid{step}\;(\conid{Look}\;\varid{x})\;(\beta_\Traces (\varid{memo}\;\varid{a}\;(\mathcal{S}_{\mathbf{need}}\denot{\varid{e}_{1}}_{\varcolor{\rho}_{1}}(\wild))\;\varcolor{\mu}_{1})){}\<[E]%
\\
\>[B]{}\mathrel{=}{}\<[BE]%
\>[5]{}\mbox{\commentbegin  Unfold \ensuremath{\varid{memo}\;\varid{a}}  \commentend}{}\<[E]%
\\
\>[B]{}\hsindent{3}{}\<[3]%
\>[3]{}\varid{step}\;(\conid{Look}\;\varid{x})\;(\beta_\Traces ((\mathcal{S}_{\mathbf{need}}\denot{\varid{e}_{1}}_{\varcolor{\rho}_{1}}(\wild)\bind \varid{upd})\;\varcolor{\mu}_{1})){}\<[E]%
\\
\>[B]{}\mathbin{⊑}{}\<[BE]%
\>[5]{}\mbox{\commentbegin  Apply \Cref{thm:by-need-bind}; see below  \commentend}{}\<[E]%
\\
\>[B]{}\hsindent{3}{}\<[3]%
\>[3]{}\varid{step}\;(\conid{Look}\;\varid{x})\;(\mathcal{S}_{\widehat{\conid{D}}}\denot{\varid{e}_{1}}_{\beta_\Domain (\varcolor{\mu}_{1})\mathbin{\lhd}\varcolor{\rho}_{1}}){}\<[E]%
\ColumnHook
\end{hscode}\resethooks

In the last step, we apply \Cref{thm:by-need-bind} under \ensuremath{\varid{step}\;(\conid{Look}\;\varid{x})}, which
yields three subgoals (under $\later$):
\begin{itemize}
  \item \ensuremath{\beta_\Traces (\mathcal{S}_{\mathbf{need}}\denot{\varid{e}_{1}}_{\varcolor{\rho}_{1}}(\varcolor{\mu}_{1}))\mathbin{⊑}\mathcal{S}_{\widehat{\conid{D}}}\denot{\varid{e}_{1}}_{\beta_\Domain (\varcolor{\mu}_{1})\mathbin{\lhd}\varcolor{\rho}_{1}}}:
    By assumption.
  \item \ensuremath{\keyword{\forall}\!\! \hsforall \;\varid{ev}\;\widehat{\varid{d'}}\hsdot{\circ }{.\ }\widehat{\varid{step}}\;\varid{ev}\;(\varid{id}\;\widehat{\varid{d'}})\mathbin{⊑}\varid{id}\;(\widehat{\varid{step}}\;\varid{ev}\;\widehat{\varid{d'}})}:
    By reflexivity.
  \item \ensuremath{\keyword{\forall}\!\! \hsforall \;\varid{v}\;\varcolor{\mu}_{2}\hsdot{\circ }{.\ }\varcolor{\mu}_{1}\progressto\varcolor{\mu}_{2}\implies\beta_\Traces (\varid{upd}\;\varid{v}\;\varcolor{\mu}_{2})\mathbin{⊑}\varid{id}\;(\beta_\Traces (\conid{Ret}\;(\varid{v},\varcolor{\mu}_{2})))}:
    If \ensuremath{\varid{v}\mathrel{=}\conid{Stuck}}, then \ensuremath{\varid{upd}\;\varid{v}\;\varcolor{\mu}_{2}\mathrel{=}\conid{Ret}\;(\varid{v},\varcolor{\mu}_{2})} and the goal follows by reflexivity.
    Otherwise, \ensuremath{\varid{upd}\;\varid{v}\;\varcolor{\mu}_{2}\mathrel{=}\conid{Step}\;\conid{Update}\;(\conid{Ret}\;(\varid{v},\varcolor{\mu}_{2}[\varid{a}\mapsto\varid{memo}\;\varid{a}\;(\varid{return}\;\varid{v})]))}.
    By $\textsc{Update}$, it suffices to show \ensuremath{\beta_\Traces (\conid{Ret}\;(\varid{v},\varcolor{\mu}_{2}[\varid{a}\mapsto\varid{memo}\;\varid{a}\;(\varid{return}\;\varid{v})]))\mathbin{⊑}\beta_\Traces (\conid{Ret}\;(\varid{v},\varcolor{\mu}_{2}))}.
    It is \ensuremath{\varcolor{\mu}_{2}\progressto\varcolor{\mu}_{2}[\varid{a}\mapsto\varid{memo}\;\varid{a}\;(\varid{return}\;\varid{v})]} (either by \progresstorefl or \progresstomemo)
    and the goal follows by \Cref{thm:heap-progress-mono}.
\end{itemize}
\end{proof}

Finally, we can prove \Cref{thm:abstract-by-need}:
\end{toappendix}

\begin{figure}
\[\ruleform{\begin{array}{c}
  α_{\mathcal{S}} : (\ensuremath{(\conid{Name}\mathbin{:\rightharpoonup}\conid{D}_{\mathbf{ne}})\to \conid{D}_{\mathbf{ne}}}) \to (\ensuremath{(\conid{Name}\mathbin{:\rightharpoonup}\widehat{\conid{D}})\to \widehat{\conid{D}}})
  \\
  α_{\Environments} : \ensuremath{\conid{Heap}_{\mathbf{ne}}} \times \pow{(\ensuremath{\conid{Name}\mathbin{:\rightharpoonup}\conid{D}_{\mathbf{ne}}})} \rightleftarrows (\ensuremath{\conid{Name}\mathbin{:\rightharpoonup}\widehat{\conid{D}}}) : γ_{\Environments}
  \\
  α_{\Domain{}} : \ensuremath{\conid{Heap}_{\mathbf{ne}}} \to \pow{\ensuremath{\conid{D}_{\mathbf{ne}}}} \rightleftarrows \ensuremath{\widehat{\conid{D}}} : γ_{\Domain{}}
  \\
  α_\Traces : \pow{\ensuremath{\conid{T}\;(\conid{Value}_{\mathbf{ne}},\conid{Heap}_{\mathbf{ne}})}} \rightleftarrows \ensuremath{\widehat{\conid{D}}} : γ_\Traces
  \qquad
  β_\Traces : \ensuremath{\conid{T}\;(\conid{Value}_{\mathbf{ne}},\conid{Heap}_{\mathbf{ne}})} \to \ensuremath{\widehat{\conid{D}}}
  \qquad
\end{array}}\]
\belowdisplayskip=0pt
\arraycolsep=2pt
\[\begin{array}{lcl}
α_{\mathcal{S}}(S)(\widehat{ρ}) & = & α_\Traces(\{\  S(ρ)(μ) \mid (μ,ρ) ∈ γ_{\Environments}(\widehat{ρ}) \ \}) \\
α_{\Environments}(E)(x) & = & α_\Traces(\{\  ρ(x)(μ) \mid (μ,ρ) ∈ E \ \}) \\
α_{\Domain{}}(μ)(D) & = & α_\Traces(\{\  d(μ) \mid d ∈ D \ \}) \\
α_\Traces(T) & = & \Lub \{\ β_\Traces(τ) \mid τ ∈ T \ \} \\
\\[-0.75em]
β_\Traces(\ensuremath{\varcolor{\tau}}) & = & \begin{cases}
  \ensuremath{\varid{step}\;\varid{e}\;(\beta_\Traces (\varcolor{\tau}'))} & \text{if \ensuremath{\varcolor{\tau}\mathrel{=}\conid{Step}\;\varid{e}\;\varcolor{\tau}'}} \\
  \ensuremath{\varid{stuck}}                         & \text{if \ensuremath{\varcolor{\tau}\mathrel{=}\conid{Ret}\;(\conid{Stuck},\varcolor{\mu})}} \\
  \ensuremath{\varid{fun}\;(\alpha_\Domain (\varcolor{\mu})\hsdot{\circ }{.\ }{\varid{f}}^{*}\hsdot{\circ }{.\ }\gamma_\Domain (\varcolor{\mu}))} & \text{if \ensuremath{\varcolor{\tau}\mathrel{=}\conid{Ret}\;(\conid{Fun}\;\varid{f},\varcolor{\mu})}} \\
  \ensuremath{\varid{con}\;\varid{k}\;(\varid{map}\;(\alpha_\Domain (\varcolor{\mu})\hsdot{\circ }{.\ }\{\wild\})\;\varid{ds})} & \text{if \ensuremath{\varcolor{\tau}\mathrel{=}\conid{Ret}\;(\conid{Con}\;\varid{k}\;\varid{ds},\varcolor{\mu})}} \\
  \end{cases} \\
\end{array}\]
\caption{Galois connection $α_{\mathcal{S}}$ for by-need abstraction derived from \ensuremath{\conid{Trace}}, \ensuremath{\conid{Domain}} and \ensuremath{\conid{Lat}} instances on \ensuremath{\widehat{\conid{D}}}}
\label{fig:abstract-name-need}
\end{figure}

In this section we prove and apply a generic abstract interpretation theorem
of the form
\[
  α_{\mathcal{S}}(\ensuremath{\mathcal{S}_{\mathbf{need}}\denot{\varid{e}}}) ⊑ \ensuremath{\mathcal{S}_{\widehat{\conid{D}}}\denot{\varid{e}}}.
\]
This statement can be read as follows:
For an expression \ensuremath{\varid{e}}, the \emph{static analysis} \ensuremath{\mathcal{S}_{\widehat{\conid{D}}}\denot{\varid{e}}}
on the right-hand side \emph{overapproximates} ($⊒$) a property of the by-need
\emph{semantics} \ensuremath{\mathcal{S}_{\mathbf{need}}\denot{\varid{e}}} on the left-hand side.
The abstraction function $α_{\mathcal{S}}$, given in
\Cref{fig:abstract-name-need}, defines the semantic property of interest in
terms of the abstract domain \ensuremath{\widehat{\conid{D}}} of \ensuremath{\mathcal{S}_{\widehat{\conid{D}}}\denot{\varid{e}}\;\varcolor{\rho}}, which is
short for \ensuremath{\mathcal{S}\denot{\varid{e}}_{\varcolor{\rho}}\mathbin{::}\widehat{\conid{D}}}.
That is: the type class instances on \ensuremath{\widehat{\conid{D}}} determine $α_{\mathcal{S}}$, and
hence the semantic property that is soundly abstracted by \ensuremath{\mathcal{S}_{\widehat{\conid{D}}}\denot{\varid{e}}_{\varcolor{\rho}}}.

We will instantiate the theorem at \ensuremath{\concolor{\mathsf{D_U}}} in order to prove that usage analysis
\ensuremath{\mathcal{S}_{\mathbf{usage}}\denot{\varid{e}}_{\varcolor{\rho}}\mathrel{=}\mathcal{S}_{\concolor{\mathsf{D_U}}}\denot{\varid{e}}_{\varcolor{\rho}}} infers absence, just as absence analysis in
\Cref{sec:problem}.
This proof will be much simpler than the proof for \Cref{thm:absence-correct},
because the complicated preservation proof is reusably contained in
the abstract interpretation theorem.

\begin{figure}
  \belowdisplayskip=0pt
  \[\begin{array}{cc}
    \multicolumn{2}{c}{\inferrule[\textsc{Mono}]
      {}
      {\text{\ensuremath{\varid{step}}, \ensuremath{\varid{stuck}}, \ensuremath{\varid{fun}}, \ensuremath{\varid{apply}}, \ensuremath{\varid{con}}, \ensuremath{\varid{select}}, \ensuremath{\varid{bind}} monotone}}} \\
    \\[-0.5em]
    \inferrule[\textsc{Step-App}]{}{%
      \ensuremath{\varid{step}\;\varid{ev}\;(\varid{apply}\;\varid{d}\;\varid{a})\mathbin{⊑}\varid{apply}\;(\varid{step}\;\varid{ev}\;\varid{d})\;\varid{a}}}
    &
    \inferrule[\textsc{Step-Sel}]{}{%
      \ensuremath{\varid{step}\;\varid{ev}\;(\varid{select}\;\varid{d}\;\varid{alts})\mathbin{⊑}\varid{select}\;(\varid{step}\;\varid{ev}\;\varid{d})\;\varid{alts}}} \\
    \\[-0.5em]
    \inferrule[\textsc{Stuck-App}]
      {\ensuremath{\varid{d}\in \{\varid{stuck},\varid{con}\;\varid{k}\;\varid{ds}\}}}
      {\ensuremath{\varid{stuck}\mathbin{⊑}\varid{apply}\;\varid{d}\;\varid{a}}}
    &
    \inferrule[\textsc{Stuck-Sel}]
      {\ensuremath{\varid{d}\in \{\varid{stuck},\varid{fun}\;\varid{x}\;\varid{f}\}\mathbin{∪}\{\varid{con}\;\varid{k}\;\varid{ds}\mid \varid{k}\not\in \varid{dom}\;\varid{alts}\}}}
      {\ensuremath{\varid{stuck}\mathbin{⊑}\varid{select}\;\varid{d}\;\varid{alts}}} \\
    \\[-0.5em]
    \inferrule[\textsc{Beta-App}]
      {\ensuremath{\varid{f}} \text{ polymorphic} \\ \ensuremath{\varid{x}}\text{ fresh}}
      {\ensuremath{\varid{f}\;\varid{a}\mathbin{⊑}\varid{apply}\;(\varid{fun}\;\varid{x}\;\varid{f})\;\varid{a}}}
    &
    \inferrule[\textsc{Beta-Sel}]
      {\ensuremath{\varid{alts}} \text{ polymorphic} \\ \ensuremath{\varid{k}\in \varid{dom}\;\varid{alts}}}
      {\ensuremath{(\varid{alts}\mathop{!}\varid{k})\;\varid{ds}\mathbin{⊑}\varid{select}\;(\varid{con}\;\varid{k}\;\varid{ds})\;\varid{alts}}} \\
    \\[-0.5em]
    \inferrule[\textsc{ByName-Bind}]
      {\ensuremath{\varid{rhs},\varid{body}}\text{ polymorphic}}
      {\ensuremath{\varid{body}\;(\varid{lfp}\;\varid{rhs})\mathbin{⊑}\varid{bind}\;\varid{rhs}\;\varid{body}}}
    &
    \fcolorbox{lightgray}{white}{$\begin{array}{c}
      \inferrule[\textsc{Step-Inc}]{}{\ensuremath{\varid{d}\mathbin{⊑}\varid{step}\;\varid{ev}\;\varid{d}}}
      \qquad
      \inferrule[\textsc{Update}]{}{\ensuremath{\varid{step}\;\conid{Upd}\;\varid{d}\mathrel{=}\varid{d}}}
    \end{array}$}
  \end{array}\]
  \caption{By-name and \fcolorbox{lightgray}{white}{by-need} abstraction laws for type class instances of abstract domain \ensuremath{\widehat{\conid{D}}}}
  \label{fig:abstraction-laws}
\end{figure}

\subsection{A Reusable Abstract By-Need Interpretation Theorem}

In this subsection, we explain and prove \Cref{thm:abstract-by-need} for
abstract by-need interpretation, which we will apply to prove usage analysis
sound in \Cref{sec:usage-sound}.
The theorem corresponds to the following derived inference rule, referring to
the \emph{abstraction laws} in \Cref{fig:abstraction-laws} by name:
\[\begin{array}{c}
  \inferrule{%
    \textsc{Mono} \\ \textsc{Step-App} \\ \textsc{Step-Sel} \\ \textsc{Stuck-App} \\
    \textsc{Stuck-Sel} \\ \textsc{Beta-App} \\ \textsc{Beta-Sel} \\ \textsc{ByName-Bind} \\
    \textsc{Step-Inc} \\ \textsc{Update}
  }{%
  α_{\mathcal{S}}(\ensuremath{\mathcal{S}_{\mathbf{need}}\denot{\varid{e}}}) ⊑ \ensuremath{\mathcal{S}_{\widehat{\conid{D}}}\denot{\varid{e}}}
  }
\end{array}\]
\noindent
In other words: prove the abstraction laws for an abstract domain \ensuremath{\widehat{\conid{D}}} of
your choosing (such as \ensuremath{\concolor{\mathsf{D_U}}}) and we give you a proof of sound abstract by-need
interpretation for the static analysis \ensuremath{\mathcal{S}_{\widehat{\conid{D}}}\denot{\wild}}!

Note that \emph{we} get to determine the abstraction function $α_{\mathcal{S}}$ based
on the \ensuremath{\conid{Trace}}, \ensuremath{\conid{Domain}} and \ensuremath{\conid{Lat}} instance on \emph{your} \ensuremath{\widehat{\conid{D}}}.
\Cref{fig:abstract-name-need} defines how $α_{\mathcal{S}}$ is thus derived.

Let us calculate \ensuremath{\alpha_\mathcal{S}} for the closed example expression
$\pe \triangleq \Let{i}{(\Lam{y}{\Lam{x}{x}})~i}{i}$:
\begin{align}
    & \ensuremath{\alpha_\mathcal{S} (\mathcal{S}_{\mathbf{need}}\denot{( \Let{i}{(\Lam{y}{\Lam{x}{x}})~i}{i} )}) (\varcolor{\varepsilon})} \notag \\
={} & \ensuremath{\beta_\Traces (\mathcal{S}_{\mathbf{need}}\denot{( \Let{i}{(\Lam{y}{\Lam{x}{x}})~i}{i} )}_{\varcolor{\varepsilon}}(\varcolor{\varepsilon}))} \label{eqn:abs-ex1} \\
={} & \ensuremath{\beta_\Traces}(
\LetIT\xhookrightarrow{\hspace{1.1ex}}\LookupT(i)\xhookrightarrow{\hspace{1.1ex}}\AppIT\xhookrightarrow{\hspace{1.1ex}}\AppET\xhookrightarrow{\hspace{1.1ex}}\UpdateT\xhookrightarrow{\hspace{1.1ex}}\langle (\lambda,[0\!\! \mapsto \!\! \wild])\rangle 
) \label{eqn:abs-ex2} \\
={} & \textstyle\ensuremath{\varid{step}\;\conid{Let}_{1}\mathbin{\$}\varid{step}\;(\conid{Look}\;\text{\ttfamily \char34 i\char34})\mathbin{\$}\mathbin{...}\mathbin{\$}\varid{fun}\;(\lambda \widehat{\varid{d}}\to \Lub\{\beta_\Traces ( \AppET \smallstep \varid{d}([0\!\!↦\!\!\wild])  )\mid \varid{d}\in \gamma_\Domain ([0\!\!↦\!\!\wild] ) (\widehat{\varid{d}})\})} \notag \\
⊑{} & \textstyle\ensuremath{\varid{step}\;\conid{Let}_{1}\mathbin{\$}\varid{step}\;(\conid{Look}\;\text{\ttfamily \char34 i\char34})\mathbin{\$}\mathbin{...}\mathbin{\$}\varid{fun}\;(\lambda \widehat{\varid{d}}\to \varid{step}\;\conid{App}_{2}\;\widehat{\varid{d}})} \label{eqn:abs-ex3} \\
={} & \ensuremath{\langle [\varid{i}\mapsto\concolor{\mathsf{U_1}}], \concolor{\mathsf{U_1}} \argcons \conid{Rep}\;\concolor{\mathsf{U_\omega}} \rangle\mathbin{::}\concolor{\mathsf{D_U}}} \label{eqn:abs-ex4} \\
={} & \ensuremath{\mathcal{S}_{\mathbf{usage}}\denot{( \Let{i}{(\Lam{y}{\Lam{x}{x}})~i}{i} )}_{\varcolor{\varepsilon}}} \notag
\end{align}
\noindent
In (\ref{eqn:abs-ex1}), $α_{\mathcal{S}}(\ensuremath{\mathcal{S}_{\mathbf{need}}\denot{\varid{e}}})(\ensuremath{\varcolor{\varepsilon}})$ simplifies to \ensuremath{\beta_\Traces (\mathcal{S}_{\mathbf{need}}\denot{\varid{e}}_{\varcolor{\varepsilon}}(\varcolor{\varepsilon}))}.
Function \ensuremath{\beta_\Traces} then folds the by-need trace (\ref{eqn:abs-ex2}) into an abstract
domain element in \ensuremath{\widehat{\conid{D}}}.
It does so by eliminating every \ensuremath{\conid{Step}\;\varid{ev}} in the trace with a call to
\ensuremath{\varid{step}\;\varid{ev}}, and every concrete \ensuremath{\conid{Value}} at the end of the trace with a
call to the corresponding \ensuremath{\conid{Domain}} method, following the structure of
types as in \citet{Backhouse:04}.
Since \ensuremath{\widehat{\conid{D}}} has a \ensuremath{\conid{Lat}} instance, \ensuremath{\beta_\Traces} is a \emph{representation
function}~\citep[Section 4.3]{Nielson:99}, giving rise to Galois connections
$α_\Traces \rightleftarrows γ_\Traces$ and
$α_\Domain(μ) \rightleftarrows γ_\Domain(μ)$.
This implies that $α_\Domain(μ) \circ γ_\Domain(μ) ⊑ \mathit{id}$,
justifying the approximation step $(⊑)$ in (\ref{eqn:abs-ex3}).
For the concrete example, we instantiate \ensuremath{\widehat{\conid{D}}} to \ensuremath{\concolor{\mathsf{D_U}}} in step
(\ref{eqn:abs-ex4}) to assert that the resulting usage type indeed
coincides with the result of \ensuremath{\mathcal{S}_{\mathbf{usage}}\denot{\wild}}, as predicted by the abstract
interpretation theorem.

The abstraction function \ensuremath{\alpha_\Domain} for by-need elements \ensuremath{\varid{d}} is a bit unusual because
it is \emph{indexed} by a heap to give meaning to addresses referenced by \ensuremath{\varid{d}}.
Our framework is carefully set up in a way that \ensuremath{\alpha_\Domain (\varcolor{\mu})} is preserved when \ensuremath{\varcolor{\mu}}
is modified by memoisation ``in the future'', reminiscent of
\citeauthor{Kripke:63}'s possible worlds.
For similar reasons, the abstraction function for environments pairs up
definable by-need environments \ensuremath{\varcolor{\rho}}, the entries of which are of the form (\ensuremath{\varid{step}\;\conid{Look}\;\varid{y})\;(\varid{fetch}\;\varid{a})}, with heaps \ensuremath{\varcolor{\mu}}.

Thanks to fixing $α_{\mathcal{S}}$, we can prove the following abstraction theorem,
corresponding to the inference rule at the begin of this subsection:

\begin{theoremrep}[Abstract By-need Interpretation]
\label{thm:abstract-by-need}
Let \ensuremath{\varid{e}} be an expression, \ensuremath{\widehat{\conid{D}}} a domain with instances for \ensuremath{\conid{Trace}}, \ensuremath{\conid{Domain}}, \ensuremath{\conid{HasBind}} and
\ensuremath{\conid{Lat}}, and let $α_{\mathcal{S}}$ be the abstraction function from \Cref{fig:abstract-name-need}.
If the abstraction laws in \Cref{fig:abstraction-laws} hold,
then \ensuremath{\mathcal{S}_{\widehat{\conid{D}}}\denot{\wild}} is an abstract interpreter that is sound \wrt $α_{\mathcal{S}}$,
that is,
\[
  α_{\mathcal{S}}(\ensuremath{\mathcal{S}_{\mathbf{need}}\denot{\varid{e}}}) ⊑ \ensuremath{\mathcal{S}_{\widehat{\conid{D}}}\denot{\varid{e}}}.
\]
\end{theoremrep}
\begin{proof}
We simplify our proof obligation to the single-trace case:
\begin{hscode}\SaveRestoreHook
\column{B}{@{}>{\hspre}l<{\hspost}@{}}%
\column{3}{@{}>{\hspre}c<{\hspost}@{}}%
\column{3E}{@{}l@{}}%
\column{5}{@{}>{\hspre}l<{\hspost}@{}}%
\column{9}{@{}>{\hspre}l<{\hspost}@{}}%
\column{E}{@{}>{\hspre}l<{\hspost}@{}}%
\>[5]{}\keyword{\forall}\!\! \hsforall \;\varid{e}\hsdot{\circ }{.\ }\alpha_\mathcal{S} (\mathcal{S}_{\mathbf{need}}\denot{\varid{e}})\mathbin{⊑}\mathcal{S}_{\widehat{\conid{D}}}\denot{\varid{e}}{}\<[E]%
\\
\>[3]{}\Longleftrightarrow{}\<[3E]%
\>[9]{}\mbox{\commentbegin  Unfold \ensuremath{\alpha_\mathcal{S}}, \ensuremath{\alpha_\Traces}  \commentend}{}\<[E]%
\\
\>[3]{}\hsindent{2}{}\<[5]%
\>[5]{}\keyword{\forall}\!\! \hsforall \;\varid{e}\;\widehat{\varcolor{\rho}}\hsdot{\circ }{.\ }\Lub\{\beta_\Traces (\mathcal{S}_{\mathbf{need}}\denot{\varid{e}}_{\varcolor{\rho}}(\varcolor{\mu}))\mid (\varcolor{\rho},\varcolor{\mu})\in \gamma_\Environments (\widehat{\varcolor{\rho}})\}\mathbin{⊑}(\mathcal{S}_{\widehat{\conid{D}}}\denot{\varid{e}}_{\widehat{\varcolor{\rho}}}){}\<[E]%
\\
\>[3]{}\Longleftrightarrow{}\<[3E]%
\>[9]{}\mbox{\commentbegin  \ensuremath{(\varcolor{\rho},\varcolor{\mu})\in \gamma_\Environments (\widehat{\varcolor{\rho}})\Longleftrightarrow\alpha_\Environments (\{\varcolor{\rho},\varcolor{\mu}\})\mathbin{⊑}\widehat{\varcolor{\rho}}}, unfold \ensuremath{\alpha_\Environments}, refold \ensuremath{\beta_\Domain}  \commentend}{}\<[E]%
\\
\>[3]{}\hsindent{2}{}\<[5]%
\>[5]{}\keyword{\forall}\!\! \hsforall \;\varid{e}\;\varcolor{\rho}\;\varcolor{\mu}\hsdot{\circ }{.\ }\beta_\Traces (\mathcal{S}_{\mathbf{need}}\denot{\varid{e}}_{\varcolor{\rho}}(\varcolor{\mu}))\mathbin{⊑}(\mathcal{S}_{\widehat{\conid{D}}}\denot{\varid{e}}_{\beta_\Domain (\varcolor{\mu})\mathbin{\lhd}\varcolor{\rho}}){}\<[E]%
\ColumnHook
\end{hscode}\resethooks
where \ensuremath{\beta_\Traces\triangleq\alpha_\Traces\hsdot{\circ }{.\ }\{\wild\}} and \ensuremath{\beta_\Domain (\varcolor{\mu})\triangleq\alpha_\Domain (\varcolor{\mu})\hsdot{\circ }{.\ }\{\wild\}} are the representation
functions corresponding to \ensuremath{\alpha_\Traces} and \ensuremath{\alpha_\Domain}.
We proceed by Löb induction and cases over \ensuremath{\varid{e}}.
\begin{itemize}
  \item \textbf{Case} \ensuremath{\conid{Var}\;\varid{x}}:
    The case $\ensuremath{\varid{x}} \not∈ \ensuremath{\varid{dom}\;\varcolor{\rho}}$ follows by reflexivity.
    Otherwise, let \ensuremath{\varcolor{\rho}\mathop{!}\varid{x}\mathrel{=}\conid{Step}\;(\conid{Look}\;\varid{y})\;(\varid{fetch}\;\varid{a})}.
    \begin{hscode}\SaveRestoreHook
\column{B}{@{}>{\hspre}l<{\hspost}@{}}%
\column{5}{@{}>{\hspre}c<{\hspost}@{}}%
\column{5E}{@{}l@{}}%
\column{9}{@{}>{\hspre}l<{\hspost}@{}}%
\column{E}{@{}>{\hspre}l<{\hspost}@{}}%
\>[9]{}\beta_\Traces ((\varcolor{\rho}\mathop{!}\varid{x})\;\varcolor{\mu}){}\<[E]%
\\
\>[5]{}\mathbin{⊑}{}\<[5E]%
\>[9]{}\mbox{\commentbegin  \Cref{thm:by-need-env-unroll}  \commentend}{}\<[E]%
\\
\>[9]{}\varid{step}\;(\conid{Look}\;\varid{y})\;(\mathcal{S}_{\widehat{\conid{D}}}\denot{\varid{e}_{1}}_{\beta_\Domain (\varcolor{\mu})\mathbin{\lhd}\varcolor{\rho}_{1}}){}\<[E]%
\\
\>[5]{}\mathrel{=}{}\<[5E]%
\>[9]{}\mbox{\commentbegin  Refold \ensuremath{\beta_\Domain}  \commentend}{}\<[E]%
\\
\>[9]{}\beta_\Domain (\varcolor{\mu}) (\varcolor{\rho}\mathop{!}\varid{x}){}\<[E]%
\ColumnHook
\end{hscode}\resethooks

  \item \textbf{Case} \ensuremath{\conid{Lam}\;\varid{x}\;\varid{body}}:
    \begin{hscode}\SaveRestoreHook
\column{B}{@{}>{\hspre}l<{\hspost}@{}}%
\column{5}{@{}>{\hspre}c<{\hspost}@{}}%
\column{5E}{@{}l@{}}%
\column{9}{@{}>{\hspre}l<{\hspost}@{}}%
\column{E}{@{}>{\hspre}l<{\hspost}@{}}%
\>[9]{}\beta_\Traces (\mathcal{S}_{\mathbf{need}}\denot{\conid{Lam}\;\varid{x}\;\varid{body}}_{\varcolor{\rho}}(\varcolor{\mu})){}\<[E]%
\\
\>[5]{}\mathrel{=}{}\<[5E]%
\>[9]{}\mbox{\commentbegin  Unfold \ensuremath{\mathcal{S}_{\mathbf{need}}\denot{\wild}_{\wild}}, \ensuremath{\beta_\Traces}  \commentend}{}\<[E]%
\\
\>[9]{}\varid{fun}\;(\lambda \widehat{\varid{d}}\to \Lub\{\varid{step}\;\conid{App}_{2}\;(\beta_\Traces (\mathcal{S}_{\mathbf{need}}\denot{\varid{body}}_{\varcolor{\rho}[\varid{x}\mapsto\varid{d}]}(\varcolor{\mu})))\mid \varid{d}\in \gamma_\Domain (\varcolor{\mu})\;\widehat{\varid{d}}\}){}\<[E]%
\\
\>[5]{}\mathbin{⊑}{}\<[5E]%
\>[9]{}\mbox{\commentbegin  Induction hypothesis  \commentend}{}\<[E]%
\\
\>[9]{}\varid{fun}\;(\lambda \widehat{\varid{d}}\to \Lub\{\varid{step}\;\conid{App}_{2}\;(\mathcal{S}_{\widehat{\conid{D}}}\denot{\varid{body}}_{\beta_\Domain (\varcolor{\mu})\mathbin{\lhd}\varcolor{\rho}[\varid{x}\mapsto\varid{d}]})\mid \varid{d}\in \gamma_\Domain (\varcolor{\mu})\;\widehat{\varid{d}}\}){}\<[E]%
\\
\>[5]{}\mathbin{⊑}{}\<[5E]%
\>[9]{}\mbox{\commentbegin  Least upper bound / \ensuremath{\alpha_\Domain (\varcolor{\mu})\hsdot{\circ }{.\ }\gamma_\Domain (\varcolor{\mu})\mathbin{⊑}\varid{id}}  \commentend}{}\<[E]%
\\
\>[9]{}\varid{fun}\;(\lambda \widehat{\varid{d}}\to \varid{step}\;\conid{App}_{2}\;(\mathcal{S}_{\widehat{\conid{D}}}\denot{\varid{body}}_{(\beta_\Domain (\varcolor{\mu})\mathbin{\lhd}\varcolor{\rho})[\varid{x}\mapsto\widehat{\varid{d}}]})){}\<[E]%
\\
\>[5]{}\mathrel{=}{}\<[5E]%
\>[9]{}\mbox{\commentbegin  Refold \ensuremath{\mathcal{S}_{\widehat{\conid{D}}}\denot{\wild}_{\wild}}  \commentend}{}\<[E]%
\\
\>[9]{}\mathcal{S}_{\widehat{\conid{D}}}\denot{\conid{Lam}\;\varid{x}\;\varid{body}}_{\beta_\Domain (\varcolor{\mu})\mathbin{\lhd}\varcolor{\rho}}{}\<[E]%
\ColumnHook
\end{hscode}\resethooks

  \item \textbf{Case} \ensuremath{\conid{ConApp}\;\varid{k}\;\varid{xs}}:
    \begin{hscode}\SaveRestoreHook
\column{B}{@{}>{\hspre}l<{\hspost}@{}}%
\column{5}{@{}>{\hspre}c<{\hspost}@{}}%
\column{5E}{@{}l@{}}%
\column{9}{@{}>{\hspre}l<{\hspost}@{}}%
\column{E}{@{}>{\hspre}l<{\hspost}@{}}%
\>[9]{}\beta_\Traces (\mathcal{S}_{\mathbf{need}}\denot{\conid{ConApp}\;\varid{k}\;\varid{xs}}_{\varcolor{\rho}}(\varcolor{\mu})){}\<[E]%
\\
\>[5]{}\mathrel{=}{}\<[5E]%
\>[9]{}\mbox{\commentbegin  Unfold \ensuremath{\mathcal{S}_{\mathbf{need}}\denot{\wild}_{\wild}}, \ensuremath{\beta_\Traces}  \commentend}{}\<[E]%
\\
\>[9]{}\varid{con}\;\varid{k}\;(\varid{map}\;((\beta_\Domain (\varcolor{\mu})\mathbin{\lhd}\varcolor{\rho})\mathop{!})\;\varid{xs}){}\<[E]%
\\
\>[5]{}\mathrel{=}{}\<[5E]%
\>[9]{}\mbox{\commentbegin  Refold \ensuremath{\mathcal{S}_{\widehat{\conid{D}}}\denot{\wild}_{\wild}}  \commentend}{}\<[E]%
\\
\>[9]{}\mathcal{S}_{\widehat{\conid{D}}}\denot{\conid{Lam}\;\varid{x}\;\varid{body}}_{\beta_\Domain (\varcolor{\mu})\mathbin{\lhd}\varcolor{\rho}}{}\<[E]%
\ColumnHook
\end{hscode}\resethooks

  \item \textbf{Case} \ensuremath{\conid{App}\;\varid{e}\;\varid{x}}:
    Very similar to the \ensuremath{\varid{apply}} case in \Cref{thm:abstract-by-name}.
    There is one exception: We must apply \Cref{thm:heap-progress-mono}
    to argument denotations.

    The \ensuremath{\varid{stuck}} case is simple.
    Otherwise, we have
    \begin{hscode}\SaveRestoreHook
\column{B}{@{}>{\hspre}l<{\hspost}@{}}%
\column{5}{@{}>{\hspre}c<{\hspost}@{}}%
\column{5E}{@{}l@{}}%
\column{7}{@{}>{\hspre}l<{\hspost}@{}}%
\column{9}{@{}>{\hspre}l<{\hspost}@{}}%
\column{E}{@{}>{\hspre}l<{\hspost}@{}}%
\>[7]{}\beta_\Traces (\mathcal{S}_{\mathbf{need}}\denot{\conid{App}\;\varid{e}\;\varid{x}}_{\varcolor{\rho}}(\varcolor{\mu})){}\<[E]%
\\
\>[5]{}\mathrel{=}{}\<[5E]%
\>[9]{}\mbox{\commentbegin  Unfold \ensuremath{\mathcal{S}_{\mathbf{need}}\denot{\wild}_{\wild}}, \ensuremath{\beta_\Traces}, \ensuremath{\varid{apply}}  \commentend}{}\<[E]%
\\
\>[5]{}\hsindent{2}{}\<[7]%
\>[7]{}\varid{step}\;\conid{App}_{1}\;((\mathcal{S}_{\mathbf{need}}\denot{\varid{e}}_{\varcolor{\rho}}(\wild)\bind \lambda \varid{v}\to \keyword{case}\;\varid{v}\;\keyword{of}\;\conid{Fun}\;\varid{f}\to \varid{f}\;(\varcolor{\rho}\mathop{!}\varid{x});\anonymous \to \varid{stuck})\;\varcolor{\mu}){}\<[E]%
\\
\>[5]{}\mathbin{⊑}{}\<[5E]%
\>[9]{}\mbox{\commentbegin  Apply \Cref{thm:by-need-bind}; see below  \commentend}{}\<[E]%
\\
\>[5]{}\hsindent{2}{}\<[7]%
\>[7]{}\varid{step}\;\conid{App}_{1}\;(\widehat{\varid{apply}}\;(\mathcal{S}_{\widehat{\conid{D}}}\denot{\varid{e}}_{\beta_\Domain (\varcolor{\mu})\mathbin{\lhd}\varcolor{\rho}})\;(\beta_\Domain (\varcolor{\mu}) (\varcolor{\rho}\mathop{!}\varid{x}))){}\<[E]%
\\
\>[5]{}\mathrel{=}{}\<[5E]%
\>[9]{}\mbox{\commentbegin  Refold \ensuremath{\mathcal{S}_{\widehat{\conid{D}}}\denot{\wild}_{\wild}}  \commentend}{}\<[E]%
\\
\>[5]{}\hsindent{2}{}\<[7]%
\>[7]{}\mathcal{S}_{\widehat{\conid{D}}}\denot{\varid{e}}_{\beta_\Domain (\varcolor{\mu})\mathbin{\lhd}\varcolor{\rho}}{}\<[E]%
\ColumnHook
\end{hscode}\resethooks

    In the $⊑$ step, we apply \Cref{thm:by-need-bind} under \ensuremath{\varid{step}\;\conid{App}_{1}}, which
    yields three subgoals (under $\later$):
    \begin{itemize}
      \item \ensuremath{\beta_\Traces (\mathcal{S}_{\mathbf{need}}\denot{\varid{e}}_{\varcolor{\rho}}(\varcolor{\mu}))\mathbin{⊑}\mathcal{S}_{\widehat{\conid{D}}}\denot{\varid{e}}_{\beta_\Domain (\varcolor{\mu})\mathbin{\lhd}\varcolor{\rho}}}:
        By induction hypothesis.
      \item \ensuremath{\keyword{\forall}\!\! \hsforall \;\varid{ev}\;\widehat{\varid{d'}}\hsdot{\circ }{.\ }\widehat{\varid{step}}\;\varid{ev}\;(\widehat{\varid{apply}}\;\widehat{\varid{d'}}\;(\beta_\Domain (\varcolor{\mu}) (\varcolor{\rho}\mathop{!}\varid{x})))\mathbin{⊑}\widehat{\varid{apply}}\;(\widehat{\varid{step}}\;\varid{ev}\;\widehat{\varid{d'}})\;(\beta_\Domain (\varcolor{\mu}) (\varcolor{\rho}\mathop{!}\varid{x}))}:
        By assumption \textsc{Step-App}.
      \item \ensuremath{\keyword{\forall}\!\! \hsforall \;\varid{v}\;\varcolor{\mu}_{2}\hsdot{\circ }{.\ }\varcolor{\mu}\progressto\varcolor{\mu}_{2}\implies\beta_\Traces ((\keyword{case}\;\varid{v}\;\keyword{of}\;\conid{Fun}\;\varid{g}\to \varid{g}\;(\varcolor{\rho}\mathop{!}\varid{x});\anonymous \to \varid{stuck})\;\varcolor{\mu}_{2})\mathbin{⊑}\widehat{\varid{apply}}\;(\beta_\Traces (\conid{Ret}\;(\varid{v},\varcolor{\mu}_{2})))\;(\beta_\Domain (\varcolor{\mu}) (\varcolor{\rho}\mathop{!}\varid{x}))}:
        By cases over \ensuremath{\varid{v}}.
        \begin{itemize}
          \item \textbf{Case \ensuremath{\varid{v}\mathrel{=}\conid{Stuck}}}:
            Then \ensuremath{\beta_\Traces (\varid{stuck}\;\varcolor{\mu}_{2})\mathrel{=}\widehat{\varid{stuck}}\mathbin{⊑}\widehat{\varid{apply}}\;\widehat{\varid{stuck}}\;\widehat{\varid{a}}} by assumption \textsc{Stuck-App}.
          \item \textbf{Case \ensuremath{\varid{v}\mathrel{=}\conid{Con}\;\varid{k}\;\varid{ds}}}:
            Then \ensuremath{\beta_\Traces (\varid{stuck}\;\varcolor{\mu}_{2})\mathrel{=}\widehat{\varid{stuck}}\mathbin{⊑}\widehat{\varid{apply}}\;(\widehat{\varid{con}}\;\varid{k}\;\widehat{\varid{ds}})\;\widehat{\varid{a}}} by assumption \textsc{Stuck-App}, for the suitable \ensuremath{\widehat{\varid{ds}}}.
          \item \textbf{Case \ensuremath{\varid{v}\mathrel{=}\conid{Fun}\;\varid{g}}}:
            Note that \ensuremath{\varid{g}} has a parametric definition, of the form \ensuremath{(\lambda \varid{d}\to \varid{step}\;\conid{App}_{2}\;(\mathcal{S}\denot{\varid{e}_{1}}_{\varcolor{\rho}[\varid{x}\mapsto\varid{d}]}))}.
            This is important to apply \textsc{Beta-App} below.
            \begin{hscode}\SaveRestoreHook
\column{B}{@{}>{\hspre}l<{\hspost}@{}}%
\column{15}{@{}>{\hspre}c<{\hspost}@{}}%
\column{15E}{@{}l@{}}%
\column{17}{@{}>{\hspre}l<{\hspost}@{}}%
\column{18}{@{}>{\hspre}l<{\hspost}@{}}%
\column{E}{@{}>{\hspre}l<{\hspost}@{}}%
\>[17]{}\beta_\Traces (\varid{g}\;(\varcolor{\rho}\mathop{!}\varid{x})\;\varcolor{\mu}_{2}){}\<[E]%
\\
\>[15]{}\mathbin{⊑}{}\<[15E]%
\>[18]{}\mbox{\commentbegin  \ensuremath{\varid{id}\mathbin{⊑}\gamma_\Domain (\varcolor{\mu}_{2})\hsdot{\circ }{.\ }\alpha_\Domain (\varcolor{\mu}_{2})}, rearrange  \commentend}{}\<[E]%
\\
\>[15]{}\hsindent{2}{}\<[17]%
\>[17]{}(\alpha_\Domain (\varcolor{\mu}_{2})\hsdot{\circ }{.\ }{\varid{g}}^{*}\hsdot{\circ }{.\ }\gamma_\Domain (\varcolor{\mu}_{2}))\;(\beta_\Domain (\varcolor{\mu}_{2}) (\varcolor{\rho}\mathop{!}\varid{x})){}\<[E]%
\\
\>[15]{}\mathbin{⊑}{}\<[15E]%
\>[18]{}\mbox{\commentbegin  \ensuremath{\beta_\Domain (\varcolor{\mu}_{2}) (\varcolor{\rho}\mathop{!}\varid{x})\mathbin{⊑}\beta_\Domain (\varcolor{\mu}) (\varcolor{\rho}\mathop{!}\varid{x})} by {thm:heap-progress-mono}  \commentend}{}\<[E]%
\\
\>[15]{}\hsindent{2}{}\<[17]%
\>[17]{}(\alpha_\Domain (\varcolor{\mu}_{2})\hsdot{\circ }{.\ }{\varid{g}}^{*}\hsdot{\circ }{.\ }\gamma_\Domain (\varcolor{\mu}_{2}))\;(\beta_\Domain (\varcolor{\mu}) (\varcolor{\rho}\mathop{!}\varid{x})){}\<[E]%
\\
\>[15]{}\mathbin{⊑}{}\<[15E]%
\>[18]{}\mbox{\commentbegin  Assumption \textsc{Beta-App}  \commentend}{}\<[E]%
\\
\>[15]{}\hsindent{2}{}\<[17]%
\>[17]{}\widehat{\varid{apply}}\;(\widehat{\varid{fun}}\;(\alpha_\Domain (\varcolor{\mu}_{2})\hsdot{\circ }{.\ }{\varid{g}}^{*}\hsdot{\circ }{.\ }\gamma_\Domain (\varcolor{\mu}_{2})))\;(\beta_\Domain (\varcolor{\mu}) (\varcolor{\rho}\mathop{!}\varid{x})){}\<[E]%
\\
\>[15]{}\mathrel{=}{}\<[15E]%
\>[18]{}\mbox{\commentbegin  Definition of \ensuremath{\beta_\Traces}, \ensuremath{\varid{v}}  \commentend}{}\<[E]%
\\
\>[15]{}\hsindent{2}{}\<[17]%
\>[17]{}\widehat{\varid{apply}}\;(\beta_\Traces (\conid{Ret}\;\varid{v},\varcolor{\mu}_{2}))\;(\beta_\Domain (\varcolor{\mu}) (\varcolor{\rho}\mathop{!}\varid{x})){}\<[E]%
\ColumnHook
\end{hscode}\resethooks
        \end{itemize}
    \end{itemize}

  \item \textbf{Case} \ensuremath{\conid{Case}\;\varid{e}\;\varid{alts}}:
    Very similar to the \ensuremath{\varid{select}} case in \Cref{thm:abstract-by-name}.

    The cases where the interpreter returns \ensuremath{\varid{stuck}} follow by parametricity.
    Otherwise, we have
    (for the suitable definition of \ensuremath{\widehat{\varid{alts}}}, which satisfies
    \ensuremath{\alpha_\Domain (\varcolor{\mu}_{2})\hsdot{\circ }{.\ }{(\varid{alts}\mathop{!}\varid{k})}^{*}\hsdot{\circ }{.\ }\varid{map}\;(\gamma_\Domain (\varcolor{\mu}_{2}))\mathbin{⊑}\widehat{\varid{alts}}\mathop{!}\varid{k}} by induction)
    \begin{hscode}\SaveRestoreHook
\column{B}{@{}>{\hspre}l<{\hspost}@{}}%
\column{5}{@{}>{\hspre}c<{\hspost}@{}}%
\column{5E}{@{}l@{}}%
\column{7}{@{}>{\hspre}l<{\hspost}@{}}%
\column{9}{@{}>{\hspre}l<{\hspost}@{}}%
\column{E}{@{}>{\hspre}l<{\hspost}@{}}%
\>[7]{}\beta_\Traces (\mathcal{S}_{\mathbf{need}}\denot{\conid{Case}\;\varid{e}\;\varid{alts}}_{\varcolor{\rho}}(\varcolor{\mu})){}\<[E]%
\\
\>[5]{}\mathrel{=}{}\<[5E]%
\>[9]{}\mbox{\commentbegin  Unfold \ensuremath{\mathcal{S}_{\mathbf{need}}\denot{\wild}_{\wild}}, \ensuremath{\beta_\Traces}, \ensuremath{\varid{apply}}  \commentend}{}\<[E]%
\\
\>[5]{}\hsindent{2}{}\<[7]%
\>[7]{}\varid{step}\;\conid{Case}_{1}\;((\mathcal{S}_{\mathbf{need}}\denot{\varid{e}}_{\varcolor{\rho}}(\wild)\bind \lambda \varid{v}\to \keyword{case}\;\varid{v}\;\keyword{of}\;\conid{Con}\;\varid{k}\;\varid{ds}\mid \varid{k}\in \varid{dom}\;\varid{alts}\to (\varid{alts}\mathop{!}\varid{k})\;\varid{ds};\anonymous \to \varid{stuck})\;\varcolor{\mu}){}\<[E]%
\\
\>[5]{}\mathbin{⊑}{}\<[5E]%
\>[9]{}\mbox{\commentbegin  Apply \Cref{thm:by-need-bind}; see below  \commentend}{}\<[E]%
\\
\>[5]{}\hsindent{2}{}\<[7]%
\>[7]{}\varid{step}\;\conid{Case}_{1}\;(\widehat{\varid{select}}\;(\mathcal{S}_{\widehat{\conid{D}}}\denot{\varid{e}}_{\beta_\Domain (\varcolor{\mu})\mathbin{\lhd}\varcolor{\rho}})\;\widehat{\varid{alts}}){}\<[E]%
\\
\>[5]{}\mathrel{=}{}\<[5E]%
\>[9]{}\mbox{\commentbegin  Refold \ensuremath{\mathcal{S}_{\widehat{\conid{D}}}\denot{\wild}_{\wild}}  \commentend}{}\<[E]%
\\
\>[5]{}\hsindent{2}{}\<[7]%
\>[7]{}\mathcal{S}_{\widehat{\conid{D}}}\denot{\varid{e}}_{\beta_\Domain (\varcolor{\mu})\mathbin{\lhd}\varcolor{\rho}}{}\<[E]%
\ColumnHook
\end{hscode}\resethooks

    In the $⊑$ step, we apply \Cref{thm:by-need-bind} under \ensuremath{\varid{step}\;\conid{Case}_{1}}, which
    yields three subgoals (under $\later$):
    \begin{itemize}
      \item \ensuremath{\beta_\Traces (\mathcal{S}_{\mathbf{need}}\denot{\varid{e}}_{\varcolor{\rho}}(\varcolor{\mu}))\mathbin{⊑}\mathcal{S}_{\widehat{\conid{D}}}\denot{\varid{e}}_{\beta_\Domain (\varcolor{\mu})\mathbin{\lhd}\varcolor{\rho}}}:
        By induction hypothesis.
      \item \ensuremath{\keyword{\forall}\!\! \hsforall \;\varid{ev}\;\widehat{\varid{d'}}\hsdot{\circ }{.\ }\widehat{\varid{step}}\;\varid{ev}\;(\widehat{\varid{select}}\;\widehat{\varid{d'}}\;\widehat{\varid{alts}})\mathbin{⊑}\widehat{\varid{select}}\;(\widehat{\varid{step}}\;\varid{ev}\;\widehat{\varid{d'}})\;\widehat{\varid{alts}}}:
        By assumption \textsc{Step-Select}.
      \item \ensuremath{\keyword{\forall}\!\! \hsforall \;\varid{v}\;\varcolor{\mu}_{2}\hsdot{\circ }{.\ }\varcolor{\mu}\progressto\varcolor{\mu}_{2}\implies\beta_\Traces ((\keyword{case}\;\varid{v}\;\keyword{of}\;\conid{Con}\;\varid{k}\;\varid{ds}\mid \varid{k}\in \varid{dom}\;\varid{alts}\to (\varid{alts}\mathop{!}\varid{k})\;\varid{ds};\anonymous \to \varid{stuck})\;\varcolor{\mu}_{2})\mathbin{⊑}\widehat{\varid{select}}\;(\beta_\Traces (\conid{Ret}\;(\varid{v},\varcolor{\mu}_{2})))\;\widehat{\varid{alts}}}:
        By cases over \ensuremath{\varid{v}}.
        \begin{itemize}
          \item \textbf{Case \ensuremath{\varid{v}\mathrel{=}\conid{Stuck}}}:
            Then \ensuremath{\beta_\Traces (\varid{stuck}\;\varcolor{\mu}_{2})\mathrel{=}\widehat{\varid{stuck}}\mathbin{⊑}\widehat{\varid{select}}\;\widehat{\varid{stuck}}\;\widehat{\varid{alts}}} by assumption \textsc{Stuck-Sel}.
          \item \textbf{Case \ensuremath{\varid{v}\mathrel{=}\conid{Fun}\;\varid{f}}}:
            Then \ensuremath{\beta_\Traces (\varid{stuck}\;\varcolor{\mu}_{2})\mathrel{=}\widehat{\varid{stuck}}\mathbin{⊑}\widehat{\varid{select}}\;(\widehat{\varid{fun}}\;\widehat{\varid{f}})\;\widehat{\varid{alts}}} by assumption \textsc{Stuck-Sel}, for the suitable \ensuremath{\widehat{\varid{f}}}.
          \item \textbf{Case \ensuremath{\varid{v}\mathrel{=}\conid{Con}\;\varid{k}\;\varid{ds}}, $\ensuremath{\varid{k}} \not∈ \ensuremath{\varid{dom}\;\varid{alts}}$}:
            Then \ensuremath{\beta_\Traces (\varid{stuck}\;\varcolor{\mu}_{2})\mathrel{=}\widehat{\varid{stuck}}\mathbin{⊑}\widehat{\varid{select}}\;(\widehat{\varid{con}}\;\varid{k}\;\widehat{\varid{ds}})\;\widehat{\varid{alts}}} by assumption \textsc{Stuck-Sel}, for the suitable \ensuremath{\widehat{\varid{ds}}}.
          \item \textbf{Case \ensuremath{\varid{v}\mathrel{=}\conid{Con}\;\varid{k}\;\varid{ds}}, $\ensuremath{\varid{k}} ∈ \ensuremath{\varid{dom}\;\varid{alts}}$}:
            Note that \ensuremath{\varid{alts}} has a parametric definition.
            This is important to apply \textsc{Beta-Sel} below.
            \begin{hscode}\SaveRestoreHook
\column{B}{@{}>{\hspre}l<{\hspost}@{}}%
\column{15}{@{}>{\hspre}c<{\hspost}@{}}%
\column{15E}{@{}l@{}}%
\column{17}{@{}>{\hspre}l<{\hspost}@{}}%
\column{18}{@{}>{\hspre}l<{\hspost}@{}}%
\column{E}{@{}>{\hspre}l<{\hspost}@{}}%
\>[17]{}\beta_\Traces ((\varid{alts}\mathop{!}\varid{k})\;\varid{ds}\;\varcolor{\mu}_{2}){}\<[E]%
\\
\>[15]{}\mathbin{⊑}{}\<[15E]%
\>[18]{}\mbox{\commentbegin  \ensuremath{\varid{id}\mathbin{⊑}\gamma_\Domain (\varcolor{\mu}_{2})\hsdot{\circ }{.\ }\alpha_\Domain (\varcolor{\mu}_{2})}, rearrange  \commentend}{}\<[E]%
\\
\>[15]{}\hsindent{2}{}\<[17]%
\>[17]{}(\alpha_\Domain (\varcolor{\mu}_{2})\hsdot{\circ }{.\ }{(\varid{alts}\mathop{!}\varid{k})}^{*}\hsdot{\circ }{.\ }\varid{map}\;(\gamma_\Domain (\varcolor{\mu}_{2})))\;(\varid{map}\;(\alpha_\Domain (\varcolor{\mu}_{2})\hsdot{\circ }{.\ }\{\wild\})\;\varid{ds}){}\<[E]%
\\
\>[15]{}\mathbin{⊑}{}\<[15E]%
\>[18]{}\mbox{\commentbegin  Abstraction property of \ensuremath{\widehat{\varid{alts}}}  \commentend}{}\<[E]%
\\
\>[15]{}\hsindent{2}{}\<[17]%
\>[17]{}(\widehat{\varid{alts}}\mathop{!}\varid{k})\;(\varid{map}\;(\alpha_\Domain (\varcolor{\mu}_{2})\hsdot{\circ }{.\ }\{\wild\})\;\varid{ds}){}\<[E]%
\\
\>[15]{}\mathbin{⊑}{}\<[15E]%
\>[18]{}\mbox{\commentbegin  Assumption \textsc{Beta-Sel}  \commentend}{}\<[E]%
\\
\>[15]{}\hsindent{2}{}\<[17]%
\>[17]{}\widehat{\varid{select}}\;(\widehat{\varid{con}}\;\varid{k}\;(\varid{map}\;(\alpha_\Domain (\varcolor{\mu}_{2})\hsdot{\circ }{.\ }\{\wild\})\;\varid{ds}))\;\widehat{\varid{alts}}{}\<[E]%
\\
\>[15]{}\mathrel{=}{}\<[15E]%
\>[18]{}\mbox{\commentbegin  Definition of \ensuremath{\beta_\Traces}, \ensuremath{\varid{v}}  \commentend}{}\<[E]%
\\
\>[15]{}\hsindent{2}{}\<[17]%
\>[17]{}\widehat{\varid{select}}\;(\beta_\Traces (\conid{Ret}\;\varid{v}))\;\widehat{\varid{alts}}{}\<[E]%
\ColumnHook
\end{hscode}\resethooks
        \end{itemize}
    \end{itemize}

  \item \textbf{Case} \ensuremath{\conid{Let}\;\varid{x}\;\varid{e}_{1}\;\varid{e}_{2}}:
    We can make one step to see
    \begin{hscode}\SaveRestoreHook
\column{B}{@{}>{\hspre}l<{\hspost}@{}}%
\column{7}{@{}>{\hspre}l<{\hspost}@{}}%
\column{E}{@{}>{\hspre}l<{\hspost}@{}}%
\>[7]{}\mathcal{S}_{\mathbf{need}}\denot{\conid{Let}\;\varid{x}\;\varid{e}_{1}\;\varid{e}_{2}}_{\varcolor{\rho}}(\varcolor{\mu})\mathrel{=}\conid{Step}\;\conid{Let}_{1}\;(\mathcal{S}_{\mathbf{need}}\denot{\varid{e}_{2}}_{\varcolor{\rho}_{1}}(\varcolor{\mu}_{1})),{}\<[E]%
\ColumnHook
\end{hscode}\resethooks
    where \ensuremath{\varcolor{\rho}_{1}\triangleq\varcolor{\rho}[\varid{x}\mapsto\varid{step}\;(\conid{Look}\;\varid{x})\;(\varid{fetch}\;\varid{a})]},
    \ensuremath{\varid{a}\triangleq\varid{nextFree}\;\varcolor{\mu}},
    \ensuremath{\varcolor{\mu}_{1}\triangleq\varcolor{\mu}[\varid{a}\mapsto\varid{memo}\;\varid{a}\;(\mathcal{S}_{\mathbf{need}}\denot{\varid{e}_{1}}_{\varcolor{\rho}_{1}})]}.

    Then \ensuremath{(\beta_\Domain (\varcolor{\mu}_{1})\mathbin{\lhd}\varcolor{\rho}_{1})\mathop{!}\varid{y}\mathbin{⊑}(\beta_\Domain (\varcolor{\mu})\mathbin{\lhd}\varcolor{\rho})\mathop{!}\varid{y}} whenever $\ensuremath{\varid{x}} \not= \ensuremath{\varid{y}}$
    by \Cref{thm:heap-progress-mono},
    and \ensuremath{(\beta_\Domain (\varcolor{\mu}_{1})\mathbin{\lhd}\varcolor{\rho}_{1})\mathop{!}\varid{x}\mathrel{=}\varid{step}\;(\conid{Look}\;\varid{x})\;(\beta_\Domain (\varcolor{\mu}_{1}) (\mathcal{S}_{\mathbf{need}}\denot{\varid{e}_{1}}_{\varcolor{\rho}_{1}}(\wild)))}.
    \begin{hscode}\SaveRestoreHook
\column{B}{@{}>{\hspre}l<{\hspost}@{}}%
\column{5}{@{}>{\hspre}c<{\hspost}@{}}%
\column{5E}{@{}l@{}}%
\column{9}{@{}>{\hspre}l<{\hspost}@{}}%
\column{15}{@{}>{\hspre}l<{\hspost}@{}}%
\column{19}{@{}>{\hspre}l<{\hspost}@{}}%
\column{E}{@{}>{\hspre}l<{\hspost}@{}}%
\>[9]{}\beta_\Traces (\mathcal{S}_{\mathbf{need}}\denot{\conid{Let}\;\varid{x}\;\varid{e}_{1}\;\varid{e}_{2}}_{\varcolor{\rho}}(\varcolor{\mu})){}\<[E]%
\\
\>[5]{}\mathrel{=}{}\<[5E]%
\>[9]{}\mbox{\commentbegin  Unfold \ensuremath{\mathcal{S}_{\mathbf{need}}\denot{\wild}_{\wild}}  \commentend}{}\<[E]%
\\
\>[9]{}\beta_\Traces (\varid{bind}\;{}\<[19]%
\>[19]{}(\lambda \varid{d}_{1}\to \mathcal{S}_{\mathbf{need}}\denot{\varid{e}_{1}}_{\varcolor{\rho}_{1}})\;(\lambda \varid{d}_{1}\to \conid{Step}\;\conid{Let}_{1}\;(\mathcal{S}_{\mathbf{need}}\denot{\varid{e}_{2}}_{\varcolor{\rho}_{1}}))\;\varcolor{\mu}){}\<[E]%
\\
\>[5]{}\mathrel{=}{}\<[5E]%
\>[9]{}\mbox{\commentbegin  Unfold \ensuremath{\varid{bind}}, $\ensuremath{\varid{a}} \not\in \ensuremath{\varid{dom}\;\varcolor{\mu}}$, unfold \ensuremath{\beta_\Traces}  \commentend}{}\<[E]%
\\
\>[9]{}\varid{step}\;\conid{Let}_{1}\;(\beta_\Traces (\mathcal{S}_{\mathbf{need}}\denot{\varid{e}_{2}}_{\varcolor{\rho}_{1}}(\varcolor{\mu}_{1}))){}\<[E]%
\\
\>[5]{}\mathbin{⊑}{}\<[5E]%
\>[9]{}\mbox{\commentbegin  Induction hypothesis  \commentend}{}\<[E]%
\\
\>[9]{}\varid{step}\;\conid{Let}_{1}\;(\mathcal{S}\denot{\varid{e}_{2}}_{\beta_\Domain (\varcolor{\mu}_{1})\mathbin{\lhd}\varcolor{\rho}_{1}}){}\<[E]%
\\
\>[5]{}\mathbin{⊑}{}\<[5E]%
\>[9]{}\mbox{\commentbegin  By \Cref{thm:by-need-env-unroll}, unfolding \ensuremath{\varcolor{\rho}_{1}}   \commentend}{}\<[E]%
\\
\>[9]{}\varid{step}\;\conid{Let}_{1}\;(\mathcal{S}\denot{\varid{e}_{2}}_{(\beta_\Domain (\varcolor{\mu}_{1})\mathbin{\lhd}\varcolor{\rho})[\varid{x}\mapsto\varid{step}\;(\conid{Look}\;\varid{x})\;(\mathcal{S}\denot{\varid{e}_{1}}_{(\beta_\Domain (\varcolor{\mu}_{1})\mathbin{\lhd}\varcolor{\rho})[\varid{x}\mapsto\beta_\Domain (\varcolor{\mu}_{1}) (\varcolor{\rho}_{1}\mathop{!}\varid{x})]})]}){}\<[E]%
\\
\>[5]{}\mathbin{⊑}{}\<[5E]%
\>[9]{}\mbox{\commentbegin  Least fixpoint  \commentend}{}\<[E]%
\\
\>[9]{}\varid{step}\;\conid{Let}_{1}\;(\mathcal{S}\denot{\varid{e}_{2}}_{(\beta_\Domain (\varcolor{\mu}_{1})\mathbin{\lhd}\varcolor{\rho})[\varid{x}\mapsto\varid{lfp}\;(\lambda \widehat{\varid{d}_{1}}\to \varid{step}\;(\conid{Look}\;\varid{x})\;(\mathcal{S}\denot{\varid{e}_{1}}_{(\beta_\Domain (\varcolor{\mu}_{1})\mathbin{\lhd}\varcolor{\rho})[\varid{x}\mapsto\widehat{\varid{d}_{1}}]}))]}){}\<[E]%
\\
\>[5]{}\mathbin{⊑}{}\<[5E]%
\>[9]{}\mbox{\commentbegin  \ensuremath{\beta_\Domain (\varcolor{\mu}_{1}) (\varcolor{\rho}\mathop{!}\varid{x})\mathbin{⊑}\beta_\Domain (\varcolor{\mu}) (\varcolor{\rho}\mathop{!}\varid{x})} by \Cref{thm:heap-progress-mono}  \commentend}{}\<[E]%
\\
\>[9]{}\varid{step}\;\conid{Let}_{1}\;(\mathcal{S}\denot{\varid{e}_{2}}_{(\beta_\Domain (\varcolor{\mu})\mathbin{\lhd}\varcolor{\rho})[\varid{x}\mapsto\varid{lfp}\;(\lambda \widehat{\varid{d}_{1}}\to \varid{step}\;(\conid{Look}\;\varid{x})\;(\mathcal{S}\denot{\varid{e}_{1}}_{(\beta_\Domain (\varcolor{\mu})\mathbin{\lhd}\varcolor{\rho})[\varid{x}\mapsto\widehat{\varid{d}_{1}}]}))]}){}\<[E]%
\\
\>[5]{}\mathrel{=}{}\<[5E]%
\>[9]{}\mbox{\commentbegin  Partially unroll fixpoint  \commentend}{}\<[E]%
\\
\>[9]{}\varid{step}\;\conid{Let}_{1}\;(\mathcal{S}\denot{\varid{e}_{2}}_{(\beta_\Domain (\varcolor{\mu})\mathbin{\lhd}\varcolor{\rho})[\varid{x}\mapsto\varid{step}\;(\conid{Look}\;\varid{x})\;(\varid{lfp}\;(\lambda \widehat{\varid{d}_{1}}\to \mathcal{S}\denot{\varid{e}_{1}}_{(\beta_\Domain (\varcolor{\mu})\mathbin{\lhd}\varcolor{\rho})[\varid{x}\mapsto\varid{step}\;(\conid{Look}\;\varid{x})\;\widehat{\varid{d}_{1}}]}))]}){}\<[E]%
\\
\>[5]{}\mathbin{⊑}{}\<[5E]%
\>[9]{}\mbox{\commentbegin  Assumption \textsc{ByName-Bind}, with \ensuremath{\widehat{\varcolor{\rho}}\mathrel{=}\beta_\Domain (\varcolor{\mu})\mathbin{\lhd}\varcolor{\rho}}  \commentend}{}\<[E]%
\\
\>[9]{}\varid{bind}\;{}\<[15]%
\>[15]{}(\lambda \varid{d}_{1}\to \mathcal{S}\denot{\varid{e}_{1}}_{(\beta_\Domain (\varcolor{\mu})\mathbin{\lhd}\varcolor{\rho})[\varid{x}\mapsto\varid{step}\;(\conid{Look}\;\varid{x})\;\varid{d}_{1}]})\;{}\<[E]%
\\
\>[15]{}(\lambda \varid{d}_{1}\to \varid{step}\;\conid{Let}_{1}\;(\mathcal{S}\denot{\varid{e}_{2}}_{(\beta_\Domain (\varcolor{\mu})\mathbin{\lhd}\varcolor{\rho})[\varid{x}\mapsto\varid{step}\;(\conid{Look}\;\varid{x})\;\varid{d}_{1}]})){}\<[E]%
\\
\>[5]{}\mathrel{=}{}\<[5E]%
\>[9]{}\mbox{\commentbegin  Refold \ensuremath{\mathcal{S}\denot{\conid{Let}\;\varid{x}\;\varid{e}_{1}\;\varid{e}_{2}}_{\beta_\Domain (\varcolor{\mu})\mathbin{\lhd}\varcolor{\rho}}}  \commentend}{}\<[E]%
\\
\>[9]{}\mathcal{S}\denot{\conid{Let}\;\varid{x}\;\varid{e}_{1}\;\varid{e}_{2}}_{\beta_\Domain (\varcolor{\mu})\mathbin{\lhd}\varcolor{\rho}}{}\<[E]%
\ColumnHook
\end{hscode}\resethooks
\end{itemize}
\end{proof}

Let us unpack law $\textsc{Beta-App}$ to see how the abstraction laws in
\Cref{fig:abstraction-laws} are to be understood.
To prove $\textsc{Beta-App}$, one has to show that
\ensuremath{\keyword{\forall}\!\! \hsforall \;\varid{f}\;\varid{a}\;\varid{x}\hsdot{\circ }{.\ }\varid{f}\;\varid{a}\mathbin{⊑}\varid{apply}\;(\varid{fun}\;\varid{x}\;\varid{f})\;\varid{a}} in the abstract domain \ensuremath{\widehat{\conid{D}}}.
This states that summarising \ensuremath{\varid{f}} through \ensuremath{\varid{fun}}, then \ensuremath{\varid{apply}}ing the summary to
\ensuremath{\varid{a}} must approximate a direct call to \ensuremath{\varid{f}};
it amounts to proving correct the summary mechanism.
In \Cref{sec:problem}, we have proved a substitution \Cref{thm:absence-subst},
which is a syntactic form of this statement.
The ``$\ensuremath{\varid{f}}\text{ polymorphic}$'' premise asserts that \ensuremath{\varid{f}} is definable at
polymorphic type \ensuremath{\keyword{\forall}\!\! \hsforall \;\varid{d}\hsdot{\circ }{.\ }(\conid{Trace}\;\varid{d},\conid{Domain}\;\varid{d},\conid{HasBind}\;\varid{d})\Rightarrow \varid{d}\to \varid{d}}, which is
important to prove \textsc{Beta-App} (in \Cref{sec:mod-subst}).

Law \textsc{Beta-Sel} states a similar property for data constructor redexes.
Law \textsc{ByName-Bind} expresses that the abstract \ensuremath{\varid{bind}} implementation must
be sound for by-name evaluation, that is, it must approximate passing the least
fixpoint \ensuremath{\varid{lfp}} of the \ensuremath{\varid{rhs}} functional to \ensuremath{\varid{body}}.
The remaining laws are congruence rules involving \ensuremath{\varid{step}} and \ensuremath{\varid{stuck}} as well as
a monotonicity requirement for all involved operations.
These laws follow the mantra ``evaluation improves approximation''; for
example, law \textsc{Stuck-App} expresses that applying a stuck term
or constructor evaluates to (and thus approximates) a stuck term, and
\textsc{Stuck-Sel} expresses the same for \ensuremath{\varid{select}} stack frames.
In the Appendix, we show a result similar to \Cref{thm:abstract-by-need}
for by-name evaluation which does not require the by-need specific laws
\textsc{Step-Inc} and \textsc{Update}.

Note that none of the laws mention the concrete semantics or the abstraction
function $α_{\mathcal{S}}$.
This is how fixing the concrete semantics and $α_{\mathcal{S}}$ pays off; the usual
abstraction laws such as
\ensuremath{\alpha_\Domain (\varcolor{\mu}) (\varid{apply}\;\varid{d}\;\varid{a})\mathbin{⊑}\widehat{\varid{apply}}\;(\alpha_\Domain (\varcolor{\mu}) (\varid{d}))\;(\alpha_\Domain (\varcolor{\mu})\;(\varid{a}))} further
decompose into \textsc{Beta-App}.
We think this is a nice advantage to our approach, because the author of
the analysis does not need to reason about by-need heaps in order to soundly
approximate a semantic trace property expressed via \ensuremath{\conid{Trace}} instance!

\begin{toappendix}
\subsection{Parametricity and Relationship to Kripke-style Logical Relations}

We remarked right at the begin of the previous subsection that the Galois
connection in \Cref{fig:abstract-name-need} is really a family of definitions
indexed by a heap \ensuremath{\varcolor{\mu}}.
It is not possible to regard the ``abstraction of a \ensuremath{\varid{d}}'' in isolation;
rather, \Cref{thm:heap-progress-mono} expresses that once an ``abstraction of a \ensuremath{\varid{d}}''
holds for a particular heap \ensuremath{\varcolor{\mu}_{1}}, this abstraction will hold for any heap \ensuremath{\varcolor{\mu}_{2}}
that the semantics may progress to.

Unfortunately, this indexing also means that we cannot apply parametricity
to prove the sound abstraction \Cref{thm:abstract-by-need}, as we did for
by-name abstraction.
Such a proof would be bound to fail whenever the heap is extended (in \ensuremath{\varid{bind}}),
because then the index of the soundness relation must change as well.
Concretely, we would need roughly the following free theorem
\[
  \ensuremath{(\varid{bind},\varid{bind})\in \later\!\;(\later\!\;({R_{\varcolor{\mu}[\varid{a}\mapsto\varid{d}]}})\to {R_{\varcolor{\mu}[\varid{a}\mapsto\varid{d}]}})\to (\later\!\;({R_{\varcolor{\mu}[\varid{a}\mapsto\varid{d}]}})\to {R_{\varcolor{\mu}[\varid{a}\mapsto\varid{d}]}})\to {R_{\varcolor{\mu}}}}
\]
for the soundness relation of \Cref{thm:abstract-by-need}
\[
  R_\ensuremath{\varcolor{\mu}}(\ensuremath{\varid{d}}, \ensuremath{\widehat{\varid{d}}}) \triangleq \ensuremath{\beta_\Domain (\varcolor{\mu}) (\varid{d})\mathbin{⊑}\widehat{\varid{d}}}.
\]
However, parametricity only yields
\[
  \ensuremath{(\varid{bind},\varid{bind})\in ({R_{\varcolor{\mu}}}\to {R_{\varcolor{\mu}}})\to ({R_{\varcolor{\mu}}}\to {R_{\varcolor{\mu}}})\to {R_{\varcolor{\mu}}}}
\]
We think that a modular proof is still conceivable by defining a custom proof
tactic that would be \emph{inspired} by parametricity, given a means for
annotating how the heap index changes in \ensuremath{\varid{bind}}.

Although we do not formally acknowledge this, the soundness relation \ensuremath{{R_{\varcolor{\mu}}}}
of \Cref{thm:abstract-by-need} is reminiscent of a \emph{Kripke logical
relation}~\citep{Ahmed:04}.
In this analogy, definable heaps correspond to the \emph{possible worlds} of
\citet{Kripke:63} with heap progression \ensuremath{(\progressto)} as the \emph{accessibility
relation}.
\Cref{thm:heap-progress-mono} states that the relation $R_\ensuremath{\varcolor{\mu}}$ is monotonic
\wrt \ensuremath{(\progressto)}, so we consider it possible to define a Kripke-style logical
relation over System $F$ types.

Kripke-style logical relations are well-understood in the literature, hence it
is conceivable that a modular proof technique just as for parametricity exists.
We have not investigated this avenue so far.
A modular proof would help our proof framework to scale up to a by-need
semantics of Haskell, for example, so this avenue bears great potential.
\end{toappendix}

\subsection{A Modular Proof for \textsc{Beta-App}: A Simpler Substitution Lemma}
\label{sec:mod-subst}

\begin{toappendix}
\subsection{Usage Analysis Proofs}

Here we give the usage analysis proofs for the main body, often deferring to
\Cref{sec:by-need-soundness}.
\end{toappendix}

In order to instantiate \Cref{thm:abstract-by-need} for usage analysis in
\Cref{sec:usage-sound}, we need to prove in particular that \ensuremath{\concolor{\mathsf{D_U}}} satisfies the
abstraction law \textsc{Beta-App} in \Cref{fig:abstraction-laws}.
\textsc{Beta-App} corresponds to the syntactic substitution
\Cref{thm:absence-subst} of \Cref{sec:problem}, and this subsection presents its
proof.

Before we discuss this proof, note that the proof for
\Cref{thm:absence-subst} has a serious drawback: It relies on knowing the
complete definition of $\semabs{\wild}$ and thus is \emph{non-modular}.
As a result, the proof complexity scales in the size of the interpreter, and
whenever the definition of $\semabs{\wild}$ changes, \Cref{thm:absence-subst}
must be updated.
The complexity of such non-modular proofs would become unmanageable for large
denotational interpreters such as for WebAssembly~\citep{Brandl:23}.

For \textsc{Beta-App}, dubbed \emph{semantic substitution}, the proof fares much
better:
\begin{toappendix}
\begin{abbreviation}[Field access]
  \ensuremath{\langle \varid{φ'}, \varid{v'} \rangle.\varcolor{\varphi}\triangleq\varid{φ'}}, \ensuremath{\langle \varid{φ'}, \varid{v'} \rangle.\varid{v}\mathrel{=}\varid{v'}}.
\end{abbreviation}
\end{toappendix}

\begin{lemmarep}[\textsc{Beta-App}, Semantic substitution]
\label{thm:usage-subst-sem}
Let \ensuremath{\varid{f}\mathbin{::}(\conid{Trace}\;\varid{d},\conid{Domain}\;\varid{d},\conid{HasBind}\;\varid{d})\Rightarrow \varid{d}\to \varid{d}}, \ensuremath{\varid{x}\mathbin{::}\conid{Name}} fresh and \ensuremath{\varid{a}\mathbin{::}\concolor{\mathsf{D_U}}}.
Then \ensuremath{\varid{f}\;\varid{a}\mathbin{⊑}\varid{apply}\;(\varid{fun}\;\varid{x}\;\varid{f})\;\varid{a}} in \ensuremath{\concolor{\mathsf{D_U}}}.
\end{lemmarep}
\begin{proof}
We instantiate the free theorem for \ensuremath{\varid{f}}
\[
  \forall A, B.\
  \forall R ⊆ A \times B.\
  \forall (\mathit{inst_1}, \mathit{inst_2}) ∈ \mathsf{Dict}(R).\
  \forall (d_1,d_2) ∈ R.\
  (f_A(\mathit{inst_1})(d_1), f_B(\mathit{inst_2})(d_2)) ∈ R
\]
as follows
\[\begin{array}{c}
  A \triangleq B \triangleq \ensuremath{\concolor{\mathsf{D_U}}}, \qquad \mathit{inst_1} \triangleq \mathit{inst_2} \triangleq \mathit{inst}, \qquad d_1 \triangleq a, \qquad d_2 \triangleq \mathit{pre}(x) \\
  R_{x,a}(d_1,d_2) \triangleq \forall g.\ d_1 = g(a) \land d_2 = g(\mathit{pre}(x)) \implies g(a) ⊑ \mathit{apply}(\mathit{fun}(x,g),a)  \\
\end{array}\]
and get (translated back into Haskell)
\[
  \inferrule
    { (\ensuremath{\varid{a}},\ensuremath{\varid{pre}\;\varid{x}}) ∈ R_{\ensuremath{\varid{x}},\ensuremath{\varid{a}}} \\ (\mathit{inst}, \mathit{inst}) ∈ \mathsf{Dict}(R_{\ensuremath{\varid{x}},\ensuremath{\varid{a}}}) }
    { (\ensuremath{\varid{f}\;\varid{a}}, \ensuremath{\varid{f}\;(\varid{pre}\;\varid{x})}) ∈ R_{\ensuremath{\varid{x}},\ensuremath{\varid{a}}} }
\]
where \ensuremath{\varid{pre}\;\varid{x}\triangleq\langle [\varid{x}\mapsto\concolor{\mathsf{U_1}}], \conid{Rep}\;\concolor{\mathsf{U_\omega}} \rangle} defines the proxy for \ensuremath{\varid{x}},
exactly as in the implementation of \ensuremath{\varid{fun}\;\varid{x}}, and $\mathit{inst}$ is the canonical
instance dictionary for \ensuremath{\concolor{\mathsf{D_U}}}.

We will first apply this inference rule and then show that the goal follows from
$(\ensuremath{\varid{f}\;\varid{a}}, \ensuremath{\varid{f}\;(\varid{pre}\;\varid{x})}) ∈ R_{\ensuremath{\varid{x}},\ensuremath{\varid{a}}}$.

To apply the inference rule, we must prove its premises.
Before we do so, it is very helpful to eliminate the quantification over
arbitrary \ensuremath{\varid{g}} in the relation $R_{x,a}(d_1,d_2)$.
To that end, we first need to factor \ensuremath{\varid{fun}\;\varid{x}\;\varid{g}\mathrel{=}\varid{abs}\;\varid{x}\;(\varid{g}\;(\varid{pre}\;\varid{x}))}, where \ensuremath{\varid{abs}}
is defined as follows:
\begin{hscode}\SaveRestoreHook
\column{B}{@{}>{\hspre}l<{\hspost}@{}}%
\column{3}{@{}>{\hspre}l<{\hspost}@{}}%
\column{E}{@{}>{\hspre}l<{\hspost}@{}}%
\>[3]{}\varid{abs}\;\varid{x}\;\langle \varcolor{\varphi}, \varid{v} \rangle\mathrel{=}\langle \varcolor{\varphi}[\varid{x}\mapsto\concolor{\mathsf{U_0}}], \varcolor{\varphi}\mathbin{!?}\varid{x} \argcons \varid{v} \rangle{}\<[E]%
\ColumnHook
\end{hscode}\resethooks
And we simplify $R_{\ensuremath{\varid{x}},\ensuremath{\varid{a}}}(d_1,d_2)$, thus
\begin{hscode}\SaveRestoreHook
\column{B}{@{}>{\hspre}l<{\hspost}@{}}%
\column{3}{@{}>{\hspre}l<{\hspost}@{}}%
\column{E}{@{}>{\hspre}l<{\hspost}@{}}%
\>[3]{}\keyword{\forall}\!\! \hsforall \;\varid{g}\hsdot{\circ }{.\ }\varid{d}_{1}\mathrel{=}\varid{g}\;\varid{a}\land\varid{d}_{2}\mathrel{=}\varid{g}\;(\varid{pre}\;\varid{x})\implies\varid{g}\;\varid{a}\mathbin{⊑}\varid{apply}\;(\varid{fun}\;\varid{x}\;\varid{g})\;\varid{a}{}\<[E]%
\\
\>[B]{}\Longleftrightarrow\mbox{\commentbegin  \ensuremath{\varid{fun}\;\varid{x}\;\varid{g}\mathrel{=}\varid{abs}\;\varid{x}\;(\varid{g}\;(\varid{pre}\;\varid{x}))}  \commentend}{}\<[E]%
\\
\>[B]{}\hsindent{3}{}\<[3]%
\>[3]{}\keyword{\forall}\!\! \hsforall \;\varid{g}\hsdot{\circ }{.\ }\varid{d}_{1}\mathrel{=}\varid{g}\;\varid{a}\land\varid{d}_{2}\mathrel{=}\varid{g}\;(\varid{pre}\;\varid{x})\implies\varid{g}\;\varid{a}\mathbin{⊑}\varid{apply}\;(\varid{abs}\;\varid{x}\;(\varid{g}\;(\varid{pre}\;\varid{x})))\;\varid{a}{}\<[E]%
\\
\>[B]{}\Longleftrightarrow\mbox{\commentbegin  Use \ensuremath{\varid{d}_{1}\mathrel{=}\varid{g}\;\varid{a}} and \ensuremath{\varid{d}_{2}\mathrel{=}\varid{g}\;(\varid{pre}\;\varid{x})}  \commentend}{}\<[E]%
\\
\>[B]{}\hsindent{3}{}\<[3]%
\>[3]{}\keyword{\forall}\!\! \hsforall \;\varid{g}\hsdot{\circ }{.\ }\varid{d}_{1}\mathrel{=}\varid{g}\;\varid{a}\land\varid{d}_{2}\mathrel{=}\varid{g}\;(\varid{pre}\;\varid{x})\implies\varid{d}_{1}\mathbin{⊑}\varid{apply}\;(\varid{abs}\;\varid{x}\;\varid{d}_{2})\;\varid{a}{}\<[E]%
\\
\>[B]{}\Longleftrightarrow\mbox{\commentbegin  There exists a \ensuremath{\varid{g}} satisfying \ensuremath{\varid{d}_{1}\mathrel{=}\varid{g}\;\varid{a}} and \ensuremath{\varid{d}_{2}\mathrel{=}\varid{g}\;(\varid{pre}\;\varid{x})}  \commentend}{}\<[E]%
\\
\>[B]{}\hsindent{3}{}\<[3]%
\>[3]{}\varid{d}_{1}\mathbin{⊑}\varid{apply}\;(\varid{abs}\;\varid{x}\;\varid{d}_{2})\;\varid{a}{}\<[E]%
\\
\>[B]{}\Longleftrightarrow\mbox{\commentbegin  Inline \ensuremath{\varid{apply}}, \ensuremath{\varid{abs}}, simplify  \commentend}{}\<[E]%
\\
\>[B]{}\hsindent{3}{}\<[3]%
\>[3]{}\varid{d}_{1}\mathbin{⊑}\keyword{let}\;\langle \varcolor{\varphi}, \varid{v} \rangle\mathrel{=}\varid{d}_{2}\;\keyword{in}\;\langle \varcolor{\varphi}[\varid{x}\mapsto\concolor{\mathsf{U_0}}]\mathbin{+}(\varcolor{\varphi}\mathbin{!?}\varid{x})\mathbin{*}\varid{a}.\varcolor{\varphi}, \varid{v} \rangle{}\<[E]%
\ColumnHook
\end{hscode}\resethooks

Note that this implies \ensuremath{\varid{d}_{1}.\varcolor{\varphi}\mathbin{!?}\varid{x}\mathrel{=}\concolor{\mathsf{U_0}}}, because \ensuremath{\varcolor{\varphi}[\varid{x}\mapsto\concolor{\mathsf{U_0}}]\mathbin{!?}\varid{x}\mathrel{=}\concolor{\mathsf{U_0}}}
and \ensuremath{\varid{a}.\varcolor{\varphi}\mathbin{!?}\varid{x}\mathrel{=}\concolor{\mathsf{U_0}}} by the scoping discipline.

It turns out that $R_{\ensuremath{\varid{x}},\ensuremath{\varid{a}}}$ is reflexive on all \ensuremath{\varid{d}} for which \ensuremath{\varid{d}.\varcolor{\varphi}\mathbin{?!}\varid{x}\mathrel{=}\concolor{\mathsf{U_0}}}; indeed, then the inequality becomes an equality.
(This corresponds to summarising a function that does not use its
argument.)
That is a fact that we need in the \ensuremath{\varid{stuck}}, \ensuremath{\varid{fun}}, \ensuremath{\varid{con}} and \ensuremath{\varid{select}} cases
below, so we prove it here:
\begin{hscode}\SaveRestoreHook
\column{B}{@{}>{\hspre}l<{\hspost}@{}}%
\column{3}{@{}>{\hspre}l<{\hspost}@{}}%
\column{E}{@{}>{\hspre}l<{\hspost}@{}}%
\>[3]{}\keyword{\forall}\!\! \hsforall \;\varid{d}\hsdot{\circ }{.\ }\varid{d}\mathbin{⊑}\langle (\varid{d}.\varcolor{\varphi})[\varid{x}\mapsto\concolor{\mathsf{U_0}}]\mathbin{+}(\varid{d}.\varcolor{\varphi}\mathbin{!?}\varid{x})\mathbin{*}\varid{a}.\varcolor{\varphi}, \varid{d}.\varid{v} \rangle{}\<[E]%
\\
\>[B]{}\Longleftrightarrow\mbox{\commentbegin  Use \ensuremath{(\varid{d}.\varcolor{\varphi}\mathbin{?!}\varid{x})\mathrel{=}\concolor{\mathsf{U_0}}}  \commentend}{}\<[E]%
\\
\>[B]{}\hsindent{3}{}\<[3]%
\>[3]{}\keyword{\forall}\!\! \hsforall \;\varid{d}\hsdot{\circ }{.\ }\varid{d}\mathbin{⊑}\langle \varid{d}.\varcolor{\varphi}, \varid{d}.\varid{v} \rangle\mathrel{=}\varid{d}{}\<[E]%
\ColumnHook
\end{hscode}\resethooks
The last proposition is reflexivity on $⊑$.

Now we prove the premises of the abstraction theorem:
\begin{itemize}
  \item $(\ensuremath{\varid{a}},\ensuremath{\varid{pre}\;\varid{x}}) ∈ R_{\ensuremath{\varid{x}},\ensuremath{\varid{a}}}$:
    The proposition unfolds to
    \begin{hscode}\SaveRestoreHook
\column{B}{@{}>{\hspre}l<{\hspost}@{}}%
\column{5}{@{}>{\hspre}l<{\hspost}@{}}%
\column{7}{@{}>{\hspre}l<{\hspost}@{}}%
\column{E}{@{}>{\hspre}l<{\hspost}@{}}%
\>[7]{}\varid{a}\mathbin{⊑}\keyword{let}\;\langle \varcolor{\varphi}, \varid{v} \rangle\mathrel{=}\varid{pre}\;\varid{x}\;\keyword{in}\;\langle \varcolor{\varphi}[\varid{x}\mapsto\concolor{\mathsf{U_0}}]\mathbin{+}(\varcolor{\varphi}\mathbin{!?}\varid{x})\mathbin{*}\varid{a}.\varcolor{\varphi}, \varid{v} \rangle{}\<[E]%
\\
\>[5]{}\Longleftrightarrow\mbox{\commentbegin  Unfold \ensuremath{\varid{pre}}, simplify  \commentend}{}\<[E]%
\\
\>[5]{}\hsindent{2}{}\<[7]%
\>[7]{}\varid{a}\mathbin{⊑}\langle \varid{a}.\varcolor{\varphi}, \conid{Rep}\;\concolor{\mathsf{U_\omega}} \rangle{}\<[E]%
\ColumnHook
\end{hscode}\resethooks
    The latter follows from \ensuremath{\varid{a}.\varid{v}\mathbin{⊑}\conid{Rep}\;\concolor{\mathsf{U_\omega}}} because \ensuremath{\conid{Rep}\;\concolor{\mathsf{U_\omega}}} is the Top element.

  \item $(\mathit{inst}, \mathit{inst}) ∈ \mathsf{Dict}(R_{\ensuremath{\varid{x}},\ensuremath{\varid{a}}})$:
    By the relational interpretation of products, we get one subgoal per instance method.
    \begin{itemize}
      \item \textbf{Case \ensuremath{\varid{step}}}.
        Goal: $\inferrule{(\ensuremath{\varid{d}_{1}},\ensuremath{\varid{d}_{2}}) ∈ R_{\ensuremath{\varid{x}},\ensuremath{\varid{a}}}}{(\ensuremath{\varid{step}\;\varid{ev}\;\varid{d}_{1}}, \ensuremath{\varid{step}\;\varid{ev}\;\varid{d}_{2}}) ∈ R_{\ensuremath{\varid{x}},\ensuremath{\varid{a}}}}$. \\
        Assume the premise $(\ensuremath{\varid{d}_{1}},\ensuremath{\varid{d}_{2}}) ∈ R_{\ensuremath{\varid{x}},\ensuremath{\varid{a}}}$, show the goal.
        All cases other than \ensuremath{\varid{ev}\mathrel{=}\conid{Look}\;\varid{y}} are trivial, because then \ensuremath{\varid{step}\;\varid{ev}\;\varid{d}\mathrel{=}\varid{d}} and the goal follows by the premise.
        So let \ensuremath{\varid{ev}\mathrel{=}\conid{Look}\;\varid{y}}. The goal is to show
        \begin{hscode}\SaveRestoreHook
\column{B}{@{}>{\hspre}l<{\hspost}@{}}%
\column{11}{@{}>{\hspre}l<{\hspost}@{}}%
\column{E}{@{}>{\hspre}l<{\hspost}@{}}%
\>[11]{}\varid{step}\;(\conid{Look}\;\varid{y})\;\varid{d}_{1}\mathbin{⊑}\keyword{let}\;\langle \varcolor{\varphi}, \varid{v} \rangle\mathrel{=}\varid{step}\;(\conid{Look}\;\varid{y})\;\varid{d}_{2}\;\keyword{in}\;\langle \varcolor{\varphi}[\varid{x}\mapsto\concolor{\mathsf{U_0}}]\mathbin{+}(\varcolor{\varphi}\mathbin{!?}\varid{x})\mathbin{*}\varid{a}.\varcolor{\varphi}, \varid{v} \rangle{}\<[E]%
\ColumnHook
\end{hscode}\resethooks
        We begin by unpacking the assumption $(\ensuremath{\varid{d}_{1}},\ensuremath{\varid{d}_{2}}) ∈ R_{\ensuremath{\varid{x}},\ensuremath{\varid{a}}}$ to show it:
        \begin{hscode}\SaveRestoreHook
\column{B}{@{}>{\hspre}l<{\hspost}@{}}%
\column{9}{@{}>{\hspre}l<{\hspost}@{}}%
\column{11}{@{}>{\hspre}l<{\hspost}@{}}%
\column{15}{@{}>{\hspre}l<{\hspost}@{}}%
\column{E}{@{}>{\hspre}l<{\hspost}@{}}%
\>[11]{}\varid{d}_{1}\mathbin{⊑}\keyword{let}\;\langle \varcolor{\varphi}, \varid{v} \rangle\mathrel{=}\varid{d}_{2}\;\keyword{in}\;\langle \varcolor{\varphi}[\varid{x}\mapsto\concolor{\mathsf{U_0}}]\mathbin{+}(\varcolor{\varphi}\mathbin{!?}\varid{x})\mathbin{*}\varid{a}.\varcolor{\varphi}, \varid{v} \rangle{}\<[E]%
\\
\>[9]{}\implies{}\<[15]%
\>[15]{}\mbox{\commentbegin  \ensuremath{\varid{step}\;(\conid{Look}\;\varid{y})} is monotonic  \commentend}{}\<[E]%
\\
\>[9]{}\hsindent{2}{}\<[11]%
\>[11]{}\varid{step}\;(\conid{Look}\;\varid{y})\;\varid{d}_{1}\mathbin{⊑}\varid{step}\;(\conid{Look}\;\varid{y})\;\langle (\varid{d}_{2}.\varcolor{\varphi})[\varid{x}\mapsto\concolor{\mathsf{U_0}}]\mathbin{+}(\varid{d}_{2}.\varcolor{\varphi}\mathbin{!?}\varid{x})\mathbin{*}\varid{a}.\varcolor{\varphi}, \varid{d}_{2}.\varid{v} \rangle{}\<[E]%
\\
\>[9]{}\Longleftrightarrow\mbox{\commentbegin  Refold \ensuremath{\varid{step}\;(\conid{Look}\;\varid{y})}  \commentend}{}\<[E]%
\\
\>[9]{}\hsindent{2}{}\<[11]%
\>[11]{}\varid{step}\;(\conid{Look}\;\varid{y})\;\varid{d}_{1}\mathbin{⊑}\langle (\varid{d}_{2}.\varcolor{\varphi})[\varid{x}\mapsto\concolor{\mathsf{U_0}}]\mathbin{+}[\varid{y}\mapsto\concolor{\mathsf{U_1}}]\mathbin{+}(\varid{d}_{2}.\varcolor{\varphi}\mathbin{!?}\varid{x})\mathbin{*}\varid{a}.\varcolor{\varphi}, \varid{d}_{2}.\varid{v} \rangle{}\<[E]%
\\
\>[9]{}\Longleftrightarrow{}\<[15]%
\>[15]{}\mbox{\commentbegin  \ensuremath{\varid{step}\;(\conid{Look}\;\varid{y})} preserves value and $\ensuremath{\varid{x}} \not= \ensuremath{\varid{y}}$ because \ensuremath{\varid{y}} is let-bound  \commentend}{}\<[E]%
\\
\>[9]{}\hsindent{2}{}\<[11]%
\>[11]{}\varid{step}\;(\conid{Look}\;\varid{y})\;\varid{d}_{1}\mathbin{⊑}\keyword{let}\;\langle \varcolor{\varphi}, \varid{v} \rangle\mathrel{=}\varid{step}\;(\conid{Look}\;\varid{y})\;\varid{d}_{2}\;\keyword{in}\;\langle \varcolor{\varphi}[\varid{x}\mapsto\concolor{\mathsf{U_0}}]\mathbin{+}(\varcolor{\varphi}\mathbin{!?}\varid{x})\mathbin{*}\varid{a}.\varcolor{\varphi}, \varid{v} \rangle{}\<[E]%
\ColumnHook
\end{hscode}\resethooks

      \item \textbf{Case \ensuremath{\varid{stuck}}}.
        Goal: $(\ensuremath{\varid{stuck}}, \ensuremath{\varid{stuck}}) ∈ R_{\ensuremath{\varid{x}},\ensuremath{\varid{a}}}$ \\
        Follows from reflexivity, because \ensuremath{\varid{stuck}\mathrel{=}\bot}, and \ensuremath{\bot.\varcolor{\varphi}\mathbin{!?}\varid{x}\mathrel{=}\concolor{\mathsf{U_0}}}.

      \item \textbf{Case \ensuremath{\varid{fun}}}.
        Goal: $\inferrule{\forall (\ensuremath{\varid{d}_{1}},\ensuremath{\varid{d}_{2}}) ∈ R_{\ensuremath{\varid{x}},\ensuremath{\varid{a}}} \implies (\ensuremath{\varid{f}_{1}\;\varid{d}_{1}}, \ensuremath{\varid{f}_{2}\;\varid{d}_{2}}) ∈ R_{\ensuremath{\varid{x}},\ensuremath{\varid{a}}}}{(\ensuremath{\varid{fun}\;\varid{y}\;\varid{f}_{1}}, \ensuremath{\varid{fun}\;\varid{y}\;\varid{f}_{2}}) ∈ R_{\ensuremath{\varid{x}},\ensuremath{\varid{a}}}}$. \\
        Additionally, we may assume $\ensuremath{\varid{x}} \not= \ensuremath{\varid{y}}$ by lexical scoping.

        Now assume the premise. The goal is to show
        \begin{hscode}\SaveRestoreHook
\column{B}{@{}>{\hspre}l<{\hspost}@{}}%
\column{11}{@{}>{\hspre}l<{\hspost}@{}}%
\column{E}{@{}>{\hspre}l<{\hspost}@{}}%
\>[11]{}\varid{fun}\;\varid{y}\;\varid{f}_{1}\mathbin{⊑}\keyword{let}\;\langle \varcolor{\varphi}, \varid{v} \rangle\mathrel{=}\varid{fun}\;\varid{y}\;\varid{f}_{2}\;\keyword{in}\;\langle \varcolor{\varphi}[\varid{x}\mapsto\concolor{\mathsf{U_0}}]\mathbin{+}(\varcolor{\varphi}\mathbin{!?}\varid{x})\mathbin{*}\varid{a}.\varcolor{\varphi}, \varid{v} \rangle{}\<[E]%
\ColumnHook
\end{hscode}\resethooks

        Recall that \ensuremath{\varid{fun}\;\varid{y}\;\varid{f}\mathrel{=}\varid{abs}\;\varid{y}\;(\varid{f}\;(\varid{pre}\;\varid{y}))} and that \ensuremath{\varid{abs}\;\varid{y}} is monotonic.

        Note that we have $(\ensuremath{\varid{pre}\;\varid{y}}, \ensuremath{\varid{pre}\;\varid{y}}) ∈ R_{\ensuremath{\varid{x}},\ensuremath{\varid{a}}}$ because of $\ensuremath{\varid{x}} \not= \ensuremath{\varid{y}}$ and reflexivity.
        That in turn yields $(\ensuremath{\varid{f}_{1}\;(\varid{pre}\;\varid{y}),\varid{f}_{2}\;(\varid{pre}\;\varid{y})}) ∈ R_{\ensuremath{\varid{x}},\ensuremath{\varid{a}}}$ by assumption.
        This is useful to kick-start the following proof, showing the goal:
        \begin{hscode}\SaveRestoreHook
\column{B}{@{}>{\hspre}l<{\hspost}@{}}%
\column{9}{@{}>{\hspre}c<{\hspost}@{}}%
\column{9E}{@{}l@{}}%
\column{11}{@{}>{\hspre}l<{\hspost}@{}}%
\column{15}{@{}>{\hspre}l<{\hspost}@{}}%
\column{E}{@{}>{\hspre}l<{\hspost}@{}}%
\>[11]{}\varid{f}_{1}\;(\varid{pre}\;\varid{y})\mathbin{⊑}\keyword{let}\;\langle \varcolor{\varphi}, \varid{v} \rangle\mathrel{=}\varid{f}_{2}\;(\varid{pre}\;\varid{y})\;\keyword{in}\;\langle \varcolor{\varphi}[\varid{x}\mapsto\concolor{\mathsf{U_0}}]\mathbin{+}(\varcolor{\varphi}\mathbin{!?}\varid{x})\mathbin{*}\varid{a}.\varcolor{\varphi}, \varid{v} \rangle{}\<[E]%
\\
\>[9]{}\implies{}\<[9E]%
\>[15]{}\mbox{\commentbegin  Monotonicity of \ensuremath{\varid{abs}\;\varid{y}}  \commentend}{}\<[E]%
\\
\>[9]{}\hsindent{2}{}\<[11]%
\>[11]{}\varid{abs}\;\varid{y}\;(\varid{f}_{1}\;(\varid{pre}\;\varid{y}))\mathbin{⊑}\varid{abs}\;\varid{y}\;(\keyword{let}\;\langle \varcolor{\varphi}, \varid{v} \rangle\mathrel{=}\varid{f}_{2}\;(\varid{pre}\;\varid{y})\;\keyword{in}\;\langle \varcolor{\varphi}[\varid{x}\mapsto\concolor{\mathsf{U_0}}]\mathbin{+}(\varcolor{\varphi}\mathbin{!?}\varid{x})\mathbin{*}\varid{a}.\varcolor{\varphi}, \varid{v} \rangle){}\<[E]%
\\
\>[9]{}\Longleftrightarrow{}\<[9E]%
\>[15]{}\mbox{\commentbegin  $\ensuremath{\varid{x}} \not= \ensuremath{\varid{y}}$ and \ensuremath{\varid{a}.\varcolor{\varphi}\mathbin{!?}\varid{y}\mathrel{=}\concolor{\mathsf{U_0}}} due to scoping, \ensuremath{\varcolor{\varphi}\mathbin{!?}\varid{x}} unaffected by floating \ensuremath{\varid{abs}}  \commentend}{}\<[E]%
\\
\>[9]{}\hsindent{2}{}\<[11]%
\>[11]{}\varid{abs}\;\varid{y}\;(\varid{f}_{1}\;(\varid{pre}\;\varid{y}))\mathbin{⊑}\keyword{let}\;\langle \varcolor{\varphi}, \varid{v} \rangle\mathrel{=}\varid{abs}\;\varid{y}\;(\varid{f}_{2}\;(\varid{pre}\;\varid{y}))\;\keyword{in}\;\langle \varcolor{\varphi}[\varid{x}\mapsto\concolor{\mathsf{U_0}}]\mathbin{+}(\varcolor{\varphi}\mathbin{!?}\varid{x})\mathbin{*}\varid{a}.\varcolor{\varphi}, \varid{v} \rangle{}\<[E]%
\\
\>[9]{}\Longleftrightarrow{}\<[9E]%
\>[15]{}\mbox{\commentbegin  Rewrite \ensuremath{\varid{abs}\;\varid{y}\;(\varid{f}\;(\varid{pre}\;\varid{y}))\mathrel{=}\varid{fun}\;\varid{y}\;\varid{f}}  \commentend}{}\<[E]%
\\
\>[9]{}\hsindent{2}{}\<[11]%
\>[11]{}\varid{fun}\;\varid{y}\;\varid{f}_{1}\mathbin{⊑}\keyword{let}\;\langle \varcolor{\varphi}, \varid{v} \rangle\mathrel{=}\varid{fun}\;\varid{y}\;\varid{f}_{2}\;\keyword{in}\;\langle \varcolor{\varphi}[\varid{x}\mapsto\concolor{\mathsf{U_0}}]\mathbin{+}(\varcolor{\varphi}\mathbin{!?}\varid{x})\mathbin{*}\varid{a}.\varcolor{\varphi}, \varid{v} \rangle{}\<[E]%
\ColumnHook
\end{hscode}\resethooks

      \item \textbf{Case \ensuremath{\varid{apply}}}.
        Goal: $\inferrule{(\ensuremath{\varid{l}_{1}},\ensuremath{\varid{l}_{2}}) ∈ R_{\ensuremath{\varid{x}},\ensuremath{\varid{a}}} \\ (\ensuremath{\varid{r}_{1}},\ensuremath{\varid{r}_{2}}) ∈ R_{\ensuremath{\varid{x}},\ensuremath{\varid{a}}}}{(\ensuremath{\varid{apply}\;\varid{l}_{1}\;\varid{r}_{1}}, \ensuremath{\varid{apply}\;\varid{l}_{2}\;\varid{r}_{2}}) ∈ R_{\ensuremath{\varid{x}},\ensuremath{\varid{a}}}}$. \\
        Assume the premises. The goal is to show
        \begin{hscode}\SaveRestoreHook
\column{B}{@{}>{\hspre}l<{\hspost}@{}}%
\column{11}{@{}>{\hspre}l<{\hspost}@{}}%
\column{E}{@{}>{\hspre}l<{\hspost}@{}}%
\>[11]{}\varid{apply}\;\varid{l}_{1}\;\varid{r}_{1}\mathbin{⊑}\keyword{let}\;\langle \varcolor{\varphi}, \varid{v} \rangle\mathrel{=}\varid{apply}\;\varid{l}_{2}\;\varid{r}_{2}\;\keyword{in}\;\langle \varcolor{\varphi}[\varid{x}\mapsto\concolor{\mathsf{U_0}}]\mathbin{+}(\varcolor{\varphi}\mathbin{!?}\varid{x})\mathbin{*}\varid{a}.\varcolor{\varphi}, \varid{v} \rangle{}\<[E]%
\ColumnHook
\end{hscode}\resethooks

        \begin{hscode}\SaveRestoreHook
\column{B}{@{}>{\hspre}l<{\hspost}@{}}%
\column{9}{@{}>{\hspre}c<{\hspost}@{}}%
\column{9E}{@{}l@{}}%
\column{11}{@{}>{\hspre}l<{\hspost}@{}}%
\column{12}{@{}>{\hspre}l<{\hspost}@{}}%
\column{18}{@{}>{\hspre}l<{\hspost}@{}}%
\column{E}{@{}>{\hspre}l<{\hspost}@{}}%
\>[11]{}\varid{apply}\;\varid{l}_{1}\;\varid{r}_{1}{}\<[E]%
\\
\>[9]{}\mathbin{⊑}{}\<[9E]%
\>[12]{}\mbox{\commentbegin  \ensuremath{\varid{l}_{1}\mathbin{⊑}\varid{apply}\;(\varid{abs}\;\varid{x}\;\varid{l}_{2})}, \ensuremath{\varid{r}_{2}\mathbin{⊑}\varid{apply}\;(\varid{abs}\;\varid{x}\;\varid{r}_{2})}, monotonicity  \commentend}{}\<[E]%
\\
\>[9]{}\hsindent{2}{}\<[11]%
\>[11]{}\varid{apply}\;{}\<[18]%
\>[18]{}(\keyword{let}\;\langle \varcolor{\varphi}, \varid{v} \rangle\mathrel{=}\varid{l}_{2}\;\keyword{in}\;\langle \varcolor{\varphi}[\varid{x}\mapsto\concolor{\mathsf{U_0}}]\mathbin{+}(\varcolor{\varphi}\mathbin{!?}\varid{x})\mathbin{*}\varid{a}.\varcolor{\varphi}, \varid{v} \rangle)\;{}\<[E]%
\\
\>[18]{}(\keyword{let}\;\langle \varcolor{\varphi}, \varid{v} \rangle\mathrel{=}\varid{r}_{2}\;\keyword{in}\;\langle \varcolor{\varphi}[\varid{x}\mapsto\concolor{\mathsf{U_0}}]\mathbin{+}(\varcolor{\varphi}\mathbin{!?}\varid{x})\mathbin{*}\varid{a}.\varcolor{\varphi}, \varid{v} \rangle){}\<[E]%
\\
\>[9]{}\mathbin{⊑}{}\<[9E]%
\>[12]{}\mbox{\commentbegin  Componentwise, see below  \commentend}{}\<[E]%
\\
\>[9]{}\hsindent{2}{}\<[11]%
\>[11]{}\keyword{let}\;\langle \varcolor{\varphi}, \varid{v} \rangle\mathrel{=}\varid{apply}\;\varid{l}_{2}\;\varid{r}_{2}\;\keyword{in}\;\langle \varcolor{\varphi}[\varid{x}\mapsto\concolor{\mathsf{U_0}}]\mathbin{+}(\varcolor{\varphi}\mathbin{!?}\varid{x})\mathbin{*}\varid{a}.\varcolor{\varphi}, \varid{v} \rangle{}\<[E]%
\ColumnHook
\end{hscode}\resethooks

        For the last step, we show the inequality for \ensuremath{\varcolor{\varphi}} and \ensuremath{\varid{v}} independently.
        For values, it is easy to see by calculation that the value is
        \ensuremath{\varid{v}\triangleq\varid{snd}\;(\varid{peel}\;\varid{l}_{2}.\varid{v})} in both cases.
        The proof for the \ensuremath{\conid{Uses}} component is quite algebraic;
        we will abbreviate \ensuremath{\varid{u}\triangleq\varid{fst}\;(\varid{peel}\;\varid{l}_{2}.\varid{v})}:
        \begin{hscode}\SaveRestoreHook
\column{B}{@{}>{\hspre}l<{\hspost}@{}}%
\column{9}{@{}>{\hspre}c<{\hspost}@{}}%
\column{9E}{@{}l@{}}%
\column{11}{@{}>{\hspre}l<{\hspost}@{}}%
\column{12}{@{}>{\hspre}l<{\hspost}@{}}%
\column{19}{@{}>{\hspre}l<{\hspost}@{}}%
\column{E}{@{}>{\hspre}l<{\hspost}@{}}%
\>[11]{}(\varid{apply}\;{}\<[19]%
\>[19]{}(\keyword{let}\;\langle \varcolor{\varphi}, \varid{v} \rangle\mathrel{=}\varid{l}_{2}\;\keyword{in}\;\langle \varcolor{\varphi}[\varid{x}\mapsto\concolor{\mathsf{U_0}}]\mathbin{+}(\varcolor{\varphi}\mathbin{!?}\varid{x})\mathbin{*}\varid{a}.\varcolor{\varphi}, \varid{v} \rangle)\;{}\<[E]%
\\
\>[19]{}(\keyword{let}\;\langle \varcolor{\varphi}, \varid{v} \rangle\mathrel{=}\varid{r}_{2}\;\keyword{in}\;\langle \varcolor{\varphi}[\varid{x}\mapsto\concolor{\mathsf{U_0}}]\mathbin{+}(\varcolor{\varphi}\mathbin{!?}\varid{x})\mathbin{*}\varid{a}.\varcolor{\varphi}, \varid{v} \rangle)).\varcolor{\varphi}{}\<[E]%
\\
\>[9]{}\mathrel{=}{}\<[9E]%
\>[12]{}\mbox{\commentbegin  Unfold \ensuremath{\varid{apply}}  \commentend}{}\<[E]%
\\
\>[9]{}\hsindent{2}{}\<[11]%
\>[11]{}(\varid{l}_{2}.\varcolor{\varphi})[\varid{x}\mapsto\concolor{\mathsf{U_0}}]\mathbin{+}(\varid{l}_{2}.\varcolor{\varphi}\mathbin{!?}\varid{x})\mathbin{*}\varid{a}.\varcolor{\varphi}\mathbin{+}\varid{u}\mathbin{*}((\varid{r}_{2}.\varcolor{\varphi})[\varid{x}\mapsto\concolor{\mathsf{U_0}}]\mathbin{+}(\varid{r}_{2}.\varcolor{\varphi}\mathbin{!?}\varid{x})\mathbin{*}\varid{a}.\varcolor{\varphi}){}\<[E]%
\\
\>[9]{}\mathrel{=}{}\<[9E]%
\>[12]{}\mbox{\commentbegin  Distribute \ensuremath{\varid{u}\mathbin{*}(\varcolor{\varphi}_{1}\mathbin{+}\varcolor{\varphi}_{2})\mathrel{=}\varid{u}\mathbin{*}\varcolor{\varphi}_{1}\mathbin{+}\varid{u}\mathbin{*}\varcolor{\varphi}_{2}}  \commentend}{}\<[E]%
\\
\>[9]{}\hsindent{2}{}\<[11]%
\>[11]{}(\varid{l}_{2}.\varcolor{\varphi})[\varid{x}\mapsto\concolor{\mathsf{U_0}}]\mathbin{+}(\varid{l}_{2}.\varcolor{\varphi}\mathbin{!?}\varid{x})\mathbin{*}\varid{a}.\varcolor{\varphi}\mathbin{+}\varid{u}\mathbin{*}(\varid{r}_{2}.\varcolor{\varphi})[\varid{x}\mapsto\concolor{\mathsf{U_0}}]\mathbin{+}\varid{u}\mathbin{*}(\varid{r}_{2}.\varcolor{\varphi}\mathbin{!?}\varid{x})\mathbin{*}\varid{a}.\varcolor{\varphi}{}\<[E]%
\\
\>[9]{}\mathrel{=}{}\<[9E]%
\>[12]{}\mbox{\commentbegin  Commute  \commentend}{}\<[E]%
\\
\>[9]{}\hsindent{2}{}\<[11]%
\>[11]{}(\varid{l}_{2}.\varcolor{\varphi})[\varid{x}\mapsto\concolor{\mathsf{U_0}}]\mathbin{+}\varid{u}\mathbin{*}(\varid{r}_{2}.\varcolor{\varphi})[\varid{x}\mapsto\concolor{\mathsf{U_0}}]\mathbin{+}(\varid{l}_{2}.\varcolor{\varphi}\mathbin{!?}\varid{x})\mathbin{*}\varid{a}.\varcolor{\varphi}\mathbin{+}\varid{u}\mathbin{*}(\varid{r}_{2}.\varcolor{\varphi}\mathbin{!?}\varid{x})\mathbin{*}\varid{a}.\varcolor{\varphi}{}\<[E]%
\\
\>[9]{}\mathrel{=}{}\<[9E]%
\>[12]{}\mbox{\commentbegin  \ensuremath{\varcolor{\varphi}_{1}[\varid{x}\mapsto\concolor{\mathsf{U_0}}]\mathbin{+}\varcolor{\varphi}_{2}[\varid{x}\mapsto\concolor{\mathsf{U_0}}]\mathrel{=}(\varcolor{\varphi}_{1}\mathbin{+}\varcolor{\varphi}_{2})[\varid{x}\mapsto\concolor{\mathsf{U_0}}]}, \ensuremath{\varid{u}\mathbin{*}\varcolor{\varphi}_{1}\mathbin{+}\varid{u}\mathbin{*}\varcolor{\varphi}_{2}\mathrel{=}\varid{u}\mathbin{*}(\varcolor{\varphi}_{1}\mathbin{+}\varcolor{\varphi}_{2})}  \commentend}{}\<[E]%
\\
\>[9]{}\hsindent{2}{}\<[11]%
\>[11]{}(\varid{l}_{2}.\varcolor{\varphi}\mathbin{+}\varid{u}\mathbin{*}\varid{r}_{2}.\varcolor{\varphi})[\varid{x}\mapsto\concolor{\mathsf{U_0}}]\mathbin{+}((\varid{l}_{2}.\varcolor{\varphi}\mathbin{+}\varid{u}\mathbin{*}\varid{r}_{2}.\varcolor{\varphi})\mathbin{!?}\varid{x})\mathbin{*}\varid{a}.\varcolor{\varphi}{}\<[E]%
\\
\>[9]{}\mathrel{=}{}\<[9E]%
\>[12]{}\mbox{\commentbegin  Refold \ensuremath{\varid{apply}}  \commentend}{}\<[E]%
\\
\>[9]{}\hsindent{2}{}\<[11]%
\>[11]{}\keyword{let}\;\langle \varcolor{\varphi}, \anonymous  \rangle\mathrel{=}\varid{apply}\;\varid{l}_{2}\;\varid{r}_{2}\;\keyword{in}\;\varcolor{\varphi}[\varid{x}\mapsto\concolor{\mathsf{U_0}}]\mathbin{+}(\varcolor{\varphi}\mathbin{!?}\varid{x})\mathbin{*}\varid{a}.\varcolor{\varphi}{}\<[E]%
\ColumnHook
\end{hscode}\resethooks

      \item \textbf{Case \ensuremath{\varid{con}}}.
        Goal: $\inferrule{\many{(\ensuremath{\varid{d}_{1}},\ensuremath{\varid{d}_{2}}) ∈ R_{\ensuremath{\varid{x}},\ensuremath{\varid{a}}}}}{(\ensuremath{\varid{con}\;\varid{k}\;(\many{\varid{d}_{1}})}, \ensuremath{\varid{con}\;\varid{k}\;(\many{\varid{d}_{2}})}) ∈ R_{\ensuremath{\varid{x}},\ensuremath{\varid{a}}}}$. \\
        We have shown that \ensuremath{\varid{apply}} is compatible with $R_{\ensuremath{\varid{x}},\ensuremath{\varid{a}}}$, and \ensuremath{\varid{foldl}}
        is so as well by parametricity.
        The field denotations \ensuremath{\many{\varid{d}_{1}}} and \ensuremath{\many{\varid{d}_{2}}} satisfy $R_{\ensuremath{\varid{x}},\ensuremath{\varid{a}}}$ by
        assumption; hence to show the goal it is sufficient to show that
        $(\ensuremath{\langle \varcolor{\varepsilon}, \conid{Rep}\;\concolor{\mathsf{U_\omega}} \rangle}, \ensuremath{\langle \varcolor{\varepsilon}, \conid{Rep}\;\concolor{\mathsf{U_\omega}} \rangle}) ∈ R_{\ensuremath{\varid{x}},\ensuremath{\varid{a}}}$.
        And that follows by reflexivity since \ensuremath{\varcolor{\varepsilon}\mathbin{?!}\varid{x}\mathrel{=}\concolor{\mathsf{U_0}}}.

      \item \textbf{Case \ensuremath{\varid{select}}}.
        Goal: $\inferrule{(\ensuremath{\varid{d}_{1}},\ensuremath{\varid{d}_{2}}) ∈ R_{\ensuremath{\varid{x}},\ensuremath{\varid{a}}} \\ (\ensuremath{\varid{fs}_{1}},\ensuremath{\varid{fs}_{2}}) ∈ \ensuremath{\conid{Tag}\mathbin{:\rightharpoonup}([\mskip1.5mu {R_{\varcolor{x},\varcolor{a}}}\mskip1.5mu]\to {R_{\varcolor{x},\varcolor{a}}})}}{(\ensuremath{\varid{select}\;\varid{d}_{1}\;\varid{fs}_{1}}, \ensuremath{\varid{select}\;\varid{d}_{2}\;\varid{fs}_{2}}) ∈ R_{\ensuremath{\varid{x}},\ensuremath{\varid{a}}}}$. \\
        Similar to the \ensuremath{\varid{con}} case, large parts of the implementation are
        compatible with \ensuremath{{R_{\varcolor{x},\varcolor{a}}}} already.
        With $(\ensuremath{\langle \varcolor{\varepsilon}, \conid{Rep}\;\concolor{\mathsf{U_\omega}} \rangle}, \ensuremath{\langle \varcolor{\varepsilon}, \conid{Rep}\;\concolor{\mathsf{U_\omega}} \rangle}) ∈ R_{\ensuremath{\varid{x}},\ensuremath{\varid{a}}}$ proved
        in the \ensuremath{\varid{con}} case, it remains to be shown that \ensuremath{\varid{lub}\mathbin{::}[\mskip1.5mu \concolor{\mathsf{D_U}}\mskip1.5mu]\to \concolor{\mathsf{D_U}}} and
        \ensuremath{(\sequ )\mathbin{::}\concolor{\mathsf{D_U}}\to \concolor{\mathsf{D_U}}\to \concolor{\mathsf{D_U}}} preserve \ensuremath{{R_{\varcolor{x},\varcolor{a}}}}.
        The proof for \ensuremath{(\sequ )} is very similar to but simpler than the \ensuremath{\varid{apply}}
        case, where a subexpression similar to \ensuremath{\langle \varcolor{\varphi}_{1}\mathbin{+}\varcolor{\varphi}_{2}, \varid{b} \rangle} occurs.
        The proof for \ensuremath{\varid{lub}} follows from the proof for the least upper bound
        operator \ensuremath{\mathbin{⊔}}.

        So let \ensuremath{(\varid{l}_{1},\varid{l}_{2}),(\varid{r}_{1},\varid{r}_{2})\in {R_{\varcolor{x},\varcolor{a}}}} and show that \ensuremath{(\varid{l}_{1}\mathbin{⊔}\varid{r}_{1},\varid{l}_{2}\mathbin{⊔}\varid{r}_{2})\in {R_{\varcolor{x},\varcolor{a}}}}.
        The assumptions imply that \ensuremath{\varid{l}_{1}.\varid{v}\mathbin{⊑}\varid{l}_{2}.\varid{v}} and \ensuremath{\varid{r}_{1}.\varid{v}\mathbin{⊑}\varid{r}_{2}.\varid{v}}, so
        \ensuremath{(\varid{l}_{1}\mathbin{⊔}\varid{r}_{1}).\varid{v}\mathbin{⊑}(\varid{l}_{2}\mathbin{⊔}\varid{r}_{2}).\varid{v}} follows by properties of least upper bound operators.

        Let us now consider the \ensuremath{\conid{Uses}} component.
        The goal is to show
        \begin{hscode}\SaveRestoreHook
\column{B}{@{}>{\hspre}l<{\hspost}@{}}%
\column{11}{@{}>{\hspre}l<{\hspost}@{}}%
\column{E}{@{}>{\hspre}l<{\hspost}@{}}%
\>[11]{}(\varid{l}_{1}\mathbin{⊔}\varid{r}_{1}).\varcolor{\varphi}\mathbin{⊑}(\keyword{let}\;\langle \varcolor{\varphi}, \varid{v} \rangle\mathrel{=}\varid{l}_{2}\mathbin{⊔}\varid{r}_{2}\;\keyword{in}\;\langle \varcolor{\varphi}[\varid{x}\mapsto\concolor{\mathsf{U_0}}]\mathbin{+}(\varcolor{\varphi}\mathbin{!?}\varid{x})\mathbin{*}\varid{a}.\varcolor{\varphi}, \varid{v} \rangle).\varcolor{\varphi}{}\<[E]%
\ColumnHook
\end{hscode}\resethooks

        For the proof, we need the algebraic identity \ensuremath{\keyword{\forall}\!\! \hsforall \;\varid{a}\;\varid{b}\;\varid{c}\;\varid{d}\hsdot{\circ }{.\ }\varid{a}\mathbin{+}\varid{c}\mathbin{⊔}\varid{b}\mathbin{+}\varid{d}\mathbin{⊑}\varid{a}\mathbin{⊔}\varid{b}\mathbin{+}\varid{c}\mathbin{⊔}\varid{d}} in \ensuremath{\conid{U}}.
        This can be proved by exhaustive enumeration of all 81 cases; the
        inequality is proper when \ensuremath{\varid{a}\mathrel{=}\varid{d}\mathrel{=}\concolor{\mathsf{U_1}}} and \ensuremath{\varid{b}\mathrel{=}\varid{c}\mathrel{=}\concolor{\mathsf{U_0}}} (or vice versa).
        Thus we conclude the proof:
        \begin{hscode}\SaveRestoreHook
\column{B}{@{}>{\hspre}l<{\hspost}@{}}%
\column{9}{@{}>{\hspre}c<{\hspost}@{}}%
\column{9E}{@{}l@{}}%
\column{11}{@{}>{\hspre}l<{\hspost}@{}}%
\column{12}{@{}>{\hspre}l<{\hspost}@{}}%
\column{E}{@{}>{\hspre}l<{\hspost}@{}}%
\>[11]{}(\varid{l}_{1}\mathbin{⊔}\varid{r}_{1}).\varcolor{\varphi}\mathrel{=}\varid{l}_{1}.\varcolor{\varphi}\mathbin{⊔}\varid{r}_{1}.\varcolor{\varphi}{}\<[E]%
\\
\>[9]{}\mathrel{=}{}\<[9E]%
\>[12]{}\mbox{\commentbegin  By assumption, \ensuremath{\varid{l}_{1}\mathbin{⊑}\varid{apply}\;(\varid{abs}\;\varid{x}\;\varid{l}_{2})} and \ensuremath{\varid{r}_{1}\mathbin{⊑}\varid{apply}\;(\varid{abs}\;\varid{x}\;\varid{r}_{2})}; monotonicity  \commentend}{}\<[E]%
\\
\>[9]{}\hsindent{2}{}\<[11]%
\>[11]{}((\varid{l}_{2}.\varcolor{\varphi})[\varid{x}\mapsto\concolor{\mathsf{U_0}}]\mathbin{+}(\varid{l}_{2}.\varcolor{\varphi}\mathbin{!?}\varid{x})\mathbin{*}\varid{a}.\varcolor{\varphi})\mathbin{⊔}((\varid{r}_{2}.\varcolor{\varphi})[\varid{x}\mapsto\concolor{\mathsf{U_0}}]\mathbin{+}(\varid{r}_{2}.\varcolor{\varphi}\mathbin{!?}\varid{x})\mathbin{*}\varid{a}.\varcolor{\varphi}){}\<[E]%
\\
\>[9]{}\mathbin{⊑}{}\<[9E]%
\>[12]{}\mbox{\commentbegin  Follows from \ensuremath{\keyword{\forall}\!\! \hsforall \;\varid{a}\;\varid{b}\;\varid{c}\;\varid{d}\hsdot{\circ }{.\ }\varid{a}\mathbin{+}\varid{c}\mathbin{⊔}\varid{b}\mathbin{+}\varid{d}\mathbin{⊑}\varid{a}\mathbin{⊔}\varid{b}\mathbin{+}\varid{c}\mathbin{⊔}\varid{d}} in \ensuremath{\conid{U}}  \commentend}{}\<[E]%
\\
\>[9]{}\hsindent{2}{}\<[11]%
\>[11]{}((\varid{l}_{2}.\varcolor{\varphi})[\varid{x}\mapsto\concolor{\mathsf{U_0}}]\mathbin{⊔}(\varid{r}_{2}.\varcolor{\varphi})[\varid{x}\mapsto\concolor{\mathsf{U_0}}])\mathbin{+}((\varid{l}_{2}.\varcolor{\varphi}\mathbin{!?}\varid{x})\mathbin{*}\varid{a}.\varcolor{\varphi}\mathbin{⊔}(\varid{r}_{2}.\varcolor{\varphi}\mathbin{!?}\varid{x})\mathbin{*}\varid{a}.\varcolor{\varphi}){}\<[E]%
\\
\>[9]{}\mathrel{=}{}\<[9E]%
\>[12]{}\mbox{\commentbegin  \ensuremath{\varcolor{\varphi}_{1}[\varid{x}\mapsto\concolor{\mathsf{U_0}}]\mathbin{⊔}\varcolor{\varphi}_{2}[\varid{x}\mapsto\concolor{\mathsf{U_0}}]\mathrel{=}(\varcolor{\varphi}_{1}\mathbin{⊔}\varcolor{\varphi}_{2})[\varid{x}\mapsto\concolor{\mathsf{U_0}}]}  \commentend}{}\<[E]%
\\
\>[9]{}\hsindent{2}{}\<[11]%
\>[11]{}((\varid{l}_{2}\mathbin{⊔}\varid{r}_{2}).\varcolor{\varphi})[\varid{x}\mapsto\concolor{\mathsf{U_0}}]\mathbin{+}((\varid{l}_{2}\mathbin{⊔}\varid{r}_{2}).\varcolor{\varphi}\mathbin{!?}\varid{x})\mathbin{*}\varid{a}.\varcolor{\varphi}{}\<[E]%
\\
\>[9]{}\mathrel{=}{}\<[9E]%
\>[12]{}\mbox{\commentbegin  Refold \ensuremath{\langle \varcolor{\varphi}, \varid{v} \rangle}  \commentend}{}\<[E]%
\\
\>[9]{}\hsindent{2}{}\<[11]%
\>[11]{}(\keyword{let}\;\langle \varcolor{\varphi}, \varid{v} \rangle\mathrel{=}\varid{l}_{2}\mathbin{⊔}\varid{r}_{2}\;\keyword{in}\;\langle \varcolor{\varphi}[\varid{x}\mapsto\concolor{\mathsf{U_0}}]\mathbin{+}(\varcolor{\varphi}\mathbin{!?}\varid{x})\mathbin{*}\varid{a}.\varcolor{\varphi}, \varid{v} \rangle).\varcolor{\varphi}{}\<[E]%
\ColumnHook
\end{hscode}\resethooks

      \item \textbf{Case \ensuremath{\varid{bind}}}.
        Goal: $\inferrule{\forall (\ensuremath{\varid{d}_{1}},\ensuremath{\varid{d}_{2}}) ∈ R_{\ensuremath{\varid{x}},\ensuremath{\varid{a}}} \implies (\ensuremath{\varid{f}_{1}\;\varid{d}_{1}}, \ensuremath{\varid{f}_{2}\;\varid{d}_{2}}), (\ensuremath{\varid{g}_{1}\;\varid{d}_{1}}, \ensuremath{\varid{g}_{2}\;\varid{d}_{2}}) ∈ R_{\ensuremath{\varid{x}},\ensuremath{\varid{a}}}}{(\ensuremath{\varid{bind}\;\varid{f}_{1}\;\varid{g}_{1}}, \ensuremath{\varid{bind}\;\varid{f}_{2}\;\varid{g}_{2}}) ∈ R_{\ensuremath{\varid{x}},\ensuremath{\varid{a}}}}$. \\
        By the assumptions, the definition \ensuremath{\varid{bind}\;\varid{f}\;\varid{g}\mathrel{=}\varid{g}\;(\varid{kleeneFix}\;\varid{f})}
        preserves \ensuremath{{R_{\varcolor{x},\varcolor{a}}}} if \ensuremath{\varid{kleeneFix}} does.
        Since \ensuremath{\varid{kleeneFix}\mathbin{::}\conid{Lat}\;\varid{a}\Rightarrow (\varid{a}\to \varid{a})\to \varid{a}} is parametric, it suffices
        to show that the instance of \ensuremath{\conid{Lat}} preserves \ensuremath{{R_{\varcolor{x},\varcolor{a}}}}.
        We have already shown that \ensuremath{\mathbin{⊔}} preserves \ensuremath{{R_{\varcolor{x},\varcolor{a}}}}, and we have also shown
        that \ensuremath{\varid{stuck}\mathrel{=}\bot} preserves \ensuremath{{R_{\varcolor{x},\varcolor{a}}}}.
        Hence we have shown the goal.

        In \Cref{sec:usage-analysis}, we introduced a widening operator
        \ensuremath{\varid{widen}\mathbin{::}\concolor{\mathsf{D_U}}\to \concolor{\mathsf{D_U}}} to the definition of \ensuremath{\varid{bind}}, that is, we defined
        \ensuremath{\varid{bind}\;\varid{rhs}\;\varid{body}\mathrel{=}\varid{body}\;(\varid{kleeneFix}\;(\varid{widen}\hsdot{\circ }{.\ }\varid{rhs}))}.
        For such an operator, we would additionally need to show that \ensuremath{\varid{widen}}
        preserves \ensuremath{{R_{\varcolor{x},\varcolor{a}}}}.
        Since the proposed cutoff operator in \Cref{sec:usage-analysis} only
        affects the \ensuremath{\concolor{\mathsf{Value_U}}} component, the only proof obligation is to show
        monotonicity:
        \ensuremath{\keyword{\forall}\!\! \hsforall \;\varid{d}_{1}\;\varid{d}_{2}\hsdot{\circ }{.\ }\varid{d}_{1}.\varid{v}\mathbin{⊑}\varid{d}_{2}.\varid{v}\implies(\varid{widen}\;\varid{d}_{1}).\varid{v}\mathbin{⊑}(\varid{widen}\;\varid{d}_{2}).\varid{v}}.
        This is a requirement that our our widening operator must satisfy anyway.
    \end{itemize}
\end{itemize}

This concludes the proof that $(\ensuremath{\varid{f}\;\varid{a}}, \ensuremath{\varid{f}\;(\varid{pre}\;\varid{x})}) ∈ R_{\ensuremath{\varid{x}},\ensuremath{\varid{a}}}$.
What remains to be shown is that this implies the overall goal
\ensuremath{\varid{f}\;\varid{a}\mathbin{⊑}\varid{apply}\;(\varid{fun}\;\varid{x}\;\varid{f})\;\varid{a}}:
\begin{hscode}\SaveRestoreHook
\column{B}{@{}>{\hspre}c<{\hspost}@{}}%
\column{BE}{@{}l@{}}%
\column{3}{@{}>{\hspre}l<{\hspost}@{}}%
\column{7}{@{}>{\hspre}l<{\hspost}@{}}%
\column{E}{@{}>{\hspre}l<{\hspost}@{}}%
\>[3]{}(\varid{f}\;\varid{a},\varid{f}\;(\varid{pre}\;\varid{x}))\in {R_{\varcolor{x},\varcolor{a}}}{}\<[E]%
\\
\>[B]{}\Longleftrightarrow{}\<[BE]%
\>[7]{}\mbox{\commentbegin  Definition of \ensuremath{{R_{\varcolor{x},\varcolor{a}}}}  \commentend}{}\<[E]%
\\
\>[B]{}\hsindent{3}{}\<[3]%
\>[3]{}\varid{f}\;\varid{a}\mathbin{⊑}\keyword{let}\;\langle \varcolor{\varphi}, \varid{v} \rangle\mathrel{=}\varid{f}\;(\varid{pre}\;\varid{x})\;\keyword{in}\;\langle \varcolor{\varphi}[\varid{x}\mapsto\concolor{\mathsf{U_0}}]\mathbin{+}(\varcolor{\varphi}\mathbin{!?}\varid{x})\mathbin{*}\varid{a}.\varcolor{\varphi}, \varid{v} \rangle{}\<[E]%
\\
\>[B]{}\Longleftrightarrow{}\<[BE]%
\>[7]{}\mbox{\commentbegin  refold \ensuremath{\varid{apply}}, \ensuremath{\varid{fun}}  \commentend}{}\<[E]%
\\
\>[B]{}\hsindent{3}{}\<[3]%
\>[3]{}\varid{f}\;\varid{a}\mathbin{⊑}\varid{apply}\;(\varid{fun}\;\varid{x}\;\varid{f})\;\varid{a}{}\<[E]%
\ColumnHook
\end{hscode}\resethooks
\end{proof}

As can be seen, its statement does not refer to the interpreter definition
\ensuremath{\mathcal{S}\denot{\wild}_{\wild}} \emph{at all}.
Instead, the complexity of its proof scales with the number of \emph{abstract
operations} supported in the semantic domain of the interpreter for a much more
\emph{modular} proof.
This modular proof appeals to parametricity~\citep{Reynolds:83} of \ensuremath{\varid{f}}'s
polymorphic type \ensuremath{\keyword{\forall}\!\! \hsforall \;\varid{d}\hsdot{\circ }{.\ }(\conid{Trace}\;\varid{d},\conid{Domain}\;\varid{d},\conid{HasBind}\;\varid{d})\Rightarrow \varid{d}\to \varid{d}}.
Of course, any function defined by the generic interpreter satisfies this
requirement.
Without the premise of \textsc{Beta-App}, the law cannot be proved
for usage analysis; we give a counterexample in the Appendix
(\Cref*{ex:syntactic-beta-app}).

\begin{toappendix}
The following example shows why we need the ``polymorphic'' premises in
\Cref{fig:abstraction-laws}.
It defines a monotone, but non-polymorphic \ensuremath{\varid{f}\mathbin{::}\concolor{\mathsf{D_U}}\to \concolor{\mathsf{D_U}}} for which
\ensuremath{\varid{f}\;\varid{a}\; \not⊑ \varid{apply}\;(\varid{fun}\;\varid{x}\;\varid{f})\;\varid{a}}.
So if we did not have the premises, we would not be able to prove usage analysis
correct.
\begin{example}
\label{ex:syntactic-beta-app}
Let \ensuremath{\varid{z}\not=\varid{x}\not=\varid{y}}.
The monotone function \ensuremath{\varid{f}} defined as follows
\begin{center}
\begin{hscode}\SaveRestoreHook
\column{B}{@{}>{\hspre}l<{\hspost}@{}}%
\column{3}{@{}>{\hspre}l<{\hspost}@{}}%
\column{E}{@{}>{\hspre}l<{\hspost}@{}}%
\>[3]{}\varid{f}\mathbin{::}\concolor{\mathsf{D_U}}\to \concolor{\mathsf{D_U}}{}\<[E]%
\\
\>[3]{}\varid{f}\;\langle \varcolor{\varphi}, \anonymous  \rangle\mathrel{=}\keyword{if}\;\varcolor{\varphi}\mathbin{!?}\varid{y}\mathbin{⊑}\concolor{\mathsf{U_0}}\;\keyword{then}\;\langle \varcolor{\varepsilon}, \conid{Rep}\;\concolor{\mathsf{U_\omega}} \rangle\;\keyword{else}\;\langle [\varid{z}\mapsto\concolor{\mathsf{U_1}}], \conid{Rep}\;\concolor{\mathsf{U_\omega}} \rangle{}\<[E]%
\ColumnHook
\end{hscode}\resethooks
\end{center}
violates \ensuremath{\varid{f}\;\varid{a}\mathbin{⊑}\varid{apply}\;(\varid{fun}\;\varid{x}\;\varid{f})\;\varid{a}}.
To see that, let \ensuremath{\varid{a}\triangleq\langle [\varid{y}\mapsto\concolor{\mathsf{U_1}}], \conid{Rep}\;\concolor{\mathsf{U_\omega}} \rangle\mathbin{::}\concolor{\mathsf{D_U}}} and consider
\[
  \ensuremath{\varid{f}\;\varid{a}\mathrel{=}\langle [\varid{z}\mapsto\concolor{\mathsf{U_1}}], \conid{Rep}\;\concolor{\mathsf{U_\omega}} \rangle\; \not⊑ \langle \varcolor{\varepsilon}, \conid{Rep}\;\concolor{\mathsf{U_\omega}} \rangle\mathrel{=}\varid{apply}\;(\varid{fun}\;\varid{x}\;\varid{f})\;\varid{a}}.
\]
\end{example}
\end{toappendix}

To prove \Cref{thm:usage-subst-sem}, we encode \ensuremath{\varid{f}}'s type in System $F$
as $f : \forall X.\ \mathsf{Dict}(X) \to X \to X$ (where $\mathsf{Dict}(\ensuremath{\varid{d}})$
encodes the type class dictionaries of \ensuremath{(\conid{Trace}\;\varid{d},\conid{Domain}\;\varid{d},\conid{HasBind}\;\varid{d})}) and
derive the following free theorem:
\[
  \forall A, B.\
  \forall R ⊆ A \times B.\
  \forall (\mathit{inst_1}, \mathit{inst_2}) ∈ \mathsf{Dict}(R).\
  \forall (d_1,d_2) ∈ R.\
  (f_A(\mathit{inst_1})(d_1), f_B(\mathit{inst_2})(d_2)) ∈ R
\]
The key to making use of parametricity is to find a useful instantiation of this
theorem, of relation $R$ in particular.
We successfully proved \textsc{Beta-App} with the following instantiation:
\[\begin{array}{c}
  A \triangleq B \triangleq \ensuremath{\concolor{\mathsf{D_U}}}, \qquad \mathit{inst_1} \triangleq \mathit{inst_2} \triangleq \mathit{inst}, \qquad d_1 \triangleq a, \qquad d_2 \triangleq \mathit{pre}(x) \\
  R_{x,a}(d_1,d_2) \triangleq \forall g.\ d_1 = g(a) \land d_2 = g(\mathit{pre}(x)) \implies g(a) ⊑ \mathit{apply}(\mathit{fun}(x,g),a)  \\
\end{array}\]
where $\mathit{pre}(x) \triangleq \ensuremath{\langle [\varid{x}\mapsto\concolor{\mathsf{U_1}}], \conid{Rep}\;\concolor{\mathsf{U_\omega}} \rangle}$ is the
argument that the implementation of \ensuremath{\varid{fun}\;\varid{x}\;\varid{f}} passes to \ensuremath{\varid{f}} and $\mathit{inst}$ is
the canonical instance dictionary at \ensuremath{\concolor{\mathsf{D_U}}}.
This yields the following inference rule:
\[
\inferrule[]
  { a ⊑ \mathit{apply}(\mathit{fun}(x,\mathit{id}),a)
  \\ (\mathit{inst},\mathit{inst}) ∈ \mathsf{Dict}(R_{x,a})}
  { f_\ensuremath{\concolor{\mathsf{D_U}}}(\mathit{inst})(a) ⊑ \mathit{apply}(\mathit{fun}(x,f_\ensuremath{\concolor{\mathsf{D_U}}}(\mathit{inst})),a) }
\]
where $(\mathit{inst},\mathit{inst}) ∈ \mathsf{Dict}(R_{x,a})$ entails showing
one lemma per type class method, such as
\[
  \forall f_1,f_2.\ (\forall d_1,d_2.\ R_{x,a}(d_1,d_2) \implies R_{x,a}(f_1(d_1),f_2(d_2))) \implies R_{x,a}(\mathit{fun}(f_1),\mathit{fun}(f_2)).
\]
Discharging each of these 7+1 subgoals concludes the proof of \Cref{thm:usage-subst-sem}.
Next, we will use \Cref{thm:usage-subst-sem} to instantiate
\Cref{thm:abstract-by-need} for usage analysis.

\subsection{A Simpler Proof That Usage Analysis Infers Absence}
\label{sec:usage-sound}

Equipped with the generic abstract interpretation \Cref{thm:abstract-by-need},
we will prove in this subsection that usage analysis from \Cref{sec:abstraction}
infers absence in the same sense as absence analysis from \Cref{sec:problem}.
The reason we do so is to evaluate the proof complexity of our approach against
the preservation-style proof framework in \Cref{sec:problem}.

Specifically, \Cref{thm:abstract-by-need} makes it very simple to relate
by-need semantics with usage analysis, taking the place of the
absence-analysis-specific preservation lemma:

\begin{corollaryrep}[\ensuremath{\mathcal{S}_{\mathbf{usage}}\denot{\wild}} abstracts \ensuremath{\mathcal{S}_{\mathbf{need}}\denot{\wild}}]
\label{thm:usage-abstracts-need}
Let \ensuremath{\varid{e}} be an expression and $α_{\mathcal{S}}$ the abstraction function from
\Cref{fig:abstract-name-need}.
Then $α_{\mathcal{S}}(\ensuremath{\mathcal{S}_{\mathbf{need}}\denot{\varid{e}}}) ⊑ \ensuremath{\mathcal{S}_{\mathbf{usage}}\denot{\varid{e}}}$.
\end{corollaryrep}
\begin{proof}
By \Cref{thm:abstract-by-need}, it suffices to show the abstraction laws
in \Cref{fig:abstraction-laws}.
\begin{itemize}
  \item \textsc{Mono}:
    Always immediate, since \ensuremath{\mathbin{⊔}} and \ensuremath{\mathbin{+}} are the only functions matching on \ensuremath{\conid{U}},
    and these are monotonic.
  \item \textsc{Stuck-App}, \textsc{Stuck-Sel}:
    Trivial, since \ensuremath{\varid{stuck}\mathrel{=}\bot}.
  \item \textsc{Step-App}, \textsc{Step-Sel}, \textsc{Step-Inc}, \textsc{Update}:
    Follows by unfolding \ensuremath{\varid{step}}, \ensuremath{\varid{apply}}, \ensuremath{\varid{select}} and associativity of \ensuremath{\mathbin{+}}.
  \item \textsc{Beta-App}:
    Follows from \Cref{thm:usage-subst-sem}.
  \item \textsc{Beta-Sel}:
    Follows by unfolding \ensuremath{\varid{select}} and \ensuremath{\varid{con}} and applying a lemma very similar to
    \Cref{thm:usage-subst-sem} multiple times.
  \item \textsc{ByName-Bind}:
    \ensuremath{\varid{kleeneFix}} approximates the least fixpoint \ensuremath{\varid{lfp}} since the iteratee \ensuremath{\varid{rhs}}
    is monotone.
    We have said elsewhere that we omit a widening operator for \ensuremath{\varid{rhs}} that
    guarantees that \ensuremath{\varid{kleeneFix}} terminates.
\end{itemize}
\end{proof}

The next step is to leave behind the definition of absence in terms of the LK
machine in favor of one using \ensuremath{\mathcal{S}_{\mathbf{need}}\denot{\wild}_{\wild}}.
That is a welcome simplification because it leaves us with a single semantic
artefact --- the denotational interpreter --- instead of an operational
semantics and a separate static analysis as in \Cref{sec:problem}.
Thanks to adequacy (\Cref{thm:need-adequate-strong}), this new notion is not a
redefinition but provably equivalent to \Cref{defn:absence}:
\begin{lemmarep}[Denotational absence]
  \label{thm:absence-denotational}
  Variable \ensuremath{\varid{x}} is used in \ensuremath{\varid{e}} if and only if there exists a by-need evaluation context
  \ensuremath{\varid{E}} and expression \ensuremath{\varid{e}'} such that the trace
  \ensuremath{\mathcal{S}_{\mathbf{need}}\denot{\varid{E}[\conid{Let}\;\varid{x}\;\varid{e}'\;\varid{e}]}_{\varcolor{\varepsilon}}(\varcolor{\varepsilon})} contains a \ensuremath{\conid{Look}\;\varid{x}} event.
  Otherwise, \ensuremath{\varid{x}} is absent in \ensuremath{\varid{e}}.
\end{lemmarep}
\begin{proof}
Since \ensuremath{\varid{x}} is used in \ensuremath{\varid{e}}, there exists a trace
\[
  (\Let{\px}{\pe'}{\pe},ρ,μ,κ) \smallstep^* ... \smallstep[\LookupT(\px)] ...
\]

We proceed as follows:
\begin{DispWithArrows}[fleqn,mathindent=0em]
                          & (\Let{\px}{\pe'}{\pe},ρ,μ,κ) \smallstep^* ... \smallstep[\LookupT(\px)] ...
                          \label{arrow:usg-context}
                          \Arrow{$\pE \triangleq \mathit{trans}(\hole,ρ,μ,κ)$} \\
  {}\Longleftrightarrow{} & \init(\pE[\Let{\px}{\pe'}{\pe}]) \smallstep^* ... \smallstep[\LookupT(\px)] ...
                          \Arrow{Apply $α_{\STraces}$ (\Cref{fig:eval-correctness})} \\
  {}\Longleftrightarrow{} & α_{\STraces}(\init(\pE[\Let{\px}{\pe'}{\pe}]) \smallstep^*, []) = \ensuremath{\mathbin{...}\conid{Step}\;(\conid{Look}\;\varid{x})\mathbin{...}}
                          \Arrow{\Cref{thm:need-adequate-strong}} \\
  {}\Longleftrightarrow{} & \ensuremath{\mathcal{S}_{\mathbf{need}}\denot{\varid{E}[\conid{Let}\;\varid{x}\;\varid{e}'\;\varid{e}]}_{\varcolor{\varepsilon}}(\varcolor{\varepsilon})} = \ensuremath{\mathbin{...}\conid{Step}\;(\conid{Look}\;\varid{x})\mathbin{...}}
\end{DispWithArrows}
Note that the trace we start with is not necessarily an maximal trace,
so step \labelcref{arrow:usg-context} finds a prefix that makes the trace maximal.
We do so by reconstructing the syntactic \emph{evaluation context} $\pE$
with $\mathit{trans}$ (\cf \Cref{thm:translation}) such that
\[
  \init(\pE[\Let{\px}{\pe'}{\pe}]) \smallstep^* (\Let{\px}{\pe'}{\pe},ρ,μ,κ)
\]
Then the trace above is contained in the maximal trace starting in
$\init(\pE[\Let{\px}{\pe'}{\pe}])$ and it contains at least one $\LookupT(\px)$
transition.

The next two steps apply adequacy of \ensuremath{\mathcal{S}_{\mathbf{need}}\denot{\wild}_{\wild}(\wild)} to the trace, making the shift
from LK trace to denotational interpreter.
\end{proof}

We define the by-need evaluation contexts for our language in the Appendix.
Thus insulated from the LK machine, we may restate and prove
\Cref{thm:absence-correct} for usage analysis.

\begin{theoremrep}[\ensuremath{\mathcal{S}_{\mathbf{usage}}\denot{\wild}_{\wild}} infers absence]
  \label{thm:usage-absence}
  Let \ensuremath{\varcolor{\rho}_e\triangleq[\many{\varid{y}\mapsto\langle [\varid{y}\mapsto\concolor{\mathsf{U_1}}], \conid{Rep}\;\concolor{\mathsf{U_\omega}} \rangle}]} be the initial
  environment with an entry for every free variable \ensuremath{\varid{y}} of an expression \ensuremath{\varid{e}}.
  If \ensuremath{\mathcal{S}_{\mathbf{usage}}\denot{\varid{e}}_{\varcolor{\rho}_e}\mathrel{=}\langle \varcolor{\varphi}, \varid{v} \rangle} and \ensuremath{\varcolor{\varphi}\mathbin{!?}\varid{x}\mathrel{=}\concolor{\mathsf{U_0}}},
  then \ensuremath{\varid{x}} is absent in \ensuremath{\varid{e}}.
\end{theoremrep}
\begin{proofsketch}
If \ensuremath{\varid{x}} is used in \ensuremath{\varid{e}}, there is a trace \ensuremath{\mathcal{S}_{\mathbf{need}}\denot{\varid{E}[\conid{Let}\;\varid{x}\;\varid{e}'\;\varid{e}]}_{\varcolor{\varepsilon}}(\varcolor{\varepsilon})} containing a \ensuremath{\conid{Look}\;\varid{x}} event.
The abstraction function $α_{\mathcal{S}}$ induced by \ensuremath{\concolor{\mathsf{D_U}}} aggregates lookups in the
trace into a \ensuremath{\varid{φ'}\mathbin{::}\conid{Uses}}, \eg
  $β_\Traces(\LookupT(i) \smallstep \LookupT(x) \smallstep \LookupT(i) \smallstep \langle ... \rangle)
    = \ensuremath{\langle [\mskip1.5mu \varid{i}\; ↦ \concolor{\mathsf{U_\omega}},\varid{x}\; ↦ \concolor{\mathsf{U_1}}\mskip1.5mu], \mathbin{...} \rangle}$.
Clearly, it is \ensuremath{\varid{φ'}\mathbin{!?}\varid{x}\mathbin{⊒}\concolor{\mathsf{U_1}}}, because there is at least one \ensuremath{\conid{Look}\;\varid{x}}.
\Cref{thm:usage-abstracts-need} and a context invariance
\Cref*{thm:usage-bound-vars-context} prove that the computed \ensuremath{\varcolor{\varphi}}
approximates \ensuremath{\varid{φ'}}, so \ensuremath{\varcolor{\varphi}\mathbin{!?}\varid{x}\mathbin{⊒}\varid{φ'}\mathbin{!?}\varid{x}\mathbin{⊒}\concolor{\mathsf{U_1}}\not=\concolor{\mathsf{U_0}}}.
\end{proofsketch}
\begin{proof}
We show the contraposition, that is,
if \ensuremath{\varid{x}} is used in \ensuremath{\varid{e}}, then \ensuremath{\varcolor{\varphi}\mathbin{!?}\varid{x}\not=\concolor{\mathsf{U_0}}}.

By \Cref{thm:absence-denotational}, there exists \ensuremath{\varid{E}}, \ensuremath{\varid{e}'} such that
\[
  \ensuremath{\mathcal{S}_{\mathbf{need}}\denot{\varid{E}[\conid{Let}\;\varid{x}\;\varid{e}'\;\varid{e}]}_{\varcolor{\varepsilon}}(\varcolor{\varepsilon})\mathrel{=}\mathbin{...}\;\conid{Step}\;(\conid{Look}\;\varid{x})\;\mathbin{...}} .
\]

This is the big picture of how we prove \ensuremath{\varcolor{\varphi}\mathbin{!?}\varid{x}\not=\concolor{\mathsf{U_0}}} from this fact:
\begin{DispWithArrows}[fleqn,mathindent=0em]
                      & \ensuremath{\mathcal{S}_{\mathbf{need}}\denot{\varid{E}[\conid{Let}\;\varid{x}\;\varid{e}'\;\varid{e}]}_{\varcolor{\varepsilon}}(\varcolor{\varepsilon})} = \ensuremath{\mathbin{...}\conid{Step}\;(\conid{Look}\;\varid{x})\mathbin{...}}
                      \label{arrow:usg-instr}
                      \Arrow{Usage instrumentation} \\
  {}\Longrightarrow{} & \ensuremath{(\alpha\;\{\mathcal{S}_{\mathbf{need}}\denot{\varid{E}[\conid{Let}\;\varid{x}\;\varid{e}'\;\varid{e}]}_{\varcolor{\varepsilon}}(\varcolor{\varepsilon})\}).\varcolor{\varphi}} ⊒ [\ensuremath{\varid{x}} ↦ \ensuremath{\concolor{\mathsf{U_1}}}]
                      \label{arrow:usg-abs}
                      \Arrow{\Cref{thm:usage-abstracts-need}} \\
  {}\Longrightarrow{} & \ensuremath{(\mathcal{S}_{\mathbf{usage}}\denot{\varid{E}[\conid{Let}\;\varid{x}\;\varid{e}'\;\varid{e}]}_{\varcolor{\varepsilon}}).\varcolor{\varphi}} ⊒ [\ensuremath{\varid{x}} ↦ \ensuremath{\concolor{\mathsf{U_1}}}]
                      \label{arrow:usg-anal-context}
                      \Arrow{\Cref{thm:usage-bound-vars-context}} \\
  {}\Longrightarrow{} & \ensuremath{\concolor{\mathsf{U_\omega}}\mathbin{*}(\mathcal{S}_{\mathbf{usage}}\denot{\varid{e}}_{\varcolor{\rho}_e}).\varcolor{\varphi}} = \ensuremath{\concolor{\mathsf{U_\omega}}\mathbin{*}\varcolor{\varphi}} ⊒ [\ensuremath{\varid{x}} ↦ \ensuremath{\concolor{\mathsf{U_1}}}]
                      \Arrow{\ensuremath{\concolor{\mathsf{U_\omega}}\mathbin{*}\concolor{\mathsf{U_0}}\mathrel{=}\concolor{\mathsf{U_0}}\mathbin{⊏}\concolor{\mathsf{U_1}}}} \\
  {}\Longrightarrow{} & \ensuremath{\varcolor{\varphi}\mathbin{!?}\varid{x}\not=\concolor{\mathsf{U_0}}}
\end{DispWithArrows}

Step \labelcref{arrow:usg-instr} instruments the trace by applying the usage
abstraction function \ensuremath{\alpha\rightleftarrows\anonymous \triangleq\varid{nameNeed}}.
This function will replace every \ensuremath{\conid{Step}} constructor
with the \ensuremath{\varid{step}} implementation of \ensuremath{\concolor{\mathsf{T_U}}};
The \ensuremath{\conid{Look}\;\varid{x}} event on the right-hand side implies that its image under \ensuremath{\alpha} is
at least $[\ensuremath{\varid{x}} ↦ \ensuremath{\concolor{\mathsf{U_1}}}]$.

Step \labelcref{arrow:usg-abs} applies the central abstract interpretation
\Cref{thm:usage-abstracts-need} that is the main topic of this section,
abstracting the dynamic trace property in terms of the static semantics.

Finally, step \labelcref{arrow:usg-anal-context} applies
\Cref{thm:usage-bound-vars-context}, which proves that absence information
doesn't change when an expression is put in an arbitrary evaluation context.
The final step is just algebra.
\end{proof}

\subsection{Comparison to Ad-hoc Preservation Proof}

Let us compare to the preservation-style proof framework in \Cref{sec:problem}.
\begin{itemize}
  \item
    Where there were multiple separate \emph{semantic artefacts} in
    \Cref{sec:problem}, such as a small-step semantics and an extension
    of the absence analysis to machine configurations $σ$ in order to
    state preservation (\Cref*{thm:preserve-absent}), our proof only has a
    single semantic artefact that needs to be defined and understood: the
    denotational interpreter, albeit with different instantiations.
  \item
    What is more important is that a simple proof for
    \Cref{thm:usage-abstracts-need} in half a page (we encourage the
    reader to take a look) replaces a tedious, error-prone and incomplete
    \emph{proof for the preservation lemma} of \Cref{sec:problem}
    (\Cref*{thm:preserve-absent}).
    Of course, in this section we lean on \Cref{thm:abstract-by-need} to prove what
    amounts to a preservation lemma; the difference is that our proof properly
    accounts for heap update and can be shared with other analyses that are
    sound \wrt by-name and by-need.
    Thus, we achieve our goal of disentangling semantic details from the proof.
  \item
    Furthermore, the proof for \Cref{thm:usage-abstracts-need} by parametricity
    in this section is \emph{modular}, in contrast to \Cref{thm:absence-subst}
    which is proven by cases over the interpreter definition.
    More work needs to be done to achieve a modular proof of
    the underlying \Cref{thm:abstract-by-need}, however.
    The (omitted) proof for abstract by-\textbf{name} interpretation in the
    Appendix (\Cref*{thm:abstract-by-name}) is already modular.
\end{itemize}

\begin{toappendix}
In the proof for \Cref{thm:usage-absence} we exploit that usage analysis is
somewhat invariant under wrapping of \emph{by-need evaluation contexts}, roughly
\ensuremath{\concolor{\mathsf{U_\omega}}\mathbin{*}\mathcal{S}_{\mathbf{usage}}\denot{\varid{e}}_{\varcolor{\rho}_e}\mathrel{=}\mathcal{S}_{\mathbf{usage}}\denot{\varid{E}[\varid{e}]}_{\varcolor{\varepsilon}}}. To prove that, we first
need to define what the by-need evaluation contexts of our language are.

\citet[Lemma 4.1]{MoranSands:99} describe a principled way to derive the
call-by-need evaluation contexts $\pE$ from machine contexts $(\hole,μ,κ)$ of
the Sestoft Mark I machine; a variant of \Cref{fig:lk-semantics} that uses
syntactic substitution of variables instead of delayed substitution and
addresses, so $μ ∈ \Var \pfun \Exp$ and no closures are needed.

We follow their approach, but inline applicative contexts,%
\footnote{The result is that of \citet[Figure 3]{Ariola:95} in A-normal form and
extended with data types.}
thus defining the by-need evaluation contexts with hole $\hole$ for our language as
\[\begin{array}{lcl}
  \pE ∈ \EContexts & ::= & \hole \mid \pE~\px \mid \Case{\pE}{\Sel} \mid \Let{\px}{\pe}{\pE} \mid \Let{\px}{\pE}{\pE[\px]} \\
\end{array}\]
The correspondence to Mark I machine contexts $(\hole,μ,κ)$ is encoded by the
following translation function $\mathit{trans}$ that translates from mark I
machine contexts  $(\hole,μ,κ)$ to evaluation contexts $\pE$.
\[\begin{array}{lcl}
  \mathit{trans} & : & \EContexts \times \Heaps \times \Continuations \to \EContexts \\
  \mathit{trans}(\pE,[\many{\px ↦ \pe}],κ) & = & \Letmany{\px}{\pe}{\mathit{trans}(\pE,[],κ)} \\
  \mathit{trans}(\pE,[],\ApplyF(\px) \pushF κ) & = & \mathit{trans}(\pE~\px,[],κ) \\
  \mathit{trans}(\pE,[],\SelF(\Sel) \pushF κ) & = & \mathit{trans}(\Case{\pE}{\Sel},[],κ) \\
  \mathit{trans}(\pE,[],\UpdateF(\px) \pushF κ) & = & \Let{\px}{\pE}{\mathit{trans}(\hole, [], κ)[\px]} \\
  \mathit{trans}(\pE,[],\StopF) & = & \pE \\
\end{array}\]
Certainly the most interesting case is that of $\UpdateF$ frames, encoding
by-need memoisation.
This translation function has the following property:
\begin{lemma}[Translation, without proof]
  \label{thm:translation}
  $\init(\mathit{trans}(\hole,μ,κ)[\pe]) \smallstep^* (\pe,μ,κ)$,
  and all transitions in this trace are search transitions ($\AppIT$, $\CaseIT$,
  $\LetIT$, $\LookupT$).
\end{lemma}
In other words: every machine configuration $σ$ corresponds to an evaluation
context $\pE$ and a focus expression $\pe$ such that there exists a trace
$\init(\pE[\pe]) \smallstep^* σ$ consisting purely of search transitions,
which is equivalent to all states in the trace except possibly the last being
evaluation states.

We encode evaluation contexts in Haskell as follows, overloading hole filling notation \ensuremath{\wild[\wild]}:
\begin{hscode}\SaveRestoreHook
\column{B}{@{}>{\hspre}l<{\hspost}@{}}%
\column{13}{@{}>{\hspre}c<{\hspost}@{}}%
\column{13E}{@{}l@{}}%
\column{16}{@{}>{\hspre}l<{\hspost}@{}}%
\column{35}{@{}>{\hspre}l<{\hspost}@{}}%
\column{E}{@{}>{\hspre}l<{\hspost}@{}}%
\>[B]{}\keyword{data}\;\conid{ECtxt}{}\<[13]%
\>[13]{}\mathrel{=}{}\<[13E]%
\>[16]{}\conid{Hole}\mid \conid{Apply}\;\conid{ECtxt}\;\conid{Name}\mid \conid{Select}\;\conid{ECtxt}\;\conid{Alts}{}\<[E]%
\\
\>[13]{}\mid {}\<[13E]%
\>[16]{}\conid{ExtendHeap}\;\conid{Name}\;\conid{Expr}\;\conid{ECtxt}\mid \conid{UpdateHeap}\;\conid{Name}\;\conid{ECtxt}\;\conid{Expr}{}\<[E]%
\\
\>[B]{}\wild[\wild]\mathbin{::}\conid{ECtxt}\to \conid{Expr}\to \conid{Expr}{}\<[E]%
\\
\>[B]{}\conid{Hole}[\varid{e}]{}\<[35]%
\>[35]{}\mathrel{=}\varid{e}{}\<[E]%
\\
\>[B]{}(\conid{Apply}\;\varid{E}\;\varid{x})[\varid{e}]{}\<[35]%
\>[35]{}\mathrel{=}\conid{App}\;\varid{E}[\varid{e}]\;\varid{x}{}\<[E]%
\\
\>[B]{}(\conid{Select}\;\varid{E}\;\varid{alts})[\varid{e}]{}\<[35]%
\>[35]{}\mathrel{=}\conid{Case}\;\varid{E}[\varid{e}]\;\varid{alts}{}\<[E]%
\\
\>[B]{}(\conid{ExtendHeap}\;\varid{x}\;\varid{e}_{1}\;\varid{E})[\varid{e}_{2}]{}\<[35]%
\>[35]{}\mathrel{=}\conid{Let}\;\varid{x}\;\varid{e}_{1}\;\varid{E}[\varid{e}_{2}]{}\<[E]%
\\
\>[B]{}(\conid{UpdateHeap}\;\varid{x}\;\varid{E}\;\varid{e}_{1})[\varid{e}_{2}]{}\<[35]%
\>[35]{}\mathrel{=}\conid{Let}\;\varid{x}\;\varid{E}[\varid{e}_{1}]\;\varid{e}_{2}{}\<[E]%
\ColumnHook
\end{hscode}\resethooks

\begin{lemma}[Used variables are free]
  \label{thm:used-free}
  If \ensuremath{\varid{x}} does not occur in \ensuremath{\varid{e}} and in \ensuremath{\varcolor{\rho}} (that is, \ensuremath{\keyword{\forall}\!\! \hsforall \;\varid{y}\hsdot{\circ }{.\ }(\varcolor{\rho}\mathop{!}\varid{y}).\varcolor{\varphi}\mathbin{!?}\varid{x}\mathrel{=}\concolor{\mathsf{U_0}}}), then \ensuremath{(\mathcal{S}_{\mathbf{usage}}\denot{\varid{e}}_{\varcolor{\rho}}).\varcolor{\varphi}\mathbin{!?}\varid{x}\mathrel{=}\concolor{\mathsf{U_0}}}.
\end{lemma}
\begin{proof}
  By induction on \ensuremath{\varid{e}}.
\end{proof}

For concise notation, we define the following abstract substitution operation:

\begin{definition}[Abstract substitution]
  \label{defn:abs-subst-usage}
  We call \ensuremath{\varcolor{\varphi}[\varid{x}\Mapsto\varid{φ'}]\triangleq\varcolor{\varphi}[\varid{x}\mapsto\concolor{\mathsf{U_0}}]\mathbin{+}(\varcolor{\varphi}\mathbin{!?}\varid{x})\mathbin{*}\varid{φ'}} the
  \emph{abstract substitution} operation on \ensuremath{\conid{Uses}}
  and overload this notation for \ensuremath{\concolor{\mathsf{T_U}}}, so that
  \ensuremath{\langle \varcolor{\varphi}, \varid{v} \rangle[\varid{x}\Mapsto\varid{φ'}]\triangleq\langle \varcolor{\varphi}[\varid{x}\Mapsto\varid{φ'}], \varid{v} \rangle}.
\end{definition}

From \Cref{thm:usage-subst-sem}, we can derive the following auxiliary lemma:
\begin{lemma}
  \label{thm:usage-subst-abs}
  If \ensuremath{\varid{x}} does not occur in \ensuremath{\varcolor{\rho}}, then
  \ensuremath{\mathcal{S}_{\mathbf{usage}}\denot{\varid{e}}_{\varcolor{\rho}[\varid{x}\mapsto\varid{d}]}\mathbin{⊑}(\mathcal{S}_{\mathbf{usage}}\denot{\varid{e}}_{\varcolor{\rho}[\varid{x}\mapsto\langle [\varid{x}\mapsto\concolor{\mathsf{U_1}}], \conid{Rep}\;\concolor{\mathsf{U_\omega}} \rangle]})[\varid{x}\Mapsto\varid{d}.\varcolor{\varphi}]}.
\end{lemma}
\begin{proof}
  Define \ensuremath{\varid{f}\;\widehat{\varid{d}}\triangleq\mathcal{S}_{\mathbf{usage}}\denot{\varid{e}}_{\varcolor{\rho}[\varid{x}\mapsto\widehat{\varid{d}}]}} and \ensuremath{\varid{a}\triangleq\varid{d}}.
  Note that \ensuremath{\varid{f}} could be defined polymorphically as
  \ensuremath{\varid{f}\;\varid{d}\mathrel{=}\mathcal{S}\denot{\varid{e}}_{\varcolor{\rho}[\varid{x}\mapsto\varid{d}]}}, for suitably polymorphic \ensuremath{\varcolor{\rho}}.
  Furthermore, \ensuremath{\varid{x}} could well be lambda-bound, since it does not occur in the
  range of \ensuremath{\varcolor{\rho}} (and that is really what we need).
  Hence we may apply \Cref{thm:usage-subst-sem} to get
  \begin{hscode}\SaveRestoreHook
\column{B}{@{}>{\hspre}l<{\hspost}@{}}%
\column{3}{@{}>{\hspre}c<{\hspost}@{}}%
\column{3E}{@{}l@{}}%
\column{5}{@{}>{\hspre}l<{\hspost}@{}}%
\column{7}{@{}>{\hspre}l<{\hspost}@{}}%
\column{E}{@{}>{\hspre}l<{\hspost}@{}}%
\>[5]{}\mathcal{S}_{\mathbf{usage}}\denot{\varid{e}}_{\varcolor{\rho}[\varid{x}\mapsto\varid{d}]}{}\<[E]%
\\
\>[3]{}\mathbin{⊑}{}\<[3E]%
\>[7]{}\mbox{\commentbegin  \Cref{thm:usage-subst-sem}  \commentend}{}\<[E]%
\\
\>[3]{}\hsindent{2}{}\<[5]%
\>[5]{}\varid{apply}\;(\varid{fun}\;\varid{x}\;(\lambda \widehat{\varid{d}}\to \mathcal{S}_{\mathbf{usage}}\denot{\varid{e}}_{\varcolor{\rho}[\varid{x}\mapsto\widehat{\varid{d}}]}))\;\varid{d}{}\<[E]%
\\
\>[3]{}\mathrel{=}{}\<[3E]%
\>[7]{}\mbox{\commentbegin  Inline \ensuremath{\varid{apply}}, \ensuremath{\varid{fun}}  \commentend}{}\<[E]%
\\
\>[3]{}\hsindent{2}{}\<[5]%
\>[5]{}\keyword{let}\;\langle \varcolor{\varphi}, \varid{v} \rangle\mathrel{=}\mathcal{S}_{\mathbf{usage}}\denot{\varid{e}}_{\varcolor{\rho}[\varid{x}\mapsto\langle [\varid{x}\mapsto\concolor{\mathsf{U_1}}], \conid{Rep}\;\concolor{\mathsf{U_\omega}} \rangle]}\;\keyword{in}\;\langle \varcolor{\varphi}[\varid{x}\mapsto\concolor{\mathsf{U_0}}]\mathbin{+}(\varcolor{\varphi}\mathbin{!?}\varid{x})\mathbin{*}\varid{d}.\varcolor{\varphi}, \varid{v} \rangle{}\<[E]%
\\
\>[3]{}\mathrel{=}{}\<[3E]%
\>[7]{}\mbox{\commentbegin  Refold \ensuremath{\wild[\wild\Mapsto\wild]}  \commentend}{}\<[E]%
\\
\>[3]{}\hsindent{2}{}\<[5]%
\>[5]{}(\mathcal{S}_{\mathbf{usage}}\denot{\varid{e}}_{\varcolor{\rho}[\varid{x}\mapsto\langle [\varid{x}\mapsto\concolor{\mathsf{U_1}}], \conid{Rep}\;\concolor{\mathsf{U_\omega}} \rangle]})[\varid{x}\Mapsto\varid{d}].\varcolor{\varphi}{}\<[E]%
\ColumnHook
\end{hscode}\resethooks
\end{proof}

\begin{lemma}[Context closure]
\label{thm:usage-bound-vars-context}
Let \ensuremath{\varid{e}} be an expression and \ensuremath{\varid{E}} be a by-need evaluation context in which
\ensuremath{\varid{x}} does not occur.
Then \ensuremath{(\mathcal{S}_{\mathbf{usage}}\denot{\varid{E}[\varid{e}]}_{\varcolor{\rho}_{E}}).\varcolor{\varphi}\mathbin{?!}\varid{x}\mathbin{⊑}\concolor{\mathsf{U_\omega}}\mathbin{*}((\mathcal{S}_{\mathbf{usage}}\denot{\varid{e}}_{\varcolor{\rho}_e}).\varcolor{\varphi}\mathbin{!?}\varid{x})},
where \ensuremath{\varcolor{\rho}_{E}} and \ensuremath{\varcolor{\rho}_e} are the initial environments that map free variables \ensuremath{\varid{z}}
to their proxy \ensuremath{\langle [\varid{z}\mapsto\concolor{\mathsf{U_1}}], \conid{Rep}\;\concolor{\mathsf{U_\omega}} \rangle}.
\end{lemma}
\begin{proof}
We will sometimes need that if \ensuremath{\varid{y}} does not occur free in \ensuremath{\varid{e}_{1}}, we have
By induction on the size of \ensuremath{\varid{E}} and cases on \ensuremath{\varid{E}}:
\begin{itemize}
  \item \textbf{Case }\ensuremath{\conid{Hole}}:
    \begin{hscode}\SaveRestoreHook
\column{B}{@{}>{\hspre}l<{\hspost}@{}}%
\column{5}{@{}>{\hspre}c<{\hspost}@{}}%
\column{5E}{@{}l@{}}%
\column{9}{@{}>{\hspre}l<{\hspost}@{}}%
\column{E}{@{}>{\hspre}l<{\hspost}@{}}%
\>[9]{}(\mathcal{S}_{\mathbf{usage}}\denot{\conid{Hole}[\varid{e}]}_{\varcolor{\rho}_{E}}).\varcolor{\varphi}\mathbin{!?}\varid{x}{}\<[E]%
\\
\>[5]{}\mathrel{=}{}\<[5E]%
\>[9]{}\mbox{\commentbegin  Definition of \ensuremath{\wild[\wild]}  \commentend}{}\<[E]%
\\
\>[9]{}(\mathcal{S}_{\mathbf{usage}}\denot{\varid{e}}_{\varcolor{\rho}_{E}}).\varcolor{\varphi}\mathbin{!?}\varid{x}{}\<[E]%
\\
\>[5]{}\mathbin{⊑}{}\<[5E]%
\>[9]{}\mbox{\commentbegin  \ensuremath{\varcolor{\rho}_e\mathrel{=}\varcolor{\rho}_{E}}  \commentend}{}\<[E]%
\\
\>[9]{}\concolor{\mathsf{U_\omega}}\mathbin{*}(\mathcal{S}_{\mathbf{usage}}\denot{\varid{e}}_{\varcolor{\rho}_{E}}).\varcolor{\varphi}\mathbin{!?}\varid{x}{}\<[E]%
\ColumnHook
\end{hscode}\resethooks
    By reflexivity.
  \item \textbf{Case }\ensuremath{\conid{Apply}\;\varid{E}\;\varid{y}}:
    Since \ensuremath{\varid{y}} occurs in \ensuremath{\varid{E}}, it must be different to \ensuremath{\varid{x}}.
    \begin{hscode}\SaveRestoreHook
\column{B}{@{}>{\hspre}l<{\hspost}@{}}%
\column{5}{@{}>{\hspre}c<{\hspost}@{}}%
\column{5E}{@{}l@{}}%
\column{9}{@{}>{\hspre}l<{\hspost}@{}}%
\column{E}{@{}>{\hspre}l<{\hspost}@{}}%
\>[9]{}(\mathcal{S}_{\mathbf{usage}}\denot{(\conid{Apply}\;\varid{E}\;\varid{y})[\varid{e}]}_{\varcolor{\rho}_{E}}).\varcolor{\varphi}\mathbin{!?}\varid{x}{}\<[E]%
\\
\>[5]{}\mathrel{=}{}\<[5E]%
\>[9]{}\mbox{\commentbegin  Definition of \ensuremath{\wild[\wild]}  \commentend}{}\<[E]%
\\
\>[9]{}(\mathcal{S}_{\mathbf{usage}}\denot{\conid{App}\;\varid{E}[\varid{e}]\;\varid{y}}_{\varcolor{\rho}_{E}}).\varcolor{\varphi}\mathbin{!?}\varid{x}{}\<[E]%
\\
\>[5]{}\mathrel{=}{}\<[5E]%
\>[9]{}\mbox{\commentbegin  Definition of \ensuremath{\mathcal{S}_{\mathbf{usage}}\denot{\wild}_{\wild}}  \commentend}{}\<[E]%
\\
\>[9]{}(\varid{apply}\;(\mathcal{S}_{\mathbf{usage}}\denot{\varid{E}[\varid{e}]}_{\varcolor{\rho}_{E}})\;(\varcolor{\rho}_{E}\mathbin{!?}\varid{y})).\varcolor{\varphi}\mathbin{!?}\varid{x}{}\<[E]%
\\
\>[5]{}\mathrel{=}{}\<[5E]%
\>[9]{}\mbox{\commentbegin  Definition of \ensuremath{\varid{apply}}  \commentend}{}\<[E]%
\\
\>[9]{}\keyword{let}\;\langle \varcolor{\varphi}, \varid{v} \rangle\mathrel{=}\mathcal{S}_{\mathbf{usage}}\denot{\varid{E}[\varid{e}]}_{\varcolor{\rho}_{E}}\;\keyword{in}{}\<[E]%
\\
\>[9]{}\keyword{case}\;\varid{peel}\;\varid{v}\;\keyword{of}\;(\varid{u},\varid{v}_{2})\to (\langle \varcolor{\varphi}\mathbin{+}\varid{u}\mathbin{*}((\varcolor{\rho}_{E}\mathbin{!?}\varid{y}).\varcolor{\varphi}), \varid{v}_{2} \rangle.\varcolor{\varphi}\mathbin{!?}\varid{x}){}\<[E]%
\\
\>[5]{}\mathrel{=}{}\<[5E]%
\>[9]{}\mbox{\commentbegin  Unfold \ensuremath{\langle \varcolor{\varphi}, \varid{v} \rangle.\varcolor{\varphi}\mathrel{=}\varcolor{\varphi}}, \ensuremath{\varid{x}} absent in \ensuremath{\varcolor{\rho}_{E}\mathbin{!?}\varid{y}}  \commentend}{}\<[E]%
\\
\>[9]{}\keyword{let}\;\langle \varcolor{\varphi}, \varid{v} \rangle\mathrel{=}\mathcal{S}_{\mathbf{usage}}\denot{\varid{E}[\varid{e}]}_{\varcolor{\rho}_{E}}\;\keyword{in}{}\<[E]%
\\
\>[9]{}\keyword{case}\;\varid{peel}\;\varid{v}\;\keyword{of}\;(\varid{u},\varid{v}_{2})\to \varcolor{\varphi}\mathbin{!?}\varid{x}{}\<[E]%
\\
\>[5]{}\mathrel{=}{}\<[5E]%
\>[9]{}\mbox{\commentbegin  Refold \ensuremath{\langle \varcolor{\varphi}, \varid{v} \rangle.\varcolor{\varphi}\mathrel{=}\varcolor{\varphi}}  \commentend}{}\<[E]%
\\
\>[9]{}(\mathcal{S}_{\mathbf{usage}}\denot{\varid{E}[\varid{e}]}_{\varcolor{\rho}_{E}}).\varcolor{\varphi}\mathbin{!?}\varid{x}{}\<[E]%
\\
\>[5]{}\mathbin{⊑}{}\<[5E]%
\>[9]{}\mbox{\commentbegin  Induction hypothesis  \commentend}{}\<[E]%
\\
\>[9]{}\concolor{\mathsf{U_\omega}}\mathbin{*}(\mathcal{S}_{\mathbf{usage}}\denot{\varid{e}}_{\varcolor{\rho}_e}).\varcolor{\varphi}\mathbin{!?}\varid{x}{}\<[E]%
\ColumnHook
\end{hscode}\resethooks
  \item \textbf{Case }\ensuremath{\conid{Select}\;\varid{E}\;\varid{alts}}:
    Since \ensuremath{\varid{x}} does not occur in \ensuremath{\varid{alts}}, it is absent in \ensuremath{\varid{alts}} as well
    by \Cref{thm:used-free}.
    (Recall that \ensuremath{\varid{select}} analyses \ensuremath{\varid{alts}} with \ensuremath{\langle \varcolor{\varepsilon}, \conid{Rep}\;\concolor{\mathsf{U_\omega}} \rangle} as
    field proxies.)
    \begin{hscode}\SaveRestoreHook
\column{B}{@{}>{\hspre}l<{\hspost}@{}}%
\column{5}{@{}>{\hspre}c<{\hspost}@{}}%
\column{5E}{@{}l@{}}%
\column{9}{@{}>{\hspre}l<{\hspost}@{}}%
\column{E}{@{}>{\hspre}l<{\hspost}@{}}%
\>[9]{}(\mathcal{S}_{\mathbf{usage}}\denot{(\conid{Select}\;\varid{E}\;\varid{alts})[\varid{e}]}_{\varcolor{\rho}_{E}}).\varcolor{\varphi}\mathbin{!?}\varid{x}{}\<[E]%
\\
\>[5]{}\mathrel{=}{}\<[5E]%
\>[9]{}\mbox{\commentbegin  Definition of \ensuremath{\wild[\wild]}  \commentend}{}\<[E]%
\\
\>[9]{}(\mathcal{S}_{\mathbf{usage}}\denot{\conid{Case}\;\varid{E}[\varid{e}]\;\varid{alts}}_{\varcolor{\rho}_{E}}).\varcolor{\varphi}\mathbin{!?}\varid{x}{}\<[E]%
\\
\>[5]{}\mathrel{=}{}\<[5E]%
\>[9]{}\mbox{\commentbegin  Definition of \ensuremath{\mathcal{S}_{\mathbf{usage}}\denot{\wild}_{\wild}}  \commentend}{}\<[E]%
\\
\>[9]{}(\varid{select}\;(\mathcal{S}_{\mathbf{usage}}\denot{\varid{E}[\varid{e}]}_{\varcolor{\rho}_{E}})\;(\varid{cont}\mathbin{\lhd}\varid{alts})).\varcolor{\varphi}\mathbin{!?}\varid{x}{}\<[E]%
\\
\>[5]{}\mathrel{=}{}\<[5E]%
\>[9]{}\mbox{\commentbegin  Definition of \ensuremath{\varid{select}}  \commentend}{}\<[E]%
\\
\>[9]{}(\mathcal{S}_{\mathbf{usage}}\denot{\varid{E}[\varid{e}]}_{\varcolor{\rho}_{E}}\sequ \varid{lub}\;(\mathbin{...}\varid{alts}\mathbin{...})).\varcolor{\varphi}\mathbin{!?}\varid{x}{}\<[E]%
\\
\>[5]{}\mathrel{=}{}\<[5E]%
\>[9]{}\mbox{\commentbegin  \ensuremath{\varid{x}} absent in \ensuremath{\varid{lub}\;(\mathbin{...}\varid{alts}\mathbin{...})}  \commentend}{}\<[E]%
\\
\>[9]{}(\mathcal{S}_{\mathbf{usage}}\denot{\varid{E}[\varid{e}]}_{\varcolor{\rho}_{E}}).\varcolor{\varphi}\mathbin{!?}\varid{x}{}\<[E]%
\\
\>[5]{}\mathbin{⊑}{}\<[5E]%
\>[9]{}\mbox{\commentbegin  Induction hypothesis  \commentend}{}\<[E]%
\\
\>[9]{}\concolor{\mathsf{U_\omega}}\mathbin{*}(\mathcal{S}_{\mathbf{usage}}\denot{\varid{e}}_{\varcolor{\rho}_e}).\varcolor{\varphi}\mathbin{!?}\varid{x}{}\<[E]%
\ColumnHook
\end{hscode}\resethooks
  \item \textbf{Case }\ensuremath{\conid{ExtendHeap}\;\varid{y}\;\varid{e}_{1}\;\varid{E}}:
    Since \ensuremath{\varid{x}} does not occur in \ensuremath{\varid{e}_{1}}, and the initial environment
    is absent in \ensuremath{\varid{x}} as well, we have \ensuremath{(\mathcal{S}_{\mathbf{usage}}\denot{\varid{e}_{1}}_{\varcolor{\rho}_{E}}).\varcolor{\varphi}\mathbin{!?}\varid{x}\mathrel{=}\concolor{\mathsf{U_0}}} by
    \Cref{thm:used-free}.
    \begin{hscode}\SaveRestoreHook
\column{B}{@{}>{\hspre}l<{\hspost}@{}}%
\column{5}{@{}>{\hspre}c<{\hspost}@{}}%
\column{5E}{@{}l@{}}%
\column{9}{@{}>{\hspre}l<{\hspost}@{}}%
\column{15}{@{}>{\hspre}l<{\hspost}@{}}%
\column{21}{@{}>{\hspre}l<{\hspost}@{}}%
\column{E}{@{}>{\hspre}l<{\hspost}@{}}%
\>[9]{}(\mathcal{S}_{\mathbf{usage}}\denot{(\conid{ExtendHeap}\;\varid{y}\;\varid{e}_{1}\;\varid{E})[\varid{e}]}_{\varcolor{\rho}_{E}}).\varcolor{\varphi}\mathbin{!?}\varid{x}{}\<[E]%
\\
\>[5]{}\mathrel{=}{}\<[5E]%
\>[9]{}\mbox{\commentbegin  Definition of \ensuremath{\wild[\wild]}  \commentend}{}\<[E]%
\\
\>[9]{}(\mathcal{S}_{\mathbf{usage}}\denot{\conid{Let}\;\varid{y}\;\varid{e}_{1}\;\varid{E}[\varid{e}]}_{\varcolor{\rho}_{E}}).\varcolor{\varphi}\mathbin{!?}\varid{x}{}\<[E]%
\\
\>[5]{}\mathrel{=}{}\<[5E]%
\>[9]{}\mbox{\commentbegin  Definition of \ensuremath{\mathcal{S}_{\mathbf{usage}}\denot{\wild}_{\wild}}  \commentend}{}\<[E]%
\\
\>[9]{}(\mathcal{S}_{\mathbf{usage}}\denot{\varid{E}[\varid{e}]}_{\varcolor{\rho}_{E}[\varid{y}\mapsto\varid{step}\;(\conid{Look}\;\varid{y})\;(\varid{kleeneFix}\;(\lambda \varid{d}\to \mathcal{S}_{\mathbf{usage}}\denot{\varid{e}_{1}}_{\varcolor{\rho}_{E}[\varid{y}\mapsto\varid{step}\;(\conid{Look}\;\varid{y})\;\varid{d}]}))]}).\varcolor{\varphi}\mathbin{!?}\varid{x}{}\<[E]%
\\
\>[5]{}\mathbin{⊑}{}\<[5E]%
\>[9]{}\mbox{\commentbegin  Abstract substitution; \Cref{thm:usage-subst-abs}  \commentend}{}\<[E]%
\\
\>[9]{}(\mathcal{S}_{\mathbf{usage}}\denot{\varid{E}[\varid{e}]}_{\varcolor{\rho}_{E}[\varid{y}\mapsto\langle [\varid{y}\mapsto\concolor{\mathsf{U_1}}], \conid{Rep}\;\concolor{\mathsf{U_\omega}} \rangle]})[\varid{y}\Mapsto\varid{step}\;{}\<[E]%
\\
\>[9]{}\hsindent{6}{}\<[15]%
\>[15]{}(\conid{Look}\;\varid{y})\;(\varid{kleeneFix}\;(\lambda \varid{d}\to \mathcal{S}_{\mathbf{usage}}\denot{\varid{e}_{1}}_{\varcolor{\rho}_{E}[\varid{y}\mapsto\varid{step}\;(\conid{Look}\;\varid{y})\;\varid{d}]}))].\varcolor{\varphi}\mathbin{!?}\varid{x}{}\<[E]%
\\
\>[5]{}\mathrel{=}{}\<[5E]%
\>[9]{}\mbox{\commentbegin  Unfold \ensuremath{\wild[\wild\Mapsto\wild]}, \ensuremath{\langle \varcolor{\varphi}, \varid{v} \rangle.\varcolor{\varphi}\mathrel{=}\varcolor{\varphi}}  \commentend}{}\<[E]%
\\
\>[9]{}\keyword{let}\;\langle \varcolor{\varphi}, {}\<[21]%
\>[21]{}\anonymous  \rangle\mathrel{=}\mathcal{S}_{\mathbf{usage}}\denot{\varid{E}[\varid{e}]}_{\varcolor{\rho}_{E}[\varid{y}\mapsto\langle [\varid{y}\mapsto\concolor{\mathsf{U_1}}], \conid{Rep}\;\concolor{\mathsf{U_\omega}} \rangle]}\;\keyword{in}{}\<[E]%
\\
\>[9]{}\keyword{let}\;\langle \varcolor{\varphi}_{2}, \anonymous  \rangle\mathrel{=}\varid{step}\;(\conid{Look}\;\varid{y})\;(\varid{kleeneFix}\;(\lambda \varid{d}\to \mathcal{S}_{\mathbf{usage}}\denot{\varid{e}_{1}}_{\varcolor{\rho}_{E}[\varid{y}\mapsto\varid{step}\;(\conid{Look}\;\varid{y})\;\varid{d}]}))\;\keyword{in}{}\<[E]%
\\
\>[9]{}(\varcolor{\varphi}[\varid{y}\mapsto\concolor{\mathsf{U_0}}]\mathbin{+}(\varcolor{\varphi}\mathbin{!?}\varid{y})\mathbin{*}\varcolor{\varphi}_{2})\mathbin{!?}\varid{x}{}\<[E]%
\\
\>[5]{}\mathrel{=}{}\<[5E]%
\>[9]{}\mbox{\commentbegin  \ensuremath{\varid{x}} absent in \ensuremath{\varcolor{\varphi}_{2}}, see above  \commentend}{}\<[E]%
\\
\>[9]{}\keyword{let}\;\langle \varcolor{\varphi}, {}\<[21]%
\>[21]{}\anonymous  \rangle\mathrel{=}\mathcal{S}_{\mathbf{usage}}\denot{\varid{E}[\varid{e}]}_{\varcolor{\rho}_{E}[\varid{y}\mapsto\langle [\varid{y}\mapsto\concolor{\mathsf{U_1}}], \conid{Rep}\;\concolor{\mathsf{U_\omega}} \rangle]}\;\keyword{in}{}\<[E]%
\\
\>[9]{}\varcolor{\varphi}\mathbin{!?}\varid{x}{}\<[E]%
\\
\>[5]{}\mathbin{⊑}{}\<[5E]%
\>[9]{}\mbox{\commentbegin  Induction hypothesis  \commentend}{}\<[E]%
\\
\>[9]{}\concolor{\mathsf{U_\omega}}\mathbin{*}(\mathcal{S}_{\mathbf{usage}}\denot{\varid{e}}_{\varcolor{\rho}_e}).\varcolor{\varphi}\mathbin{!?}\varid{x}{}\<[E]%
\ColumnHook
\end{hscode}\resethooks
  \item \textbf{Case }\ensuremath{\conid{UpdateHeap}\;\varid{y}\;\varid{E}\;\varid{e}_{1}}:
    Since \ensuremath{\varid{x}} does not occur in \ensuremath{\varid{e}_{1}}, and the initial environment
    is absent in \ensuremath{\varid{x}} as well, we have
    \ensuremath{(\mathcal{S}_{\mathbf{usage}}\denot{\varid{e}_{1}}_{\varcolor{\rho}_{E}[\varid{y}\mapsto\langle [\varid{y}\mapsto\concolor{\mathsf{U_1}}], \conid{Rep}\;\concolor{\mathsf{U_\omega}} \rangle]}).\varcolor{\varphi}\mathbin{!?}\varid{x}\mathrel{=}\concolor{\mathsf{U_0}}} by
    \Cref{thm:used-free}.
    \begin{hscode}\SaveRestoreHook
\column{B}{@{}>{\hspre}l<{\hspost}@{}}%
\column{5}{@{}>{\hspre}c<{\hspost}@{}}%
\column{5E}{@{}l@{}}%
\column{9}{@{}>{\hspre}l<{\hspost}@{}}%
\column{15}{@{}>{\hspre}l<{\hspost}@{}}%
\column{21}{@{}>{\hspre}l<{\hspost}@{}}%
\column{E}{@{}>{\hspre}l<{\hspost}@{}}%
\>[9]{}(\mathcal{S}_{\mathbf{usage}}\denot{(\conid{UpdateHeap}\;\varid{y}\;\varid{E}\;\varid{e}_{1})[\varid{e}]}_{\varcolor{\rho}_{E}}).\varcolor{\varphi}\mathbin{!?}\varid{x}{}\<[E]%
\\
\>[5]{}\mathrel{=}{}\<[5E]%
\>[9]{}\mbox{\commentbegin  Definition of \ensuremath{\wild[\wild]}  \commentend}{}\<[E]%
\\
\>[9]{}(\mathcal{S}_{\mathbf{usage}}\denot{\conid{Let}\;\varid{y}\;\varid{E}[\varid{e}]\;\varid{e}_{1}}_{\varcolor{\rho}_{E}}).\varcolor{\varphi}\mathbin{!?}\varid{x}{}\<[E]%
\\
\>[5]{}\mathrel{=}{}\<[5E]%
\>[9]{}\mbox{\commentbegin  Definition of \ensuremath{\mathcal{S}_{\mathbf{usage}}\denot{\wild}_{\wild}}  \commentend}{}\<[E]%
\\
\>[9]{}(\mathcal{S}_{\mathbf{usage}}\denot{\varid{e}_{1}}_{\varcolor{\rho}_{E}[\varid{y}\mapsto\varid{step}\;(\conid{Look}\;\varid{y})\;(\varid{kleeneFix}\;(\lambda \varid{d}\to \mathcal{S}_{\mathbf{usage}}\denot{\varid{E}[\varid{e}]}_{\varcolor{\rho}_{E}[\varid{y}\mapsto\varid{step}\;(\conid{Look}\;\varid{y})\;\varid{d}]}))]}).\varcolor{\varphi}\mathbin{!?}\varid{x}{}\<[E]%
\\
\>[5]{}\mathbin{⊑}{}\<[5E]%
\>[9]{}\mbox{\commentbegin  Abstract substitution; \Cref{thm:usage-subst-abs}  \commentend}{}\<[E]%
\\
\>[9]{}(\mathcal{S}_{\mathbf{usage}}\denot{\varid{e}_{1}}_{\varcolor{\rho}_{E}[\varid{y}\mapsto\langle [\varid{y}\mapsto\concolor{\mathsf{U_1}}], \conid{Rep}\;\concolor{\mathsf{U_\omega}} \rangle]})[\varid{y}\Mapsto\varid{step}\;{}\<[E]%
\\
\>[9]{}\hsindent{6}{}\<[15]%
\>[15]{}(\conid{Look}\;\varid{y})\;(\varid{kleeneFix}\;(\lambda \varid{d}\to \mathcal{S}_{\mathbf{usage}}\denot{\varid{E}[\varid{e}]}_{\varcolor{\rho}_{E}[\varid{y}\mapsto\varid{step}\;(\conid{Look}\;\varid{y})\;\varid{d}]}))].\varcolor{\varphi}\mathbin{!?}\varid{x}{}\<[E]%
\\
\>[5]{}\mathrel{=}{}\<[5E]%
\>[9]{}\mbox{\commentbegin  Unfold \ensuremath{\wild[\wild\Mapsto\wild]}, \ensuremath{\langle \varcolor{\varphi}, \varid{v} \rangle.\varcolor{\varphi}\mathrel{=}\varcolor{\varphi}}  \commentend}{}\<[E]%
\\
\>[9]{}\keyword{let}\;\langle \varcolor{\varphi}, {}\<[21]%
\>[21]{}\anonymous  \rangle\mathrel{=}\mathcal{S}_{\mathbf{usage}}\denot{\varid{e}_{1}}_{\varcolor{\rho}_{E}[\varid{y}\mapsto\langle [\varid{y}\mapsto\concolor{\mathsf{U_1}}], \conid{Rep}\;\concolor{\mathsf{U_\omega}} \rangle]}\;\keyword{in}{}\<[E]%
\\
\>[9]{}\keyword{let}\;\langle \varcolor{\varphi}_{2}, \anonymous  \rangle\mathrel{=}\varid{step}\;(\conid{Look}\;\varid{y})\;(\varid{kleeneFix}\;(\lambda \varid{d}\to \mathcal{S}_{\mathbf{usage}}\denot{\varid{E}[\varid{e}]}_{\varcolor{\rho}_{E}[\varid{y}\mapsto\varid{step}\;(\conid{Look}\;\varid{y})\;\varid{d}]}))\;\keyword{in}{}\<[E]%
\\
\>[9]{}(\varcolor{\varphi}[\varid{y}\mapsto\concolor{\mathsf{U_0}}]\mathbin{+}(\varcolor{\varphi}\mathbin{!?}\varid{y})\mathbin{*}\varcolor{\varphi}_{2})\mathbin{!?}\varid{x}{}\<[E]%
\\
\>[5]{}\mathrel{=}{}\<[5E]%
\>[9]{}\mbox{\commentbegin  \ensuremath{\varcolor{\varphi}\mathbin{!?}\varid{y}\mathbin{⊑}\concolor{\mathsf{U_\omega}}}, \ensuremath{\varid{x}} absent in \ensuremath{\varcolor{\varphi}}, see above  \commentend}{}\<[E]%
\\
\>[9]{}\keyword{let}\;\langle \varcolor{\varphi}_{2}, \anonymous  \rangle\mathrel{=}\varid{step}\;(\conid{Look}\;\varid{y})\;(\varid{kleeneFix}\;(\lambda \varid{d}\to \mathcal{S}_{\mathbf{usage}}\denot{\varid{E}[\varid{e}]}_{\varcolor{\rho}_{E}[\varid{y}\mapsto\varid{step}\;(\conid{Look}\;\varid{y})\;\varid{d}]}))\;\keyword{in}{}\<[E]%
\\
\>[9]{}\concolor{\mathsf{U_\omega}}\mathbin{*}\varcolor{\varphi}_{2}\mathbin{!?}\varid{x}{}\<[E]%
\\
\>[5]{}\mathrel{=}{}\<[5E]%
\>[9]{}\mbox{\commentbegin  Refold \ensuremath{\langle \varcolor{\varphi}, \varid{v} \rangle.\varcolor{\varphi}}  \commentend}{}\<[E]%
\\
\>[9]{}\concolor{\mathsf{U_\omega}}\mathbin{*}(\varid{step}\;(\conid{Look}\;\varid{y})\;(\varid{kleeneFix}\;(\lambda \varid{d}\to \mathcal{S}_{\mathbf{usage}}\denot{\varid{E}[\varid{e}]}_{\varcolor{\rho}_{E}[\varid{y}\mapsto\varid{step}\;(\conid{Look}\;\varid{y})\;\varid{d}]}))).\varcolor{\varphi}\mathbin{!?}\varid{x}{}\<[E]%
\\
\>[5]{}\mathrel{=}{}\<[5E]%
\>[9]{}\mbox{\commentbegin  \ensuremath{\varid{x}\not=\varid{y}}  \commentend}{}\<[E]%
\\
\>[9]{}\concolor{\mathsf{U_\omega}}\mathbin{*}(\varid{kleeneFix}\;(\lambda \varid{d}\to \mathcal{S}_{\mathbf{usage}}\denot{\varid{E}[\varid{e}]}_{\varcolor{\rho}_{E}[\varid{y}\mapsto\varid{d}]})).\varcolor{\varphi}\mathbin{!?}\varid{x}{}\<[E]%
\\
\>[5]{}\mathrel{=}{}\<[5E]%
\>[9]{}\mbox{\commentbegin  Argument below  \commentend}{}\<[E]%
\\
\>[9]{}\concolor{\mathsf{U_\omega}}\mathbin{*}(\mathcal{S}_{\mathbf{usage}}\denot{\varid{E}[\varid{e}]}_{\varcolor{\rho}_{E}[\varid{y}\mapsto\langle [\varid{y}\mapsto\concolor{\mathsf{U_1}}], \conid{Rep}\;\concolor{\mathsf{U_\omega}} \rangle]}).\varcolor{\varphi}\mathbin{!?}\varid{x}{}\<[E]%
\\
\>[5]{}\mathbin{⊑}{}\<[5E]%
\>[9]{}\mbox{\commentbegin  Induction hypothesis, \ensuremath{\concolor{\mathsf{U_\omega}}\mathbin{*}\concolor{\mathsf{U_\omega}}\mathrel{=}\concolor{\mathsf{U_\omega}}}  \commentend}{}\<[E]%
\\
\>[9]{}\concolor{\mathsf{U_\omega}}\mathbin{*}(\mathcal{S}_{\mathbf{usage}}\denot{\varid{e}}_{\varcolor{\rho}_e}).\varcolor{\varphi}\mathbin{!?}\varid{x}{}\<[E]%
\ColumnHook
\end{hscode}\resethooks
    The rationale for removing the \ensuremath{\varid{kleeneFix}} is that under the assumption that
    \ensuremath{\varid{x}} is absent in \ensuremath{\varid{d}} (such as is the case for \ensuremath{\varid{d}\triangleq\langle [\varid{y}\mapsto\concolor{\mathsf{U_1}}], \conid{Rep}\;\concolor{\mathsf{U_\omega}} \rangle}), then it is also absent in \ensuremath{\varid{E}[\varid{e}]\;\varcolor{\rho}_{E}[\varid{y}\mapsto\varid{d}]} per
    \Cref{thm:used-free}.
    Otherwise, we go to \ensuremath{\concolor{\mathsf{U_\omega}}} anyway.

    \ensuremath{\conid{UpdateHeap}} is why it is necessary to multiply with \ensuremath{\concolor{\mathsf{U_\omega}}} above;
    in the context $\Let{x}{\hole}{x~x}$, a variable $y$ put in the hole
    would really be evaluated twice under call-by-name (where
    $\Let{x}{\hole}{x~x}$ is \emph{not} an evaluation context).

    This unfortunately means that the used-once results do not generalise
    to arbitrary by-need evaluation contexts and it would be unsound
    to elide update frames for $y$ based on the inferred use of $y$ in
    $\Let{y}{...}{\pe}$; for $\pe \triangleq y$ we would infer that $y$
    is used at most once, but that is wrong in context $\Let{x}{\hole}{x~x}$.
\end{itemize}
\end{proof}
\end{toappendix}

\section{Related Work}
\label{sec:related-work}


\subsubsection*{Call-by-need, Semantics}
Arguably, \citet{Josephs:89} described the first denotational by-need semantics,
predating the work of \citet{Launchbury:93} and \citet{Sestoft:97}, but not
the more machine-centric (rather than transition system centric) work on the
G-machine~\citep{Johnsson:84}.
We improve on \citeauthor{Josephs:89}'s work in that our encoding is
simpler, rigorously defined (\Cref{sec:totality}) and proven adequate \wrt
\citeauthor{Sestoft:97}'s by-need semantics (\Cref{sec:adequacy}).
\citet{HackettHutton:19} define a denotational cost semantics for call-by-need,
but unfortunately we fail to see how their approach can be extended to
totally generate detailed by-need small-step traces, \cf \Cref*{sec:clair}.

\citet{Sestoft:97} related the derivations of
\citeauthor{Launchbury:93}'s big-step natural semantics for our language to
the subset of \emph{balanced} small-step LK traces.
Balanced traces are a proper subset of our maximal LK traces that --- by nature
of big-step semantics --- excludes stuck and diverging traces.

Our denotational interpreter bears strong resemblance to
a denotational semantics~\citep{ScottStrachey:71},
or to a definitional interpreter~\citep{Reynolds:72}
featuring a finally encoded domain~\citep{Carette:07}
using higher-order abstract syntax~\citep{Pfenning:88}.
The key distinction to these approaches is that we generate small-step traces,
totally and adequately, observable by abstract interpreters.
%
\citet{AgerDanvyMidtgaard:04} successively transform a partial denotational
interpreter into a variant of the LK machine, going the reverse route of
\Cref{sec:adequacy}.

\subsubsection*{Coinduction and Fuel}
\citet{LeroyGrall:09} show that a coinductive encoding of big-step semantics
is able to encode diverging traces by proving it equivalent to a small-step
semantics, much like we did for a denotational semantics.
The work of \citet{Atkey:13,tctt} had big influence on our use of the later
modality and Löb induction.

Our \ensuremath{\conid{Trace}} type class is appropriate for tracking ``pure'' transition events,
but it is not up to the task of modelling user input, for example.
A redesign of \ensuremath{\conid{Trace}} inspired (and instantiated) by guarded interaction
trees~\citep{interaction-trees,gitrees} would help with that.



\subsubsection*{Abstract Interpretation and Relational Analysis}
\citet{Cousot:21} recently condensed his seminal work rooted in \citet{Cousot:77}.
The book advocates a compositional, trace-generating semantics and then derives
compositional analyses by calculational design, and inspired us to attempt the same.
However, while \citet{Cousot:94,Cousot:02} work with denotational semantics
for a higher-order language, it was unclear to us how to derive a compositional,
\emph{trace-generating} semantics for a higher-order language.
The required changes to the domain definitions seemed daunting, to say the
least.
Our solution delegates this complexity to the underlying theory of guarded
recursive type theory~\citep{tctt}.


\subsubsection*{Abstractions of Reachable States}
CFA~\citep{Shivers:91} computes a useful control-flow graph abstraction for
higher-order programs, thus lifting classic intraprocedural analyses such as
constant propagation to the interprocedural setting.
\citet{MontaguJensen:21} derive CFA from small-step traces.
We think that a variant of our denotational interpreter would be a good fit for
their collecting semantics.
Specifically, the semantic inclusions of Lemma 2.10 that govern the transition
to a big-step style interpreter follow simply by adequacy of our interpreter,
\Cref{thm:need-adequate-strong}.

Abstracting Abstract Machines~\citep{aam} derives
a computable \emph{reachable states semantics}~\citep{Cousot:21} from any
small-step semantics, by bounding the size of the heap.
Many analyses such as control-flow analysis arise as abstractions of reachable
states.
\citet{adi} and others apply the AAM recipe to big-step interpreters in the style
of \citeauthor{Reynolds:72}.

Whenever AAM is involved, abstraction follows some monadic structure inherent to
dynamic semantics~\citep{Sergey:13,adi}.
In our work, this is apparent in the \ensuremath{\conid{Domain}\;(\conid{D}\;\varcolor{\tau})} instance depending on
\ensuremath{\conid{Monad}\;\varcolor{\tau}}.
Decomposing such structure into a layer of reusable monad transformers has been
the subject of \citet{Darais:15} and \citet{Keidel:19}.
The trace transformers of \Cref{sec:interp} enable reuse along a different dimension.

A big advantage of the big-step framework of \citet{Keidel:18} is that
soundness proofs are modular in the sense of \Cref{sec:mod-subst}.
In the future, we hope to modularise the proof for
\Cref{thm:abstract-by-need}.

\subsubsection*{Summaries of Functionals \vs Call Strings}
\citet{Lomet:77} used procedure summaries to capture aliasing effects,
crediting the approach to untraceable reports by \citet{Allen:74} and
\citet{Rosen:75}.
\citet{SharirPnueli:78} were aware of both \cite{Cousot:77} and \cite{Allen:74},
and generalised aliasing summaries into the ``functional approach'' to
interprocedural data flow analysis, distinguishing it from the ``call strings
approach'' (\ie $k$-CFA).

That is not to say that the approaches cannot be combined;
inter-modular analysis led \citet[Section 3.8.2]{Shivers:91} to implement
the $\mathit{xproc}$ summary mechanism.
He also acknowledged the need for accurate intra-modular summary mechanisms for
scalability reasons in Section 11.3.2.
We are however doubtful that the powerset-centric AAM approach could integrate
summary mechanisms; the whole recipe rests on the fact that the set of
expressions and thus evaluation contexts is finite.

\citet{Mangal:14} have shown that a summary-based analysis can be equivalent
to $\infty$-CFA for arbitrary complete lattices and outperform 2-CFA in both
precision and speed.
%
%


\subsubsection*{Cardinality Analysis} More interesting cardinality
analyses involve the inference of summaries called \emph{demand
transformers}~\citep{Sergey:14}, such as implemented in the Demand Analysis of
the Glasgow Haskell Compiler.
We intend to use our framework to describe improvements to Demand Analysis in
the future.
A soundness proof would require a slightly different Galois connection than
\Cref{fig:abstract-name-need}, because Demand Analysis is not sound \wrt by-name
evaluation; a testament to its precision.

\clearpage
\bibliography{references}

\end{document}